\documentclass[12pt,a4paper,oneside,oldfontcommands]{memoir}

\usepackage[utf8]{inputenc}
\usepackage[T1]{fontenc}
\usepackage{stmaryrd}

\usepackage{graphicx,epstopdf,epsfig,psfrag}
\usepackage[usenames,dvipsnames]{xcolor}
\usepackage[colorlinks=true,linkcolor=blue,citecolor=magenta,linktocpage=true]{hyperref}
\makeatletter
\renewcommand{\@dotsep}{1000}

\usepackage{amssymb}

\usepackage[mathscr]{euscript}
\usepackage{dsfont}
\usepackage{braket}
\usepackage{slashed}
\usepackage{fullpage}
\usepackage{currvita}
\usepackage{tocloft}
\usepackage{comment}
\usepackage{pict2e}
\usepackage{cancel}

\usepackage{mathrsfs}
\usepackage{amsmath,amssymb,amsfonts,amsthm}
\usepackage{mathtools}
\usepackage{multirow}

\usepackage{appendix}
\usepackage{etoolbox}
\usepackage{comment}
\usepackage{import}
\usepackage{caption}

\usepackage{pgf,tikz,pgfplots}
\usetikzlibrary{calc}
\usetikzlibrary{decorations.pathmorphing}
\usetikzlibrary{shapes.geometric}
\usetikzlibrary{arrows,decorations.markings}

\graphicspath{{figures/}}

\usepackage[draft]{todonotes} 
\usepackage[final]{showlabels}	

\tikzset{->-/.style={decoration={
			markings,
			mark=at position #1 with {\arrow{>}}},postaction={decorate}}}	
\tikzset{-<-/.style={decoration={
			markings,
			mark=at position #1 with {\arrow{<}}},postaction={decorate}}}

\definecolor{airforceblue}{rgb}{0.36, 0.54, 0.66}
\definecolor{antiquefuchsia}{rgb}{0.57, 0.36, 0.51}
\definecolor{blush}{rgb}{0.87, 0.36, 0.51}
\definecolor{bondiblue}{rgb}{0.0, 0.58, 0.71}
\definecolor{MyGreen}{rgb}{0.0,0.5,0}
\definecolor{MyDarkRed}{rgb}{0.7,0,0}
\definecolor{MyBlue}{rgb}{0.0,0.0,.5}
\def\be#1\ee{\begin{align}#1\end{align}}
\def\bsub#1\esub{\begin{subequations}#1\end{subequations}}
\def\bg#1\eg{\begin{gather}#1\end{gather}}
\def\ba{\begin{eqnarray}}
\def\ea{\end{eqnarray}}

\def\sgn{\mathrm{sgn}}

\def\grav{\text{grav}}

\def\q{\qquad}
\def\f{\frac}
\def\df{\dfrac}

\def\eps{\varepsilon}

\def\lb{\big\lbrace}
\def\rb{\big\rbrace}

\def\ip{\lrcorner\,}
\def\ipp{\ip\!\!\!\ip}

\def\rm#1{\mathrm{#1}}

\def\lp{\ell_\text{Pl}}

\def\mb{\bar{\mu}}
\def\Tr{\mathrm{Tr}}

\def\de{\mathrm{d}}

\DeclareRobustCommand{\loplus}{\mathbin{\mathpalette\dog@lsemi{+}}}
\DeclareRobustCommand{\lotimes}{\mathbin{\mathpalette\dog@lsemi{\times}}}
\DeclareRobustCommand{\roplus}{\mathbin{\mathpalette\dog@rsemi{+}}}
\DeclareRobustCommand{\rotimes}{\mathbin{\mathpalette\dog@rsemi{\times}}}

\newcommand{\dog@rsemi}[2]{\dog@semi{#1}{#2}{-90,90}}
\newcommand{\dog@lsemi}[2]{\dog@semi{#1}{#2}{270,90}}
\newcommand{\dog@semi}[3]{%
  \begingroup
  \sbox\z@{$\m@th#1#2$}%
  \setlength{\unitlength}{\dimexpr\ht\z@+\dp\z@\relax}%
  \makebox[\wd\z@]{\raisebox{-\dp\z@}{%
    \begin{picture}(1,1)
    \linethickness{\variable@rule{#1}}
    \roundcap
    \put(0.5,0.5){\makebox(0,0){\raisebox{\dp\z@}{$\m@th#1#2$}}}
    \put(0.5,0.5){\arc[#3]{0.5}}
    \end{picture}%
  }}%
  \endgroup
}
\newcommand{\variable@rule}[1]{%
  \fontdimen8  
  \ifx#1\displaystyle\textfont3\else
    \ifx#1\textstyle\textfont3\else
      \ifx#1\scriptstyle\scriptfont3\else
        \scriptscriptfont3\relax
  \fi\fi\fi
}
\makeatother

\newcommand{\R}{{\mathbb R}}
\newcommand{\Z}{{\mathbb Z}}	
\newcommand{\I}{{\mathbb I}}	

\newcommand{\cA}{{\mathcal A}}
\newcommand{\cB}{{\mathcal B}}
\newcommand{\cC}{{\mathcal C}}
\newcommand{\cD}{{\mathcal D}}
\newcommand{\cE}{{\mathcal E}}
\newcommand{\cF}{{\mathcal F}}
\newcommand{\cG}{{\mathcal G}}
\newcommand{\cH}{{\mathcal H}}

\newcommand{\cJ}{{\mathcal J}}
\newcommand{\cK}{{\mathcal K}}
\newcommand{\cL}{{\mathcal L}}
\newcommand{\cM}{{\mathcal M}}
\newcommand{\cN}{{\mathcal N}}
\newcommand{\cO}{{\mathcal O}}
\newcommand{\cP}{{\mathcal P}}
\newcommand{\cQ}{{\mathcal Q}}
\newcommand{\cR}{{\mathcal R}}
\newcommand{\cS}{{\mathcal S}}
\newcommand{\cT}{{\mathcal T}}
\newcommand{\cU}{{\mathcal U}}
\newcommand{\cV}{{\mathcal V}}

\newcommand{\SU}{\mathrm{SU}}

\newcommand{\SL}{\mathrm{SL}}

\renewcommand{\O}{\mathrm{O}}

\newcommand{\ISO}{\mathrm{ISO}}

\newcommand{\BMS}{\mathrm{BMS}}

\renewcommand{\c}{{\mathfrak{c}}}

\newcommand{\su}{{\mathfrak{su}}}

\renewcommand{\sl}{{\mathfrak{sl}}}
\newcommand{\so}{{\mathfrak{so}}}
\newcommand{\iso}{{\mathfrak{iso}}}

\newcommand{\bms}{{\mathfrak{bms}}}
\newcommand{\h}{{\mathfrak{h}}}
\usepackage{amsthm} 

\newtheorem{theorem}{Theorem}
\newtheorem*{Prop*}{Proposition}

\theoremstyle{definition}
\newtheorem*{Def*}{Definition}

\usepackage{colortbl}

\usepackage{indentfirst} 
\usepackage{layout}
\usepackage{changepage}
\usepackage{pdfpages}
\usepackage{setspace}
		
\setsecnumdepth{subsection}
\chapterstyle{ger}

\makeatletter
\renewcommand\part{%
	\if@openright
	\cleardoublepage
	\else
	\clearpage
	\fi
	\thispagestyle{empty}%
	\if@twocolumn
	\onecolumn
	\@tempswatrue
	\else
	\@tempswafalse
	\fi
	\null\vfil
	\secdef\@part\@spart}
\makeatother

\makepagestyle{Contents}
\makeevenhead{Contents}{\small{\it{Contents}}}{}{\thepage}
\makeoddhead{Contents}{\small{\it{Contents}}}{}{\thepage}

\makepagestyle{Introduction}
\makeevenhead{Introduction}{\small{\it{Introduction}}}{}{\thepage}
\makeoddhead{Introduction}{\small{\it{Introduction}}}{}{\thepage}

\makepagestyle{Conclusion}
\makeevenhead{Conclusion}{\small{\it{Conclusion}}}{}{\thepage}
\makeoddhead{Conclusion}{\small{\it{Conclusion}}}{}{\thepage}

\makepagestyle{Regular}
\makeevenhead{Regular}{{\it{\leftmark}}}{}{\thepage}
\makeoddhead{Regular}{{\it{\rightmark}}}{}{\thepage}

\makepagestyle{Bibliography}
\makeevenhead{Bibliography}{\small{\it{Bibliography}}}{}{\thepage}
\makeoddhead{Bibliography}{\small{\it{Bibliography}}}{}{\thepage}

\nouppercaseheads

\usepackage[pagestyles,raggedright]{titlesec}
\usepackage{titletoc}

\usepackage{etoolbox}

\makeatletter
\patchcmd{\ttlh@hang}{\parindent\z@}{\parindent\z@\leavevmode}{}{}
\patchcmd{\ttlh@hang}{\noindent}{}{}{}
\makeatother

\newcommand{\secmark}{}
\newcommand{\marktotoc}[1]{\renewcommand{\secmark}{#1}}

\titleformat{\section}{\normalfont\Large\bfseries}{\makebox[1.5em][l]{\llap{\secmark}\thesection}}{0.4em}{}
\titlecontents{section}[3.7em]{}
{\contentslabel[\llap{\secmark}\thecontentslabel]{2.3em}}
{\hspace*{-2.3em}}{\titlerule*{.}\contentspage}

\titleformat{\subsection}{\normalfont\large\bfseries}
{\makebox[2.3em][l]{\llap{\secmark}\thesubsection}}{0.4em}{}
\titlecontents{subsection}[6.7em]{}
{\small\contentslabel[\llap{\secmark}\thecontentslabel]{3.1em}}
{\hspace*{-2.3em } }{\titlerule*{.}\contentspage}

\usepackage[numbers,sort,compress]{natbib}

\begin{document}

\thispagestyle{empty}
\begin{tikzpicture}[overlay, remember picture] 
\node at (current page.center) 
    [
    anchor=center,
    xshift=0mm,
    yshift=8mm
    ] 
{
		\includegraphics[width=1.3\textwidth]{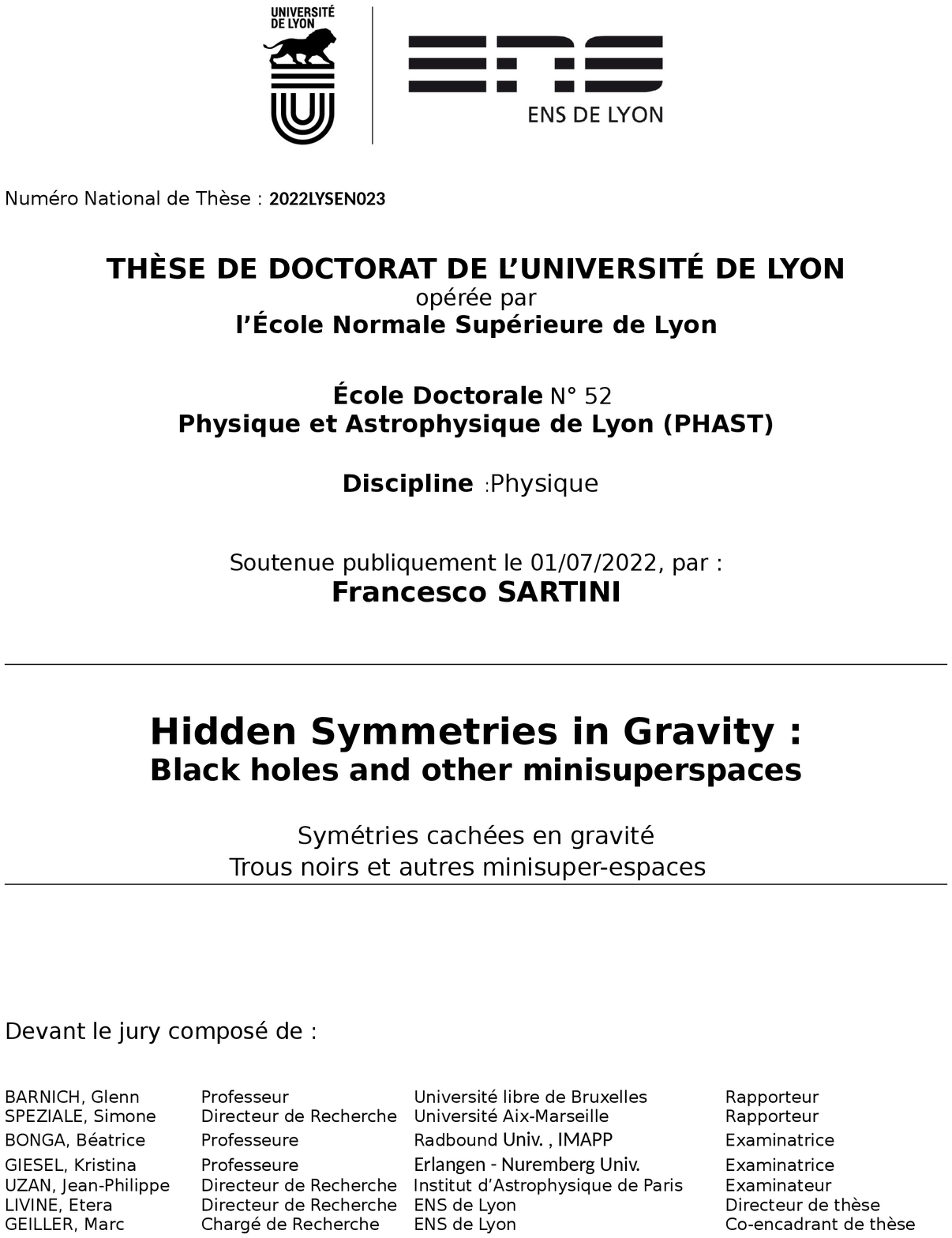}
};
\end{tikzpicture}


\frontmatter

\thispagestyle{empty}

\newpage
~
\thispagestyle{empty}


\vspace*{4cm}
\begin{flushright}
\textit{I think Nature's imagination is\\ so much greater than man's,\\ she's never gonna let us relax!}

R.P. Feynmann
\end{flushright}



%
\cleardoublepage

\ \newpage

\section*{Acknowledgements}
\addcontentsline{toc}{section}{Acknowledgements}

First and foremost, I would like to thank my supervisors. Marc has been the first one to welcome me when I was a master's student, and to accept guiding me through my years as a PhD student. He has been a passionate teacher during my first steps to discover the world of quantum gravity. He has always been present to help me enlighten my doubts and to answer my questions. I sincerely thank him for his constant enthusiasm, his pedagogical ability and for proposing many interesting project opportunities. It has been a pleasure to work with Etera, gifted with a quick mind and impressive intuition, he never ceased to push me to have a more global vision of my research. Whenever I felt stuck during my work, a simple discussion with him relaunched my research with many stimulating ideas. 

I am grateful to both of them for according me a lot of freedom in the choices of the directions that I wanted to pursue in my work, still guiding me in my path as a young researcher. Their complementary supervision has created a constructive and inspiring environment for my scientific growth.
 
I would like to express my thanks to the members of my committee. To Glenn Barnich and Simone Speziale, who accepted to take precious time to examine my manuscript. To Béatrice Bonga, Kristina Giesel and Jean--Philippe Uzan for being part of my jury.

On the scientific level, I would also like to thank Alexandre Arbey and Jérémy Auffinger. Our collaboration has been a valuable and enriching experience. Although the pandemic has prevented me from travelling and personally meeting other members of the scientific community, I would like to thank all those who accept to discuss my work, letting me give some talks, even if only virtually. 

Looking at the past I have to give credit to my middle and high school Maths and Physics professors, Patrizia and Ivan. They have been crucial for the development of my interest in mathematical problems and my curiosity about the scientific world. \newline

On a more private level, a special thought goes to my parents and my brother. They are undoubtedly the determining factor for my personal evolution, providing the peaceful and comfortable spacetime where I grew up. I will always be thankful to them, for pushing me to cultivate my passions and for being a solid anchor to hold on to, while facing the vortices of life.

I am extremely grateful to all the people that I met during these years in France, and those that supported me from Italy. A heartfelt thank goes to my friends in Florence, for the tenacious remote support, to Chiara, my flatmates and all those who made my life in Lyon more exciting. 

In parallel to research, during these years I also dedicate myself to another of my passion, defying gravity in a less scientific way. Climbing is not only an entertainment and a break from work, but also an incredible opportunity for personal growth. I will never cease to be amazed by the warmth and spirit of acceptance that the climbing community continuously shows.

Finally, I gratefully acknowledge all the staff of the \textit{Jardin Singulier}, for providing a wonderful and peaceful corner of the world, where this manuscript has been written.

\cleardoublepage
\vspace*{-2cm}

\section*{Summary of the Thesis (French version below)}
\addcontentsline{toc}{section}{Summary of the Thesis}

This thesis is dedicated to the study of symmetries in reduced models of gravity, with some frozen degrees of freedom. In particular, we focus on the minisuperspace reduction where the number of remaining degrees of freedom is finite. This work takes its place in the quest of understanding the role of physical and gauge symmetries in gravity. Consistent efforts have been made in this direction by the study of boundary symmetries. They can turn gauge into physical symmetries, this being related to the appearance of an infinite-dimensional algebra of conserved charges.

I choose here a simplified setup provided by minisuperspaces. They are treated as mechanical models, evolving in one spacetime direction. This evolution parameter represents the orthogonal coordinate to the homogeneous foliation of the spacetime. I investigate their classical symmetries and the algebra of the corresponding Noether charges.

After presenting the formalism allowing us to describe the reduced models in terms of an action principle, we discuss the condition for having an (extended) conformal symmetry. The $\SL(2,\R)$ group can be enhanced with the semidirect product with some abelian extension. In particular, the black hole model enlightens the subtle role of the spacelike boundary of the homogeneous slice. The latter interplays with the conformal symmetry, being associated with a conserved quantity from the mechanical point of view.

The absence of the infinite tower of charges, characteristic of the full theory, is due here to a symmetry-breaking mechanism. This is made explicit by looking at the infinite-dimensional extension of the symmetry group. In particular, this allows looking at the equation of motion of the mechanical system in terms of the infinite-dimensional group, who in turn has the effect of rescaling the coupling constants of the theory.

Finally, the presence of the finite symmetry group allows defining a quantum model in terms of the corresponding representation theory. At the level of the effective theory, accounting for the quantum effects, the request that the symmetry is protected provides a powerful tool to discriminate between different modifications. In the end, the conformal invariance of the black hole background opens the door to its holographic properties and might have important consequences on the study of the propagation of test fields on it and the corresponding perturbation theory.
 \newpage
 
 \section*{Résumé de la Thèse }

Cette thèse est dédié à l'étude des symétries des modèles à symétrie réduite en gravité, où on a gelé  certains degrés de liberté. En particulier, on va se concentrer sur la réduction à \textit{minisuper-espace}, où ils nous reste un nombre fini de degrés de liberté. Ce travail s'inscrit dans le cadre de la compréhension du rôle en gravité des symétries physiques et celles de jauge. Nombreux progrès ont été faits dans ce domaine, grâce à l'étude des symétries de bord. Elles peuvent changer des symétries de jauge en symétries physiques, et cela est souvent accompagné par la présence d'une algèbre de charges, de dimension infinie.

J'ai choisi ici un contexte plus simple: celui des minisuper-espaces. Ils sont considérés comme des modèles mécaniques, qui évoluent dans une dimension spatio-temporelle. Cette direction est représentée par la coordonné perpendiculaire aux hypersurfaces homogènes dans le quelle l'espace temps est folié. On va donc étudier les symétries de ces modèles est l’algèbre de leurs charges de Noether.

Après avoir présenté le formalisme nécessaire pour définir un principe de moindre action pour ces modèles en partant de la relativité générale, on va étudier les conditions nécessaires pour avoir une symétrie conforme (avec SL(2,R) comme groupe de symétrie). Celle ci peut aussi être étendue par produit semi-direct avec un autre groupe (abélien). En particulier, l'étude du modele des trous noirs révèle un rôle subtile joué par le bord de l'hypersurface spatiale. Il est relié à la symétrie du modele mécanique en étant associé a une des charges conservées.

L'absence d'une tour infinie de charge, typique de la théorie complète, est due ici à une brisure de symétrie. Cela est visible en regardant à l'extension infinito-dimensionnelle du groupe de symétrie. En particulier, cela va nous permettre de voir les équations du mouvement du système en termes d'orbite coadjointe du groupe de dimension infinie, qui, à son tour, peut changer les constants de couplage de la théorie.

La présence d'un groupe de dimension finie permet aussi de définir une théorie quantique, en utilisant la théorie des représentations. Au niveau effectif, la condition de conserver la symétrie est un outil puissant pour discriminer les différents régularisations. Finalement, l'invariance conforme du modèle des trous noirs ouvre la voie pour étudier ses propriétés holographiques et peut avoir des conséquences importantes pour la propagations des champs test et des perturbations.



\pagestyle{Contents}
\renewcommand\cftsubsectionfont{\small}
\newpage

\tableofcontents*

\newpage
~
\thispagestyle{empty}

\mainmatter
\pagestyle{Introduction}
\chapter*{Introduction}
\addcontentsline{toc}{part}{Introduction}

The cornerstone of modern high energy physics is the existence of four fundamental interactions: strong and weak nuclear interaction, electromagnetic and gravitational forces. A complex composition of these elementary interactions should in principle give a mathematical description for any natural phenomena. The first three interactions are used to describe the atomic and subatomic world. The strong force is responsible for binding together the bricks of atomic nuclei, the weak interaction mediates the radioactive decay, while the electromagnetic force describes the attraction between the nucleus and the electrons, the chemical bonds, as well as the electromagnetic waves, including visible light. Gravity, on the other hand, describes the way that objects with mass or energy attract each other. It is by far weaker than the other forces at the atomic scale but becomes the only relevant force for phenomena of astronomic size. This is because gravity has an infinite effective range, unlike strong and weak interaction, and celestial bodies carry are globally neutrally charged and they do not interact electromagnetically.

The study of gravitation is maybe the oldest branch of physics, very early civilisations had a predictive knowledge of the motion of the stars. The scientific method of Galileo lead to the Newtonian formulation of gravity, then superseded by Einstein's General Relativity and its Equivalence Principle, the theory used by a modern physicist to describe the gravitational interaction. 

This is formulated using the language of differential geometry, where the effect of mass is to curve a smooth spacetime, like a heavy ball on a trampoline. This is completely different from the physics of the Standard Model of particles, the theory used to describe the other three forces. It speaks the language of quantum field theories, where the spacetime is flat and everything is described by discrete quanta, elementary particles mediating the interactions. The two theories depict two completely different worlds and this makes it impossible to combine them, in particular, we cannot describe the quantum behaviour of gravity. 

Despite the global expectation that the Standard Model and General Relativity are both incomplete and effective theories, they successfully pass any experimental test. The energy limit where violations of the Standard Model could appear is constantly pushed further by particles accelerators: we have never seen the signature of a supersymmetric or exotic particle emerging in an experiment. On its side General Relativity has passed many tests among the years, starting from the right accounting for the perihelion precession of Mercury, the time dilation measured by GPS, the very recent detection of gravitational waves \cite{Abbott:2016blz} and image of the matter near the horizon of the black hole at the centre of the M87 galaxy \cite{EventHorizonTelescope:2019dse}.

On the other hand, because of the incapacity to represent the quantum nature of spacetime, we lack predictivity for phenomena taking place at the Planck scale, like for example near the black hole singularity or the very early cosmology, as well as the scattering between particles with very high energy and small impact parameter. According to quantum indetermination, to sharply localise the position of a particle, say with precision $\Delta x < L$ we need to hit it with an impulsion $p^2 > (\Delta p)^2 > (\hbar/L)^2$. High impulsion means high energy, because in the relativistic limit $E\approx c p > c \hbar / L$. We need to have an energy $E$ concentrated in a region of space of size $L \approx c \hbar/E$. If now we turn on the deformation of spacetime by the energy, acting as a mass $M \approx E/c^2$, we know that if this mass is concentrated in a radius smaller than $R =2 G M/c^2$, a black hole forms. The better is the precision desired (the smaller is $L$), the bigger is the energy that is required. At a certain threshold, the radius of the black hole is larger than $L$, this defines the Planck Length
\be
L \approx \f{G M}{c^2} \approx \f{ \hbar G}{L c^3} \q \Rightarrow\q \lp = \sqrt{\f{\hbar G}{c^3}} \approx 10^{-33} cm\,.
\ee
It is the shortest distance we can probe, beyond it we will inevitably produce a black hole, larger than the distance to measure. This simple argument, mixing a fundamental property of quantum mechanics and a prediction of General Relativity, is quite enlightening in many aspects. Besides the needing for a quantum theory of gravity, it also points out the fundamental role played by black holes in the quest for such a theory.

The impossibility to measure the gravitational field with an arbitrary precision lead to setting aside the view of standard quantum field theory, where the fields are defined over a smooth flat spacetime. On a more technical playground, this is translated into the non-renormalizability of General Relativity as a perturbative quantum field theory. Taking very seriously the lesson of Einstein, where matter and spacetime are interlaced together we must obtain a quantum theory \textit{of} spacetime itself, that is not defined on a fixed background.

General Relativity describes the geometric properties of spacetime, equipped with a metric tensor that measures distances between points. Einstein fields equations describe how the metric tensor and the matter content are related. Thanks to Lovelock's theorem \cite{Lovelock:1971yv}, we know that General Relativity is the only theory that depends on the metric up to second order. The beauty and at the same time the difficulty of General Relativity lies precisely in the fact that it comes directly as the only possible theory of gravity once we assume a small number of postulates. Einstein's theory is rooted in the \textit{general covariance} principle: the physical laws must be expressed by equations that hold good for any observer, represented by a system of coordinates. This translates mathematically into the choice of differential geometry as natural playground to describe spacetime, where the general covariance is expressed by the fact that diffeomorphisms represent the local gauge symmetry of general relativity. 

The invariance under a set of transformations of the mathematical description of Nature is the cornerstone of any physical theory, making symmetries the most powerful tool in theoretical physics. Classical mechanics is based on the fact that physical laws are the same for all inertial observers. The fact that the speed of lights is the same, regardless of the speed of the frame, lead Einstein to the formulation of special relativity, where the group of symmetries is known as the Poincar\'e group. The generalisation to a non-inertial observer gives General Relativity, as already discussed. But also the quantum world strongly relies on the concept of symmetries because, through the representation theory of the symmetry groups, we organise the states and the spectra of the quantum theories. The understanding of the symmetries of a theory is then a crucial step in the way to describe its quantum properties, and their role in gravity is far from being understood.

Symmetries are intimately related to conserved charges, through Noether's theorem \cite{1971TTSP....1..186N}. For example, invariance under time translations implies energy conservation. We usually distinguish the classical symmetries in global (physical) and gauge, assigning to the first ones some physical measurable quantity, while the gauge symmetries are an artefact of redundancy of different ways of describing the same physical system. This is made more precise by the difference between first and second Noether theorems \cite{Banados:2016zim}. Both are based on the action principle, for which the solutions of the equations of motion (EOM) are equivalent to stationary points of the action functional. The first theorem states that to every differentiable symmetry of the action there corresponds a conserved current. Whenever this current is non-zero, the symmetry is called global (or rigid), and it is usually associated to finite dimensional Lie groups. While the second theorem concerns infinite dimensional sets of transformations, parametrized by $n$ \textit{arbitrary} functions of spacetime, and gives $n$ relations among the equation of motions, that in turn implies that not all the variables are fixed by the equation of motion. At the Hamiltonian level this is translated by the presence of first-class constraint that allows reducing the degrees of freedom.

We shall notice that gauge invariance is one of the common aspects of all the descriptions of the fundamental interactions: the Standard Model of particle physics is invariant under a $\SU(3)\times \SU(2)\times \rm{U}(1)$ transformation at every point of spacetime, and General Relativity possesses diffeomorphisms covariance. Nevertheless, their precise role in gravity is far from being understood. Whereas in Einstein's theory, as it is usually presented in textbooks, the diffeomorphisms play the role of gauge symmetries, they relate two physically different observers. From a quantum mechanics point of view, where the relation between system and observer is crucial, they cannot be just a redundancy in the description. This is even more dramatic in the presence of boundaries, where the diffeomorphisms acquire a non-trivial charge and become physical. 

It has been known for a while \cite{Bondi:1962px,Sachs:1962wk,Sachs:1962zza} that the asymptotic symmetry group of Einstein gravity at null infinity is provided by an infinite dimensional extension of the Poincar\'e transformation, called BMS group. In recent years the study of asymptotic or, more in general, boundary symmetries have received lots of attention, due to the newly discovered connections with gravitational scattering, memory effects \cite{Strominger:2013jfa,He:2014laa,Compere:2016gwf,Choi:2017ylo,Ashtekar:2018lor,Rahman:2019bmk}, as well as the conjectural ability of these structures to account for the Hawking radiation or the prospect of constructing holographic theories \cite{Arcioni:2003xx,Arcioni:2003td,Barnich:2010eb,Fotopoulos:2019vac,Laddha:2020kvp,Donnay:2020guq,Puhm:2019zbl,Guevara:2021abz}. The asymptotic symmetries are placed at one of the vertices of a triangular equivalence, while memory effects and soft theorems lie on the other corners \cite{Strominger:2017zoo}. There are several different approaches to this boundary symmetries, and a big zoology exists. They depend on whether the boundary is spacelike or null, but also on the geometry and dimension of the internal spacetime, the boundary condition, the location of the boundary, but also on the description of gravity being considered.

The world described by Einstein's theory is four-dimensional, and this brings in lots of complexities, a good strategy is to turn to the study of toy models. This can be achieved in many ways, for example by considering lower dimensional systems. In three or two spacetime dimensions there are no local degrees of freedom, meaning that there is no propagation of gravitational waves. Their topological nature makes them natural candidates to explore the boundary structures and their holographic properties. Ever since the work of Brown and Henneaux \cite{Brown:1986nw}, which pointed out the existence of a double Virasoro algebra for 3D gravity with a particular class of fall-off conditions, numerous generalisation has been revealed  \cite{Ashtekar:1996cd,Barnich:2006av,Barnich:2012aw,Oblak:2016eij,Carlip:2017xne,Barnich:2014kra,Barnich:2015uva,Barnich:2016lyg,Barnich:2017jgw,Penna:2017vms,Barnich:2019vzx,Carlip:2019dbu,Ruzziconi:2020wrb,Freidel:2021fxf,Geiller:2021vpg}, enlightening a richer structure in spite of what the topological nature might suggest. 

But the toy models being toy models, there are aspects of the four-dimensional setup that are not captured and that prevent from generalising conclusions valid in three dimensions. We have already cited the absence of local degrees of freedom and radiative properties. Alongside it, in three dimensions we lack a good Newtonian limit, the force between particles being zero, making the dynamics of the four and three-dimensional cases completely different.

A different possibility to simplify the study is to stick with the four-dimensional case, but freeze some degrees of freedom, by imposing some symmetry condition on the spacetime, reducing the class of possible solutions, choosing some gauge condition, or a combination of both. The most extreme way of doing so is provided by \textit{minisuperspaces} \cite{ryan_1972, MACCALLUM1972385,Christodoulakis:1991ky,Christodoulakis:2018swq}, where the infinite-dimensional phase space of General Relativity is reduced to a finite-dimensional mechanical model. This is for instance what happens when we consider the cosmological approximation of a homogeneous scale factor, which is left as a single variable that evolves in time, like the position of one particle on a line. This thesis takes its place in the framework of these symmetry-reduced models. From now on we refer to the coordinate describing the evolution as \textit{time} although it is not necessarily time-like, the hypersurfaces at constant time are called \textit{slices}. While the term \textit{minisuperspace} is used to denote a reduction to a finite number of degrees of freedom, we define as \textit{midisuperpsaces} the symmetry reduced models that are still described by a field theory, with an infinite number of degrees of freedom.

An important role of minisuperspace is played in the building of solutions for quantum theories. For any known quantum theory there is indeed no straightforward way to obtain the states representing classical solutions, even the very simple case of spherically symmetric static black holes or homogeneous cosmologies. On the other hand, is possible to incorporate some features of the full theories into reduced models with a finite number of degrees of freedom. This is for example what has been done to build the so-called Loop Quantum Cosmology (see \cite{Ashtekar:2011ni} for a review), where we import into the cosmological minisuperspace the holonomy quantization of Loop Quantum Gravity (LQG hereafter) \cite{Thiemann:2007pyv,Ashtekar:2004eh,rovelli_2004,Perez:2012wv}. The last one is a canonical and non-perturbative approach to quantum gravity, written in terms of Wilson loops of the Ashtekar connection.

The apparent simplicity of minisuperspaces is the tree hiding in the forest. Besides their importance for model building, they possess surprising symmetry properties. It has been recently shown that homogeneous and isotropic cosmology, coupled with a massless scalar field exhibits a rigid $\SL(2,\R)$ conformal invariance, in the form of a M\"obius transformation of proper time \cite{BenAchour:2017qpb,BenAchour:2019ywl,BenAchour:2019ufa,BenAchour:2020njq,BenAchour:2020xif} (see also \cite{Dimakis:2015rba,Dimakis:2016mpg,Christodoulakis:2018swq,Pailas:2020xhh,Dussault:2020uvj}). This is because of the close relationship with the conformal particle model \cite{Pioline:2002qz}. The existence of the $\SL(2,\R)$ structure has been rapidly generalised for anisotropic Bianchi I models \cite{BenAchour:2019ywl} and in the presence of a cosmological constant \cite{BenAchour:2020xif}. The relationship with the latter was already pointed out by Gibbons \cite{Gibbons:2014zla}. What is quite surprising, and granted them the status of \textit{hidden} symmetries, is the fact that they exist on top of the residual spacetime diffeomorphism. 

Moreover, at this stage, it seemed that there were no systematic ways of understanding their origin, and their relationship with some boundary structure was quite obscure because at first glance it looks like no boundary is involved.

This thesis points towards the direction of improving our understanding of the minisuperspaces symmetry, providing a more general construction based on the dynamical structure of the theory. By the very definition of minisuperspace, once the symmetry reduction is performed, we are left with a mechanical phase space, with a finite number of degrees of freedom. The remaining metric free fields, which now depend only on one variable, can be treated as particles moving on a curved field space, endowed with a supermetric, with some potential. The evolution of the metric coefficients will then inherit the symmetry structure of the supermetric. 

The  cosmological $\SL(2,\R)$ group can be enhanced up to a symmetry group that is isomorphic to the semi-direct product $\left (\SL(2,\R)\times \R\right ) \ltimes \R^4$ \cite{Geiller:2022baq}. The corresponding $\sl(2,\R)$ piece of the algebra is connected to what was called CVH generators, namely the generator of phase space dilatations $C$, the volume of the three-dimensional slice $V$ and the Hamiltonian $H$. 

The simplicity of the cosmological model blurs the role of the boundary. Through this generalisation, we will see that the latter has actually to be considered, even if it awkwardly enters the game. Whenever the slice is non-compact, we require an IR cutoff in order to normalize the integration of the reduced action. This fiducial length plays a subtle role and enters as a shift in the Hamiltonian, meaning that symmetry and boundary speak each other.

The most interesting application of this construction is given by black holes \cite{Geiller:2020xze}. Due to their fundamental role in the quest for a quantum theory, we seek a deep understanding of their symmetries. In this case, it turns out to be useful to rewrite the symmetry generators in a way that we make contact with the (2+1) dimensional Poincar\'e group $\ISO(2,1)$. Interestingly, it is possible to show that it descends from its infinite dimensional extension the 3-dimensional BMS$_3$ group \cite{Geiller:2021jmg}. This ubiquitous group is not a symmetry in the homogeneous setup. Although this might seem disappointing at first, we know that symmetry breaking is as important as symmetry invariance, because it reveals physical degrees of freedom. In geometric theories, there is a similar situation in which the boundary partially breaks the symmetry, instead of enlarging it. For example the case of the Schwarzian theory leaving on the edge of JT gravity, where the Virasoro group is broken to its $\SL(2,\R)$ global part.   

Moreover, for minisuperspaces, we are in an intermediate situation, where the infinite-dimensional group act as a (time-dependent) rescaling of the coupling constant of the theory, somehow in analogy to what we might expect for a renormalization group flow. Furthermore, despite not being symmetries, we will be able to identify the integrable generators for the whole infinite dimensional set of transformations.

These symmetries have an important role also as guiding principle in quantization. Provided that the symmetry group is finite-dimensional, we can use its representation theory to uniquely choose a Hilbert space, well defined, in which the group structure would be realized by construction. This has been done for the black hole minisuperspace and has given striking predictions like, among others, the
continuity of the spectrum for the mass operator \cite{Sartini:2021ktb}. Finally, the requirement that the symmetry must be protected upon quantization gives a powerful criterion to discriminate between different regularization schemes, in particular for the so-called \textit{polymerisation}, inspired by Loop Quantum Gravity.

\newpage
\section*{Original contributions and plan of the thesis}
The purpose of this thesis is to give a detailed description of the minisuperspaces symmetries. This manuscript collects thoughts and results from four articles \cite{Geiller:2020xze,Geiller:2021jmg,Sartini:2021ktb,Geiller:2022baq} which originated from a three years research project. It also contains some comments inspired by a collaboration with researchers of Lyon 1 University on the consequences on Hawking radiation of quantum effective modification of the black hole metric \cite{Arbey:2021jif, Arbey:2021yke}.

\begin{itemize}

\item[\cite{Geiller:2022baq}]
M.~Geiller, E.~R. Livine and F.~Sartini, \emph{{Dynamical symmetries of
  homogeneous minisuperspace models}},
  [\href{https://arxiv.org/abs/2205.02615}{{\ttfamily 2205.02615}}].

\item[\cite{Geiller:2020xze}]
M.~Geiller, E.~R. Livine and F.~Sartini, \emph{{Symmetries of the black hole
  interior and singularity regularization}},
  \href{http://dx.doi.org/10.21468/SciPostPhys.10.1.022}{\emph{SciPost Phys.}
  {\bfseries 10} (2021) 022},
  [\href{https://arxiv.org/abs/2010.07059}{{\ttfamily 2010.07059}}].

\item[\cite{Geiller:2021jmg}]
M.~Geiller, E.~R. Livine and F.~Sartini, \emph{{BMS$_{3}$ mechanics and the
  black hole interior}},
  \href{http://dx.doi.org/10.1088/1361-6382/ac3e51}{\emph{Class. Quant. Grav.}
  {\bfseries 39} (2022) 025001},
  [\href{https://arxiv.org/abs/2107.03878}{{\ttfamily 2107.03878}}].

\item[\cite{Sartini:2021ktb}]
F.~Sartini, \emph{{Group quantization of the black hole minisuperspace}},
  \href{http://dx.doi.org/10.1103/PhysRevD.105.126003}{\emph{Phys. Rev. D}
  {\bfseries 105} (2022) 126003},
  [\href{https://arxiv.org/abs/2110.13756}{{\ttfamily 2110.13756}}].

\item[\cite{Arbey:2021jif}]
A.~Arbey, J.~Auffinger, M.~Geiller, E.~R. Livine and F.~Sartini, \emph{{Hawking
  radiation by spherically-symmetric static black holes for all spins:
  Teukolsky equations and potentials}},
  \href{http://dx.doi.org/10.1103/PhysRevD.103.104010}{\emph{Phys. Rev. D}
  {\bfseries 103} (2021) 104010},
  [\href{https://arxiv.org/abs/2101.02951}{{\ttfamily 2101.02951}}].

\item[\cite{Arbey:2021yke}]
A.~Arbey, J.~Auffinger, M.~Geiller, E.~R. Livine and F.~Sartini, \emph{{Hawking
  radiation by spherically-symmetric static black holes for all spins. II.
  Numerical emission rates, analytical limits, and new constraints}},
  \href{http://dx.doi.org/10.1103/PhysRevD.104.084016}{\emph{Phys. Rev. D}
  {\bfseries 104} (2021) 084016},
  [\href{https://arxiv.org/abs/2107.03293}{{\ttfamily 2107.03293}}].
\end{itemize}
 
 In the following there is a brief summary of each article, in order to show the original contributions of my work. This also provide an overview of the contents of the thesis.
\begin{itemize}
\item \textbf{Poincar\'e symmetry for black holes} \cite{Geiller:2020xze}. 
We have revealed an $\iso(2,1)$ Poincar\'e algebra of conserved charges associated with the dynamics of the interior of black holes. This symmetry corresponds to Möbius transformations of the proper time together with translations, and extends the known results for cosmologies. Remarkably, this is a physical symmetry changing the state of the system, by e.g. interplaying with the black hole mass.
\item \textbf{Infinite dimensional extension} \cite{Geiller:2021jmg}. 
The $\iso(2,1)$ structure presented in the previous work naturally forms a subgroup of the larger $\bms_3$ group. This is done by reinterpreting the action for the model as a geometric action for $\BMS_3$, where the configuration space variables are elements of the algebra $\bms_3$ and the equations of motion transform as coadjoint vectors. The Poincar\'e subgroup then arises as the stabilizer of the vacuum orbit. This symmetry breaking is analogous to what happens with the Schwarzian action in AdS$_2$ JT gravity, although in the present case there is no direct interpretation in terms of boundary symmetries. 
\item \textbf{Group quantization of black holes}  \cite{Sartini:2021ktb}.
We use the well-known $\ISO(2,1)$ irreducible representations to build a consistent quantum theory of black hole minisuperspace. This has, among others, the striking consequence of implying a continuous spectrum for the mass operator. Following loop quantum cosmology, we also obtain a regularization scheme compatible with the symmetry structure. It is possible to study the evolution of coherent states following the classical trajectories in the low curvature regime. We show that this produces an effective metric where the singularity is replaced by a Killing horizon merging two asymptotically flat regions. The quantum correction comes from a fundamental discreteness of spacetime, and the uncertainty on the energy of the system. Remarkably the effective evolution of semiclassical states is described by an effective (polymerised) Hamiltonian, related to the original one through a canonical transformation. 
\item \textbf{Generalisation to other minisuperspaces} \cite{Geiller:2022baq}. Through a more systematic approach to minisuperspace, we reveal the existence of rigid symmetries for a wide class of models. We review the homogeneous reduction of gravity, separating  the homogeneous dynamical fields from the background geometric structure of the various minisuperspace models. This enables us to identify an internal metric on field space (or superspace). The geometry of the superspace controls the dynamics of the homogeneous symmetry-reduced action and is directly connected with the existence of rigid symmetries.

\item \textbf{Hawking radiation on modified backgrounds}  \cite{Arbey:2021jif, Arbey:2021yke}. In the first paper we have derived the short-ranged potentials for the Teukolsky equations for massless spins (0,1/2,1,2) in general spherically-symmetric and static metrics, focusing on the case of polymerized BHs arising from models of quantum gravity. The second paper is devoted to the numerical application, computing the Hawking radiation spectra. In order to ensure the robustness of our numerical procedure, we show that it agrees with newly derived analytic formulas for the cross-sections in the high and low energy limits. We show how the short-ranged potentials and precise Hawking radiation rates can be used inside the code \texttt{BlackHawk} to predict future primordial BH evaporation signals. In particular, we derive the first Hawking radiation constraints on polymerized BHs from AMEGO. We prove that the mass window $10^{16}-10^{18}g$ for all dark matter into primordial BHs can be reopened with high values of the polymerization parameter, which encodes the typical scale and strength of quantum gravity corrections. 
\end{itemize}

\paragraph{Plan of the thesis} The present manuscript is organised into five chapters divided into two parts.

The first part focuses on the classical analysis of the minisuperspace symmetries. We start [\ref{chap1}] by reviewing the framework of General relativity, in particular its Hamiltonian ADM formulation. This will allow setting the stage to properly define the symmetry reduction and to go from a field theory to a mechanical model with a finite number of degrees of freedom. Therein we will develop the concept of supermetric, regulating the dynamics of the field space.

In the following chapter [\ref{chap2}], after a brief review of the results for the homogeneous cosmology, we will discuss how the rigid symmetries appear in a more general setting. We will see how the global $\left (\SL(2,\R) \times \R\right ) \ltimes \R^4$ is connected to the geometry of the field space. We will then dwell on the application to the black hole minisuperspace, as it provides interesting insights on how these symmetries are related to the boundary structure. We will briefly spend some words on the potential implication of the Poincar\'e symmetry on the perturbation of the black hole background and the Hawking radiation. This will also allow to enlighten the relationship between the boundary and the symmetries.

We will conclude this part by dedicating some attention to the infinite dimensional extensions of the symmetry group [\ref{chap3}]. We define the transformation on the fields through the adjoint representation, and use this to act on the action and derive the integrable Hamiltonian generators. This shows that only the Poincar\'e subgroup defines symmetries, while a more general transformation rescales the coupling constants. A BMS$_3$ fully invariant theory is obtained via a Stueckelberg mechanism. The study of coadjoint representations allows understanding the broken symmetry and the equation of motion in terms of coadjoint orbits

The second part is devoted to the consequences of the symmetry in the quantum regime. After a brief review of the general features of Loop Quantum Gravity [\ref{chap4}] and how to incorporate them into the cosmological model (LQC), we use the Poincar\'e symmetry as a guiding principle towards a group quantization of the black hole minisuperspace in chapter [\ref{chap5}], and its consequences on the polymerisation in the last chapter [\ref{chap6}]. 

Note that the first chapter and the fourth one are mainly for the completeness of the thesis.

\paragraph{Conventions and notations}Through all the manuscript we use the units where $c=1$, $\hbar=1$, but we leave the explicitly the Newton constant $G$. In this units it is equivalent to the Planck length, in the sense that $G= \lp^2$. If not specified differently, the metric tensor is always taken with Lorentzian $(-,+,+,\dots)$ signature, the spacelike directions have a positive measure, and the temporal direction has a negative one.

\pagestyle{Regular}
\newpage
~
\thispagestyle{empty}
\renewcommand{\afterpartskip}{}
\part*{Part I\\[.3cm]
Classical symmetries of minisuperspaces} 
\addcontentsline{toc}{part}{I\ Classical symmetries of minisuperspaces } \label{part:partI}
\newpage
~
\thispagestyle{empty}

\chapter{Gravity and its symmetries}
\label{chap1}

By General Relativity we usually mean the geometric theory of gravitation, as it was described by Einstein in the first place \cite{Einstein:1915by}. The key aspect of Einstein's theory is to give a unified description of space and time on a four-dimensional manifold, that is curved by the presence of energy and momentum, provided by both matter and radiation. The equations describing this relationship are second-order partial differential equations, the Einstein field equations.

Ever since Einstein's seminal work, numerous approaches to gravity have been developed. This thesis contains by no means a review of the different approaches to gravity, the motivation behind this section being simply to recall the fundamental elements necessary to address the symmetry reduction, necessary to deal with minisuperspaces. The interested readers shall refer to the already existing excellent reviews in the literature. See for instance \cite{Misner:1973prb,wald1984general} for excellent textbooks on General Relativity or the lecture notes \cite{Compere:2018aar} for a more modern approach. We will focus here on the Hamiltonian formulation of General Relativity, in particular on its properties of constrained (gauge) theory. 

In the second part, we will examine how and when it is possible to reduce the degrees of freedom of the theory, going from a four-dimensional field theory to the equivalent of a mechanical model. This is done by selecting a particular class of solution of general relativity, and by choosing an appropriate gauge fixing. The first consequence of doing so is to map the Einstein field equations to a set of ordinary differential equations (ODEs), much simpler to deal with. 

The Lagrangian of the minisuperspace model will take the form of the functional describing the motion of a particle on a $n$ dimensional curved spacetime, with a generically non-trivial potential. The number $n$ stands for the amount of free dynamical fields left in the metric. The \textit{dynamics} should be thought of as the evolution with respect to the coordinate labelling different slices of the original four-dimensional manifold. The geometrical property of the field space will play a crucial role in the determination of the minisuperspace symmetry in chapter \ref{chap2}.

The final section of this chapter is devoted to the effect of boundaries on symmetries. Providing a complete review of the boundary or asymptotic symmetry is beyond the scope of this thesis, and we advise the reader to look at the reference therein for sake of completeness. The focus will be on the appearance of the BMS group in flat three-dimensional spacetimes, but we postpone the review of its representation theory to the chapter \ref{chap3}, dedicated to the study of the infinite-dimensional extension of the minisuperspace symmetries.

\section{Hamiltonian formulation of General Relativity}
\label{sec1.1:Hamiltonian_GR}
The Einstein field equation can be obtained via the variational principle, starting from the Einstein-Hilbert action. Consider a $D+1$ dimensional manifold $\cM$, in the absence of boundaries, the action functional reads
\be
\label{Einst_Hilb_action}
\cS_\rm{EH} = \f{1}{\kappa}\int_\cM \de^{D+1} x \sqrt{|g|} \cR\,,
\ee
where $g$ is the determinant of the invertible metric tensor $g_{\alpha\beta}$, $\cR$ its Ricci scalar and $\kappa^2 = 16 \pi G$,  with the Newton constant $G$. The space time indices are denoted here with the greek letters $\alpha,\beta=0,1,2,3$,  keeping the mid-alphabet letters for the field space in the following sections. In the presence of a cosmological constant we shall add a term $-2\sqrt{|g|}\Lambda$ to the Lagrangian \eqref{Einst_Hilb_action}. In the absence of sources (matter or radiation) the field equation are given by
\be 
\label{EFE}
\cG_{\alpha\beta}:= \cR_{\alpha\beta}-\f{1}{2} \cR\, g_{\alpha\beta}=0\,,
\ee
where $\cR_{\alpha\beta}$ is the Ricci tensor and $\cG_{\alpha\beta}$ is called Einstein's tensor. The cosmological constant adds a term $+\Lambda g_{\alpha\beta}$ to the Einstein's tensor.

In the presence of boundaries, it is known that some terms must be added for the variational principle to be well defined. This is intimately interlaced with choice of boundary condition we want to impose \cite{Odak:2021axr}, the most common choice being to fix the metric on the boundary. In this case, defining $h$ the pull-back of the metric $g$ to the boundary $\partial \cM$, the term that we must add is provided by the Gibbons-Hawking-York term
\be
\cS_\rm{GHY} =\f{2}{\kappa^2}\int_{\partial\cM} \de^{D}x\, \sqrt{h} K\,,
\ee
where $K$ is the trace of the extrinsic curvature. For any hypersurface with normal vector $n^\mu$, the trace of the extrinsic curvature tensor is defined as
\be
K := \nabla_\alpha n^\alpha\,.
\label{Extrinsic curvature_trace}
\ee
In the aim of quantizing the theory, in the fifties, Dirac first developed an Hamiltonian approach to General relativity \cite{doi:10.1098/rspa.1958.0141}. This was then simplified by the change of variables introduced by Arnowit, Deser and Misner \cite{Arnowitt:1962hi} (ADM) in the early sixties. The latter strongly relies on the 3+1 decomposition on the manifold. This is because any Hamiltonian approach needs the concept of evolution of some physical quantities (here the metric on the $(D-1)$-dimensional slice) with respect to  \textit{a certain time}. From this point of view, known as \textit{gravitodynamics}, we aim to formulate the quantum theory in terms of wave-functions of the induced metric and consider transition amplitudes between metrics on the slices. We remark that other polarisations are obtained through a Fourier transform in the ($q,K$) (super)space, which amount to the modification of the boundary term. The so-called ADM variables consist in separating the metric coefficient in three groups
\bsub
\be 
q_{ab} &=g_{ab}\,,\\
\cN_a &= g_{0a}\,,\\
\cN&=1/\sqrt{-g^{00}}\,,\\
\de s^2 =& -(\cN^2-\cN_a \cN^a) \de t^2 + 2 \cN_a \de x^a \de t + q_{ab} \de x^a \de x^b\,,
\ee
\label{ADM_variables}
\esub
where $a,b =1,2\dots D$, while $\cN$ and $\cN^a$ take respectively the names of \textit{lapse} and \textit{shift}. This comes from their geometric interpretation in terms of spacetime foliations. For this, we think $t$ as the time labelling the ($D-1$)-dimensional hypersurfaces that foliate the spacetime. The unit normal to the slice is the vector 
\be
n^\alpha\partial_\alpha = \f{1}{\cN} \left (\partial_t -\cN^a\partial_a\right )\,.
\label{normal_slice}
\ee 
It is easy to convince ourself that the lapse function measures the proper time elapsed between the events $x^\alpha$, lying on the slice at time $t$, and $x^\alpha+\cN n^\alpha \de t$, at coordinate time $t + \de t$. On the other hand the shift vector $\cN^a := q^{ab} \cN_b$ accounts for the spatial separation (on the $t+ \de t$ slice) between the two events \cite{Gourgoulhon:2007ue}. In the following figure \eqref{fig:slices} we have schematically represented the foliation, the coordinate choice and the meaning of shift and lapse.

\vbox{
\begin{center}
\begin{minipage}{0.5\textwidth}
		\includegraphics[width=\textwidth]{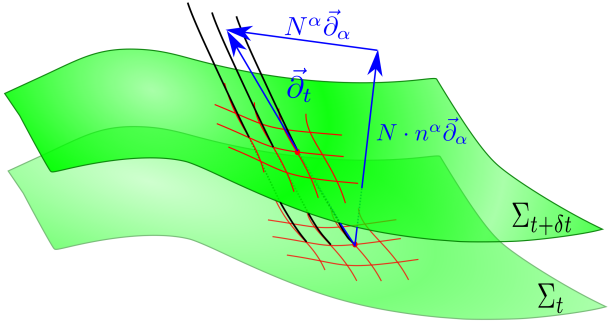}
\end{minipage}
~\hspace{0.5cm}
\begin{minipage}{0.4\textwidth}
\captionof{figure}{\small{The red lines on each slices represents the coordinate $x^i$ on the hypersurfaces $\Sigma$, each line at $x^i = const$ (in black) cuts across the foliation and defines the time vector $\partial_t$, the unit normal is $n^\alpha \partial_\alpha$ and the shift vector is then  the difference between the two.
} \label{fig:slices}}
\end{minipage}
\end{center}
}

The main advantage of these variables is that it makes evident that the lapse and the shift play the role of Lagrange multipliers enforcing a set of constraints. To better see this we rewrite the General Relativity action (with the GHY term) in the variables \eqref{ADM_variables}. The boundary that should be considered for the computation of the GHY term is provided by the slices at the initial and final time, assuming for the moment that there is no spatial boundary for the slices, we will come back to this point later on this chapter. The extrinsic curvature  (or second fundamental form) of the hypersurface is given explicitly by
\be
K_{ab} =  \f{1}{2\cN} (\dot q_{ab} -\pounds_{\vec{\cN}} q_{ab})= \f{1}{2\cN} (\dot q_{ab} - D_{(a} \cN_{b)}) \,,
\label{Extrinsic curvature}
\ee
where the dot represent the derivative with respect to time, $D_a$ the covariant derivative corresponding to the measure on the hypersurface $q_{ab}$ and $\pounds_{\vec\cN}$ is the Lie derivative along the vector field $\vec{\cN}$. It is easy to convince ourselves that this agrees with the formula above \eqref{Extrinsic curvature_trace}. This makes the action functional to read \cite{Gourgoulhon:2007ue}
\bsub
\be
\cS_\rm{ADM} =& \cS_\rm{EH} + \cS_\rm{GHY} = \int_{t_{\rm i}}^{t_\rm{f}} \de t \int_\Sigma \de^{D} x \left (\pi^{ab} \dot{q}_{ab} - \cN \cH - \cN^a \cH_a\right )\,,
\label{ADM_Hamilton}\\
\pi^{ab}&= \f{\partial \cL}{\partial \dot q_{ab}} =\f{\sqrt{q}}{\kappa^2}\left (K^{ab} - K q^{ab}\right )\,,\\
\cH&= -\f{\sqrt{q}}{\kappa^2} \left (K^2 - K^{ab}K_{ab}+R\right ) = -\f{\kappa^2}{\sqrt{q}} \left (\f{\pi^2}{2} - \pi^{ab}\pi_{ab}\right )-\f{\sqrt{q}R}{\kappa^2}  
\label{ADM_scalar_cstr}\,,\\
\cH_a&= -2q_{ab}D_c \left (\f{\pi^{bc}}{\sqrt{q}}\right ) = \f{2}{\kappa^2} q_{ab} D_c \left(K q^{bc} -K^{bc}\right )
\label{ADM_vect_cstr}\,,
\ee
\esub
where $q$ is the determinant of the metric on the slice and $R$ its Ricci scalar. We have decomposed the original manifold $\cM$ into slices,  $\cM= \bigcup\limits_{t=t_\rm{i}}^{t_\rm{f}} \Sigma_t$. We have also supposed the absence of time like boundaries or corner terms, meaning that the GHY term accounts only for the initial and final slices. 

We would like to remark that $\pi$ is not a tensor (it has a non zero weight density), and that is why we need to scale it with the determinant of $q$, to make sense of the covariant derivative . It is manifest that the action do not depend on the derivatives of lapse and shift, stationarity of $\cS$ with respect of their variations imposes respectively the constraints $\cH\approx 0$ and $\cH_a\approx 0$. The first one is called \textit{scalar constraint}, and the other one is the \textit{vector constraint}.

The action \eqref{ADM_Hamilton} is already written in terms of a canonical phase space, endowed with a poisson structure 
\be
\label{ADM_phase_space}
\{q_{ab}(t,\vec{x}),\pi^{cd}(t,\vec{y})\}= \delta_{(a}^c \delta_{b)}^d \delta(\vec{x}-\vec{y})\,.
\ee
The relationship between the momenta and the extrinsic curvature is easily understood by looking at the definitions \eqref{Extrinsic curvature_trace} and \eqref{Extrinsic curvature}. The tensor $K$ naively represents the \textit{velocity} of the metric on the hypersurface. It makes it natural that it is related to the conjugate momentum, up to a trace factor, which corresponds to the longitudinal modes.

The temporal evolution of any phase space functional $\cO[q,\pi]$ is given by the bracket with the Hamiltonian
\be
\dot \cO = \{\cO, \cN \cH+\cN^a\cH_a\}\,.
\ee

The theory to deal with constrained systems has been developed by Dirac \cite{doi:10.1098/rspa.1958.0141}, and strongly relies on the second Noether's theorem. We denote by \textit{first class}, the set of constraints whose Poisson brackets vanish on the constraint surface (a surface defined by the simultaneous vanishing of all the constraints). Equivalently this means that the first class constraints form a closed algebra. It is a fact that they are associated to gauge symmetries, for a constraint $\cC$, the corresponding smeared generator is given by
\bsub
\be
G[\alpha] =&\int \de^{D} x\, \alpha\, \cC\,,\\
\delta_\alpha \cO &= \{\cO,G[\alpha]\}\,,
\ee
\esub
for a generic function $\alpha$, $\delta_\alpha\cO$ gives the infinitesimal gauge transformation of the phase space functional $\cO$. In gravity the scalar and vector constraint are respectively associated with the time and spatial reparametrization of the manifold. Defining the corresponding generators as $H[\cN]=\int \cN \cH$, $D[\vec{\cN}]=\int \cN^a \cH_a$ they form the algebra \cite{Thiemann:2007pyv}
\bsub
\be
\{H[\cN_1,H[\cN_2]\} 		&=-D[q^{ab}(\cN_1 \partial_b\, \cN_2-\cN_2 \partial_b\, \cN_1)]\,,\\
\{D[\vec{\cN_1}],H[\cN]\} 		&= H[\pounds_{\vec {\cN_1}} \cN]	\,,\\
\{D[\vec{\cN_1}],D[\vec{\cN_2}]\} 	&= D[\pounds_{\vec{\cN_1}}\vec{\cN_2}]\,,
\ee
\label{constraint_algebra_GR}
\esub
where $\pounds_X$ represents the Lie derivative with respect to the vector field $X$. The existence of constraints is a typical feature of gauge system (it is actually their very definition), but what makes General relativity different, is the fact that the Hamiltonian itself vanishes on-shell. 

The gauge symmetry is crucial in the determination of the number of degrees of freedom. For each first-class constraint, its imposition $\cC\approx 0$ subtracts one degree of freedom, plus, because there is a gauge symmetry we shall subtract another degree of freedom, accounting for the redundancy in the physical description. In general relativity in $D+1$ dimensions we have a $D(D+1)$-dimensional phase space given by the symmetric coefficients of the metric $q_{ab}$ on the hypersurface and its momenta, and $D+1$ first-class constraints (the scalar constraint and $D$ vector constraints), which leaves us with $D(D+1)-2(D+1)=(D+1)(D-2)$ integration constants (per point) that physically determine the dynamics, we shall then divide by two to obtain the number of degrees of freedom. This gives the two local d.o.f representing gravitational waves in the four-dimensional case ($D=3$) and makes the topological nature of 3D gravity manifest ($D=2$).

Dirac had also developed an approach to quantising constrained systems. The basic idea consists in quantising the unconstrained phase space, giving a \textit{kinematical} phase space, and then one obtain the subset of physical states as the kernel of the constraints operators. Unfortunately, because of the high non-linearity of the constraints \eqref{ADM_scalar_cstr} and \eqref{ADM_vect_cstr} in terms of canonical metric variables, the application of Dirac's program to the ADM formulation is cumbersome and had not found a consistent mathematical formulation to date. The way out of this puzzle has been found thanks to the introduction of a new set of variables by A. Ashetkar, which recast GR in the form of a Yang-Mills theory, easier to handle. This launched the canonical quantization procedure that lead to the formulation of Loop Quantum Gravity in the late nineties. We will come back to this in the second part of the thesis, devoted to the study of the quantum world.

As already discussed in the introduction, the reduction to a subclass of solutions, with a high number of spacetime symmetries (spacetime Killing vectors) simplifies a lot the discussion. This symmetry reduction takes the name of \textit{midisuperspace} if we are still left a field theory, with local degrees of freedom, or \textit{minisuperspace} if the fields are defined only globally. This simplification is however not harmless, we cannot always partially solve the equation of motion (Einstein's field equation) and substitute it into the GR action to obtain the functional describing the dynamics of the remaining fields. We will discuss this issue in the next section. Moreover, in the case of minisuperspaces, we will see that they introduce a new set of symmetry, independent from the residual diffeomorphism, somehow analogously to the effect of the introduction of boundaries, we will dedicate the central part of the thesis to this subject.

\section{Symmetry reduction and minisuperspaces}
\label{sec1.2:symm_reduc}
Minisuperspaces provide the most extreme simplification of general relativity and yet contain some dynamical information and some non-trivial symmetry structure. The infinite-dimensional phase space is reduced to a finite-dimensional mechanical model, simplifying very much the canonical analysis of the classical theory, and the corresponding quantization. Classically these minisuperspace models are viable in the so-called adiabatic approximation, where variations in one spacetime direction differ by orders of magnitude with respect to the other ones. This is believed to be valid in the early stages of cosmology or near the black hole singularity.

A minisuperspace is thus defined by a spacetime metrics where the lapse only depends on time and the temporal and spatial dependence of the three metric and the shift split like
\bsub\be
q_{ab}&= \gamma_{ij}(t) e^i_a(\vec{x}) e^j_b(\vec{x})\,,\\
\cN_a&=N_i(t) e^i_a(\vec{x})\,,\\
\cN &= N(t)\,,\\
\de s^2 =& -(N(t)^2 -N_i N^i) \de t^2 + 2 N_i e^i_a \de x^a \de t + \gamma_{ij} e^i_a e^j_b \de x^a \de x^b\,.
\ee
\label{minisup_metric}
\esub
We use the mid-alphabet latin indices ($i,j,k,\dots$) for the internal tensors: the shift vector $N^i$ and the metric $\gamma_{ij}$. We use the latter to rise and lower the internal indices, and we refer to its coefficients as \textit{scale factors}. We also call the frame $e^i_a$ and its inverse $e^a_i$ as \textit{tetrad}. They are fixed once and for all and they select a particular class of solutions with a set of spacetime symmetries encoded by the isometries of the frame field. The dynamical information is left into $N$, $N_i$ and $\gamma_{ij}$. We have in principle 10 time-dependent fields, nonetheless, it could occur that for a given minisuperspace, not all these fields are independent, but we need to consider a section of the ten-dimensional field space. This is to assure compatibility between the symmetry reduction, the gauge choice and the action principle.

Indeed, we need to verify what happens if we calculate the reduced action for the line element \eqref{minisup_metric}, does the corresponding action principle give the correct Einstein equations? 

To answer this question we start by rewriting the constraints in terms of the frame and the dynamical quantities. We start with the extrinsic curvature tensor
\bsub
\be
K_{ab}&=\f{1}{2N}\left (\dot{q}_{ab} - D_{(a} \cN_{b)}\right )= \f{1}{2N} \left (e^i_a e^j_b \dot\gamma_{ij} - N_i D_{(a} e^i_{b)} \right )\,,\\
K&= \f{1}{2N} \left (\gamma^{ij} \dot\gamma_{ij} - N^i D_a e^a_i\right )\,.
\ee\label{extrins_curv_tetrad}
\esub
We need here to be careful about what is meant to be a boundary term. All the dynamical information is left in the time dependence, meaning that, in calculating the reduced action, we explicitly carry out the spatial integration and blur the spatial dependence of the metric into some global integration factor. The reduced action is then a mechanical action of the kind $\cS[\gamma_{ij}, N , N^i]= \int \de t\, \cL[\gamma_{ij}, N , N^i]$. This means that the boundary term must be total derivatives with respect to time only.

On the other hand, we have that the GHY term in the ADM formalism is given by
\be
\int \de t \int_\sigma \de^D x \partial_\alpha (\sqrt{|g|} K n^\alpha)\,,
\ee
where $n^\alpha$ is given at \eqref{normal_slice}. For this to be a total derivative in $t$, we need to fix the shift such that:
\be 
0\doteq \sum_{a=1}^D \partial_a (\sqrt{|g|} K n^a)= \sqrt{q} D_a ( K \cN^a) = \sqrt{q} N^i D_a ( K e^a_i) \,.
\ee
We can ensure that the latter is satisfied by demanding the stronger condition
\be
0\doteq \pounds_{\vec{\cN}} q_{ab} = N_i D_{(a} e^i_{b)}\,,
\label{shift_condition}
\ee
in particular the latter has the consequence of making disappear the terms linear in the velocities $\dot \gamma_{ij}$ form the action, by the elimination of the shift term in the extrinsic curvature \eqref{extrins_curv_tetrad}. This do not necessary imply that the shift is zero, only that the combination \eqref{shift_condition} vanishes. 

This, on the other hand, has the side effect of cancelling out all the shift terms from the reduced action, meaning that $\delta \cL/\delta N^i =0$. We recall that the vector constrain is obtained in the ADM form precisely by a combination of the integration by part of the GHY term, and the contribution from the Lie derivative of the three metrics along $\vec{\cN}$ in the extrinsic curvature (see \cite{Gourgoulhon:2007ue} for a review on 3+1 decomposition in gravity), both being zero in the cases under consideration. 

This leads us to the necessity of imposing the vector constraint by hand. The imposition of the condition \eqref{shift_condition}, together with the minisuperspace reduction, make the constraint \eqref{ADM_vect_cstr} to be
\be 
\label{minisup_vect_constr}
0\doteq D_a ( K q^{ab} - K^{ab})= \f{(\gamma^{k\ell}\dot{\gamma}_{k\ell}) \gamma^{ij} + \dot{\gamma}^{ij}}{2N} D_a(e^a_i e^b_j)\,.
\ee
For a given triad, the internal metric $\gamma_{ij}$ must be chosen such that this equation is satisfied, by fixing relationships between different coefficients. Nonetheless, we would like to stress that there is no straightforward solution to none of the conditions \eqref{minisup_vect_constr} nor \eqref{shift_condition}, and the solutions must be searched case by case. We report in the appendix \ref{app:2_bianchi} the results for a wide class of models, provided by Bianchi cosmologies and black hole minisuperspace. For some of them (e.g. black holes) only some of the shift components must vanish in order to satisfy \eqref{shift_condition} and \eqref{minisup_vect_constr}, meaning that we can obtain the same reduced action for different choices of the shift vector. In order to see the geometrical interpretation of this fact, imagine that $N^1$ is not constrained to be zero. Let the dual vector to the frame field $e^1$  be $\vec{\partial_1}$, then $N^1$ is not constrained. It do not appear in \eqref{shift_condition} if and only if $\pounds_{\vec{\partial}_1} q_{ab} =0$. This is interpreted as the fact that we can freely translate different slices at constant time $t$ along $\vec{\partial_1}$, without affecting the dynamics of the reduced model. 

At the end of the day, the result is that the vector constraint is imposed by hand, meaning that we are left with only the scalar constraint imposed by the lapse. But it could happen that we are left with more than one Lagrange multiplier in the metric. As the condition \eqref{shift_condition} does not necessarily fix all the shift components, the same action, containing only the lapse, is obtained for any choice of the unconstrained shift. This will correspond to different coordinates on the slice within the four-dimensional manifold, and we will see its consequences for the black hole case.

One last remark should be done at this point. When we integrate out the spatial dependence in the action, we can be forced to introduce a fiducial size on it, in order to regulate the divergent integration. This will introduce the presence of a spatial boundary, that will play a crucial role in the minisuperspace symmetry, in spite of what might meet the eye. We define thus the fiducial volume as
\be
\cV_0 =\f{1}{16\pi} \int_\Sigma \de^3 x |e|\,, \q\q e=\det(e^i_a)\,,
\label{fiducial_volume}
\ee
over the compact support $\Sigma$. Upon imposing the vector and shift constraint the ADM action reduces to 
\be\label{Einstein_mini}
\cS_{\rm{ADM}} &=\f{1}{\kappa }\int_{t_\rm f}^{t_\rm i} \de t \int_\Sigma\de^D x\,\sqrt{-g}\,\cR -\f{2}{\kappa^2} \left .\int_\Sigma \de^D x\, \sqrt{q}\, K\right |_{t_\rm{i}}^{t_\rm{f}}\cr
&=\f{1}{\kappa} \int_{t_\rm f}^{t_\rm i}\de t \int_\Sigma \de^D x \,N\sqrt{q}\left (K^2-K_{ab}K^{ab} + R^{(3)} - 2 \Lambda \right )    \\
&= \f{\cV_0}{G} \int_{t_\rm f}^{t_\rm i}\de t\,\sqrt{\gamma}\left[\f{1}{4N}\Big((\gamma^{ij} \dot \gamma_{ij})^2+ \dot \gamma_{ij} \dot \gamma^{ij}\Big)  \right ] +\f{1}{\kappa} \int_{t_\rm f}^{t_\rm i} \de t \int_\Sigma \de^D x \,N\sqrt{q}\,R^{(3)} \notag\,,
\ee
where the Lagrangian splits in two terms: the kinetic term coming from the extrinsic curvature, while the potential is given by the three dimensional Ricci scalar. For the former we have a simple expression in terms of the internal metric and the triad, but for the scalar curvature the two do not split in a simple fashion. See appendix \ref{app:1_ADM_triad} for a complete list of formulas. Moreover the vector constraint imposes some relations between the internal metric coefficients, meaning that the theory is actually defined on a slice of the original 6 dimensional configuration space, and the $\gamma_{ij}$ are actually functions of a lower dimensional field space, that is going to be studied in the following chapter.  

For the sake of clarity, we will detail the previous construction and verify the equivalence with the Einstein field equation on a specific example, namely the four-dimensional black hole minisuperspace. The latter is defined through the triad\footnote{The coordinates $\theta$ and $\phi$ parametrize the sphere at fixed time and radius in Schwarzschild coordinates, while $t$ is going to be the radius of the sphere. However, from the point of view of the interior, where the minisuperspace has been built in the first place, it is a time-like coordinate. The interpretation of $x$ is less straightforward and depends on how we choose the foliation, this is discussed just below.}
\be
e^1 = \de x\,,\q\q e^2 =L_s\, \de \theta\,,\q\q e^3=L_s\, \sin \theta\, \de\phi\,.
\ee
The length scale $L_s$ is introduced here to give the same dimension to the frame fields, and to have dimensionless coefficients for the internal metric. The slice has the topology of $S^2\times \R$ and then we need to regulate the divergent $x$ integration via the cutoff $L_0$. The fiducial volume is here
\be
\cV_0 = \f{L_0 L_s^2}{4}\,, \q\rm{with} \q x\in[0,L_0]\,,\theta\in[0,\pi]\,, \phi \in[0,2\pi].
\label{fiducial_volume_KS}
\ee
The dependence on multiples length scales is a typical feature of minisuperspaces. On the one hand, the need to regulate the spatial integration introduces an infrared cutoff, on the other hand, the Planck length $\lp = \sqrt{G}$ signals the UV regime.

We need to solve the vector constraint \eqref{minisup_vect_constr} and the condition \eqref{shift_condition}. The calculation is quite involved, because of the non-linear dependence of the equations on the internal metric and the shift vector. On the other hand, it is quite easy to verify that it is satisfied for 
\be
\gamma_{ij} = \rm{diag}\,\big(a^2(t),b^2(t), b^2(t)\big)\,,\q\q N^2=0=N^3\,.
\ee
The result is easily understood, by the fact that the triad only depends on the coordinate $\theta$ via $e^3$ and $e^2$ is proportional to $\de \theta$. In this case the field space is two dimensional and is spanned by the variables $a,b$. This is not the only solution, but it is the one that we are going to study in the following, and it defines the black hole minisuperspace
\be
\de s^2 = -N^2 \de t^2 +a^2 (\de x + N^1 \de t)^2 + L_s^2\, b^2\, \de \Omega^2\,,\q\q \de \Omega^2 =\de \theta^2 + \sin^2 \theta\, \de \phi^2\,. 
\ee
The non-zero components of the Einstein tensors are
\bsub
\be 
\cG_{xx} =\f{1}{N^1}\cG_{tx}&= -\f{a^2}{L_s^2 b^2 N^3}\left (N^3-2 L_s^2 b \dot b \dot N +L_s^2 N(\dot b^2 +2b \ddot b)\right )\,, \label{Einst_tens_KS_1}\\
\cG_{tt}-N^1 \cG_{tx}&= \f{\dot b^2}{ b^2} + 2\f{\dot a \dot b}{ ab} + \f{N^2}{L_s^2 b^2} \label{Einst_tens_KS_2}\,,\\
\dot b \cG_{\theta\theta} = \f{\dot b }{\sin^2 \theta} \cG_{\phi\phi} &= - \f{L_s^2 b}{2a} \f{\de}{\de t}\left (\f{a b^2(\cG_{tt}-N^1 \cG_{tx})}{N^2}\right ) -\f{L_s^2 b^3 \dot a}{2a^3}  \cG_{xx} \,.
\ee
\label{Einst_tens_KS}
\esub
The last equation reflects the residual gauge invariance of time reparametrization, giving a relationship between the components of the Einstein tensor, while the first two reflects the freedom in the shift choice.  The reduced action gives:
\be
\cS_\rm{ADM} = \f{\cV_0}{G} \int \de t \left [-\f{4 b\, \dot a \dot b+ 2 a 	\dot b^2}{N} +2N\f{a}{L_s^2} \right ]\,.
\label{KS_action_ab}
\ee
This has the form of a mechanical action for the fields $a$ and $b$, the kinetic term coming from the extrinsic curvature, while the potential is given by the three dimensional Ricci scalar.

In order to complete the construction of the minisuperspace, we still need to choose a particular lapse and shift. The latter has completely disappeared from the Lagrangian, it is impossible to recover it from the action principle, while de lapse is always present, but simply reflect the residual time reparametrization invariance. The absence of the shift in the action is not a problem, as also the Einstein equations are independent on it, meaning that the vanishing of the r.h.s. of equations \eqref{Einst_tens_KS_1}, \eqref{Einst_tens_KS_2} implies the vanishing of all the components of the Einstein tensor, regardless of the value of $N^1$. The equation of motion arising from \eqref{KS_action_ab} are completely equivalent to the vanishing of the Einstein tensor \eqref{Einst_tens_KS}, for all lapse and shift, even if they are field dependent. The Euler Lagrange equations for the action \eqref{KS_action_ab} are
\bsub
\be
\cE_a &= -2\f{\cV_0}{G}\f{b^2 N}{a^2} \cG_{xx}\,,\\
\cE_b &= -4\f{\cV_0}{G}\f{a N}{b L_s^2} \cG_{\theta\theta}\,,\\
\cE_N &= 2 \f{\cV_0}{G}\f{a b^2}{N^2}(\cG_{tt}-N^1 \cG_{tx})\,.
\ee
\label{Eom_BH_minisup}
\esub
The gauge symmetry is here represented as the relationship between equations of motion, equivalent to the one between terms of the Einstein tensor
\be
N \dot \cE_N = \dot b \cE_b + \dot a \cE_a\,.
\ee
The redundancy of the shift vector is not translated into a gauge symmetry, but conversely into the fact that the mechanical model is the same for different choices of the slicing of the black hole. For example, we can take $N^1=0$, representing a $x$ space-like coordinate, or $N^1 =\pm N/a$, giving a null coordinate. The slices at constant time $t$ are always the same (we will see that they actually correspond to slices at constant Schwarzschild radius), but the point at constant $x$ can move from one slice to the other. We will come back to the meaning of this freedom in the next chapter.

Beside the one presented here, there are also other, more systematic, ways to discuss symmetry reduction in gravity \cite{Ashtekar:1991wa,Fels:2001rv,Torre:2010xa} . However, these strongly rely on the structure of the isometry group of the foliation \cite{Fels:2001rv,Torre:2010xa}, or the topology of the slice \cite{Ashtekar:1991wa}. We choose the approach presented in this chapter because of its easier representation in terms of spacetime metric: we can roughly summarize the conditions \eqref{shift_condition} and  \eqref{minisup_vect_constr} as respectively the requirement that the evolution in $\tau$ is the same as the evolution in the orthogonal direction (there is no variation along the shift) and the freezing of the spatial diffeomorphisms (the vector constraint is satisfied). We were able to find at least one set of (diagonal) internal metric whose symmetry reduced action gives the right Einstein equations, for all the minisuperspace we have considered. Then we content ourself with this approach. We are aware that a more systematic study is needed if we want to consider all the possible internal metrics.

Before moving to the study of the minisuperspace symmetries, let us propose a brief review of the boundary symmetries in General Relativity.

\section{Boundary symmetries and covariant phase space}
\label{sec1.3:boundary}

We finish this introductive chapter with a brief presentation of how boundary symmetries play a role in gravity. We have stressed since the beginning that boundaries (may them be at finite distance or asymptotic) should be treated with care in the case of a gauge theory. It may occur that some transformations, which are gauge in the bulk, get promoted to physical symmetries on the boundary, by the acquisition of a non-trivial charge. This is crucial for topological theories, where an infinite number of boundaries degrees of freedom could arise, despite the absence of local bulk degrees of freedom, but it has also the appealing prospect of establishing holographic dualities  \cite{Arcioni:2003xx,Arcioni:2003td, Barnich:2010eb, Fotopoulos:2019vac, Laddha:2020kvp, Donnay:2020guq, Puhm:2019zbl, Guevara:2021abz}. An important feature of the charge algebra is that it might differ from the algebra of gauge symmetries, often by a central extension, that must be present on a quantum level. On the contrary, by definition, the gauge transformations must be implemented with an anomaly-free algebra.

Even setting aside their relationship with gauge transformations, boundaries have another crucial role in field theory. Maybe the greatest achievement of theoretical physics is the description of dynamics in terms of the \textit{principle of least action}. This is the most general physical \textit{rule}, and describes an enormous amount of phenomena, going from String Theory and General Relativity to fluid dynamics or optics. It is a variational principle states that the trajectories of a given system (i.e. the solution of the equation of motions, EOM) are stationary points\footnote{The adjective \textit{least} finds historical roots in the minimization of the length of the trajectory followed by a light ray, where the principle has been formulated in the first place, no minimality is required in the principle, just stationarity} of an action functional. We denote the latter schematically by $I[\Psi] = \int_\cM \cL[\Psi]$, where $\Psi$ represents the collection fields of the theory, defined on a manifold $\cM$. For the principle to be well-defined the action must be \textit{differentiable}, meaning that its functional derivative must be of the form $\delta I =\int_\cM \rm{EOM} \cdot \delta \Psi $ + a boundary term that is zero if we held fixed some boundary condition. It is evident that the boundary has its words to say. 

To deal with the last issue, the covariant phase space formalism has been developed, which hides the choice of boundary condition into a pre-symplectic potential, that in turn enlightens a procedure to define integrable charges. This last ingredient is fundamental to determining whether the gauge symmetry becomes physical or not, being unable to define the associate charges makes the discussion of boundaries and gauge symmetries pointless. This observation points out the strong relationship between boundary, boundary conditions, gauge symmetries and the formulation of the theory that is chosen, upon which the symplectic structure relies.

We start this section by resuming the main features of the covariant phase space, and the formulation of Noether's theorems in its terms. We then go on to review its application to asymptotically AdS or flat 3D gravity. In the latter the BMS$_3$ group appears, which unravels the double-fold scope of this section, that is precisely to introduce the infinite-dimensional extension of the 2+1 Poincar\'e group that will be employed in chapter \ref{chap3}.

\subsection{Covariant phase space}
The covariant phase space allows to deal at the same time with the definition of the Poisson structure on the field space and the account for boundary conditions. We choose to follow the notation as in \cite{Geiller:2020edh}, using $\delta$ for derivation on the field space, and $\de$ for the exterior derivative of space-time forms, the interior product on field space is $\ipp$, while for spacetime forms is $\ip$. The variation of the Lagrangian $\cL$ (that might contain some boundary term to control polarization, boundary condition or renormalization informations) is given by
\be
\delta \cL = \rm{EOM}\cdot \delta \Psi  + \de \theta \,,
\label{var_lagr}
\ee
where $\theta $ is the pre-symplectic potential, and defines the symplectic form\footnote{By variation of \eqref{var_lagr} we have that $d \delta \theta =0$, so $\int_{\partial \cM} \delta \theta =0$. If the spatial boundary do not contribute, the symplectic structure is conserved} through
\be
\Omega = \int_\Sigma \delta \theta\,,
\ee
where $\Sigma$ is a slice at constant time (as in any Hamiltonian approach, we meed a decomposition between slices and evolution time). The first Noether theorem states that if a transformation $\delta_\epsilon$ is a symmetry in the sense that the variation of the Lagrangian is a total derivative
\be
\delta_\epsilon \cL = \rm{EOM}\cdot\delta_\epsilon \Psi + \de \theta[\delta_\epsilon] = \de b\,,
\ee
then the current $J:=\theta[\delta_\epsilon] - b$ is conserved, when the equation of motion holds\footnote{We denote the \textit{on-shell} equalities by $\approx$.}:
\be
\de J \approx 0 \q \Rightarrow \q J=\de Q\,.
\ee
For gauge symmetries, on the other hand we have the \textit{off-shell equality} that reflects the relationship between different equation of motion ($P$ vanishes \textit{on-shell})
\be
\rm{EOM}\cdot\delta_\epsilon \Psi = -\de P\,.
\ee
The Hamiltonian charges are then defined as
\be
\cancel{\delta} \cH[\epsilon] &= \Omega[\delta,\delta_\epsilon] = -\int_\sigma \delta_\epsilon\ipp\delta\theta[\delta]\\
&=\int_\Sigma \delta P+ \oint_{\partial \Sigma}(\delta Q +f[\delta])\,,\q\q\de f=\delta b-\delta_\epsilon \theta\,. 
\notag
\ee
For gauge symmetries, since $P\approx 0$, all the information is captured on the co-dimension 2 corner $\partial \Sigma$. The symbol $\cancel\delta$ stands for the fact that in general, the charges might not be integrable, because of the contribution of $f[\delta]$. The quest is then to find a clever set of choices, playing on the symplectic form, or the field dependence of the gauge parameter, to make them integrable. This is a quite involved task and might be even impossible to succeed, for example in the presence of a symplectic flux. For 3D gravity, consistent efforts have been made in this direction, and a huge set of boundary symmetry has been discovered. We review here the most historical ones, provided by the asymptotic $\BMS_3$ group.

There are of course transformations for which the codimension 2 integration is also zero. Using the terminology of Regge and Teitelboim \cite{REGGE1974286}, we refer to the latter as \textit{proper gauge transformation}, while if the charge is non zero we call them \textit{improper gauge transformation}. The asymptotic (or boundary) symmetry group is then defined by the quotient between improper and proper transformations that preserve the boundary conditions.

We just end this section by recalling that the Poisson's structure descends directly from the symplectic form. For integrable charges, their algebra is calculated through
\be
\{H[\epsilon_1],H[\epsilon_2]\}=\Omega[\delta_{\epsilon_1},\delta_{\epsilon_1}]\,.
\ee
We refer the interested reader to \cite{Iyer:1994ys,Compere:2018aar,Harlow:2019yfa} for a more detailed introduction to covariant phase space, and \cite{REGGE1974286,HenneauxTeitelboim2020,Banados:2016zim}  for the relationship between Noether's theorem, gauge symmetries and boundaries.
\subsection{3D gravity}
Three-dimensional Lorentzian gravity offers an easy handle to deal with boundary structure and holographic properties, because of its topological nature. The solutions are indeed parametrized by a finite set of global quantities (like mass or angular momentum), on the other hand, the presence of a bounded region unravels new structures. To see how they appear let us consider the parametrization of the manifold through the Bondi-Sachs form:
\be
\label{BS_gauge_3D}
\de s^2 = e^{2\beta} \f{V}{r} \de u^2 - 2e^{2\beta} \de u \de r + r^2 (\de \phi - U \de u)^2\,. 
\ee
Einstein equations in the presence of a negative cosmological constant $\Lambda = -1/\ell^2$, together with Dirichlet boundary conditions (we fix the metric at the boundary) at $r\to \infty$, give the following fall off conditions
\be 
\f{V}{r}=-\f{r^2}{\ell^2}+2\cM(u,\phi) - r^{-2} \cN(u,\phi) + \cO(r^{-3})\,,\q
U=- r^{-2} \cN + \cO(r^{-3})\,,\q
\beta = \cO(1/r)\,,
\label{eq:fall_off_3d}\ee
where 
\be
\partial_u\cN = \partial_\phi\cM\,,\q\q
\partial_u \cM = \f{1}{\ell^2} \partial_\phi \cN\,.
\label{eq:on_shell_3d}
\ee
The subset of diffeomorphism that preserves the conditions \eqref{eq:fall_off_3d} and \eqref{eq:on_shell_3d} are generated by the vector $\xi$:
\be
\xi^u =f\,,\q
\xi^r =\partial_\phi^2 f -r \partial_\phi g -\cN \f{\partial_\phi f}{r}\,,\q
\xi^\phi =g -\f{\partial_\phi f}{r}\,.
\ee
with $\partial_u f = \partial_\phi g$ and $\partial_u g = \f{1}{\ell^2} \partial_\phi f$. As the vectors are field dependent, their Lie algebra is evaluated by the introduction of a modified Lie bracket, that accounts for the filed variation of one transformation upon the other (see \cite{Barnich:2010eb} for definition and discussion about the modified bracket\footnote{Sometimes this is called \textit{adjusted} Lie bracket \cite{Compere:2015knw}, but has not to be confused with the modified Barnich-Troessaert bracket at the level of the charge
algebra \cite{Barnich:2011mi}.}):
\be
\label{modif_bracket}
[\xi_1,\xi_2]_\star =[\xi_1,\xi_2] -\delta_{\xi_1} \xi_2 +\delta_{\xi_2} \xi_1\,.
\ee
This gives the Lie algebra $[\xi(f_1,g_1),\xi(f_2,g_2)]_\star = \xi(f_{12},g_{12})$, with
\bsub\be
f_{12} &= f_1 \partial_\phi g_2 + g_1 \partial_\phi f_2 -\left (1\leftrightarrow 2\right )\,,\\
g_{12} &= \f{f_1}{\ell^2} \partial_\phi f_2 + g_1 \partial_\phi g_2  -\left (1\leftrightarrow 2\right )\,.
\ee\label{BMS_vector_algebra_3d}\esub
It is now time to discuss the existence of the corresponding charges. In order to apply the covariant phase space approach, we need to choose a formulation of the theory, may it be tetrad formalism with Einstein Cartan Lagrangian \cite{Carlip:2017xne,Carlip:2019dbu,Geiller:2020edh}, or metric gravity with the opportune renormalized symplectic potential \cite{Ashtekar:1996cd,Barnich:2006av,Barnich:2017jgw,Ruzziconi:2020wrb,Geiller:2021vpg}. The review of the different approaches is beyond the scope of this thesis, the interested reader might look at \cite{Brown:1986nw} for a more complete historical review of 3D asymptotic symmetries, \cite{Geiller:2020edh} for computation of the charges at finite distance, \cite{Ruzziconi:2020wrb} for a generalisation to other boundary condition than Dirichlet, and \cite{Geiller:2021vpg} for a more general gauge, that unravels a new part of the algebra. Let us go back to the charges, no matter the approach that we use,  we find for the asymptotic boundary charges 
\be
\cQ_\xi =\f{1}{\kappa} \oint_0^{2\pi} \de \phi \left [f \cM + g \cN \right ]\,.
\ee
Let us now restrict to the flat case, where the BMS group arises. This is given by $\ell \to \infty$, which means that equation \eqref{eq:on_shell_3d} and the condition on the $f,g$ functions implies:
\be\left|
\begin{array}{rl}
\cM &= \cJ(\phi)\\
\cN &= \cJ(\phi)+u\cP'(\phi)
\end{array}\right .\,,&\q\q\left |
\begin{array}{rl}
f &= \alpha(\phi) + u X'(\phi)\\
g &= X(\phi)
\end{array}\right .\,, \label{BMS_gravity_charges}\\
\notag \cQ[X,\alpha] &=\f{1}{\kappa} \oint \left [X\cJ + \alpha \cP \right] \,.
\ee 
The prime $'$ denotes the derivative with respect to the angle $\phi$. Defining the mode expansion of the symmetries generators as $\cL_n =\cQ[X=e^{in\phi}]$, $\cT_n =\cQ[\alpha=e^{in\phi}]$, the charge algebra is given by
\bsub\be
\{ \cL_n,\cL_m\} &= i(n-m)\cL_{n+m} \,,\\
\{ \cL_n,\cT_m\} &= i(n-m)\cT_{n+m} + \f{i m^3}{4\pi G} \delta^0_{n+m}\,,\\
\{ \cT_n,\cT_m\} &=0\,.
\ee\label{bms_modes}\esub
This is the centrally extended BMS algebra (\cite{Oblak:2016eij} for a review). We will come back to its properties in the third chapter. 

For a non vanishing cosmological constant, we can redefine the symmetry generators as $f=\f{\ell}{2}(Y^+ + Y^-)$ and $g=\f{1}{2}(Y^+-Y^-)$, whose mode the mode expansion is given by $\cL_n^\pm =\cQ[Y^\pm=e^{\pm in x^\pm}]$, with $x^\pm = t/\ell \pm \phi$, that corresponds to a double copy of the Virasoro algebra
\bsub\be
\{ \cL_n^\pm,\cL_m^\pm\} &= i(n-m)\cL_{n+m}^\pm +i \f{3\ell}{24 G} m^3 \delta^0_{n+m}\,,\\
\{ \cL_n^\pm,\cL_m^\mp\} &=0\,.
\ee\label{2vira_modes}\esub \newline \

In this introductive chapter, we have reviewed some aspects of the symmetries in General Relativity. This is a gauge theory, invariant under diffeomorphism, and moreover, the only possible one depending on the metric coefficient up to second order. On the other hand, it is not the only theory invariant under diffeomorphism, what has just been discussed about the boundary structures allows the glimpse of something more peculiar to gravity that makes diffeomorphism something more than a mere gauge symmetry. In the next chapter, we will see how the simplicity of minisuperspace makes manifest the rising of this kind of new structure.

\chapter{Symmetries in minisuperspaces}
\label{chap2}

We have just seen how boundaries play an important role in the understanding of gravitational symmetries and degrees of freedom. However, even setting aside the role of boundaries, bulk symmetries present their own subtleties. Among the many formulations of general relativity, the gauge content is different, diffeomorphism freedom may either be partially fixed or supplemented with additional gauge invariances \cite{Peld_n_1994, Aldrovandi:2013wha, Gielen:2018pvk, krasnov_2020}. It is still not clear whether this freedom is harmless or not \cite{Matschull:1999he, Lusanna:2003im, Kiriushcheva:2008sf}. If, as it appears with the introduction of boundaries,  gauges contain physical information \cite{Rozali:2008ex, Rovelli:2013fga, Kapec:2014opa, Rovelli:2020mpk}, we should aim to understand in detail the symmetry content of a given formulation. 

Minisuperspaces furnish a good testbed for this. Despite their apparent simplicity, they contain a richer structure than what meets the eye. The simplest reduced model we could think off is the homogeneous and isotropic picture of the Universe. Freezing almost all the degrees of freedom, except for a global conformal factor for flat 3-dimensional slices, results in the Friedmann-Lema\^itre-Robertson-Walker (FLRW) metric. The dynamics of this model is too simple, and we need to add some matter content to make the scale factor evolve in time. The simplest one, and yet very relevant for the study of early stages of inflation in cosmology, is given by a massless scalar field. 

It was in this setup that the conformal $\SL(2,\R)$ rigid symmetry was find in the first place \cite{BenAchour:2017qpb,BenAchour:2019ywl,BenAchour:2019ufa,BenAchour:2020njq,BenAchour:2020xif} (see also \cite{Dimakis:2015rba,Dimakis:2016mpg,Christodoulakis:2018swq,Pailas:2020xhh,Dussault:2020uvj}), to be rapidly extended to Bianchi I models \cite{BenAchour:2019ywl} and in the presence of a cosmological constant \cite{BenAchour:2020xif}. We start this chapter with a brief review of the cosmological flat model (along the lines of \cite{BenAchour:2019ufa}), as it is the precursor for the generalisation that we develop in the following, where we are going to extend the symmetry group to the semidirect product  $\left (\SL(2,\R) \times \R\right ) \ltimes\R^4$. The last part of the chapter is devoted to the application to black hole minisuperspace and anisotropic cosmological systems.

The main part of the chapter collects results and thoughts of \cite{Geiller:2020xze} and \cite{Geiller:2022baq}

\section{FLRW model and conformal symmetry}
\label{sec2.1:FLRW}

In the minisuperspace formalism introduced in section \ref{sec1.2:symm_reduc} the flat spacetime is encoded in the diagonal triad $e^i_a =\delta^i_a$ and the homogeneity corresponds to the diagonal internal metric $\gamma_{ij} = a^2 \delta_{ij}$. The shift is usually set to zero, but we add a minimally coupled massless scalar field $\phi$. The reduced action is given by
\be
\label{FLRW action}
\cS= \cV_0 \int_\R \de t \left [a^3 \f{\Phi'^2}{2N} - \f{3}{8\pi G} \f{a a'^2}{N} \right ]\,,
\ee
where $\cV_0$ is the fiducial cell volume $\cV_0 = \int_\Sigma \de^3 x$, meaning that the spatial coordinates are constrained into a compact support, e.g. a cube of edge $\cV_0^{1/3}$. The derivative with respect to coordinate time is hereafter denoted with a prime $a' = \de a/\de t$.	Because of the symmetry reduction we have killed almost all the freedom of coordinates reparametrization, with the exception of the time diffeomorphisms, encoded into the lapse function and the scalar constraint. The corresponding gauge symmetry is given by
\be
\label{time_diffeos}
\left|\begin{array}{ccrl}
t &\mapsto&  \tilde{t} &=f(t)\\
N &\mapsto & \tilde N (\tilde{t}) &= h^{-1} N(t)\\
a &\mapsto & \tilde a (\tilde{t}) &= a(t)\\
\Phi &\mapsto&  \tilde \Phi (\tilde{t}) &= \Phi(t)\\
\end{array}\right .
\q\q \rm{with}\q h(t) := f'(t)\,.\ee
A straightforward calculation of the Euler-Lagrange equation corresponding to \eqref{FLRW action} gives 
\bsub \be
\cE_a &= \f{3 \cV_0 \left (a'^2+4\pi G a^2 \Phi'^2+2a a'')\right )}{8\pi G } \approx 0\,,\\
\cE_\Phi &= -\cV_0\f{ \de}{\de t} (a^3 \Phi') \approx 0\,,\\
\cE_N &= \cV_0 \left [\f{3}{8\pi G} \f{a a'^2}{N^2} - a^3  \f{\Phi'^2}{2N^2} \right ] \approx 0\,.
\ee
\label{eom_FLRW}
\esub
They correspond to the Friedmann equations, meaning in the first place that the variational principle of \eqref{FLRW action} is equivalent to the Einstein field equations of the FLRW metric in the presence of a scalar field. Respectively, the first one corresponds to the acceleration equation or Raychaudhury’s equation, the second to the energy conservation of the scalar field, and the latter to the first Friedmann equation, relating the Hubble factor to the energy density. It is known that the three equations are not independent, and in the present framework this can be seen as a direct consequence of the gauge freedom \eqref{time_diffeos}:
\be
\label{relation_eom_FLRW}
N  \cE_N' =  \Phi' \cE_\Phi + a' \cE_a\,.
\ee
This means that $N$ is not specified by the equation of motion and we can freely choose it. 

We would like to dig a little further into the relationship between time diffeomorphisms, gauge choices and charges. In the general setting, we have claimed that the scalar constraint is associated with the time diffeomorphisms, but actually, this relationship is less straightforward than it meets the eye. The main problem is how we infer the lapse transformation in \eqref{time_diffeos} from the Poisson's bracket with the Hamiltonian if the lapse is not in the phase space?

\paragraph{Integrable charges for time reparametrization} In order to properly build a relationship between the infinitesimal time reparametrization and the Hamiltonian constraint we shall add the lapse to the phase space \cite{Pons:1996av,Pons:1999az,Garcia:2000yt,Pons:2003uc}. This means that we need to consider the following modified action, for the time coordinate $t$:
\be 
\cS_\rm{Lapse} = \cS + \int \de t(N'  - \mu) \Pi\,.
\ee
The action $\cS$ is given at \eqref{FLRW action}, and we have added form the onset a momentum $\Pi$ for the lapse $N$, together with a multiplier $\mu$ enforcing the constraint $\Pi\approx 0$. The variational principle for the new action is completely equivalent to the evolution described by the action principle for $\cS$. We have the general variation of the Lagrangian
\be
\delta \cL_\rm{Lapse} = \delta a \cE_a +\delta \Phi \cE_\Phi + \delta N (\cE_N - \Pi') -\delta \mu \Pi + (N'-\mu) \delta \Pi + \theta'\,,
\ee 
where the $\cE$ are defined in \eqref{eom_FLRW}. From these we see the equivalence between the two action principles and we can infer presymplectic potential 
\be
\theta=p_a \delta a +\Pi\delta N + p_\Phi \delta\Phi\,,\q\q \left\{
\begin{array}{rl}
p_a &= -\df{3 \cV_0 a a'}{4 \pi G N}\\[9pt]
p_\Phi &= \df{\cV_0 a^3 \Phi'}{N}
\end{array}\right.\,.
\ee
We now want to make use of the covariant phase space formalism to study the existence of integrable charges associated with the symmetries \eqref{time_diffeos}, their infinitesimal version for the new action $\cS_\rm{Lapse}$ reads:
\be\label{alpha symmetry}
\delta_\alpha a=\alpha a'\,, 
\q
\delta_\alpha \Phi=\alpha \Phi'\,,\q
\delta_\alpha N=\partial_t( \alpha N)\,,
\q
\delta_\alpha\Pi=0\,,
\q
\delta_\alpha\mu=\partial_t^2( \alpha N)\,.
\ee
Using the symplectic form $\Omega =\delta\theta$, we can compute the generator of the symmetries \eqref{alpha symmetry}. Thanks to the equations of motion, we can get rid of the second derivatives and obtain
\be
\cancel \delta \cH = -\delta_\alpha \ipp \Omega\approx \alpha \delta\left (a^3 \f{\Phi'^2}{2N} - \f{3}{8\pi G} \f{a a'^2}{N} \right ) + \partial_t(\alpha N) \delta \Pi\,.
\ee 
One can see that this is not integrable, and integrability can be achieved with the field-dependent redefinition $\alpha = \eps/N$, for which we find the generator
\be
G(\eps) = \f{\eps}{\cV_0} \left [ \f{p_\Phi^2	}{2a^3} - \f{2 \pi G}{3} \f{p_a^2}{a}\right ] + \eps' \Pi\,.
\ee

We would like to remark here that in the one-dimensional context here (all the fields depend only on one spacetime direction) there is no codimension 2 corner that captures the non-trivial information of the boundary. Everything lives in the bulk and the integrable generators are zero on-shell $G\approx 0$, synonymous with a gauge transformation. The Poisson bracket of this generator with $q$ and $N$ does indeed generate \eqref{alpha symmetry} with $\alpha=\eps/N$.

We can perform a Legendre transform of the Lagrangian $\cL_\rm{Lapse}$ to obtain the corresponding Hamiltonian, this gives:
\be
H_\rm{Lapse} = \f{N}{\cV_0} \left [ \f{p_\Phi^2	}{2a^3} - \f{2 \pi G}{3} \f{p_a^2}{a}\right ] + \mu \Pi := N \cH_\rm{ADM} + \mu \Pi\,,
\ee
where the first part $\cH_\rm{ADM}$ corresponds to the first class constraint of the initial gravitational action \eqref{FLRW action}, and the $\mu \Pi$ term accounts for the inclusion of $N$ into the phase space.  

This has the following interpretation: the original action \eqref{FLRW action} exhibits a time reparametrization gauge invariance \eqref{time_diffeos}, that is associated to the first-class constraint $\cH_\rm{ADM}\approx 0$ and to the relationship between equations of motion \eqref{relation_eom_FLRW}. However, if we want to recover the transformation from the constraint, we must add $N$ to the phase space, which in turn implies that the integrable generator is given by $G(\eps)$, nonetheless, there is a priori no obstruction to acting with 
\be
G(\alpha N)=\alpha H_\rm{Lapse}+\dot{\alpha}N\Pi = \alpha N \cH_\rm{ADM} + (\alpha \mu + \dot{\alpha}N )\Pi \,,
\ee which generates the initial (non-integrable) transformation \eqref{alpha symmetry}.

We now discuss the opposite approach, which consists of gauge fixing the time. It is indeed convenient to describe the dynamics in terms of a gauge fixed time corresponding to a choice of lapse. 
For example we can take here the proper time $\tau$, such that $\de \tau = N \de t$. This makes $N$ to disappear from the action, which becomes
\be
\label{FLRW action_gfix}
\cS_0= \cV_0 \int_\R \de \tau \left [a^3 \f{\dot \Phi^2}{2} - \f{3}{8\pi G} a \dot a^2 \right ]\,,
\ee
where now the dot represents proper time derivatives $\dot a = \de a/\de \tau$. This is exactly the same thing as if we had chosen from the onset $N=1$. This is of course not the only possible choice for the lapse, it could be for instance field dependent (e.g. $N=1/a$, etc...), as $N$ is completely free.

This turns the gauge theory into a simple mechanical model without constraint, but this means that we loose the Euler Lagrange equation imposed by the lapse $\cE_N$, which is precisely the scalar constraint and it cannot be recovered from the relationship \eqref{relation_eom_FLRW}. The latter just says that from the equation of motion of $a$ and $\Phi$ we infer that $\cE_N$ is constant, but not that it is zero. Another way of saying it is that the action \eqref{FLRW action_gfix} has a true Hamiltonian, that is a constant of motion, whose value must be fixed to zero by hand. Introducing the canonical momenta $p_\Phi ,p_a$, the Legendre transform\footnote{The canonical momenta are the same as for the gauge Hamiltonian, for a configuration variable q: $\delta S/\delta q' = \delta \cS_0 /\delta \dot q$. e.g.:
\be\notag
p_a = \f{\delta S}{\delta a'} = -\f{3\cV_0}{4\pi G} \f{ aa'}{N} = -\f{3\cV_0}{4\pi G} a \dot{a} = \f{\delta S_0}{\delta \dot a}\,.
\ee} of \eqref{FLRW action_gfix} gives the Hamiltonian
\be
H = \f{1}{\cV_0} \left [ \f{p_\Phi^2	}{2a^3} - \f{2 \pi G}{3} \f{p_a^2}{a}\right ]\,.
\ee

The main message that shall be retained here is that for a mechanical model with time reparametrization invariance like \eqref{FLRW action}, we can explicitly build a relationship between the constraint and the integrable generator of the gauge symmetry only by adding the lapse to the phase space.

Conversely, we can recast the model in a gauge-fixed form by choosing a lapse (here we pick $N=1$, but it could also be field dependent), paying the price of losing the Hamiltonian constraint, that must be reintroduced by fixing by hand the energy value of a true Hamiltonian. This assures the classical equivalence between gauge fixed \eqref{FLRW action_gfix} and unfixed \eqref{FLRW action} actions but begs the question of the contribution of quantum fluctuations on different time choices. We leave this discussion for the second part of the thesis, and we continue here with the study of the symmetries of \eqref{FLRW action_gfix}.

\subsection{M\"obius symmetry and conformal particle}
On top of the time diffeomorphism invariance it exists an other symmetry \cite{BenAchour:2019ufa}, in the form of M\"obius transformations (parametrized by a $\SL(2,\R)$ group element) of proper time
\be
\label{Mobius_FLRW}
\left|\begin{array}{ccrl}
\tau &\mapsto&  \tilde{\tau} &=f(\tau) = \df{\alpha  \tau +\beta }{\gamma \tau +\delta}\\
a &\mapsto & \tilde a (\tilde{\tau}) &= h(\tau)^{1/3} a(\tau)\\
\Phi &\mapsto&  \tilde \Phi (\tilde{\tau}) &= \Phi(\tau)\\
\end{array}\right .
\q\q \rm{with}\q 
\begin{array}{rl}
&\alpha \delta -\beta \gamma =1\\
&h(t) := f'(t) = \df{1}{(\gamma \tau + \delta)^2}
\end{array}\,.\ee
The scale factor $a$ represents what is called a \textit{ conformal field} of weight $1/3$, where the weight is given by the power of the derivative that appears in the transformation law. This terminology will be used throughout the whole thesis. A straightforward calculation leads to verifying that the action \eqref{FLRW action} is invariant (up to a total derivative), under the transformations above \eqref{Mobius_FLRW}:
\be
\Delta \cS_0 = \cS_0[\tilde{a}(\tilde{\tau}),\tilde{\Phi}(\tilde{\tau})] - \cS_0[a(\tau),\Phi(\tau)] = \int \de \tau \f{\de}{\de \tau} \left [\f{\cV_0}{6 \pi G} \f{\gamma a^3}{\gamma \tau + \delta}  \right ]\,, 
\ee
proving that the M\"obius transformation is indeed a symmetry for the action \eqref{FLRW action_gfix}, leaving the corresponding Euler-Lagrange equation invariant. An interesting comparison could be made with the one dimensional conformal particle. The simplest model of conformal mechanics introduced by de Alfaro, Fubini and Furland \cite{deAlfaro:1976vlx}, describing a particle moving in one direction $q$, submitted to a conformal potential $-\beta/q^2$, exhibits the same conformal symmetry, in terms of Mobius transformation. The analogy is made explicit if we restrict ourself to the study of the gravitational degree of freedom $a(t)$. 

The equation of motion for the scalar field imposes that
\be
0\approx \dot{p_\Phi} = \f{\de}{\de \tau} (\cV_0 a^3\dot \Phi )\,.
\ee
We can directly substitute this into the action and this will not spoil the dynamics of $a$. The action in terms of $p_\Phi$ gives the same equation of motion for the scale factor as \eqref{FLRW action_gfix}. With the  change of variables 
\be
q = \sqrt{\f{\cV_0}{6\pi G}} a^{3/2}\,,\q\q \beta = -\f{1}{12 \pi G} \f{p_\Phi^2}{4\cV_0^2} =- \f{1}{12 \pi G} \left (\f{a^3 \dot \Phi}{2}\right )^2\,,
\ee
we recover exactly the dAFF action for conformal mechanics \cite{deAlfaro:1976vlx}
\be
\cS_\beta[q] =\int \de \tau \left [\dot q^2-\f{\beta}{q^2}\right ]\,.
\ee
This provides the simplest example of a one-dimensional conformal model, and has been widely studied. The analogy with the conformal particles will be the guiding principle to generalise the construction of the symmetries to other minisuperspaces, where the generalisation will be given by the increasing number of directions in which the particle can move. This will be the subject of the next section, but let us focus again on the FLRW model and discuss the Noether charges associated with the symmetry.

\subsection{Noether's theorem and conserved charges}
We have already discussed the first Noether theorem in the context of gauge symmetry. We will review it here, in order to apply it to the rigid M\"obius symmetry \eqref{Mobius_FLRW}. The pre-symplectic form for the gauge fixed model is given by $\theta = p_a \delta a  + p_\Phi \delta_\Phi $, the infinitesimal variation on the volume and the scalar field are defined as
\bsub 
\be
\delta a &= \tilde{a} (\tau) -a(\tau)  = \f{1}{3}\,\dot \xi a - \xi \dot a\,, \\
\delta \Phi &= \tilde{\Phi} (\tau) -\Phi(\tau)  = - \xi \dot \Phi\,,
\ee
\esub
for $\delta\tau =\xi(\tau)$ a second degree polynomial in time, whose coefficients are associated respectively the time translations, dilations and special conformal transformation:
\be
\tilde \tau &= \tau + \eps &&\rm{translation\;} (\alpha=\delta =1, \beta =\eps, \gamma=0)\notag\,,\\
\tilde \tau &= \tau + \eps \tau &&\rm{dilatations\;}  (\beta=\gamma =0, \alpha  =\delta^{-1}=\sqrt{1+\eps^2})\,,\\
\tilde \tau &= \tau + \eps \tau^2 &&\rm{special\; conformal\; transformation\;} (\beta= 0,\alpha  =\delta=1,\gamma =-\eps)\notag\,.
\ee
The corresponding infinitesimal variation of the Lagrangian, at leading order in $\eps$ is given by
\be
\delta L = \f{\de}{\de t}\left (\eps \cV_0\f{3 a \dot a^2 - 4 \pi G a^3 \dot \Phi^2}{8\pi G} t^n - \f{\eps \cV_0 a^3}{12 \pi G} n(n-1)\right ):= \dot b\,,
\ee
for $\delta \tau = \eps \tau^n$, with $n\in \{0,1,2\}$. In the one dimensional context there is no codimension 2 surface, so all the information is encoded into what is called \textit{current} in the field theory context, that takes here the status of charge. It is evaluated as $Q_\eps = \theta[\delta_\eps] -b$.
\be
Q[\eps t^n] = \eps \cV_0 \f{4 a^3}{3\kappa^2} n(n-1) t^{n-2}  - 4\eps\cV_0 \f{a^2 \dot a}{\kappa^2} n t^n +\eps \cV_0\left (6 \f{a \dot a^2}{\kappa^2} -\f{a^3 \dot \Phi^2}{2}\right ) t^n\,.
\ee

We can also read the $\SL(2,\R)$ structure from the Hamiltonian formulation of the dynamics. For this it turns out to be useful to consider a new set of canonical variables, measuring the physical volume of the slice, and its rate of variation
\be
V = a^3 \cV_0\,, \q\q B= -\f{1}{3\cV_0} \f{p_a}{a^2} = \f{1}{12 \pi G}\f{\dot V}{V}\,,
\ee
which inherit the Poisson structure\footnote{The reason behind the weird sign in the canonical structure, that swap what we usually assign to configuration variables and canonical momenta, find its roots in literature about quantum cosmology (see chapter \ref{chap4} and appendix \ref{app:4_LQG_phase_space}). We keep this choice here to make contact with this literature}  $\{B,V\}=1$. The Hamiltonian then reads:
\be
\label{Ham_prop_time_flrw}
H =\f{1}{2} \left [\f {p_\Phi^2}{V} - 12 \pi G V\, B^2\right ]\,.
\ee
The evolution of any phase space $\cO$ function is given by its bracket with the Hamiltonian, as in any mechanical model $\dot \cO = \partial_\tau \cO + \{\cO,H\}$. On top of this, we must impose the condition $H \approx 0$, once we solve the equation of motion. 

Now, we dig deeper into the algebraic structure of the evolution equations for the gravitational degrees of freedom. As just discussed, the variation of the volume $V$ in terms of proper time is given by its bracket with the Hamiltonian \eqref{Ham_prop_time_flrw}. The second-order time derivative by the iterated brackets $\{\{V, H\}, H\}$, and so on. We can thus check if at which stage the iteration closes and form a Lie algebra, moreover this gives the order of the differential equation that must be solved to integrate the evolution of $V$. In this case, it turns out that the iteration stops at the second step and the Lie algebra is isomorphic to $\sl(2,\R)$. 	The first derivative gives
\be
C:= \f{4}{3} \kappa^{-2} \{V,H\} = VB\,,
\ee
where $\kappa^2=16 \pi G$. The phase space function $C$ can also be interpreted as the integrated extrinsic curvature, upon the fiducial cell on the slice:
\be
\f{3}{4} \kappa^2 C = \int_\Sigma \sqrt{q} K = \cV_0 a^3 \f{3 \dot a}{a}	 = \f{3}{4} \kappa^2 VB\,.
\ee
Then the second time derivative is proportional to the Hamiltonian, as it trivially commutes with itself. The iteration stops here.
\be 
\dot C =\{C,H\} = - H\,.
\ee
The quantities $C$, $V$ and $H$, respectively the trace of the extrinsic curvature\footnote{Sometimes it is also called complexifier, from which the name $C$, because of its role in shifting the value of the Immirzi parameter, see \cite{BenAchour:2017qpb} for a more detailed discussion}, the volume of the slice, and the Hamiltonian form then an algebra, to which we refer as CVH, that is isomorphic to $\sl(2,\R)$
\be
\{V,H\} = \f{3\kappa^2}{4}C\,,\q\q \{C,H\} = - H \,,\q\q \{C,V\} = V\,.
\ee
We have seen the conformal structure appearing in two different ways, in the first place as a symmetry of the action, and then in the algebraic structure encoding the dynamics. In order to see the equivalence, it is useful to recast the CVH generators into the more familiar $\sl(2,\R)$ form
\be
L_{-1} = \f{\sigma}{3}\kappa H\,,\q
L_0 &= -C\,,\q
L_1 = \f{8}{2 \sigma \kappa^3}V \,,\q \\
\{L_n, L_m\}&= (n-m) L_{n+m}\,,
\ee
for a real parameter $\sigma$ and $\kappa= \sqrt{16 \pi G}$. We can also rewrite the Noether charges in terms of phase space variables, namely the volume and its momentum, to make contact with the CVH generators
\bsub 
\be
\cQ_- &= \f{2}{\eps}	Q[\eps ]	= -2 H \,,\\
\cQ_0 &= \f{1}{\eps}	Q[\eps t]	= -C -H \tau \,,\\
\cQ_+ &= \f{1}{2\eps}	Q[\eps t^2]	= \f{V}{12 \pi G} -C \tau - H \f{\tau^2}{2}\,.
\ee
\label{noether_charges_FLRW}
\esub
We see that the CVH generators are the evolving counterpart of the constant of motion. Despite the explicit time dependence of the conserved charges, they provide indeed constants of motion, corresponding to the initial condition of the CVH generators (i.e. their value at $\tau=0$). On the other hand $C$ and $V$ are not constant themselves. 
\be
\cQ[n]_{(\tau=0)} = L_n[\sigma = 3/\sqrt{4 \pi G}]\,.
\ee
Moreover, it is straightforward to see that the Poisson algebra of the charges is the same as the Lie algebra $\sl(2,\R)$ of the infinitesimal transformation. At the quantum level we should aim to represent the algebra without anomalies. This is easily achievable using the well know representation theory for $\SL(2,\R)$ \cite{BenAchour:2019ywl}. We can compute the charge algebra both using the standard definition of Poisson's bracket or using a double contraction of the symmetries with the symplectic potential.

The equivalence between the two structure will become a key point in the study of the quantization \ref{sec5.1:WDW}, we will define here the following notation: the charges $\cQ$ that corresponds to the Noether charges are sometimes referred to as \textit{evolving constants of motion}, because of the explicit time dependence. They measure the initial condition of the \textit{kinematical} algebra (here CVH) and the map between the two is given by the Hamiltonian flow $e^{\{\bullet, \tau H\}}$. By construction, this map is canonical and preserves the Poisson structure so that the kinematical CVH algebra is the same $\sl(2,\R)$ as the charge algebra. Moreover, the map is implemented quantum mechanically by the unitary transformation that swaps between Shr\"odinger and Heisenberg pictures, the former is given by the quantization of the CVH algebra without explicit time dependence (this is hidden in the wavefunction), and the latter is obtained by the realization of the charge $\cQ$ algebra on a suitable Hilbert space where the states are kept fixed in time.

We would like now to address the question of whether this symmetry exists from the point of view of the gauge un-fixed perspective. For this it is interesting to translate the M\"obius transformation in terms of coordinate time. Since the relation $\de \tau = N \de t$, the mapping between $\tau$ and $\tilde{\tau}$ can be achieved through an infinite number of different choices of transformations from $t,N$ to $\tilde t, \tilde N$. The simplest one is assuming that the time coordinate is unchanged and all the transformation is encoded into the lapse \cite{BenAchour:2019ufa}. This gives the finite transformations
\be
a 		&\mapsto \f{a(t)}{(\gamma \tau + \delta)^{2/3}}\,,\notag\\
\Phi 	&\mapsto \Phi(t) \,,\\
N		&\mapsto \f{N(t)}{(\gamma \tau + \delta)^2}\notag\,,
\label{Mobius_t}
	\ee
where we would like to remark that the transformations are non-local, given that they involve the proper time $\tau(t) = \int^t_{t_0}  N(s)\de s$. First of all we see that the presence of the initial time $t_0$ is necessary to recover a three parameter family as expected for the M\"obius transformation, while the transformations \eqref{Mobius_t} seemed to depend only on two parameters. Secondly the non-locality is present also in the conserved quantities \eqref{noether_charges_FLRW}, explicitly involving the proper time.  The infinitesimal version  of \eqref{Mobius_t} reads
\be
\delta_\eps a = \f{1}{3}\f{\eps'}{N} a\,,\q\q \delta \Phi =0\,,\q\q \delta N = \eps'\,. 
\ee
For $\eps = \epsilon (\tau (t))^n$, $n \in\{0,1,2\}$. A straightforward calculation verifies that they really are symmetries for the gauge action \eqref{FLRW action}, because the variation of the action is a total derivative
\be
\delta S =  -\cV_0\f{n(n-1)}{12 \pi G} \int \de t\, \f{\de q^3}{\de t} \,, \q\q n \in \{0,1,2\}\,.
\ee
On the other hand, we could ask ourselves if it is possible to find integrable generators on the extended phase space including $N$, for the transformations \eqref{Mobius_t}. Once contracted with the symplectic form $\Omega$, including the $\delta N \delta \Pi$ term, they give the non-integrable generators:
\be
\cancel \delta \cH = \f{\eps'}{N} \delta \left (\f{\cV_0 \eps' a^2 a'}{4 \pi G N} \right )+ \eps' \f{\cV_0 a^2 N'}{4 \pi G N^3}\delta a -\eps''  \f{\cV_0 a^2}{4\pi G N^2} \delta a+\eps ' \delta \Pi\,.
\ee
Unfortunately, there is no way of making them integrable, by a local (in the sense that does not contain integrals over time) field-dependent redefinition of the parameter $\eps$. 

It has been recently shown that the CVH kinematical generators are associated with volume-preserving diffeomorphisms \cite{Achour:2021lqq} and the corresponding Kodama charges. The relationship however holds only for some specific choices of clocks, to whom belongs the proper time in the flat case. This means that the M\"obius symmetry that lives on top of the time reparametrization is associated with a finite set of non-local charges, depending on the history of the system by the factor $\int N\de t$, however, these become local for a particular choice of clock.

Nonetheless, the physical interpretation of these charges and the corresponding clock is still quite obscure. This is because the simplicity of the cosmological models happens to blur the role of the infrared regulator and the respective boundary. For this, we will study the presence of the conformal symmetry in a more general setup.

\section{Field space approach and geodesics}
\label{sec2.2:mink&minisup}
We have shown how the CVH algebra has been introduced in the first place. Using the spacetime interpretation of the CVH generator we could wonder if the presence of the conformal structure in cosmology is somehow a reduction of a more general property. It is actually well known for a while that the volume, extrinsic curvature, and kinetic term of the Hamiltonian constraint form a closed algebra, even in the inhomogeneous setting. This plays a crucial role in the regularization procedure in Loop Quantum Gravity\footnote{In Loop quantum gravity we use a different set of canonical variables, namely the connection-triad variables, and the CVH algebra is provided by the smeared triad determinant, complexifier and Lorentzian part of the constraint. The two sets of generators are related by a canonical transformation} \cite{Thiemann:2007pyv}. See appendix \ref{app:4_LQG_phase_space} for details. Consider the following functionals on the ADM phase space \eqref{ADM_phase_space}:
\bsub \be
H_\rm{Kin} &= \kappa^2 \int_\Sigma \de^3 x\, \f{1}{\sqrt{q}}\left (\pi^{ab}\pi_{ab}-\f{\pi^2}{2}  \right ) \,,\\  
C &= -\f{1}{2}\int_\Sigma \de^3 x\, \pi\,,\\
V &= \int_\Sigma \de^3 x\, \sqrt{q}\,.
\ee \esub
They form a closed algebra, dropping the boundary terms:
\be
\{V,H_\rm{Kin}\} = \kappa^2 C\,,\q\q \{C,H_\rm{Kin}\} = -\f{1}{2} H_\rm{Kin} \,,\q\q \{C,V\} = \f{1}{2}V\,,
\label{CVH_ADM_kinem}
\ee
that is of course isomorphic to the CVH algebra in cosmology. Unfortunately, the introduction of the three dimensional Ricci curvature, playing the role of potential in the Hamiltonian constraint, spoils completely the algebra, making it impossible to relate it to the symmetry of the system \cite{Geiller:2022baq}. The kinematical algebra \eqref{CVH_ADM_kinem} of course still holds, but it does not contain the physical Hamiltonian, and it is not related to a dynamical symmetry.  Hopefully, the minisuperspace reduction provides an easier playground to test how and when it is however possible to make a contact between the kinematical algebra and a true dynamical symmetry. 

We have already discussed in the first chapter, once the symmetry reduction is well performed by fixing the shift and choosing an appropriate internal metric, the ADM action takes the form of a mechanical action describing the motion of a particle moving in a $n$-dimensional curved spacetime, under a potential $\cU$, where $n$ is the number of independent fields in the internal metric $\gamma_{ij}$ and the potential is proportional to the three-dimensional curvature. The Lagrangian has the form
\be
\label{minisup_mech_lagr}
\cL =\f{1}{2N} \tilde g_{\mu\nu} \dot q^\mu \dot q ^\nu -N \cU(q)\,.
\ee
A remark should be done at this point. The tensor $\tilde g_{\mu\nu}(q)$ is an invertible metric on the field space, and the $q$'s are coordinates on it. We will denote by $\de \tilde s^2_\rm{mini}$ the corresponding line element, to whom we refer hereafter as \textit{supermetric}. These should not be confused with the original spacetime metric and coordinates\footnote{We sill use the greek letters to denote the coordinates on the field space. To avoid confusions we have used the first part of the greek alphabet for the spacetime indices. To resume we have many kinds of indices: 
\begin{itemize}
\item The $D+1$ spacetime coordinates $x^\alpha$, $q_{\alpha\beta}$, where $\alpha,\beta,\gamma \dots = 0 \dots D$
\item The coordinates and the metric on the hypersurface $x^a$, $q_{ab}$, where $a,b= 1\dots D$
\item The internal metric and index for the frame $i,j, k \dots = 1\dots D$
\item The field space coordinates and metric $q^\mu$, $g_{\mu\nu}$. Its dimension varies on a case by case basis, but we will concentrate in the following on the case where it is two dimensional.
\end{itemize}
To them, we will add in this section an index for the Killing vectors, still denoted by the latin letters, but between parentheses: $(i)$, $(j)$,...
}. The two are related awkwardly by the fact that the vector constraint must be satisfied, and of course once the minisuperspace model is chosen, supermetric and potential are uniquely determined. In the following, we will focus on the case in which the field space is two dimensional, which is the relevant one for the majority of the gravitational minisuperspaces.

As in cosmology, the lapse function $N$ is a remnant of the diffeomorphism invariance of general relativity as it assures that the model is invariant under time reparametrization. It still plays the role of a Lagrangian multiplier enforcing the Hamiltonian constraint, which takes here the form  
\be
\cH = \f{1}{2} \tilde g^{\mu \nu} p_\mu p_\nu + \cU\,\approx 0\,, 
\ee
on the phase space $\{q^\mu, p_\nu\} =\delta_\nu^\mu$. The matrix $\tilde{g}^{\mu\nu}$ gives the inverse of the supermetric. Classically we can freely choose the lapse as any phase space function and this will give equivalent descriptions of the dynamics, provided that we remember to impose the constraint on the energy value. The Lagrangian \eqref{minisup_mech_lagr} is indeed invariant under the symmetry
\be
\delta_f q^\mu = f \f{\de q^\mu}{\de t} \,,\q \delta_f N = f \f{\de N}{\de t} + \f{\de f}{\de t} N\,, 
\ee
for a generic function $f$, that could also be field dependent. It is precisely by playing on this freedom that we make the extended conformal symmetry appear. We will see that the symmetry will depend on a particular choice of clock (\textit{i.e.} time variable), for which the charges will have a local form, the whole game of unravelling the symmetry is then to find the smart lapse choice.

We will take here the approach of taking the lapse as a field space function $N(q)$, which is seen as a conformal factor in front of the field space metric. 
We define again the \textit{gauge-fixed} time as $\tau$, even if the field space dependence of the lapse make it to differ from the usual definition of proper-time. From now on, we will denote by a dot the derivative with respect to $\tau$. And we will use the metric with the conformal factor to raise and lower the indices.
\be
g_{\mu\nu} := \f{1}{N} \tilde g_{\mu\nu}\,,\q p^\mu =g^{\mu\nu} p_\nu = \dot q^\mu\,, \q \xi_\mu = g_{\mu\nu} \xi^\mu\, \dots \,,
\ee
for any vector on the field space $\xi^\mu$. The evolution of a phase space observable $\cO$ with respect to $\tau$ is given by
\be
\dot \cO:=\f{\de \cO}{\de \tau} = \partial_\tau \cO + \{\O, H\}\,,\q\q H:= N \cH = \f{p^\mu p_\mu}{2} +N\cU\,.  
\ee

It is now time to turn to the study of the symmetries of the motion associated to the Lagrangian \eqref{minisup_mech_lagr}. For this let us begin with the \textit{free} case, i.e. without potential. A straightforward computation of the equation of motion in the free case gives the geodesic equation
\be
\ddot q^{\,\mu} + \Gamma^\mu_{\nu \rho} \dot q^\nu \dot q^\rho =0\,, 
\ee
where the $\Gamma$'s are the Christoffel symbols of the field space metric $g_{\mu\nu}$. The scalar constraint $H =p^\mu p_\mu/2 \approx 0$ restrict ourself to the null geodesics. In this case the evolution of phase space functions along any trajectory is given by
\be
\f{\de \cO}{\de \tau} = \partial_\tau \cO + p^\mu \nabla_\mu \cO\,, \q\q p^\mu \nabla_\mu p^\nu =0\,.
\ee

\subsection{Killing vectors and conserved quantities}

In the free case, the most intuitive way of searching for dynamical symmetries, is to look at the symmetries of the supermetric. We recall that, if the potential is zero, we are describing geodesic motion over the field space, and Killing vectors of the metric are associated with conserved quantities along the geodesics. We are going to describe a slightly more general structure,  for this, we consider the \textit{conformal} Killing vectors of $\tilde g_{\mu\nu}$. These are such that
\be 
\pounds_\xi \tilde g_{\mu\nu} = \tilde \nabla_\mu \xi_\nu + \tilde \nabla_\nu \xi_\mu \doteq \tilde \lambda \tilde g_{\mu\nu}\q\Rightarrow\q   \pounds_\xi \tilde g^{\mu\nu} =-\tilde \lambda \tilde g^{\mu\nu} \,,
\ee
where $\tilde \nabla$ is the covariant  derivative with respect to $\tilde g$. Whenever $\tilde \lambda$ is a constant and non zero we refer to the corresponding vector as \textit{homothetic Killing vector}. If it vanishes, we will call the vector simply \textit{Killing vector}. We see that the conformal factor $N$ do not change the fact that a vector is conformal but it changes the value of $\lambda$:
\be
\pounds_\xi g_{\mu\nu} =  \nabla_\mu \xi_\nu +  \nabla_\nu \xi_\mu =  \pounds_\xi \left (\f{1}{N} \tilde g_{\mu\nu}\right ) = \left (\tilde \lambda - \xi^\mu \partial_\mu (\log N)  \right ) \left (\f{1}{N} \tilde  g_{\mu\nu}\right ) := \lambda {g}_{\mu\nu}\,.
\label{HKV_g_def}
\ee
So that homothetic Killing vectors can be mapped to true Killing vectors or just conformal vectors with a non-constant $\lambda$, and the other way round. This is a first hint about the importance of the lapse in this analysis. 

In the following, we imagine having fixed the lapse and considering the homothetic Killing vectors of $g$, such that $\lambda$ is constant. The whole set of homothetic and true Killing vectors $\left \{ \xi_{(i)} \right\}$ form an algebra
\be
[\xi_{(i)}, \xi_{(j)}] = \left (\xi_{(i)}^\mu \nabla_\mu \xi_{(j)}^\nu  -\xi_{(j)}^\mu \nabla_\mu \xi_{(i)}^\nu \right )\partial_\nu = c_{ij}^{\;\;k} \xi_{(k)}\,,
\ee
where the coefficient $c_{ij}^{\;\;k}$ are the structure constants. It turns out that they are also solutions of the geodesic deviation equation (see \ref{app:5_HKV_prop} for the proof)
\be
p^\mu p^\nu \nabla_\mu \nabla_\nu \xi_\rho = - R_{\rho\mu\sigma\nu} p^\mu p^\nu \xi^\sigma \,.
\ee
With these vectors at hand we can consider the following phase space functions:
\be
C_{(i)} : = p_\mu \xi_{(i)}^\mu\,.
\ee
By construction we have that
\be
\{C_{(i)},H\} = p^\nu \nabla_\nu \left ( p^\mu \xi_{(i)\,\mu} \right ) = \f{\lambda_{(i)}}{2}  g_{\mu \nu} p^\mu p^\nu =\lambda_{(i)} H\,,
\ee
so that the following quantities are conserved
\be
\cQ_{(i)} =  C_{(i)} - \tau \lambda_{(i)} H\,,\q\q \dot \cQ_{(i)} =0\,, 
\label{linear_charges}
\ee
and thus they generate infinitesimal symmetries on the action corresponding to \eqref{minisup_mech_lagr}, given by the transformations on the phase space
\bsub
\be
\delta x^\mu = \{x^\mu, \cQ_i\} &= \xi_{(i)}^\mu - \lambda_{(i)}  \tau\, p^\mu = \xi_{(i)}^\mu - \lambda_{(i)}  \tau\, \dot x^\mu \,,\\
\delta p_\mu = \{p_\mu, \cQ_i\} &= -p_\nu \partial_\mu \xi_{(i)}^\nu+\f{\lambda_{(i)}  \tau}{2}\, p_\nu p_\rho \partial_\mu g^{\nu\rho}\,.
\ee
\esub
The interpretation is as follows: let us take a geodesic, with an affine parameterization. The transformation moves the point $x^\mu$ to a nearby geodesic, recalling that the conformal killing vectors are solutions of the geodesic deviation equation, and then along the new geodesic by $- \lambda_i t\, p^\mu$, that account for an eventual dilation of the spacetime generated by $\xi$ to obtain again an affine parameterization of the new solution. The algebra of the charges reproduces the Lie algebra of the vectors
\be
\{\cQ_{(i)},\cQ_{(j)}\} = -c_{ij}^{\;\;k} \cQ_{(k)}\,.
\ee
Once again we would like to stress the fact that the size of the algebra, or even its existence, depends on the choice of the conformal factor $N$, as the latter determines the number of homothetic killing vectors.

However, as already discussed in the case of cosmology, these charges happen to be non-local once we want to restore the time parametrization invariance, in the sense that they are still conserved but they depend on the history of the system through $\tau$.

Inspired by the symmetry that has been exhibited for the cosmological minisuperspace, we would like to extend the construction to charges that are quadratic in time. Unfortunately, to date, we are not able to find a straightforward generalisation, valid for any dimension of the field space. Only in the simple (and yet very relevant for gravitation) example of two-dimensional field space, we can say something more about the quadratic charges.

\subsection{Inclusion of a potential}

Before moving to the two-dimensional case we shall discuss what happens in the presence of a non-vanishing potential. Here the lapse comes again into the game. Let's use the following notation
\be 
N(q)\cU(q) = U(q) + U_0\,,
\label{potential splitting}
\ee
where $U_0$ captures the constant part of the potential (i.e. $\partial_q U_0=0$) and $U$ is the \textit{effective} potential, this separation of course depends on the lapse. The constant part doesn't play any fundamental role from the point of view of the algebra. If $U=0$, all the discussion above is still valid, except for the fact that we shift from null to massive or tachyonic trajectories. The constant piece provides indeed a boundary term for the corresponding action. We just need to replace the Hamiltonian in the algebra with the charge \footnote{This is exactly what is going to happen for the black hole algebra \cite{Geiller:2020xze}, as it is going to be discussed later in section \ref{sec2.3:BH_minisup}}
\be
\cQ_0 := H -U_0\,,\q\q \{\cO,H\} =\{\cO,\cQ_0\}\,, 
\ee
for any phase space function $\cO$. On the other hand, whenever there is a non-constant potential we should add a condition on the killing vectors
\be 
\pounds_\xi U = \xi^\mu\partial_\mu U \doteq -\lambda U\,.
\label{potential_condition}
\ee
In this case, we still have
\be
\{C_{(i)}, \cQ_0\} = \{C_{(i)}, \f{p^\mu p_\mu}{2} +  U \} = \f{\lambda_{(i)}}{2} p^\mu p_\mu - \xi^\mu \partial_\mu U = \lambda_{(i)} \cQ_0\,.  
\ee
So that again the quantities 
\be
\cQ_{(i)} = C_{(i)} -\lambda_{(i)} \tau \cQ_0 
\ee
are conserved. This set the stage for a quite puzzling quest: given a minisuperspace model we would like to find the lapse that gives the maximal algebra, provided that almost every step in the construction depends on the choice of the conformal factor. The number of homothetic vectors indeed depends on $N$, as well as the decomposition of the potential into its constant and no constant part, so finding a general solution to \eqref{potential_condition} is cumbersome.

Another method to deal with a non-trivial potential would be the  Eisenhart lift \cite{Cariglia:2015bla}, that consists in adding one dimension (named $y$) to the field space in order to map  the $n$-dimensional problem with potential to a $n+1$-dimensional free system. The Eisenhart lift is explicitly given by a new enhanced supermetric $G_{AB}$, where
\be
G_{\mu\nu} =  g_{\mu\nu}\,,\q G_{\mu, n+1} = 0\,,\q G_{n+1,n+1} = \df{1}{2 U}\,,\q\q x^A = (q^\mu, y)\,.
\ee
The null geodesic equation of $G_{AB}$ on the configuration space $(q^\mu, y)$ gives exactly the same equation of $H = p^\mu p_\mu + U$, as can be easily verified by computing the Hamiltonian equations for $H_\rm{Lift}= G^{AB} p_A p_B$ with $p_A =(p_\mu,p_y)$. In  particular the equation of motion for $p_y$ states that it is a constant. Once we pull back this into the equations for $q^\mu$ and $p_\mu$, these become equivalent to the original equations for the $n$-dimensional problem. 

It is possible to show that searching for the homothetic Killing vectors of the supermetric of the enlarged space is the same as finding a simultaneous solution of the rescaling equations \eqref{potential_condition} and \eqref{HKV_g_def} for the supermetric and potential, $\pounds_\xi g_{\mu\nu}  ={\lambda} g_{\mu\nu}$ and  $\pounds_\xi U  =-{\lambda} U$. Although this apparently look equivalent, it allows for vector fields $\xi$ depend on the new coordinate $y$ and have non-trivial components along $\partial_{y}$. Unfortunately, this new coordinate does not have a physical meaning and there is no way to extract non-trivial Dirac observable  from such a homothetic Killing vector on the Eisenhart-lift metric, so that this does not come into help for the systems we want to consider in the following.

\subsection{Two-dimensional field space geometries}
In the vast majority of the general relativity minisuperspaces (Bianchi cosmologies, black holes, etc..), because of the vector constraint, we must reduce the number of independent fields. At the end of the day, we will deal with two-dimensional dynamical space and a two-dimensional supermetric $g_{\mu\nu}$. What makes this case special is the fact that in two dimensions every metric is conformally flat, such that it exists a set of coordinates (on the field space) on which the metric takes the form \footnote{This is for a Lorentzian signature of $g_{\mu\nu}$, that is the relevant case for all the examples treated here}
\be
\de s^2_\rm{mini} = -2 e^{-\varphi} \de \tilde u \de \tilde v\,, \q\q  \cL = \f{1}{2} g_{\mu\nu} \dot q^\mu \dot q^\nu =  - e^{-\varphi} \dot{ \tilde u }\dot {\tilde  v}\,. 
\label{conformal_2d}
\ee
We focus for the moment on the free case, all the information is actually carried by the conformal factor (i.e. the lapse after an appropriate coordinate transformation). The structure of the homothetic killing vectors depends on the scalar curvature
\be
R= -2e^{\varphi} \partial_{\tilde u} \partial_{\tilde v} \varphi\,.
\ee
Whenever $\varphi$ is separable into functions depending only on one of the null directions, the field space is flat, and the conformal structure is inherited from the 1+1 Minkowski space\footnote{If $\varphi(u,v) = f(u) + g(v)$ we obtain the Minkowski space by the transformation that integrates
\[ \begin{pmatrix}
e^{-f(\tilde u)} \de \tilde u \\
e^{-g(\tilde v)} \de \tilde v 
\end{pmatrix} =
\begin{pmatrix}
\de u \\
\de v 
\end{pmatrix}
\]}
\be
\de s^2_\rm{Mink} = -2 \de u \de v\,, \label{Mink_de_s}
\ee
that has three independent killing vectors (two translation and one boost) and one homothetic vector, corresponding to the dilation of the plane. In null coordinates they have the form
\bsub
\be
\rm{time\, translation}&	& \xi_{(\rm t)}^\mu =&\{1,1\}			& \lambda_{\rm{t}}=&0 
&C_\rm{t} &=p_u + p_v\,,  \\
\rm{space\, translation}&	& \xi_{(\rm x)}^\mu =&\{1,-1\}			& \lambda_{\rm{x}}=&0
&C_\rm{x} &=p_u - p_v\,,  \\
\rm{boost}&					& \xi_{(\rm b)}^\mu =&\{u,-v\}			& \lambda_{\rm{b}}=&0
&C_\rm{b} &=u p_u -v  p_v\,,  \\
\rm{dilation}&				& \xi_{(\rm d)}^\mu =&\f{1}{2}\{u,v\}	& \lambda_{\rm{d}}=&1
&C_\rm{d} &=\f{u p_u + v p_v}{2}\,.
\ee
\label{HKV_mink}
\esub 
On the other hand if $R\neq 0$ the number of independent homothetic killing vectors varies on a case by case basis, e.g.
\begin{itemize}	
\item (A)dS$_2$ $\Rightarrow$ $\varphi = 2 \log(u\mp v)$, there are only 3 killing vector,
\be
\xi^\mu\partial_\mu =\left( \mathfrak{c}_1 u^2 + \mathfrak{c}_2 u +\mathfrak{c}_3\right )\partial_u + \left (\pm \mathfrak{c}_1 v^2 + \mathfrak{c}_2 v \pm \mathfrak{c}_3 \right )\partial_v \,,\q \pounds_\xi g_{\mu\nu} = 0\,,\q \mathfrak{c}_i = const\,,
\ee
the upper sign is for AdS, the lower one is for dS.
\item $\varphi = u^2 v$, there is only 1 homothetic killing vector,
\be
\xi^\mu\partial_\mu = \lambda (- u\partial_u + 2v \partial_v) \,,\q\q \pounds_\xi g_{\mu\nu} = \lambda g_{\mu\nu} \,.
\ee
\item $\varphi = u v$, there is only 1 killing vector,
\be
\xi^\mu\partial_\mu = - u \partial_u+  v\partial_v \,,\q\q \pounds_\xi g_{\mu\nu} = 0\,.
\ee
\end{itemize}
There are also examples without any killing vector. Let us concentrate on the cases where the conformal factor gives a flat space. In this situation, we can easily extend the algebra and find the quadratic charges. Indeed it turns out that when the Riemann tensor vanishes, the third time derivative of the following quantities is zero (\ref{app:5_HKV_prop}):
\be
V_{ij} := g_{\mu\nu} \xi_{(i)}^\mu \xi_{(j)}^\nu\,, \q\q R=0\; \Rightarrow\;\dddot{V}_{ij} =0\,.
\ee
The notation $V_{ij}$ and $C_{(i)}$ is in order to make contact with the kinematica CVH algebra, as we will illustrate at the end of this chapter. We recall that in the two-dimensional case the vanishing of the Ricci scalar and of the Riemann tensor are equivalent. With the Killing vectors \eqref{HKV_mink} at hand we can explicitly calculate the scalar products $V_{ij}$ and find the non trivial combinations
\be
\begin{array}{rcl}
V_{\rm{bt}}&=-2 V_{\rm{dx}}=&v-u\,,\\
V_{\rm{bx}}&=-2 V_{\rm{dt}}=&u+v\,,\\
V_{\rm{bb}}&=-4 V_{\rm{dd}}=&2 u v\,.\\
\end{array}
\ee
Together with the central element $V_{\rm{xx}}=-V_{\rm{tt}}=2$. By construction the second derivatives of the $V_{ij}$ is a constant and in particular the following charges are conserved for a constant potential:
\bsub
\be
\cQ_\rm{bt} &:=V_\rm{bt} + \tau\, C_\rm{x}\,\\
\cQ_\rm{bx} &:=V_\rm{bx} + \tau\, C_\rm{t}\,, 	\\
\cQ_\rm{bb} &:=V_\rm{bb} + 4 \tau\, C_\rm{d} -2 \tau^2\, H\,.
\ee
\esub
Together with the charges associated with the $C$'s functions, presented in the previous section, they form an algebra. This can be checked using the fact that the brackets with the Hamiltonian are $\lb C_{(i)},\cQ_0\rb=\lambda_{(i)}\cQ_0$ and
\be\label{Vij Q0 algebra}
\lb V_1,\cQ_0\rb=-C_\rm{x}\,,\q\q\lb V_2,\cQ_0\rb=-C_\rm{t}\,,\q\q\lb V_3,\cQ_0\rb=-4C_\rm{d}\,.
\ee
The charges $(\cQ_0,\cQ_\rm{t},\cQ_\rm{x},\cQ_\rm{b},\cQ_\rm{d})$ found previously and $(\cQ_\rm{bb},\cQ_\rm{bx},\cQ_\rm{bt})$ presented here form an 8-dimensional algebra isomorphic to the semi-direct sum $\big(\sl(2,\R) \oplus\R\big)\loplus\mathfrak{h}_2$, where $\mathfrak{h}_2$ is the two-dimensional Heisenberg algebra\footnote{Note that this algebra has four generators, which are analogous to the two positions and two momenta of a two-dimensional space obeying the Heisenberg commutation relations, the central element is provided by the identity on the phase space $V_{\rm{xx}}=-V_{\rm{tt}}=2$}. For this it is convenient to introduce the linear combinations
\bsub
\be
\label{sl_2+R} \cL_{-1} = -\cQ_0\,,\q\q \cL_{0} = \cQ_\rm{d}\,, \q\q \cL_{1} = \f{1}{2} \cQ_\rm{bb}\,, 
 \q\q \cD_0 = \cQ_\rm{b}\,,\ee
 which span $\sl(2,\R) \oplus \R $, as well as the four generators
 \be
\label{h_2} \cS_{1/2}^\pm = \f{\cQ_\rm{bx} \pm \cQ_\rm{bt}}{2}\,, \q\q \cS_{-1/2}^\pm = \f{\cQ_\rm{t} \pm \cQ_\rm{x}}{2}\,,
\ee
\label{charges}
\esub
which span $\h_2$. These definition can be repackaged in the more compact form
\bsub
\be
\cL_n 			&= \f{\ddot f}{4} V_\rm{bb} + \dot f C_{\rm d} - f \cQ_0  	
			&f(\tau) &= \tau^{n+1} 		& n & \in \{-1,0,1\}\,, \label{sl_nform}\\
\cS^\pm_s		&= \f{(\dot f V_\rm{bx} + f C_\rm{t}) \pm(\dot f V_\rm{bt} + f C_\rm{x})}{2}  					  	 						&f(\tau) &= \tau^{s+1/2} 	& s & \in \{-1/2,1/2\}\,.
\ee
\esub
They form the algebra
\bsub
\be
\{\cL_n,\cL_m\} &= (n-m) \cL_{n+m}\,,\\
\{\cL_n,\cD_0\}&=0\,,\\
\{\cS_s^\eta,\cS_{s'}^\epsilon\} &= 2 n \delta_{\epsilon+\eta} \delta_{s+s'}\,,\\
\{\cL_n,\cS_s^\pm\} &= \left (\f{n}{2}-s\right ) \cS_{n+s}^\pm\,,\\
\{\cD_0,\cS_s^\pm\} &= \pm \cS_{s}^\pm\,.
\ee
\label{charge_algebra}
\esub
The first line define the $\sl(2,\R)$ algebra, commuting with $\cD_0$. The third line is the Heisenberg algebra with central element  $V_{\rm{xx}}=-V_{\rm{tt}}=2$, and the last two lines show how the semidirect sum is built. The charges \eqref{charges} measure the initial condition of the corresponding phase space functions, i.e $\cQ_{ij}(\tau=0) = V_{ij}$ and $\cQ_{i}(\tau=0) = C_{i}$ meaning that the two sets have the same algebra, exactly as for the CVH kinematical algebra and the Noether charges \eqref{noether_charges_FLRW} in homogeneous cosmology.

We stress here that the existence of the symmetry and the charges in a local form is connected with the lapse. Choosing the lapse in a way that the supermetric is flat makes indeed the quantities \eqref{charges} to be conserved, where we read $\tau$ as the gauge fixed time. This is always possible in the free case, and there is actually an infinite set of choices of clocks for which $N$ satisfy the flatness condition.  On the other hand, if $N$ is such that the supermetric is curved, the charges \eqref{charges} still exist but can be only partially recast in a local form, depending now on the conformal structure of the supermetric.

This procedure, which we have outlined here in the two-dimensional case, can be extended to a field space of arbitrary dimension, although it is of course not guaranteed to then lead to the same results. If the field space is non-flat and higher-dimensional we are indeed not guaranteed that homothetic Killing vectors exist, nor that the observables $C_{(i)}$ and $V_{(ij)}$, if they exist, form a closed algebra with the Hamiltonian. We give in appendix \ref{app:5_HKV_prop} the equation \eqref{condition on xi's} which has to be satisfied by the Killing vectors in order for this to be the case. Although this equation has a simple form, one cannot find an a priori condition on the field space geometry (and on the potential if this latter is non-trivial) which guarantees the existence of an algebra on phase space. With the construction presented here, we have however the tools to investigate this on a case by case basis.


\paragraph{Finite transformations}

Looking at the conserved quantities as Noether charges, we can integrate the infinitesimal action of the 
$\cQ$ on the phase space to reconstruct the corresponding symmetries. This is the inverse construction to the one presented for cosmology, therein we have started from the symmetry to obtain the charges, here we already have the conserved quantities at hand and we try to reconstruct the symmetries. The finite transformation on the fields $(u,v)$ is obtained by exponentiating the action of the charges, given by the Poisson bracket with the generator, replacing then the canonical momenta in terms of the first derivatives of the fields. For example $\tilde u = e^{\{\bullet,\cQ\}} \triangleright u$.

The first part \eqref{sl_2+R} gives the scaled M\"obius transformations, represented by the $\R \times \SL (2, \R)$ group acting on the null coordinates as
\be
\label{mobius_sym} 
\begin{pmatrix}
u\\
v\\
\end{pmatrix}\mapsto \dot f^{1/2}
\begin{pmatrix}
k\, u(\tau)\\
v(\tau)/k\\
\end{pmatrix} \circ f^{(-1)}(\tau)=  \f{1}{c \tau + d}
\begin{pmatrix}
k\, u(\tau)\\
v(\tau)/k\\
\end{pmatrix}\circ f^{(-1)}(\tau)\;\;\text{where}\;\; f(\tau)=\f{a \tau + b}{c \tau + d}\,.
\ee
For constants coefficients $k$, $a$, $b$, $c$, $d$, such that $ad-bc=1$. The null variables transforms as conformal fields of weight $1/2$ form the point of view of M\"obius transformations. For the Heisenberg algebra \eqref{h_2} the transformations actually exponentiate into the abelian group $\R^4$:
\be 
\label{abelian_symm}
\begin{pmatrix}
u \\
v \\
\end{pmatrix}\mapsto 
\begin{pmatrix}
u + \alpha \tau + \beta \\
v + \gamma \tau + \delta \\
\end{pmatrix}\,.
\ee
A straightforward calculation leads to verify that these are indeed symmetries for the action 
\be
\cS= \int \de \tau \dot u \dot v\,.
\ee
Applying Noether's theorem, as shown for cosmology, provides exactly the charges $\cL_n$, $\cS^\pm_n$, $\cD_0$ as conserved quantities. The Lie algebra of infinitesimal transformations can be written in the compact form:
\bg \label{null_conf_infinit_transf}\left |
\begin{array}{rl}
\delta u = - f \dot u + \df{1}{2}\dot f u + k u+ g(\tau)\\[9pt]
\delta v = - f \dot v + \df{1}{2}\dot f v - k v + h(\tau)\\
\end{array} \right .\,,\\
\big[\delta(f_1, g_1, h_1, k_1 ),\delta(f_2,g_2,h_2,k_2)\big]_{\rm{Lie}} = \delta(f_{12},g_{12},h_{12},k_{12})\,,\notag
\eg
with
\bsub
\be
f_{12}&= \dot f_1 f_2 -(1\leftrightarrow 2)\,,\\
k_{12}&=0\,,\\
g_{12}&= -\f{1}{2} g_1 \dot f_2 + f_2 \dot g_1 -k_1 g_1 -(1\leftrightarrow 2)\,,\\
h_{12}&= -\f{1}{2} h_1 \dot f_2 + f_2 \dot h_1 +k_1 h_1 -(1\leftrightarrow 2)\,.
\ee
\label{Lie_commutator}
\esub
The $\sl(2,\R)$ charges are associated to the function $f$, that is at most a two degree polynomial in time, while the abelian transformation are generated by the linear functions $g$, $h$. This is equivalent to the charge algebra \eqref{charge_algebra}, up to the Heisenberg central extension, that is absent on the Lie algebra level, and is replaced by the commuting functions $g$ and $h$. Each $\cL_n$ generate the $\tau^{n+1}$ coefficient in the function $f$, while $\cS^-_n$ generates the $\tau^{n+1/2}$ coefficient of $h$, and $\cS^+_n$ generates the $\tau^{n+1/2}$ coefficient of $g$. The symmetry group is isomorphic to $\left (\SL(2,\R) \times \R\right ) \ltimes \R^4$, where the semidirect product is built according to the last two lines of \eqref{Lie_commutator}, that shows how the $\left (\SL(2,\R) \times \R\right )$ generators $(f(\tau),k)$ act on $\R^4$.

\paragraph{Size of the algebra}
The previous construction consists in translating the symmetry properties of the supermetric, encoded into the Killing vectors, into a dynamical symmetry. To them there correspond a set of conserved charges, whose number is equal to the dimension of the symmetry group. At the same time, we have only a four-dimensional phase space $\{u,v,p_u,p_v\}$, so that not all the charges are independent. We have indeed the following identities:
\bsub
\be
\cL_n &= \sum_{s=-1/2}^{1/2} s \left (\f{n}{2}+s\right ) (\cS^+_s \cS^-_{n-s}+\cS^-_s \cS^+_{n-s})\,,\\
\cD_0 &=  \sum_{s=-1/2}^{1/2} 2s\,\cS^-_s \cS^+_{-s}\,.
\ee
\label{identities}
\esub
The previous identities could also be interpreted differently: the fact that all the generators are quadratic or linear in the Heisenberg algebra charges makes the four $\cS^\pm_n$ as the fundamental building blocks for the symmetry transformations. They are conserved and provide the four necessary integration constants. Any (non-linear) combination of them will be trivially conserved on the trajectory and will generate a symmetry transformation by calculating its Poisson bracket.

The presence of the Heisenberg algebra plays also a crucial role in the quantization. Promoting the Poisson bracket to the commutator in the usual quantization procedure will provide an extremely easy way of quantizing the system, such that the symmetry is faithfully represented. 

Furthermore we can consider two other set of transformations that are quadratic in the $\cS^\pm_n$. We define them as
\be
\label{trans_gen_null}
\cT^\pm_n = \sum_{s=-1/2}^{1/2} s \left (\f{n}{2}+s\right ) \cS^\pm_s \cS^\pm_{n-s}\,,\q n\in \{0,\pm 1\}\,.
\ee
They form an algebra with the other quadratic\footnote{Here \textit{quadratic} is referred to the power of $\cS$, they can be both linear or quadratic in time} charges:
\bsub
\be
\{\cL_n,\cT_m^\pm\} &= (n-m) \cT^\pm_{n+m}\,,\\
\{\cD_0,\cT_n^\pm\}&= \pm \cT_n^\pm\,,\\
\{\cT_n^+,\cT_n^-\} &= \f{n-m}{2} \cL_{n+m} -\f{m^2-mn+n^2-1}{4} \cD_{0}\,,\\
\{\cT_n^+,\cT_n^+\} &= 0=\{\cT_n^-,\cT_n^-\}\,.
\ee
\label{charge_algebra_2}
\esub
In particular the subset of $\cL$ with one of the two sets $\cT^+$ or $\cT^-$ form a closed subalgebra, isomorphic to the Poincar\'e algebra in 2+1 dimensions. hereafter denoted by $\ISO(2,1) = \SL(2,\R) \ltimes R^3$. Exponentiating the Poisson bracket we can read the corresponding finite symmetry, these are given by two sets of coupled abelian transformations
\be 
\label{abelian_symm_quadratic}
\begin{pmatrix}
u \\
v \\
\end{pmatrix}\mapsto 
\begin{pmatrix}
u + 2 g \dot v - v \dot g  \\
v + 2 h \dot u - u \dot h  \\
\end{pmatrix}\,,
\ee
for second-degree polynomials in time $g$ and $h$.

\paragraph{Inclusion of a potential}

With the explicit expression for the killing vectors at hand, we can dig here a little further the discussion about the choice of the lapse to get the maximal algebra. Let us imagine that we have managed to write the supermetric of a certain minisuperspace in the null conformally flat gauge \eqref{conformal_2d}. In particular let us set\footnote{Here the lapse is left explicit and only the $\tilde g$ metric is recast into the conformal null form, on the contrary to the free case above \eqref{conformal_2d}}
\be
\de s^2 &= \f{1}{N(q)} \tilde g_{\mu\nu} \de q ^\mu \de q^\nu = \f{1}{N[q(\tilde u,\tilde v)]} e^{-\varphi} \de \tilde u \de \tilde v \,,\label{conf_2d_potential}\\
\cL &=  \f{1}{2N} e^{-\varphi} \dot {\tilde u} \dot{ \tilde v} +N \cU\,.\notag
\ee
We can follow two different directions:
\begin{enumerate}
\item  We set $N=1/\cU$. As discussed previously, the potential does not play any role on the dynamical level, except for the fact that restricts the solution to represent trajectories of massive (or tachyonic) particles, and we can discuss the symmetries as in the free case, but now the lapse is fixed. 
\begin{enumerate}
\item If  $\cU e^{-\varphi}$ has the form of the product of functions depending separately on the null directions, we inherit the symmetries from the conformal structure of the flat 2d Minkowski spacetime.
\item If  $\cU e^{-\varphi}$ do not satisfy the flatness condition we must evaluate the existence and the size of the symmetry case by case. If we are in this situation is better to proceed in the other direction \textit{(ii)}.
\end{enumerate}
\item  We choose $N$ such that the supermetric is flat, and the non-constant part of the potential does not spoil the existence of algebra, satisfying \eqref{potential_condition}. 
\end{enumerate} 
The application of this procedure to black holes will be the subject of the next section, we spend there also some words about the Bianchi models. In general, we will see that in case \textit{(ii)} it is not possible to preserve the whole set of Killing vectors, but the case  \textit{(i.b)} is even worse, being without any killing vector. In the following we neglect it and focus only on the case \textit{(i.a)} and \textit{(ii)}. 

We should remark that the (conformal) Killing vectors of the supermetric and the potential are not related to the isometries of the original spacetime description,. We will see in the following the charges presented in this chapter generate a set of physical symmetries that changes the value of the integration constant, allowing to move into the solution space of the theory, while the spacetime diffeomorphism simply change the coordinates on the same solution. There is of course a one to one correspondence between the minisuperspace model, with its internal isometries ($e.g.$ spherical symmetries for black holes) and the supermetric, it we can imagine to build a dictionary between the isometry group of the homogeneous slice and the group of dynamical symmetries, although the logic behind it is still unclear.

We will see that we are going to be able to cast directly the supermetric in a form with $\varphi = 0$ so that the whole game is played by the lapse and the potential. In this case, the algorithm that we use is as follows:
\begin{itemize}
\item \textbf{Step 0}. We have managed to rewrite the Lagrangian and the supermetric as \eqref{conf_2d_potential} with $\varphi =0$.
\item \textbf{Step 1}. We look at the form of the potential in terms of $\tilde u$ and $\tilde v$: if it is separable in functions depending on only one direction, we can choose $N=1/\cU$ and by construction, we have that $ \cU \de \tilde u \de \tilde v$ is flat and there is a new change of variables that put it into the Minkowski form \eqref{Mink_de_s}.

We inherit the symmetry structure of the flat 2d free case and using backwards the map from $u,v$ to $\tilde u, \tilde v$ and then to the original configuration space, we find the expression for the local conserved charges. The algorithm stops here.

\item \textbf{Step 2}. The potential does not satisfy the flatness condition. In this case, we choose $N$ such that the supermetric is flat and hopefully we set part of the potential to a constant, such that, as in \eqref{potential splitting}
\be
N \cU = U + U_0\,.
\ee
We can again find the new change of variables that put the supermetric into the Minkowski form \eqref{Mink_de_s},  but we need now to discuss the condition \eqref{potential_condition} on $U$.
\end{itemize}

In particular, in the following, it turns out to be useful to study the case in which $U = u^n v^m$, where $u,v$ are the null Minkowski coordinates, for which the supermetric is \eqref{Mink_de_s} (this can be easily generalised to a general function, that can be expanded in a series of monomials). We have three possibilities:
\begin{enumerate}
\bsub
\item $n=0$, $m\neq 0$. There are two vector satisfying $\pounds_\xi U = \lambda U$  \eqref{potential_condition}:
\be\xi = \xi_\rm{t}+\xi_\rm{x}\,,\q\q \xi = (2+m)\xi_\rm{b}+2m\xi_\rm{d}\,,\ee
the first one is a true Killing vector with $\lambda =0$, while the second one has $\lambda = 2 m$.
\item $m=0$, $n\neq 0$, it is the complementary case to the previous one, so that we have the vectors
\be\xi = \xi_\rm{t}-\xi_\rm{x}\,,\q\q \xi = (2+n)\xi_\rm{b}-2n\xi_\rm{d}\,,\ee
with respectively $\lambda=0$ and $\lambda = -2n$.
\item $n,m\neq 0$, we have only one vector 
\be (2+n+m)\xi_\rm{b}+2(m-n)\xi_\rm{d}\,,\q \lambda= 2(m-n)\,.\ee
\esub
\end{enumerate}
This gives the relation between the potential and the existence of the linear charges \eqref{linear_charges}, on the other hand for the quadratic charges to do not spoil the algebra we need to have closed brackets between the products $\xi_{(i)}^\mu \xi_{(j)\,\mu}$ and the $C_{(i)}$'s. Let us focus on the third case ($n,m\neq 0$) that will be relevant in the following. We have only one homothetic Killing vector, thus 
\bsub
\be
\xi^2 &= (2+n+m)^2 \xi_\rm{b}^2 -4(m-n)^2 \xi_\rm{d}^2 = 4(1+m)(1+n) V_\rm{bb}\,,\\
p_\mu \xi^\mu &= (2+n+m) C_\rm{b} -2(m-n) C_\rm{d}\,, \\
\{p_\mu \xi^\mu, \xi^2\} &= 2(n-m)\xi^2\,,\\
\{p^\mu p_\mu, \xi^2\} &= 16(1+m)(1+n) C_\rm{d}\,,
\ee
\esub
so that the last bracket is proportional to $p_\mu \xi^\mu$, and $\xi^2$ is non-trivial if and only if $2+n+m=0$. In this case the algebra of the charges reduces to its $\sl(2,\R)$ subset. This gives exactly the two dimensional equivalent of the conformal potential $1/q^2$, discussed at the beginning of this chapter. 

For a two dimensional minisuperspace, whenever we have access to the full algebra and in particular to the $\cS^\pm$ generators, we can fully integrate the motion on the phase space. The $\cS^\pm$ are always independent and provide the four necessary constant integration, and their flow under the Hamiltonian is easily obtained in terms of group action. Conversely, if we are left with the $\sl(2,\R)$ sector alone, we can integrate only half of the phase space variables evolution.

It is now time to see how this construction applies to the black hole minisuperspace.

\section{Black holes minisuperspaces}
\label{sec2.3:BH_minisup}

Black holes are one of the most fascinating predictions of General Relativity, but at the same time are inevitably associated with singularities, signals of the breakdown of the theory. It is a common expectation that approaching the classical singularity, in a region where the curvature become Planckian, quantum effects become relevant and that in a complete quantum description of spacetime singularities must be replaced by a unitary evolution through a fuzzy geometry. Hope is that with the increasing experimental ability, by gravitational wave detection \cite{Abbott:2016blz} or black hole imagery \cite{EventHorizonTelescope:2019dse}, black holes will provide an experimental probe for quantum deformations of spacetime. 

Any quantum gravity theory should then provide some information about the fate of black hole singularity. Unfortunately, to date, for any known quantum gravity
theory, this remains an outstanding challenge, there is no straightforward way to determine the physical quantum states representing black holes, even for the simplest case of spherically symmetric geometries.

It is nonetheless possible to incorporate some features of the full theories into reduced symmetry models, and this is why it is worthy to dig a little further into the structure of black hole minisuperspaces. 

In the first chapter \ref{sec1.2:symm_reduc} we have already introduced the model and its mechanical action. We will discuss here the corresponding solution to show that the model describes well the black holes and then we will apply the previous construction to build the symmetry algebra.

\subsection{Classical solution}
We set in a first place the cosmological constant to zero. Exhaustive discussions to how the latter enters the game can be found in \cite{Achour:2021dtj}, and will be briefly reviewed in the next chapter. We recall that the black hole minisuperspace is characterised by the line element
\be
\de s^2_\rm{BH} = -N^2 \de t^2 + a^2 (\de x + N^1 \de t)^2 +L_s^2 b^2\, \de \Omega^2\,,
\label{Bh_minisup_line}
\ee
where the $x$ coordinate is constrained into the finite interval $[0,L_0]$. The corresponding action (see \eqref{KS_action_ab}) takes the minisuperspace form \eqref{minisup_mech_lagr}, where the supermetric and the potential are given by
\be
\de \tilde s^2 = -\f{\cV_0}{G} \left (4 a\, \de b^2 +8 b\, \de a\, \de b\right )\,, \q\q \cU = -2\f{\cV_0}{G} \f{a}{L_s^2}
\ee
and the field space is parametrized by $a$ and $b$. We recall that in order to study the symmetry of the model in the terms presented in the previous section we need to recast the field space in the  null conformal parametrization \eqref{conformal_2d}. This allow to study the potential and see if choosing $N \propto 1/\cU$ gives a flat or curved spacetime.
For this we study the $\de \tilde s$ line element, without lapse, and consider the change of variables
\be
\label{family_1_conf}
\left|
\begin{array}{rl}
\tilde u &= 2  \sqrt{\df {2\cV_0}{3G}}\, a \sqrt{b} \\
\tilde v &= 2 \sqrt{\df{2\cV_0}{3G}}\, b^{3/2} \\
\end{array} \right.\,,\q\q \de \tilde s^2 = -2 \de \tilde u \de \tilde v\,.
\ee
As announced it is flat and $\varphi=0$. However, to discuss the existence of an algebra encoding the dynamics of the minisuperspace and the equivalent realization of the symmetries \eqref{mobius_sym}, \eqref{abelian_symm} and \eqref{abelian_symm_quadratic} we shall write the potential in terms of the conformal null coordinates. It turns out that the latter takes the form of a product of two functions depending only on one null direction:
\be
\cU = -\f{3^{1/3}\cV_0}{L_s^2 G^{2/3}} \tilde{u} \tilde{v}^{-1/3}\,.
\ee
It is then smart to take $N =\df{\tilde{v}^{1/3}}{3^{1/3} G^{1/3}\tilde{u}} = \df{1}{2a}$, so that the potential becomes a constant 
\be 
U_0 = -\f{\cV_0}{G L_s^2} = -\f{L_0}{16\lp^2}\,,
\ee
using the expression for the fiducial volume \eqref{fiducial_volume_KS} and the Planck length $\lp^2=G$. As we have just discussed, the lapse choice will correspond to a flat supermetric, and the system will inherit the whole $\left (\SL(2,\R) \times \R\right) \ltimes \R^4$ symmetry group. We should stress that the choice of lapse is not exactly the inverse of the potential, but it differs form a dimensionful prefactor, because we want to keep $N$ dimensionless, while the Hamiltonian as an inverse length ad $\tau$ as a time.

Let us now rapidly discuss the equation of motion corresponding to the gauge fixed time. We will stick for the moment to the $a,b$ coordinates on the field space. The gauge fixed effective Lagrangian and the corresponding Hamiltonian are given by 
\be
\label{Bh_lagran} 
\cL =  -\f{\cV_0}{G}\left (4 a^2\, \dot b^2 +8 b a\, \dot a\, \dot b\right )\,, \q\q H=\f{G}{\cV_0} \f{p_a(a p_a -2 b p_b)}{16 a b^2}\,.
\ee
The action principle from $\cS_0 = \int \de \tau \cL$ will give the Euler-Lagrange equation \eqref{Eom_BH_minisup}, with fixed lapse. Equivalently we can use the Hamilton equation from the Hamiltonian above, to whom we shall add the energy constraint $H=\df{\cV_0}{G L_s^2} = \df{L_0}{16\lp^2}$. Conversely, we can also study directly the equation of motion in the null conformal gauge. According to what is presented in the previous section we search for the field redefinition that recast the supermetric into the two dimensional  Minkowski metric. This is for example provided by
\be
\label{conf_to_null_bh}
\left|\begin{array}{rl}
u &= \df{3}{4}\sqrt{\df{G}{\cV_0}}  \tilde u^2    =2 \sqrt{\df{\cV_0}{G}} a^2 b	  \\[9pt]
v &= 3^{1/3}\left (\df{\cV_0}{G}\right )^{1/6} \tilde v^{2/3}  =2 \sqrt{\df{\cV_0}{G}} b  \\
\end{array}\right .\,.
\ee
This makes the equation of motion extremely simple. In terms of Minkowski supermetric they are given by tachyonic geodesics with mass $-H=2 \dot u \dot v =-\df{\cV_0}{G L_s^2}$:
\bsub
\be 
u(\tau) &= u_0 - \sigma \sqrt{\f{\cV_0}{2G L_s^2}} (\tau -\tau_0) \,,\\
v(\tau) &= \sqrt{\f{\cV_0}{2G L_s^2 \sigma^2}}  (\tau -\tau_0) \,.
\ee
\label{Bh_sol_conf}
\esub
From this we can infer the solution for the scale factors $a,b$, and the respective momenta
\bsub
\be
a &= \left(-\f{\cV_0^2}{2\cA L_s^2 G^2} + \f{ \cB \cV_0}{\cA (\tau - \tau_0) G}\right )^{1/2} \label{Bh_sol_a}\,,\\
b &=\f{G \sqrt{\cA}}{\sqrt{2}\cV_0} (\tau- \tau_0)\,,\label{Bh_sol_b}
\\
p_a &=-8\f{\cV_0}{G}\left ( b a\, \dot b\right ) 
= -\left (8 \cA (\tau - \tau_0) \f{\cV_0(\tau_0-\tau) + 2\cB L_s^2 G}{L_s^2 \cV_0} \right )^{1/2}\,,\\
p_b &=-8\f{\cV_0}{G}\left ( a^2\, \dot b + b a\, \dot a\right )
= -2\sqrt{2} \cV_0 \f{\cV_0(\tau_0 -\tau) + \cB L_s^2 G}{\sqrt{\cA} L_s^2 (\tau - \tau_0) G^2}\,.
\ee
\label{Bh_sol}
\esub
The energy constraints already imposes one initial condition, the three other integration quantities are $\cA$, $\cB$ (or equivalently $u_0$, $v_0$) and $\tau_0$. The last one simply reflect the residual time translation freedom, while the first can be defined as phase space functions
\be
\cA=\f{\cV_0^2}{L_s^2 G^2 \sigma^2} =\f{p_a^2}{32 a^2 b^2}\,, \q\q \cB=\f{\sqrt{ \cV_0}}{L_s \sqrt{2G} \sigma} u_0 =\f{b p_b -a p_a}{2}\,. 
\ee
Without surprise they both commute with the Hamiltonian $\{\cA,H\} =0=\{\cB,H\}$. Let us insert this solution into the line element \eqref{Bh_minisup_line}, and consider the change of variables on the space time coordinates
\be
\tau -\tau_0 =\f{\sqrt{2} \cV_0}{L_s G \sqrt{\cA}} R\,,\q\q x= \f{L_s G \sqrt{2 \cA}}{\cV_0} T\,,
\label{change_classic}
\ee
we then get the static spherically symmetric black hole metric 
\be
\de s^2_\rm{BH} &= -f(R) \left (\de T +\f{N^1\cV_0^2}{G^2 \cA L_s^2} \de R\right )^2 + f(R)^{-1} \de R^2 + R^2 \de \Omega^2\,,\notag\\
f(R) &=1- \f{2M}{R}\,,\q\q M= \f{\sqrt{\cA} L_s^3 G^2 \cB}{\sqrt{2} \cV_0^2}\,.
\label{classic mass}
\ee

A note should be made here about the range of validity of the model. We have started from the line-element \eqref{Bh_minisup_line}, for which we need real values of the fields $a,b, N, N^1$. This happens only for $\tau \in [\tau_0,\tau_1]$ so that $R$ is timelike, and the solution describes the black hole interior. On the other hand, via the transformation \eqref{conf_to_null_bh}, and the gauge fixing we can extend the range of validity of the model, similar to what happens for the change of coordinates from Schwarzschild to Kruskal charts. Inverting the transformation \eqref{conf_to_null_bh}, gauge fixing the lapse and inserting it into the line element \eqref{Bh_minisup_line}, gives:
\be
\de s^2_\rm{BH} = -\f{v(\tau)}{4u(\tau)} \de \tau^2 + \f{u(\tau)}{v(\tau)} (\de x + N^1 \de \tau)^2 + \f{G L_s^2}{4 \cV_0} v(\tau)^2 \, \de \Omega^2\,.
\label{Bh_minisup_line:null_conf}
\ee
With this parametrization, the $g_{xx}$ coefficient can change its sign. This is translated into the fact that the solution can be extended beyond the horizon of the black hole. Furthermore,  from the point of view of mechanical model describing the evolution in the configuration space spanned by the null fields, nothing special happens when the values of fields are zero. The classical mechanical evolution smoothly crosses not only the point representing the horizon but also the one that represents the singularity. For $\tau<\tau_0$, corresponding to a negative radius, we recover a white hole solution (with a negative effective mass). Extending the time $\tau$ on the whole real line covers an extension of the black hole solution of Einstein's equations, which however differs from the Kruskal–Szekeres one. The subtlety is that this solution is not connected in the sense that there is a hole in the conformal Penrose diagram. The singularity appears to be spacelike approaching from the positive radius and timelike if we arrive from $R<0$ (see figures below), this is synonymous with a naked singularity on the white hole side, that moreover is behind the singularity and not attached to the past asymptotically flat region $\cJ^-$ as in the usual Kruskal–Szekeres coordinates. This also results in an effective negative mass on the white hole side. We will come back to this point many times in this thesis, and we will see how the introduction of a regularisation capturing the quantum effects can erase the discontinuity in the Penrose diagram.

Let us go back to the line element \eqref{classic mass}. It is parametrized with a radial coordinate $R$ and a retarded coordinate $T$. This means that the mechanical model under consideration describes the evolution of the metric coefficients with respect to the coordinate Schwarzschild radius. The evolving parameter $\tau$ runs backwards from the point of view of an infalling observer. The choice of retarded coordinate depends on the shift, for example:
\begin{itemize}
\item $N^1=0$ gives the Schwarzschild coordinates,
\item $N^1=\df{1}{2 a^2} = -\df{v}{2u}$ gives the Eddington-Finkelstein coordinates,
\item $N^1= \left (\df{a p_a - b p_b}{2 a^4 (a p_a -2 b p_b)}\right )^{1/2} =\left (\df{v^2 (p_v v-p_u u)}{4 u^2 p_v v}\right )^{1/2}$ gives the Gullstrand-Painlev\'e coordinates.
\end{itemize}
In the following figures \eqref{fig:coord_level} we draw the lines at constant $x$ (or retarded time $T$) for the three choices of shift.

\vbox{
\begin{center}
\small{(a)}\begin{minipage}{0.4\textwidth}
		\includegraphics[width=\textwidth]{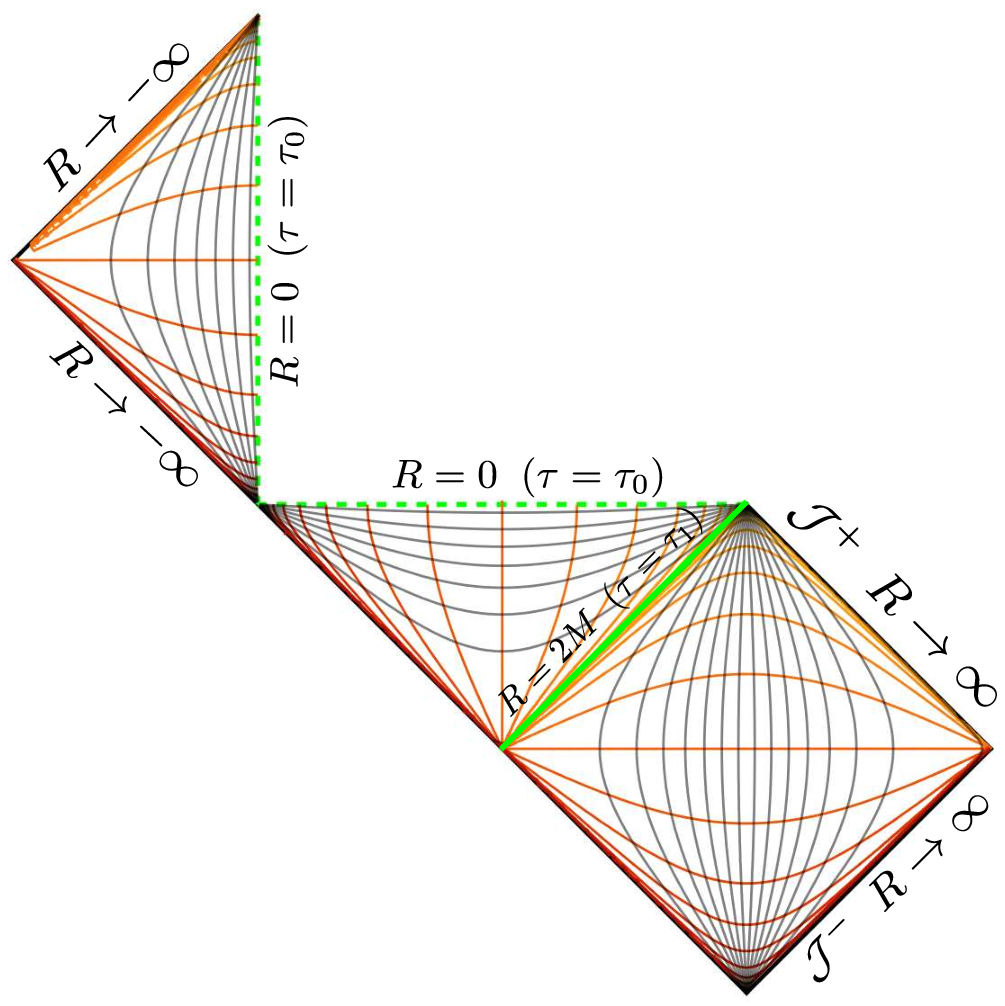}
\end{minipage}
~\hspace{0.5cm}
\small{(b)}
\begin{minipage}{0.4\textwidth}
		\includegraphics[width=\textwidth]{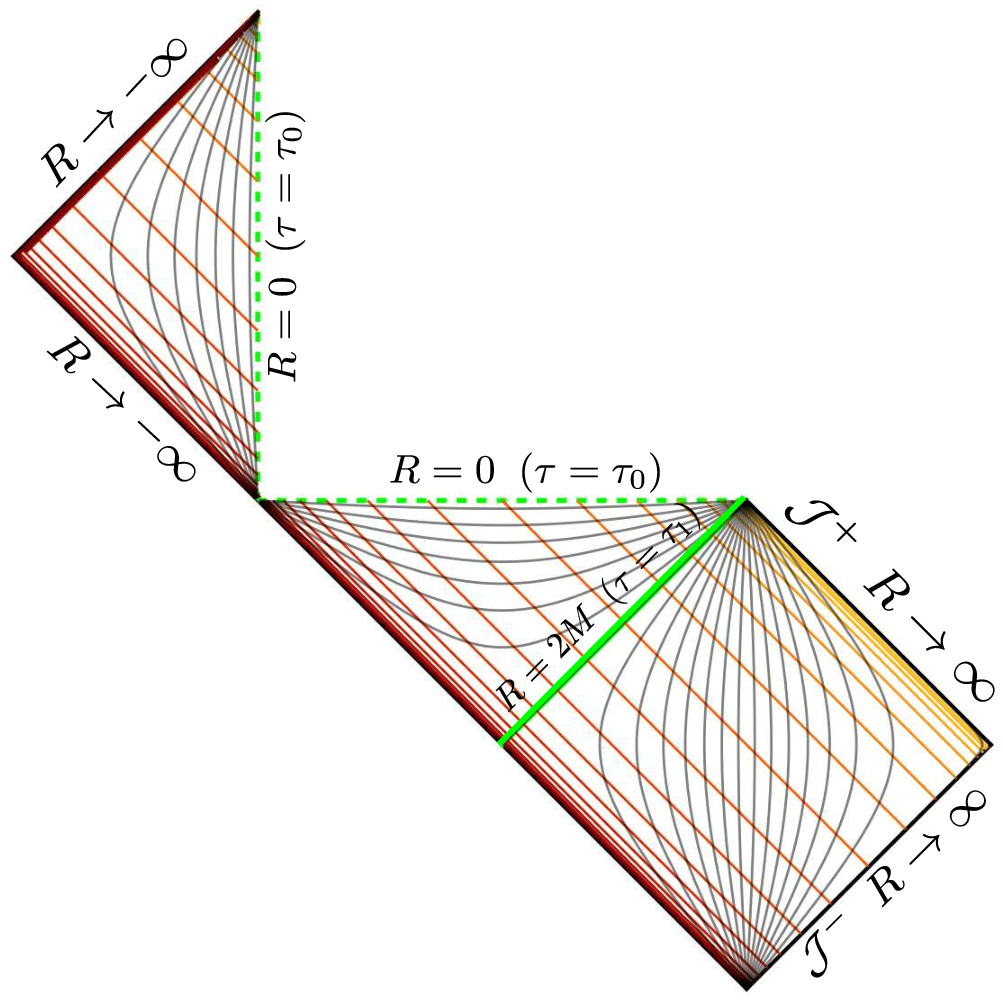}
\end{minipage}\\
\small{(c)}\begin{minipage}{0.4\textwidth}
		\includegraphics[width=\textwidth]{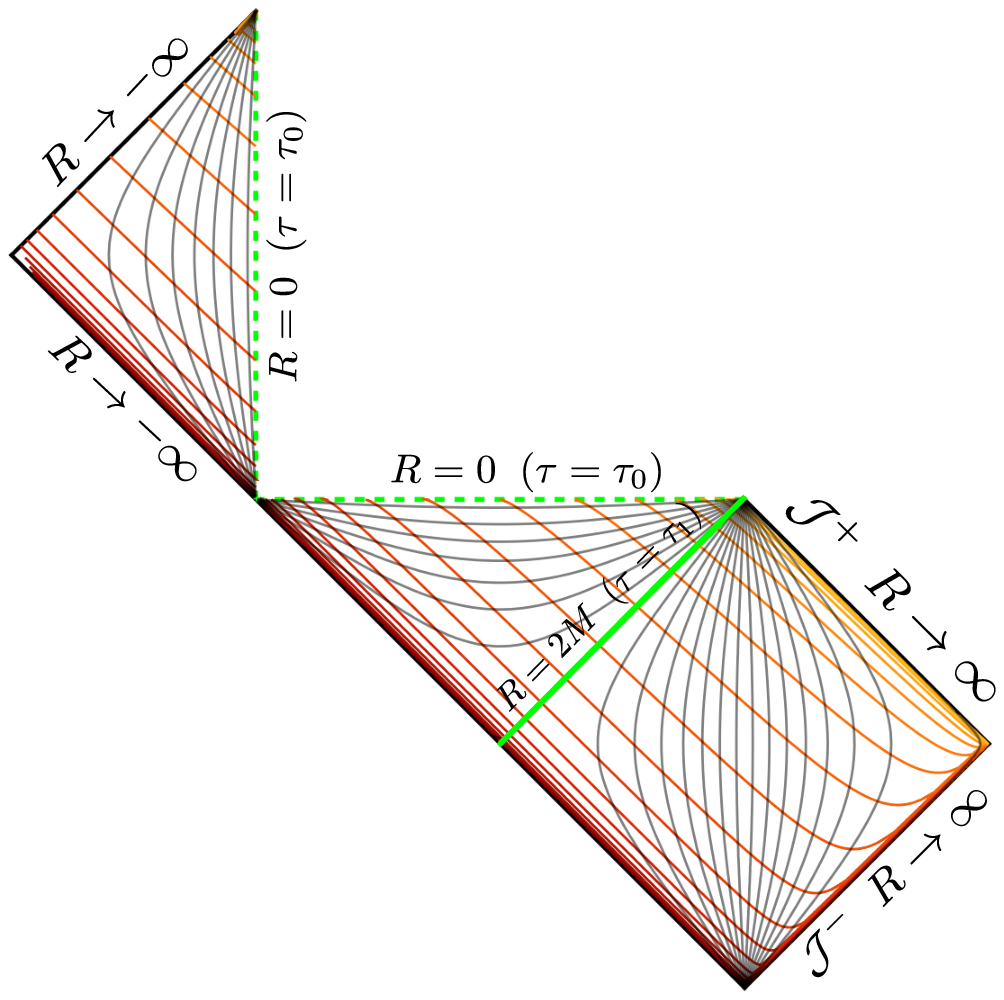}
\end{minipage}\hspace{0.5cm}
\begin{minipage}{0.5\textwidth}
\captionof{figure}{\small{In these figures we have drawn the lines at constant $\tau \sim R$ and $x \sim T$. The first one are represented in black, and they deinfine the homogeneous slices labelled by the evolution parameter of the model $\tau$. In the leftmost figure (a) we have the line at constant Schwarzschild time, the color of the lines goes from a darker one for negative values of $T$, and gets lighter when it grow. We see that all the slices are well spanned by the $T$ coordinate, both in the interior and the exterior, with the exception of the horizon. The central figure (b) represent the case in which the shift select a null Eddington-Finkelstein coordinate, while the rightmost one is the raindrop time. For the cases (b) and (c) the coordinate $T$ is well defined also on the horizon. We recall that each point on the diagram actually represents a two-sphere, spanned by the remaining angular coordinates $\de \Omega$ in \eqref{Bh_minisup_line}. }\label{fig:coord_level}}
\end{minipage}
\end{center}}

The fact that the metric is well defined at the horizon depends on the choice of shift. We define for what follows the point of the evolution that represents the horizon as 
\be
\tau_1 = \f{2 L_s^2 G \cB}{ \cV_0} + \tau_0\,.
\label{horizon_time}
\ee
The homogeneous hypersurface labelled by $\tau=\tau_1$,  spanned by $x$, $\theta$ and $\phi$ is always the horizon of the black hole, no matter the shift that we choose, as long as it does not diverge faster than $(\tau-\tau_1)^{-1}$ close to the horizon.

In the figures \ref{fig:kruskal_level} we schematically draw the levels of the functions $a,b,u,v$ in terms of $R$ on a Kruskal diagram. During the evolution, $\tau = \tau_0$ corresponds to the true singularity at $R = 0$, where $v$ vanish, while at the horizon $R = 2M$ it $u$ that vanishes. In terms of $a,b$ we have at the singularity a vanishing $b$ and a divergent $a$, while only the latter vanishes at the horizon.  From the figures \ref{fig:kruskal_level} and the equations \eqref{Bh_sol_conf} and \eqref{Bh_sol} we also see explicitly that $a,b$ describe only the interior of the black hole, for $\tau \in [\tau_0, \tau_1]$, because $a$ is well defined only therein, while $u,v$ are well defined all along the real evolution line $\tau \in \R$.

\vbox{
\begin{flushleft}
\small{(a)}\begin{minipage}{0.30\textwidth}
		\includegraphics[width=\textwidth]{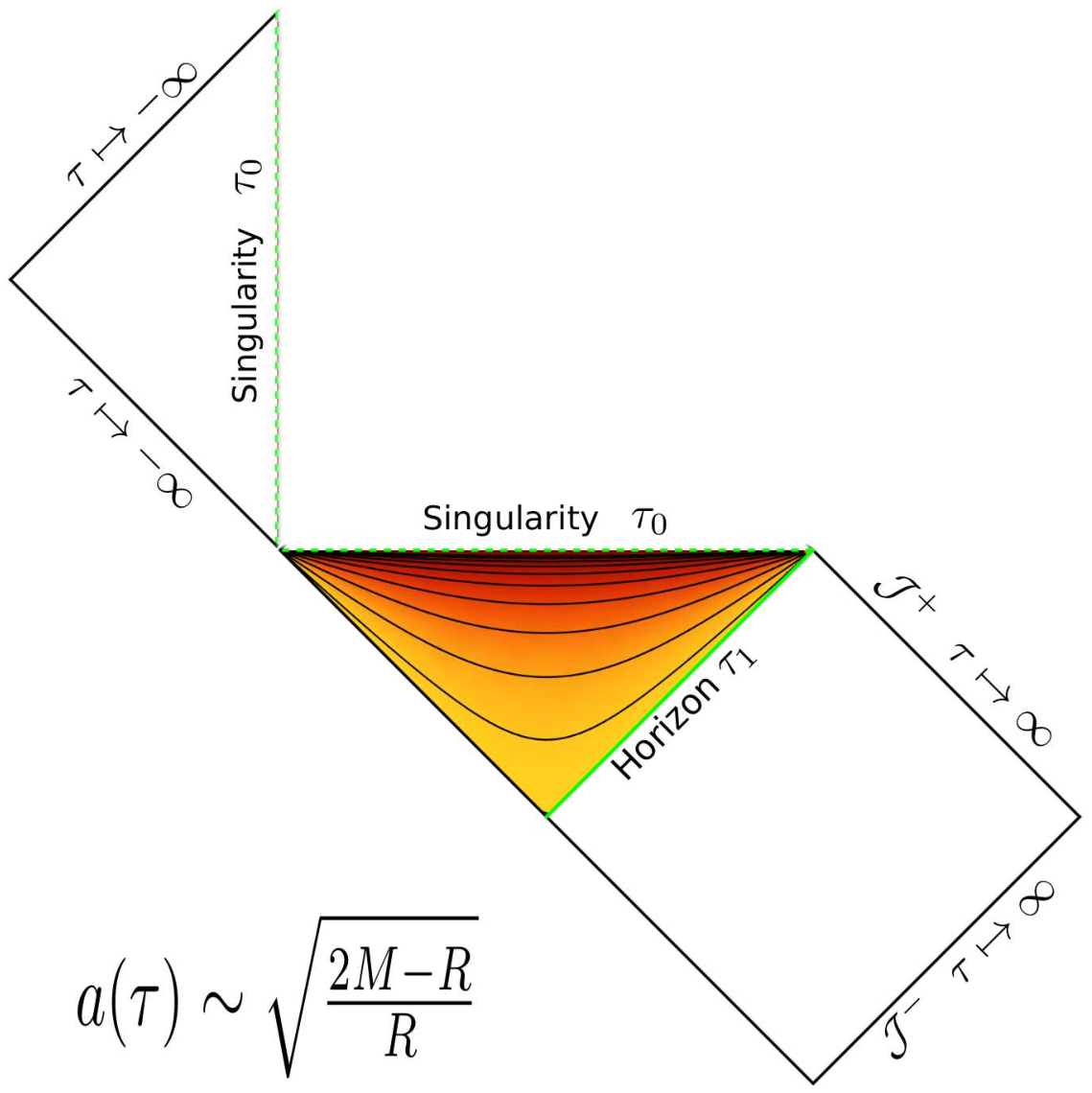}
\end{minipage}
~\hspace{2cm}
\small{(b)}
\begin{minipage}{0.30\textwidth}
		\includegraphics[width=\textwidth]{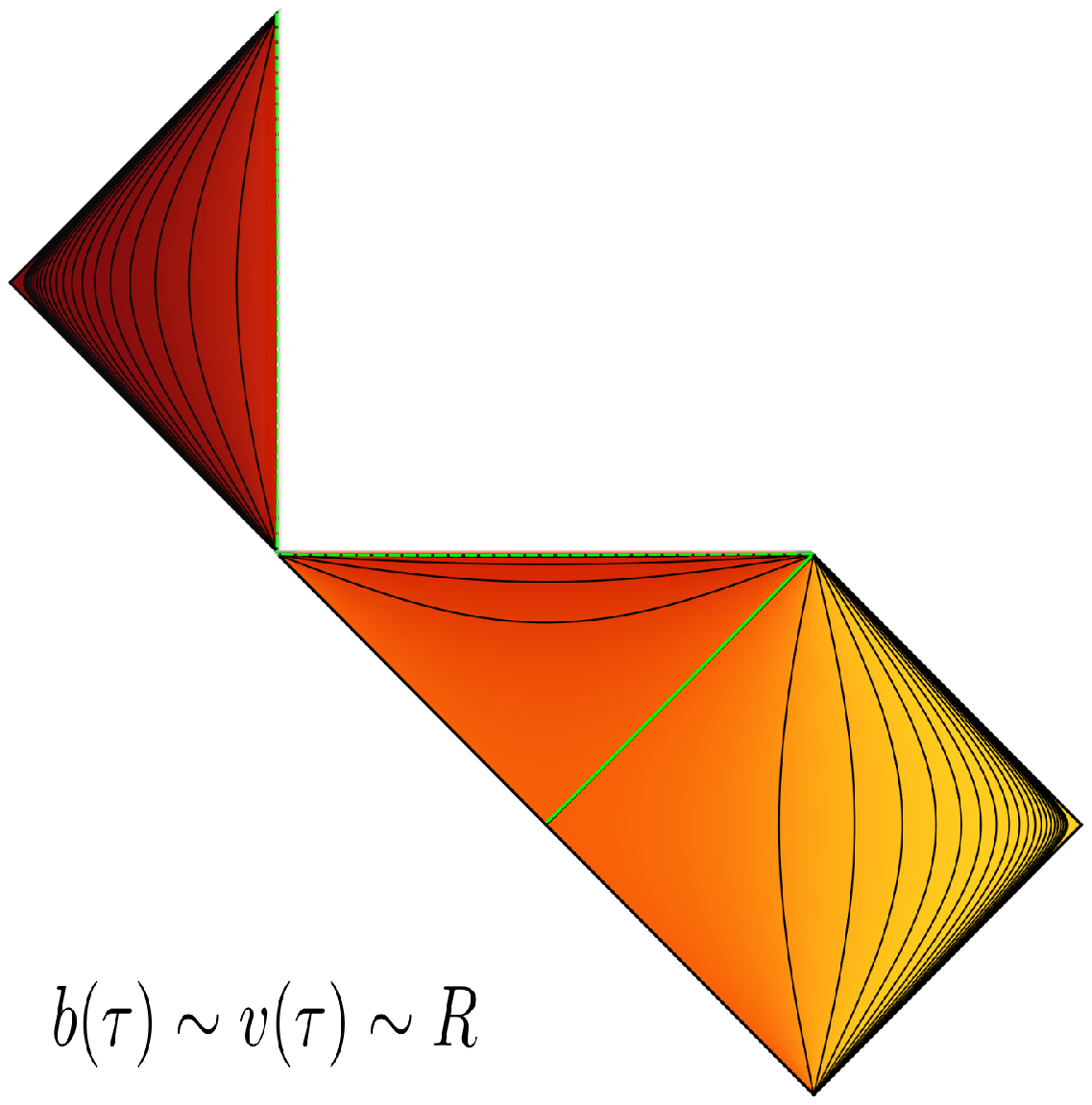}
\end{minipage}\\
\small{(c)}\begin{minipage}{0.30\textwidth}
		\includegraphics[width=\textwidth]{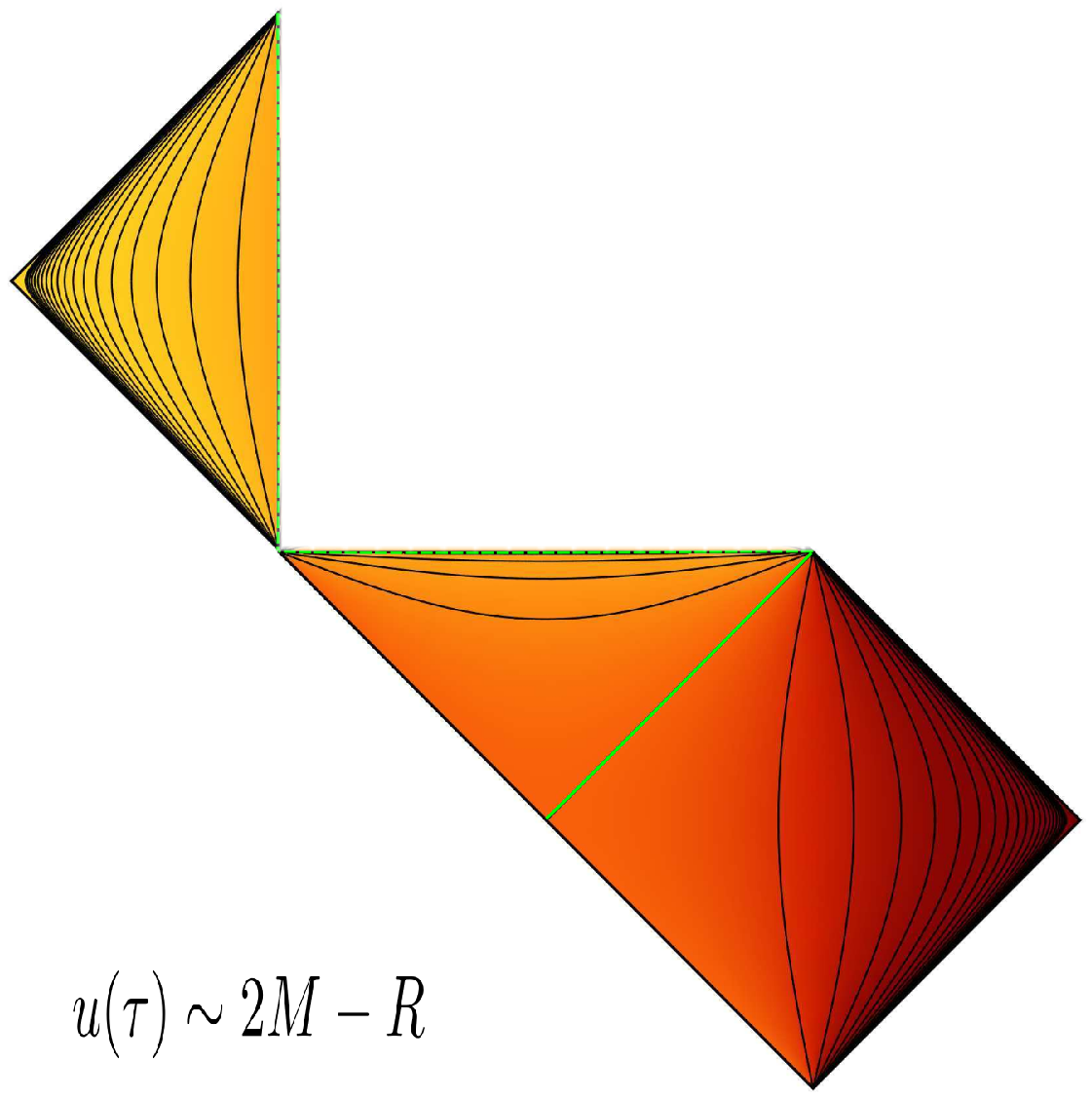}
\end{minipage}
~\hspace{0.5cm}
\begin{minipage}{0.60\textwidth}
\captionof{figure}{\small{In these figures we have drawn the levels of the fuctions $a,b \sim v,u$ in terms of the Schwarzschild radius $R$ on a Penrose diagram, the black lines are hypersurfaces that equally space the levels of the respective functions, the lighter background color represent the higher values of the functions and it gets darker as the function decreases, we recall that $R \propto (\tau-\tau_0)$ and the Penrose diagram is extented behind the singularity if we add the patch with negative coordinate radius. This part is given by the upper triangle, where there is no horizon and the singularity is connected to the asymptotically flat region at $\tau\to -\infty$. The first one starting from the left (a) represents $a(\tau)$ as in \eqref{Bh_sol_a} that is defined only for $\tau\in[\tau_0, \tau_1]$. Then in (b) we have $b$ as in \eqref{Bh_sol_b}, this is also proportional to  $v$, according to \eqref{conf_to_null_bh}, the rightmost plot represents $u(\tau)$.} \label{fig:kruskal_level}}
\end{minipage}
\end{flushleft}}

Finally, an important point should be stressed here. One of the most important theorems in general relativity is the no-hair theorem, which states that black hole solutions of Einstein equations, coupled to an electromagnetic field, are characterised only by three observables: mass, electric charge and angular momentum. In the case under consideration here we don't have an electric charge and the spherical symmetry makes the angular momentum vanish, and we are left with the mass alone and a one-parameter family of solutions.

On the other hand, the mechanical model that we have considered is defined as a system with two degrees of freedom, where does this mismatch come from? The answer to this puzzle is hidden in the dynamical symmetries and without surprise it is the boundary that has something to say about it. For this, we recall that the boundary exists here as a consequence of the infrared regulator constraining the homogeneous coordinate $x$ (or $T$) to a finite range, and is then represented by two of the lines in figure \ref{fig:coord_level}. Our model does not represent the whole black hole solution but only a patch bounded by the regularization of the $x$ coordinate. The size of the fiducial cell will reappear as a conserved charge, suggesting that the additional degree of freedom really accounts for the boundary. Furthermore the ability to freely choose this boundary by playing on the shift is a remarkable property that deserves a deeper analysis. A more natural framework to deal with this question is given by the so-called \textit{midisuperspace} model, where the dynamical fields still depends on two coordinates $x,t$. This still contains both spatial and temporal diffeomorphisms, allowing for a clearer description of the equivalence between different choices of $x$ (or shift). The price to pay is the loss of simplicity, trading the mechanical model for a field theory. The question is there that of the relationship between the boundary and some conformal rigid symmetry.  

If we want to make sense of the physical meaning of the boundary sticking with the one dimensional model, we need to use the results of the previous section \ref{sec2.2:mink&minisup}. We are going to apply them to black holes, to explicitly see the relationship between symmetries and fiducial cell. 

\subsection{Symmetries and Poincar\'e invariance}
In the previous section \ref{sec2.2:mink&minisup} we have shown the existence of dynamical symmetries and the corresponding conserved charges, for a particle with any mass moving in a two dimensional Minkowski background, possibly with a conformal potential. We have just shown that is also possible to recast the black hole minisuperspace in these terms so that it will automatically inherit the $\left (\SL(2,\R) \oplus \R\right ) \loplus \R^4$ symmetry. Using the change of variables \eqref{conf_to_null_bh}, we can read the finite transformation on the scale factors. We shall start with the rescaled M\"obius transformations:
\bsub
\be
\tau	& \mapsto \tilde{\tau} = f(\tau) = \f{\alpha \tau +\beta}{\gamma \tau + \delta}\,,\q\q \alpha\delta -\gamma\beta = 1 \label{Mobius_BH_time}\,,\\
a 		& \mapsto \tilde{a}(\tilde \tau)=k\, a(\tau)\,,\q\q\q\q\q k=const\,,\\
b 		& \mapsto \tilde{b}(\tilde \tau)= k^{-1} \dot f^{1/2}(\tau) b(\tau)\,.
\ee
\label{Mobius_BH}
\esub
The $\SL(2,\R)$ act as a conformal rescaling of the unit sphere metric, recalling that $b$, that has a $1/2$ conformal weight, is the coefficient in front of the $\de \Omega^2$ in the line element \eqref{Bh_minisup_line}. On the other hand the $g_{xx}$ term is unaffected by the M\"obius transformation. For the sake of completeness we verify again that the transformations are indeed symmetries. The corresponding variation of the action is given by the total derivative, leaving the equation of motion unchanged
\be
\Delta \cS = \cS[\tilde{a}(\tilde{\tau}), \tilde{b}(\tilde{\tau})] - \cS[a(\tau),b(\tau)] =
\int \de \tau \f{\de}{\de \tau} \left [ \f{4\cV_0}{G} \f{\gamma a^2 b^2}{\gamma t +\delta} \right ]\,.
\ee
On its side, the abelian subset of transformations \eqref{abelian_symm}, generated by $\cS^\pm$, acts on the metric coefficients as
\bsub
\be
a 		& \mapsto \tilde{a}= \sqrt{\f{a^2 b  +g(\tau)}{ b + h(\tau)}}\,,\\
b 		& \mapsto \tilde{b}= b+h(\tau)\,,
\ee
\esub
for two linear functions $g$ and $h$. It is of course a symmetry and the variation of the Lagrangian \eqref{Bh_lagran} gives a total derivative
\be
\Delta \cS = \cS[\tilde{a}(\tilde{\tau}), \tilde{b}(\tilde{\tau})] - \cS[a(\tau),b(\tau)] = - \int \de \tau \f{\de}{\de \tau} \left [ \f{4\cV_0}{G} \left (\dot g (h+b^2) +\dot g a^2 b \right)\right  ]\,.
\ee
On top of them, we have also the transformations generated by the charges $\cT^\pm$. This exponentiates into two abelian $\R^3$ transformations, that do not commute with each other. As already discussed in the previous section, this gives two subgroup of symmetry that are isomorphic to the Poincar\'e group in 2+1 dimensions, $\ISO(2,1)$. Their action on the metric coefficients is given by
\bsub
\be
a 		& \mapsto \tilde{a}= \sqrt{\f{a^2 b+ 2g \dot b - b \dot g}{2a h \dot b +b(1+ 4 a h \dot a- a^2 \dot h)}}\,,\\
b 		& \mapsto \tilde{b}= b +2 a^2 h \dot b + 4 a b h \dot a - b a^2 \dot h	\,,
\ee
\esub
for second-degree polynomials in time $g$ and $h$. Without loss of generality we will consider here the $\ISO(2,1)$ subset of symmetries as the one generated by $\cL$  and $\cT^+$ (i.e. the function $g$ in the formulas above), hereafter we refer to the abelian transformations as \textit{field translation}. The corresponding transformations takes a very compact form written in terms of the metric coefficients $a^2$ and $b^2$. The 6-dimensional group is parametrized by a M\"obius function $f$ and a two degree polynomial $g$:
\bsub \be
\tilde \tau &= f(\tau)\,,\\
g_{xx} = a^2 &\mapsto \tilde{a}^2(\tilde{\tau}) = a^2 + 2 g \f{\dot {b}}{b} - \dot g \,, \\
g_{\Omega\Omega} \propto b^2 &\mapsto \tilde{b}^2(\tilde{\tau}) = \dot f b^2\,.
\ee \label{iso_bh}\esub
For the rest of this section we will concentrate on this subset of symmetry, to come back to the full set of charges in the following. As a direct consequence of the identities \eqref{identities} relating the generators, we have that the two Poincar\'e Casimir vanish
\be
\mathfrak{C}_1 = \cT_1^+ \cT_{-1}^+ - {\cT_0^+}^2= 0\,, \q\q \mathfrak{C}_2 =\cT_1^+ \cL_{-1} +\cT_{-1}^+ \cL_1 -2{\cT_0^+}\cL_0= 0\,.
\label{poincar_casimir}
\ee
These two Casimir conditions have a two-fold effect. One one side they reduce the 6-dimensional Lie algebra $\iso(2,1)$ back to the original 4-dimensional phase space generated by the configuration and momentum variables $a,b,p_a,p_b$. This also means that the minisuperspace black hole interior carries a massless and zero spin unitary representation of the Lie algebra $\iso(2,1)$. Preserving this structure at the quantum level is a synonym for quantizing the black hole interior dynamics in terms of Poincar\'e representations.

\paragraph{Conserved charges}
Before moving on a closer investigation of the consequences of the symmetry, we would like to discuss the Noether charges, and their relationship with the integration constants. For this, we just need to use the form of the generators $\cL$ \eqref{sl_2+R} and $\cT$ \eqref{trans_gen_null}, combined with the change of variables \eqref{conf_to_null_bh}, in order to obtain the expression of the charges in terms of $a$ and $b$.  Finally, we use the solutions \eqref{Bh_sol} to obtain their \textit{on-shell} values
\bsub \be
\cL_n &\approx -(n+1)\tau_0^n \cB - \tau_0^{n+1} \f{\cV_0}{L_s^2 G} \label{gen&integral_1}\,,\\
\cT_n^+ &\approx \f{\cA G}{\cV_0} \tau_0^{n+1} \label{gen&integral_2}\,.
\ee	
\esub
In particular, we see that the $\sl(2,\R)$ Casimir operator measures the first integral $\cB$
\be
\mathfrak{C}_{\sl 2} =\cL_1 \cL_{-1} - \cL_0^2 = -\cB^2\,,
\ee
where the last equality holds for each time $\tau$. As a remark, we note that here the Hamiltonian $H=-\cL_{-1}$ belongs to the non-Abelian $\mathfrak{sl}(2,\R)$ subalgebra. This is very different from the usual spacetime picture, where the Hamiltonian is part of the Abelian generators and can be seen as a consequence of the fact that the Poincar\'e structure which we have discovered here has nothing to do per se with the isometries of spacetime. In the group quantization of the model using representation theory, this will have important consequences as it will determine how the Hamiltonian is represented.

\paragraph{Physical interpretation}
It is now time to understand how these transformations act on the physical trajectories of the system. In particular, we want to determine if they simply consist of a different time parametrization of the same trajectory, or if they provide a map between different trajectories and thus contain physically relevant information. 

From the mechanical point of view, different trajectories are labelled by the values of the first integrals $\cA$ and $\cB$, as well as the energy $H$ and the initial time $\tau_0$. We will study how these are modified by the conformal transformations in proper time and field translations.

By explicitly computing the action of the $\SL(2,\R)$ reparametrizations of the proper time \eqref{Mobius_BH} on the classical solutions \eqref{Bh_sol} we obtain
\be 
\tilde a (\tilde \tau) &= \sqrt{\f{ \cB (\alpha -\gamma \tilde{\tau})(\alpha -\gamma \tilde{\tau}_0) \cV_0}{\cA (\tilde \tau - \tilde \tau_0) G} -\f{\cV_0^2}{2\cA L_s^2 G^2}} \,,\\ 
\tilde{b} (\tilde \tau) &= \f{G \sqrt{\cA}}{\sqrt{2}\cV_0} \f{\tilde \tau- \tilde \tau_0}{\alpha - \gamma \tilde \tau_0}\,,\\
\tilde{p}_a &= -8 \f{\cV_0}{G} \tilde b \tilde a \f{\de \tilde b}{\de \tilde \tau} 
= -\left (8 \cA (\tilde \tau - \tilde \tau_0) \f{\cV_0(\tilde \tau_0-\tilde \tau) + 2\cB L_s^2 G(\alpha - \gamma \tilde{\tau})(\alpha - \gamma \tilde{\tau_0})}{L_s^2 \cV_0 (\alpha - \tilde \tau_0 \gamma)^4} \right )^{1/2}\,,\\
\tilde{p}_b &= -8 \f{\cV_0}{G} \left (\tilde a^2  \f{\de \tilde b}{\de \tilde \tau} + \tilde b \tilde a \f{\de \tilde a}{\de \tilde \tau}\right )  
= -2\sqrt{2} \cV_0 \f{ \cV_0(\tilde \tau_0 -\tilde \tau) + \cB L_s^2 G\left (\alpha - \gamma( \tilde{\tau}-2\tilde{\tau}_0)\right )(\alpha - \gamma \tilde{\tau_0})}{\sqrt{\cA} L_s^2 (\tilde \tau - \tilde \tau_0)(\alpha - \tilde \tau_0 \gamma) G^2}\,,
\ee
where we have used the inverse M\"obius transformation to \eqref{Mobius_BH_time}, given by $\tau=(\delta\tilde{\tau}-\beta)/(\alpha-\gamma\tilde{\tau})$. The new metric components $\tilde a$, $\tilde b$ are obviously still solutions to the equations of motion, but the values of the constants of the motion need to be slightly adjusted. A closer inspection reveals that the value of $\cB$ remains the same, while  $\cA$ acquires a transformation-dependent factor. Explicitly, we have
\be
\cA &\mapsto \widetilde{\cA}=\f{\tilde{p}_a}{32 \tilde a \tilde{b}}= \f{\cA}{(\alpha - \tilde \tau_0 \gamma)^2} = \cA (\tau_0 \gamma + \delta)^2 \,,\\ 
\cB &\mapsto \widetilde{\cB}= \f{\tilde b\tilde{p}_b - \tilde{a}\tilde{p}_a}{2} = \cB\,.
\ee
This result is consistent with the fact that $\cB$ is the Casimir invariant of the $\mathfrak{sl}(2,\R)$ algebra, and as such is expected to be conserved under $\SL(2,\R)$ transformations. On the other hand, $\cA$ belongs to the translational sector of $\iso(2,1)$ and is therefore naturally modified by this symmetry.

However, the story of the conformal mapping is actually more subtle because it happens to shift the Hamiltonian constraint, which therefore does not seem to vanish anymore. Indeed, we have a variation in the energy value
\be
H=\f{G}{\cV_0} \f{p_a(a p_a -2 b p_b)}{16 a b^2}
\,\, \mapsto\,\,
\widetilde{H} =\f{G}{\widetilde \cV}_0 \f{\tilde p_a(\tilde a \tilde p_a -2 \tilde b \tilde p_b)}{16 \tilde a \tilde b^2}=  \f{\cV_0^2\widetilde{\cV}_0^{-1}}{ L_s^2 G (\alpha  - \tilde \tau_0 \gamma)^2} + \f{2\cB \cV_0\widetilde{\cV}_0^{-1} \gamma} {\alpha - \tilde \tau_0 \gamma}
\,.\nonumber
\ee
This apparent puzzle is resolved by the fact that the fiducial lengths do actually need to change under conformal transformations. In order to interpret the new trajectory as a black hole solution, we need to restore the on-shell condition and redefine the fiducial scales. The most reasonable way of doing it is to keep the fiducial volume $\cV_0$ to be always constant (i.e. $\cV_0= \widetilde{\cV}_0$), and modify accordingly the length unit on the sphere $L_s$, and the regulator $L_0$.
\bsub \be
L_{0} &\mapsto\widetilde L_{0} =\f{L_0}{(\alpha - \gamma \tilde{\tau}_0)(\alpha - \gamma \tilde{\tau}_1)} = L_0 (\delta +\gamma \tau_0)(\delta + \gamma \tau_1)\,,\\
L_{s}^2 &\mapsto\widetilde L_{s}^2 = L_s^2 (\alpha - \gamma \tilde{\tau}_0)(\alpha - \gamma \tilde{\tau}_1) = \f{L_s^2}{(\delta +\gamma \tau_0)(\delta + \gamma \tau_1)} \,,
\ee \esub
where the horizon locations $\tilde \tau_1 = f(\tau_1)$ (see \eqref{horizon_time}) appears in a surprising manner. One can see that the last requirement ensures that
\be
H \approx \f{\cV_0}{G} L_s^{-2} \mapsto \widetilde{H} \approx \f{\cV_0}{G} \widetilde{L_s}^{-2}   .
\ee
Because of the presence of the fiducial length in the line element \eqref{Bh_minisup_line}, the 4-metric gets modified as well, and becomes
\be
\de s^2_\rm{BH} = -\f{1}{4 \tilde{a}^2} \de \tilde \tau^2 + \tilde a^2 (\de x + \tilde N^1 \de t)^2 +\widetilde {L_s}^2 \tilde b^2\, \de \Omega^2\,.
\ee
Inserting the explicit trajectories in $\tilde{\tau}$ in this expression and performing the change of coordinates
\be
\tau -\tau_0 =\f{\sqrt{2} \cV_0}{\widetilde{L_s} G \sqrt{\widetilde \cA}} R\,,\q\q x= \f{\widetilde{L_s} G \sqrt{\widetilde {2\cA}}}{\cV_0} T\,,
\ee
gives the static spherically symmetric black hole metric with a dilated mass 
\be
\widetilde{M}= \f{\sqrt{\widetilde \cA} \widetilde{L_s}^3 G^2 \widetilde{\cB}}{\sqrt{2} \cV_0^2} = M (\alpha - \gamma \tilde{\tau}_0)^{1/2}(\alpha - \gamma \tilde{\tau}_1)^{3/2} =\f{M}{(\delta +\gamma \tau_0)^{1/2}(\delta + \gamma \tau_1)^{3/2}}\,.
\ee
We can treat the field translations in the same manner. They act only on the field $a$ and the resulting flow between trajectories corresponds to a shift in the constants of motion
\be
\widetilde{\cA} =\cA \,,\q\q  \widetilde{\cB} =\cB + \f{2\cA G}{ \cV_0} g(\tau_0)\,,
\ee
where $g$ is a second degree polynomial. In order to preserve the Hamiltonian constraint we redefine the fiducial length as
\be
\widetilde{L_0} &= \f{L_0^2 L_s^2 - 32 \cA G^2 \dot{g}(\tau_0)}{L_0 L_s^2}\,,\\
\widetilde{L_s}^2 &= \f{L_0^2 L_s^4}{L_0^2 L_s^2 - 32 \cA G^2 \dot{g}(\tau_0)}\,,
\ee
which in turn leads to the shifted mass
\be
\tilde{M}= \f{ M \cV_0^3 +\sqrt{2\cA^3}L_s^3 G^3 g(\tau_0)}{ \left ( \cV_0^2-2\cA L_s^2 G^2 \dot g(\tau_0)\right )^ {3/2}}\,.
\ee
This shows that the Poincar\'e symmetry transformations are not mere time reparametrizations of the classical solutions, but map physical black hole trajectories onto different trajectories with different initial conditions (the values of $\cA$ and $\cB$) and a different black hole mass. In this sense we speak about \textit{hidden} rigid symmetries, leaving on top of the residual time-diffeomorphisms.

Furthermore, this makes clearer which are the degrees of freedom described by the minisuperspace Lagrangian \eqref{Bh_lagran}. While from the naive point of view, the only physical parameter is the mass $M$, the presence of the symmetry makes the fiducial length to be physically relevant, in the sense that they are to be considered as labels of the solutions, interplaying with the conserved charges. Conversely, the introduction of the regulation scale $L_0$, which in turn defines the boundary of the homogeneous slice, is unavoidable from the very beginning, as it is necessary to properly define the dynamics of the model via its action principle. On the one side, the symmetry unravels the role of the boundary as a degree of freedom, on the other side it is the boundary itself that makes the symmetry manifest. 

Together with the mass $M$ and the size of the slice $L_0$, we can completely span the phase space and the corresponding solution space by considering also the quantities under which the M\"obius transformation takes a very suggestive form, namely the singularity $\tau_0$ and horizon location $\tau_1$. The Poincar\'e symmetry transformations thus govern the conformal properties of the main features of black holes: mass, horizon and singularity. This turns it into a very powerful tool to deal with black hole physics, and hopefully, helps to shed a light on its holographic properties.

Another point in favour of the relevance of the Poincar\'e symmetry has been recently raised in the study of static perturbation of black holes \cite{Achour:2022syr}. It has been shown that any static linear perturbation of Schwarzschild background exhibits an $\SL(2,\R)$ symmetry transformation, that is connected to the vanishing of Love numbers. In particular, the finite transformation for the zero modes of perturbations stands for conformal reparametrization of coordinate radius. We recall that the latter is indeed proportional to the evolution parameter $\tau$ chosen in this section. Despite the difficulty of generalising the connection to higher multipoles (see \cite{Achour:2022syr}), this nevertheless makes first contact between the Poincar\'e invariance of the background and its perturbations.

From the same point of view, the Poincar\'e symmetry can also have important consequences on Hawking radiation. The quasi-thermal emission spectra of black holes \cite{Hawking1975}, and their quasi-normal modes are indeed intimately related to their response to perturbations.   For spherically-symmetric static metrics the wave equations for the perturbation fields can be written in the form known as Regge--Wheeler equation, \emph{i.e.} as a one-dimensional Schr\"odinger-like radial wave equations with a short-ranged potential, that depends on the spin of the perturbation field (see \cite{Arbey:2021jif} for a review) and on the metric coefficients.  
The conformal reparametrization of the radial coordinate, and the corresponding $\SL(2,\R)$ symmetry must act on the perturbation fields, but finding the corresponding conserved quantities and their physical interpretation might be quite cumbersome. This is already the case for the static perturbation, studied in \cite{Achour:2022syr}.

Furthermore, we can push the logic even further and consider that the conformal symmetry presented here must be protected upon quantization and that the consequences of modifying gravity by the inclusion of quantum effects should keep the conformal invariance of the radial coordinate. This will translate again into a symmetry for the perturbation on a modified background, as studied in \cite{Arbey:2021jif,Arbey:2021yke}. 

The Poincar\'e symmetry exhibited here provides, in the end, a powerful tool to deal with both the propagation of test fields on an (eventually modified) black hole and the corresponding perturbation theory.

\section{Bianchi Models and Cosmologies}
\label{sec2.4:Bianchi}
The black hole is not the only interesting minisuperspace that we can study with the formalism developed in section \eqref{sec2.2:mink&minisup}. The very nature of the minisuperspace models, with the metric evolution described in only one direction, is perfectly suited to describe cosmological models. In particular, whenever we believe the cosmological principle to be valid, we can approximate the Universe with a homogeneous and isotropic FLRW line element evolving in time. Furthermore, the minisuperspace (or adiabatic) approximation has been widely used to describe the early stage of the evolution of the Universe. From a quantum cosmology perspective, where we assign a quantum wavefunction to the Universe, the quantization of the minisuperspace will correspond to quantising the largest wavelength mode, of the size of the Universe.  

The relaxation of the isotropy property, by keeping only the homogeneity will open the door to the whole set of Bianchi cosmologies

\subsection{Symmetries of Bianchi cosmologies}

At the end of the 19th century, the mathematician Luigi Bianchi managed to classify all the three dimensional Lie algebras into nine different equivalence classes \cite{bianchi1897sugli}. This classification was imported into geometry and physics, using the associated Lie groups as symmetry groups of 3-dimensional Riemannian manifolds. The Bianchi spaces are then defined as manifolds admitting a set of three Killing vectors, whose Lie brackets reproduce one of the Bianchi algebras. 

The original Bianchi classification \cite{bianchi1897sugli} distinguishes nine classes containing a single Lie algebra and two infinite-sized families, nowadays known as VI$_h$ and VII$_h$. The inclusion of two classes (VI$_0$ and VII$_0$) into the one-parameter families reduces the number from eleven to nine groups. Hereafter, when the index $h$ is present, it should be understood that it is non zero $h \neq 0$, while we denote with VI and VII the groups that contain the zero value of $h$. Furthermore, the VI$_{1}$ algebra is equivalent to III, so we can consider $h\neq 1$ for the Bianchi VI spaces.  

The main application of this classification is in the context of cosmology, where the Bianchi spaces represent a homogeneous slice evolving in time. The whole four-dimensional manifold is then named Bianchi spacetime. For example, the FLRW metrics are isotropic and they are particular cases of types I, V, VII$_h$ and IX. The Bianchi type I model contains the Kasner metric as a special case, and the Taub metric is a particular case of Bianchi IX cosmologies \cite{wald1984general}. They also play a crucial role in the dynamics near the singularity, which is approximately governed by a series of successive Kasner (Bianchi I) phases in the presence of a scalar field, or by type VIII and IX in vacuum. The complicated dynamics which is referred to as BKL (Belinskii, Khalatnikov and Lifshitz \cite{Belinsky:1970ew, Belinskii:1972sg}), consists of a billiard motion in hyperbolic space and exhibits a chaotic behaviour \cite{Damour:2002et}.

In the terms presented in the first chapter of this thesis, the fact that the homogeneous slice belongs to the Bianchi classification fixes the triads $e^i$ to be left-invariant forms under the Killing vectors. The reader shall refer to appendix \eqref{app:2_bianchi} for the complete list of Bianchi hypersurfaces. On the other hand, the internal metric $\gamma$ is left completely free by the Bianchi classification, we shall fix some of its coefficients by the requirement that the vector constraint \eqref{minisup_vect_constr} is satisfied.

We would like to stress here that the black hole minisuperspace, also known as Kantowski-Sachs cosmology does not belong to the Bianchi classification, as the isometries of the spatial slice do not form a three dimensional Lie algebra.

The importance of Bianchi models in cosmology and the dynamics near singularities raised interest in the structure of their symmetry \cite{Christodoulakis:1991ky, Christodoulakis:2018swq}. These works and the reference therein were already based on the isometries of the supermetric providing the existence of the linear charge $\cQ_{(i)}$, but the application of the construction in section \ref{sec2.2:mink&minisup} will allow to extend to the existence of quadratic charges in time, and eventually completely solve the dynamics in terms of the algebra.

For this, we will restrict ourselves to the so-called \textit{diagonal models} \cite{MACCALLUM1972385}, where the adjective ``diagonal" stands for the internal metric. It turns out that for each Bianchi family it is possible to define the triad such that the coefficients of $\gamma_{ij}$ are diagonal and the vector constraint is satisfied. We refer to appendix \ref{app:2_bianchi} for the discussion on the shift vector. The Bianchi models happen to be divided into four classes according to the form of the internal metric: 
\bsub
\be
\gamma_{ij} &= \rm{diag}\left (a^2,b^2,c^2\right)\,, & \rm{Bianchi}&\; \rm{I}\,,\; \rm{II}\,,\\
\gamma_{ij} &= \rm{diag}\left(a^2,b^2,b^2\right)\,, & \;\phantom{\rm{Bianchi}}&\; \rm{III}\,,\; \rm{VI}_{0}\,,\; \rm{VIII}\,,\; \rm{IX}\,,\\
\gamma_{ij} &= \rm{diag}\left (a^2,\f{a^4}{b^2},b^2\right)\,, & \phantom{\rm{Bianchi}}&\; \rm{IV}\,,\; \rm{V}\,,\; \rm{VII}\,,\\
\gamma_{ij} &= \rm{diag}\left (a^2, b^{\f{2(1+h)}{1 - h}} a^{\f{4 h}{h -1}},b^2\right)\,, & \phantom{\rm{Bianchi}}&\; \rm{VI}_h\,.
\ee
\esub
We focus our attention on the cases with a two-dimensional field space, a brief discussion on Bianchi I and II can be found in \cite{Geiller:2022baq}. We recall that the supermetric only depends on the form of the internal metric, up to the fiducial scale $\cV_0$, that on the contrary depends on the topology of the slice. The second class has the same internal metric as the black hole model, and we already know how to recast the field space into its null conformal parametrization \eqref{family_1_conf}. For the other two families the kinetic term of the reduced action reads (see appendix \ref{app:2_bianchi} for the computation of the fiducial volumes):
\be
\f{\cV_0 \sqrt{\gamma}}{G} \f{(\gamma^{ij} \dot \gamma_{ij})^2+ \dot \gamma_{ij} \dot \gamma^{ij}}{4N} &= 
- \df{2 \cV_0}{N G} \left\{
\begin{array}{lc}
\df{2a b^2 \dot a^2 + 2 a^2 b\, \dot a \dot b - a^3 \dot b^2}{ b^2}\,, &\; \rm{IV}\,,\; \rm{V}\,,\; \rm{VII}\,,\\[10pt]
a^{\f{2h}{h-1}}b^{\f{1+h}{1-h}} \left (\df{2b \dot a^2 h}{a(h-1)} +2 \dot a \dot b + \df{a \dot b^2(1+h)}{b(1-h)}\right )\,, &\;\rm{VI}_h\,.
\end{array}
\right. 
\ee
With a constant lapse both the supermetrics are flat. For the IV, V and VII models the following field transformations map the field space to the null conformal parametrization \eqref{conformal_2d}:
\bsub
\be
\label{family_2_conf}
&\left|
\begin{array}{rl}
\tilde u &= 2 \sqrt{\df{2\cV_0}{3G}} a^{\f{3+\sqrt 3}{2}} b^{-\f{\sqrt{3}}{2}} \\
\tilde v &= 2 \sqrt{\df{2\cV_0}{3G}} a^{\f{3-\sqrt 3}{2}} b^{\f{\sqrt{3}}{2}}\\
\end{array}
\right.\,, &&\; \rm{IV}\,,\; \rm{V}\,,\; \rm{VII}\,,
\ee
while for the VI$_h$ family we use
\be
\label{family_3_conf}
&\left|
\begin{array}{rl}
\tilde u &= 2 \sqrt{\df{2\cV_0}{3G}} a^\f{1 - 3 h- \sqrt{1 + 3 h^2}}{2 -2 h}  b^\f{2 + \sqrt{1 + 3 h^2}}{2 -2 h}\\
\tilde v &= 2 \sqrt{\df{2\cV_0}{3G}} a^\f{1 - 3 h+ \sqrt{1 + 3 h^2}}{2 -2 h}  b^\f{2 - \sqrt{1 + 3 h^2}}{2 -2 h} \\
\end{array} \right.\,,
&&\; \rm{VI}_h\,.
\ee
\esub
The conformal factor $\varphi$ in \eqref{conformal_2d} is related to the choice of clock of the Bianchi cosmology through $N(a,b)G/\cV_0 = \log \varphi(\tilde u,\tilde v)$. In order to discuss the existence of the symmetry groups \eqref{mobius_sym}, \eqref{abelian_symm} and \eqref{abelian_symm_quadratic} we should rewrite the potentials in terms of the conformal null coordinates. The potentials do depend on both the internal metric and the frame field (see  \ref{app:1_ADM_triad} and \ref{app:2_bianchi}).

We distinguish two families of potentials: the ones that separate in the product of two functions of the null directions (III, V, VI) and the ones that don't (IV, VII, VIII, IX). For the first family we simply need to take $N\propto1/\cU$, then find the reparametrization that maps to the Minkowski coordinates, inverse the two transformations and obtain the full algebra. We have already illustrated the procedure for the black hole example in the previous section and the reader shall refer to the appendix \ref{app:2_bianchi} for a more detailed calculation. 

The study of the second family of models, where the conformal factor that set the potential to a constant doesn't give a flat supermetric, is a little more involved. It turns out that in the cases IV, VII, VIII and IX, choosing the factor $N=1/\cU$, the resulting curved metric doesn't have any Killing vector. Nevertheless, the potential is a sum of two or three monomials $\tilde u^n \tilde v^m$. This means that we can set one of them to a constant, find the reparametrization that maps to the Minkowski coordinates and then study the condition \eqref{potential_condition} for the remaining piece(s). Let us see this on the action
\be
\cS= \f{\cV_0}{G} \int \de \tau \left [-\f{1}{N}\dot {\tilde u} \dot {\tilde v}+ N\left (\c_1 \tilde u^{n_1} \tilde v^{m_1} + \c_2 \tilde u^{n_2} \tilde v^{m_2}\right )\right ] \,.
\ee
Without loss of generality we take (with $m_i, n_i \neq -1$)
\be
N = \f{1}{\tilde u^{n_1} \tilde v^{m_1}}\,,\q\q  
\left|\begin{array}{rl}
u &= \df{\tilde u^{n_1+1}}{1+n_1}  \\[9pt]
v &= \df{\tilde v^{m_1+1}}{1+m_1}  \\
\end{array}\right .\,.
\ee
The effective potential reads
\be
U= \f{\cV_0}{G}\c_2 (u (1+n_1))^{\f{n_2-n_1}{1+n_1}} (v (1+m_1))^{\f{m_2-m_1}{1+m_1}}\,,
\ee
so that the algebra reduces to $\sl(2,\R)$ but still closes if and only if 
\be
\label{power_condition}0=2+\f{n_2-n_1}{1+n_1}+\f{m_2-m_1}{1+m_1}\q \Leftrightarrow\q2+ m_1 +m_2+ n_1 +n_2 + m_1 n_2 + n_1 m_2 =0\,.
\ee
It turns out that the equation above is satisfied for the potential of the VIII and IX models, but it is not for IV and VII. 

To summarize, we have the full algebra for the black hole dynamics and the III, V and VI Bianchi models, while the VIII and IX potentials preserve only the conformal part of the algebra and spoil the Heisenberg $\cS^{\pm}_n$ subalgebra. The symmetry does not exist for the IV and VII models. We would like to stress that in \cite{Christodoulakis:2018swq} (and references therein) the authors single out the same two models (IV and VII) as the ones for which a diagonal ansatz for the scales factors is incompatible with Einstein's equation. This is not the case here, where we manage to have a diagonal ansatz that gives the correct equation of motion, but it spoils the conformal symmetry. 

\subsection{Revisiting the CVH symmetry}
We end this chapter by remarking that the curved FLRW model can also be reviewed in terms of a two dimensional field space, parametrized by the scale factor and the scalar field. Because of the inclusion of matter degrees of freedom, we cannot apply blindly the construction presented in \ref{sec1.2:symm_reduc} that ensures the well definition of the symmetry reduction. Luckily, as we have seen in the beginning of this chapter \ref{sec2.1:FLRW}, in the presence of an homogeneous scalar field the reduced action \eqref{FLRW action} gives the correct  Friedmann equations. Furthermore the minisuperspace Lagrangian takes the form of a two dimensional mechanical action \eqref{minisup_mech_lagr}. The inclusion of a non-zero curvature in the FLRW metric 
\be
\de s^2_{\rm{FLRW}} = -N^2 \de t^2 + a(t)^2 \left (\f{\de r^2}{1-k r^2} + r^2 (\de \theta^2+\sin^2 \theta\, \de \phi^2)\right )\,,
\ee
contributes with a potential term, turning the Lagrangian \eqref{FLRW action} into \cite{Achour:2021lqq}
\be
\label{FLRW_action_k}
\cS_{\rm{FLRW}}= \cV_0 \int \de t\left [\f{a^3 \dot \Phi^2}{2 N} -\f{3 a \dot a^2}{8 \pi G N} + \f{3 k a N}{8 \pi G}  \right ]\,, \q\q \cV_0 = \int_\Sigma \f{r^2 \sin \theta}{\sqrt{1-k r^2}}\,.
\ee
As announced, the reduced action is of the same form as the others treated above, the supermetric and the potential being:
\be
\de s^2_\rm{mini} = \f{\cV_0 a^3 }{N} \de \Phi^2 -\f{3 \cV_0 a}{4 \pi G N} \de a^2\,, \q\q \cU = -\f{3 \cV_0 k a }{8 \pi G}\,.
\ee
Along the lines of the previous sections, we choose the lapse that set the potential to a constant
\be
\label{lapse_FLRW_k}
N = -\f{8 \pi }{3 a }\,,\q\q U_0  = -k \f{\cV_0}{G}\,,
\ee
and we find the reparametrization of the field space that maps to the null coordinates \eqref{Mink_de_s}
\be
\label{FLRW_conf}
&
\begin{array}{rlcrl}
 u &= \df{3}{4\kappa} \sqrt{\df{\cV_0}{\pi}}  a^2 e^{ \f{\kappa \Phi}{\sqrt{3}}}\,,   &\q\q & p_u &= \sqrt{\df{\pi}{\cV_0}} e^{ -\f{\kappa \Phi}{\sqrt{3}}} \df{a p_a \kappa + 2\sqrt{3} p_\Phi}{3 a^2}\,,    \\[9pt]
 v &= \df{3}{4\kappa} \sqrt{\df{\cV_0}{\pi}} a^2 e^{ -\f{\kappa \Phi}{\sqrt{3}}}\,,  &\q\q & p_v &=\sqrt{\df{\pi}{\cV_0}} e^{\f{\kappa \Phi}{\sqrt{3}}} \df{a p_a \kappa - 2\sqrt{3} p_\Phi}{3 a^2}\,,  \\
\end{array} 
\ee
with $\kappa^2 = 16 \pi G$. We see here that we can extend the conformal symmetry presented in \ref{sec2.1:FLRW} in two directions. In a first place, the algebra $\sl(2,\R)$ is enlarged to $\left (\sl(2,\R) \oplus \R\right ) \loplus \h_2$. The full algebra is again given by the replacement of the transformation \eqref{FLRW_conf} above in the charges \eqref{charges}. In addition, we can extend the results to the presence of a non zero spatial curvature \cite{Achour:2021lqq}. The symmetries that we have added here, allow to integrate algebraically the evolution of the scalar field and are somehow analogous to the spinorial construction in \cite{BenAchour:2020njq}. As for the black hole, we have that the motion on the four dimensional phase space is completely determined by the four independent charges $\cS^\pm_n$. They furnish the necessary number of initial condition and any non-linear combination of this functional will still generate a symmetry, once we exponentiate its Poisson brackets. In the present work we have focused on the quadratic combinations \eqref{identities} and \eqref{trans_gen_null}, but we can equivalently recover the $\so(2,3)$ generators in \cite{BenAchour:2020njq}, as combinations of the Heisenberg charges.

\paragraph{CVH algebra}
At the beginning of this chapter, we have motivated the study of the minisuperpace symmetry by seeking a generalisation of the CVH algebra in the presence of a non-trivial three-dimensional curvature. This is the last question that we need to address, before moving to the study of the infinite-dimensional extensions and the quantum theory.

For this we would like to remark that the choice of lapse \eqref{lapse_FLRW_k} that allows to recast the potential of the curved FLRW model into a constant is completely different from the constant proper time ($N=1$) choice that defines the CVH algebra. Conversely we find that the change of variables \eqref{FLRW_conf} into the charges gives the $\sl(2,\R)$ generators \eqref{sl_2+R}
\bsub 
\be
\cL_n &= \tau^n\left ( \f{p_a^2 \pi \kappa^2}{9 a^2 \cV_0} - 4 \pi \f{p_\Phi^2}{3a^4 \cV_0} \right ) + n \tau^{n-1} \f{a p_a}{4} + n(n-1) \tau^{n-2} \f{9 a^4 \cV_0}{32 \pi \kappa^2}\\
&= \tau^n\left ( B^2 \pi V^{2/3} \cV_0^{1/3} \kappa^2 - 4 \pi \f{p_\Phi^2 \cV_0^{1/3}}{3 V^{4/3}} \right ) - n \tau^{n-1} \f{3 B V}{4} + n(n-1) \tau^{n-2} \f{9 V^{4/3}}{32 \pi \cV_0^{1/3} \kappa^2}\,,
\ee 
\esub
so that they measures the initial condition of the smeared CVH generators. Putting in evidence the lapse factor $N = -8 \pi \cV_0/3 G a = -8 \pi \cV_0^{4/3} /3 G V^{1/3}$ we have indeed
\bsub 
\be
\cL_{-1}(\tau=0) &= N \left (\f{p_\Phi^2}{V} -\f{3\kappa^2}{4} V B^2\right ) = - N H\,,\\
\cL_{0}(\tau=0) &= -\f{3 BV}{4}\,,\\
\cL_{1}(\tau=0) &= N^{-1} \f{3}{2\kappa^2} V \,,
\ee 
\esub
which are exactly proportional to the smeared Hamiltonian and volume, together with the complexifier. Despite the fact that the algebra that we have built here relies only on the properties of the supermetric via its isometries, the gravitational origin of the minisuperspace seems to be somehow carried over because at the end of the day we recover the kinematical CVH algebra. The subtlety is that the latter become a local symmetry only through a good choice of clock, that fixes the lapse as in \eqref{lapse_FLRW_k}, which in turns happens to rescale the physical volume by an inverse lapse factor.

The reason behind this coincidence remains mysterious. Due to the very different origin of the $\sl(2,\R)$ sector in terms of homothetic Killing vector \eqref{sl_2+R} with respect to the kinematical CVH algebra of general relativity we do not expect the two to coincide, even up to a rescaling factor. The only generator that is expected to be the same is of course the Hamiltonian $\cL_{-1}$. For example, it turns out that for the black hole the $\cL_{-1}$ and $\cL_{1}$ generators do measure the initial condition for the smeared Hamiltonian and volume, but $\cL_0$ is not proportional to the trace of the extrinsic curvature (see appendix \ref{app:4_LQG_phase_space}). For some of the Bianchi models (e.g. VI$_0$) we do not even have that $\cL_1$ is proportional to the volume of the slice. 

\ \newline 

In this chapter, we have reviewed the symmetry structure of some minisuperspaces in General Relativity. The simplicity of the models, coming from the reduction to a mechanical model has allowed us to probe the existence of a conformal structure living on top of the residual diffeomorphism invariance. This is associated with a set of conserved charges, that take a local form only for a particular choice of the clock, and are non-local otherwise. 

These rigid symmetries are associated with a transformation of physical parameters, like the mass of the black hole, or the location of the boundary. They somehow encode the scaling properties of the minisuperspace.  This signals the presence of a richer structure in gravity than the diffeomorphism covariance alone. The relationship between the appearance of new physical degrees of freedom and the boundary is not new and is known in general relativity for a while. Usually, the boundary degrees of freedom are however associated with an infinite tower of corner charges, provided for example by the ubiquitous $\bms_3$ group.  The natural question is to what extent we could make contact with these infinite-dimensional groups, and this is going to be the subject of the next chapter.

	\chapter{Infinite-dimensional extension}
\label{chap3}

We have already recalled that, in numerous situations, boundaries are responsible for the appearance of an infinite set of charges, manifesting itself in the ubiquitous BMS group.
Indeed, it is natural to investigate if the finite symmetry group exhibited in the previous chapter are maximal or whether they extend to a larger symmetry group. For instance we are tempted to generalise the $\ISO(2,1)$ subgroup for black holes to the BMS$_3$ group. 

We have already presented in the first chapter \ref{sec1.3:boundary} how it appears in the context of asymptotic flat three-dimensional spacetime, and we rapidly review some basic facts about its properties and its central extensions. The interested reader shall look at \cite{Oblak:2016eij} for a more exhaustive discussion.

This chapter is developed along the lines of the article \cite{Geiller:2021jmg}.

\section{Virasoro and BMS$_3$ groups}
\label{sec3.1:vir_bms_prop}
We recall that the  group BMS$_3$ is defined as $\text{BMS}_3=\text{Diff}(S^1) \ltimes \text{Vect}(S^1)_\text{ab}$, where $\text{Diff}(S^1)$ is the group of diffeomorphism of the unit circle, while $\text{Vect}(S^1)_\text{ab}$ is its Lie algebra seen as an Abelian vector group. The central extension of the diffeomorphisms of the circle $S^1$ is known as Virasoro group, and we start by reviewing its main properties. We denote by $\text{Vect}(S^1)$ the space of vector fields on the circle endowed with the bracket
\be
\big[f(\theta)\,\partial_\theta ,\,g(\theta)\,\partial_\theta\big]=(f g' -g f')\partial_\theta \,,\q\q 
f\,\partial_\theta ,\,g\,\partial_\theta \in \text{Vect}(S^1)\,,
\ee
where $f$ and $g$ are periodic functions of $\theta$ with period $2 \pi$. A standard basis is given by the Fourier modes $\ell_n\coloneqq  i e^{in \theta}\partial_\theta$, with the commutation relations $[\ell_n,\ell_m] = -i(n-m) \ell_{n+m}$. This is the Witt algebra. Moreover, the vector space $\text{Vect}(S^1)$ can been seen as the space of generators of the (orientation preserving) diffeomorphisms on the circle, denoted by $\text{Diff}^+(S^1)$. Conversely, the diffeomorphisms can be endowed with a Lie group structure, whose algebra is $\text{Vect}(S^1)$. 

We denote by $\cF_\lambda$ the space of tensor densities of degree $\lambda\in \R$. It consists of elements of the form $\alpha=\alpha(\theta) \de \theta^\lambda$. Diffeomorphisms act by the adjoint action as
\be
\alpha(\theta) \xrightarrow{\varphi\, \in\, \text{Diff}^+(S^1)}  \text{Ad}^*_{\varphi^{-1}} \alpha(\theta) = \left (\varphi'(\theta)\right )^\lambda \alpha\left (\varphi(\theta)\right ).
\ee
The infinitesimal version gives the adjoint action of $\text{Vect}(S^1)$ on densities, i.e. the Lie derivative of the differential along a vector field. Taking $\varphi=\mathbb{I}+\epsilon X$, with $\epsilon \to 0$ and $X=X(\theta)\partial_\theta$, we get
\be
L^{(\lambda)}_X \alpha\coloneqq \text{ad}^*_{X} \alpha = X \alpha'+\lambda X' \alpha.
\ee
A first remark is due here, we would like indeed to stress the formal analogy with this definition and the conformal weight in the transformation law of the black hole metric coefficients under a M\"obius transformation (see e.g. \eqref{null_conf_infinit_transf}). We can naturally define a bilinear form on $\cF_\lambda \times \cF_{1-\lambda}$ which is invariant under Lie derivative as
\be
\braket{\alpha,\beta}\coloneqq \int_{S^1} \alpha \otimes \beta\,, \q\q  \forall\ \alpha \in \cF_\lambda\,,\ \beta \in \cF_{1-\lambda}\,.
\ee
We then identify the dual of $\text{Vect}(S^1)$ with the space of quadratic differentials $p= p (\theta)\de \theta^2 \in \cF_2$, and we have
\be
\braket{p,v}\coloneqq \int_{0}^{2\pi}\de \theta\, p(\theta) v(\theta)\,, \q\q \forall\ v \in \cF_{-1}\,.
\ee
This means that the action of the vector fields on $\cF_2$ coincides with the co-adjoint action, i.e.
\be\label{coadjoint_def}
\braket{\text{ad}^*_{X}(p),Y}=-\braket{p,[X,Y]}\,,\q  \forall\ X,Y \in \cF_{-1}\,,\ p \in \cF_{2}\,.
\ee

We recall then that the only non-trivial cocycle on the algebra $\text{Vect}(S^1)$ is given by the Gelfand--Fuchs 2-cocycle
\ba
\omega : \text{Vect}(S^1) \times \text{Vect}(S^1) &\to& \R\\
(X,Y) &\mapsto& \int \de \theta\, X' Y'' = \f 1 2 \int \de \theta (X' Y''-X''Y')\,.  \notag
\ea
The Virasoro algebra is the central extension of the algebra of vector fields, defined on the vector space $\widehat{\text{Vect}(S^1)}\coloneqq\text{Vect}(S^1) \oplus \R$, with the bracket
\be
\big[(X,a),(Y,b)\big]\coloneqq \big([X,Y], \omega(X,Y)\big)\,, \q\q \forall\ a,b \in \R\,,\ X,Y \in \cF_{-1}\,.
\ee
The corresponding group is introduced via the Bott--Thurston 2-cocycle in the group $\rm {Diff}^+(S^1)$, which is
\ba
B : \text{Diff}^+(S^1) \times \text{Diff}^+(S^1) &\to& \R\\
(\varphi,\psi) &\mapsto&\f{1}{2} \int \de \theta \log (\varphi' \circ \psi) \left  (\log(\psi')\right )'\,. \notag
\ea
We then define the Virasoro group as $\widehat{\text{Diff}^+}(S^1)=\text{Diff}^+(S^1) \times \R$ equipped with the group law
\be
(\varphi,a) \circ (\psi,b) = \big(\varphi \circ \psi, b + B(\varphi,\psi)\big) \,.
\ee
Consistently, we recover the Virasoro algebra by the observation that the infinitesimal limit of $B$ is $\omega$. Indeed, defining the flows $\varphi_t$ and $\psi_s$ corresponding respectively to the vector fields $X(\theta)\partial_\theta$ and $Y(\theta)\partial_\theta$, we have
\be
\omega(X,Y) = \f{\de^2}{\de s\,\de t}  \big (B(\varphi_t,\psi_s)-B(\psi_s,\varphi_t)\big ) \Big|_{t=s=0}\,.
\ee

With the extended Lie bracket and the bilinear form at our disposal, we can compute the coadjoint action of the Virasoro algebra on quadratic forms. Considering the adjoint action of the algebra on itself, we define its coadjoint action as in \eqref{coadjoint_def}. A straightforward calculation leads to
\be
\text{ad}^*_{(X,a)} (p,c) = \big(p'X +2 X'p - c X^{(3)} ,0\big)\,, \q\q \forall\ X \in \cF_{-1}\,,\ p \in \cF_{2}\,, \ a,c \in \R\,.
\ee 
Its exponential gives the group coadjoint action
\be
\text{Ad}^*_{(f^{-1},a)} (p,c) = \big(p \bullet f - c \,\cS[f],\,c \big)\,, \q\q \forall\ f \in \rm {Diff}^+(S^1)\,,\ p \in \cF_{2}\,, \ a,c \in \R\,,
\ee
where
\be
p \bullet f = f'(\theta)^2(p \circ f)\, \de \theta^2\,, \q\q \cS[f]= \text{Sch}[f]\de \theta^2= \left [\f{f^{(3)}}{f'}-\f{3}{2}\left (\f{f''}{f'}\right )^2\right ]\, \de \theta^2\,.
\ee
Finally, introducing $h=f'$, we recall that the Schwarzian derivative satisfies
\bsub
\be
\text{Sch}[f]&=\f{h''}{h}-\f{3}{2}\left (\f{h'}{h}\right )^2 = (\log h)'' -\f{1}{2} \big((\log h)'\big)^2\,,\\
\text{Sch}[f\circ g]&=\text{Sch}[g] +(g')^2\,\text{Sch}[f]\circ g\,,\\
\text{Sch}[f^{-1}]&=-\f{\text{Sch}[f]}{h^2}\,.
\ee
\label{schw_dev}
\esub
The Schwarzian derivative is the unique solution of the cocycle equation \cite{Matone:2005qc}.

\paragraph{The BMS$_3$ group}

The centrally-extended BMS$_3$ group is the semidirect product, under the adjoint action, of the Virasoro group and its algebra seen as an Abelian vector group. This is
\be
\widehat{\text{BMS}_3}=\widehat{\text{Diff}^+}(S^1) \ltimes_\text{Ad} \widehat{\text{Vect}}(S^1)_\text{ab}\,.
\ee
We will (abusively) refer to the first factor as superrotations, and to the second one as supertranslations. Its elements are quadruples $(f,a;g,b)$ with $f \in \text{Diff}^+(S^1)$, $g \in \text{Vect}(S^1)_\text{ab}$, and $a,b \in \R$. The adjoint action of the Virasoro group reads
\be
\text{Ad}_f (g,b)=\left ( (g f') \circ f^{-1}, b - \int \de \theta\, \text{Sch}[f] g  \right )\,.
\ee
The corresponding algebra is
\be
\widehat{\mathfrak{bms}_3}=\widehat{\text{Vect}}(S^1) \loplus_\text{ad} \widehat{\text{Vect}}(S^1)_\text{ab}\,,
\ee
whose elements are again quadruples $(X,a; \alpha,b)$ with $X \in \text{Vect}(S^1)$, $\alpha \in \text{Vect}(S^1)_\text{ab}$, and $a,b \in \R$. The commutation relations are
\be
\label{BMS_algebra}
\big[ (X,a;\alpha,b), (Y,r;\beta,s) \big ]=
\big ([X,Y], \omega(X,Y); [X,\beta]-[Y,\alpha],\omega(X,\beta) -\omega(Y,\alpha) \big )\,.
\ee
As always, we can define an adjoint action of the group on its algebra by conjugation, which simply consists in exponentiating the commutation relation above. In the specific case studied here this gives
\be
\label{BMS_adjoint}
\text{Ad}_{(f; g)} (X;\alpha)=\left ( \text{Ad}_f X; \text{Ad}_f \alpha + [\text{Ad}_f X,g]\right )\,.
\ee
We are now interested in the coadjoint representation. The dual of the algebra $\widehat{\mathfrak{bms}_3}$ is the space $\widehat{\text{Vect}}(S^1)^* \loplus \widehat{\text{Vect}}(S^1)_\text{ab}^*$, paired with the algebra element via the bilinear form
\be
\label{BMS_bilinear}
\lb (\cJ,c_1;\cP,c_2),(X,a;\alpha,b) \rb\coloneqq \int_{0}^{2\pi} \left(\cJ X+ \cP\alpha\right ) \de \theta +c_1 a +c_2 b\,.
\ee
This leads to the coadjoint representations of $\widehat{\text{BMS}_3}$, which is
\be
\text{Ad}^*_{f^{-1},g} (\cJ,c_1;\cP,c_2) =\left (\tilde \cJ,c_1; \tilde \cP,c_2 \right)\,,
\ee
with
\be
\label{BMS_coadj}
\tilde \cP=f'^2(\cP \circ f) - c_2 \,\text{Sch}[f]\,,\q\q \tilde{\cJ} =   f'^2\left(\cJ+g\cP' +2 g'\cP - c_2 g^{(3)}\right )\circ f - c_1 \,\text{Sch}[f]\,.
\ee

This is what we use in the section \ref{sec1.3:boundary}, where the mass and angular momentum aspect are rewritten in terms of $\cJ$ and $\cP$ transforming exactly as covectors under the BMS group. Therein the central charges are $c_2 =0$ and $c_1 = \f{1}{4\pi G}$. 

Given an element of the coadjoint algebra, we define its orbit as the subset of coadjoint element that are attainable from the starting point via the action of a group element. For a Lie algebra $\h$ of the group H and a coadjoint vector $p_0 \in \h^*$, its orbit is 
\be
\cO_{p_0} = \left  \{p\in \h^* | \exists\, h \in H\; \rm{such\; that}\; p= Ad^*_H p_0\right \} 	\,.
\ee
The stabilizer of the orbit, or little group, is the subset of elements that leaves the orbit points invariant
\be
H_0(p_0) = \left  \{ h \in H | Ad^*_H p_0 =p_0 \right \} 	\,.
\ee
The orbit is then isomorphic to $\cO_{p_0} = H/H_0(p_0)$. The understanding of coadjoint orbits of the symmetry groups is crucial in order to go towards the quantum theory. In particular, the classification of the orbits of semi-direct product groups by their little groups will naturally lead to irreducible representations of the full group. In the case of the BMS$_3$ group, a detailed discussion is presented in \cite{Oblak:2016eij}. The key ingredient that makes the contact between coadjoint classes and representation theory is a result from Kirillov, Kostant and Souriau \cite{souriau1970structure, BFb0079068, Kirillov_2004, Kostant2009}. This states that the coadjoint orbits are naturally endowed with a symplectic structure. Without surprise, in gravity this happens to coincide with the gravitational phase space, as is manifest by comparing the form of the conserved charges \eqref{BMS_gravity_charges} with the bilinear form \eqref{BMS_bilinear}. This implies that the infinitesimal Lie algebra of variation seen on the coadjoint vectors (i.e. the mode expansion of the infinitesimal version of \eqref{BMS_coadj}) reproduces the charge algebra \eqref{bms_modes}.

\section{Black holes and adjoint representation of BMS${}_3$}
\label{sec3.2:bms_adjoint}

We now set out to explain the relationship between the action \eqref{Bh_lagran} and BMS${}_3$. In the previous chapter \ref{sec2.3:BH_minisup} we have shown that the dynamical system describing the black hole interior via \eqref{Bh_lagran} admits a set of symmetries whose finite form is isomorphic to the 3-dimensional Poincar\'e group ISO$(2,1)$. This defines the global part of the BMS$_3$, i.e. its maximal finite-dimensional subalgebra. The idea is then to try to extend the generators $\cL_n$ \eqref{sl_2+R} and $\cT^+_n$ \eqref{trans_gen_null} to have any possible integer value $n$.

Before starting, we shall remark that the mode expansion is defined in two different ways in the usual gravitational picture and the minisuperspaces. While the diffeomorphisms that define the $\BMS_3$ group live on a circle, the M\"obius transformations and translations $\cT_n$ for the black hole depend on a time coordinate with no periodic condition. The two are mapped onto each other by the (de-)compactification 
\be
\tau = \arctan \theta/2\,.
\label{decompact}
\ee
This pushes the periodicity in $\theta$ at $\pm \infty$ for $\tau$, opening the circle to the real line. It changes the mode expansion from the Fourier series to the Laurent expansion in powers of time changing by a complex factor the charge algebra \eqref{bms_modes}, but has also the effect of changing the map between coadjoint orbits and irreducible representations. For example the vacuum orbit $\cP_0=-\f{c_2}{2}$ of 3d gravity \cite{Oblak:2016eij} is mapped to $\cP_0=0$ in the minisuperspace.

From this perspective, it becomes then crucial that the time coordinate spans the whole real line, and does not restrict itself to the interior region of the black hole (as already discussed in \ref{sec2.3:BH_minisup}).

The natural question is now to understand if and how it is possible to reformulate the Kantowski--Sachs minisuperspace model in terms of the (centrally-extended) BMS group. For this, we will first discuss again the symmetries of the action \eqref{Bh_lagran} using the adjoint representation, and then explain in the next section the relationship with coadjoint representations. The central extensions were defined in the previous section \ref{sec3.1:vir_bms_prop} by the notation $\widehat{G}$ for the extension of the group $G$. In the following, however, we drop this hat in order to deal with lighter notations.

\subsection{Vector fields and transformations of the action}

In the previous chapter \ref{sec2.3:BH_minisup} we have discussed the symmetries of the action \eqref{Bh_lagran} starting from the conformal properties of the supermetric. They take the form \eqref{iso_bh} for the Poincar\'e group acting on the square root of	the metric coefficients $a$ and $b$. In particular, they transform as conformal fields of respective weight $0$ and $1/2$ under a M\"obius reparametrization of proper time. In this section it happens to be easier to work with a different parametrization of the field space:
\be
V_1= b^2= \f{G v^2}{4 \cV_0} \,,\q\q V_2 = 2a^2 b^2 = \f{G uv}{2 \cV_0}\,.
\label{new_fields_v}
\ee
 For sake of completeness we  report here the relationship between the $V_i$'s and the null coordinates ($u,v$) on the configuration space\footnote{The reason for the notation $V_i$ is two-fold. On the one side it turns out that $V_2 \propto V_\rm{bb}$ measuring the smeared volume on the slice. Secondly the fields $V_i$ are going to be embedded into the Lie algebra $\bms_3 =\rm{Vect}(S^1) \loplus \rm{Vect}(S^1)_\rm{ab}$.}. The mechanical action \eqref{Bh_lagran} takes then the form 
\be
\label{new_lapse_action}
\cS_0 =\f{\cV_0}{G} \int  \de \tau \left [ \f{\dot V_1 (V_2 \dot V_1 - 2 V_1 \dot V_2)}{2V_1^2} \right ]\,.	
\ee 
From the transformations laws \eqref{iso_bh} together with the reparametrization \eqref{new_fields_v}, we read the variations on the $V_i$'s:
\bsub
\be
\cL[f]:\,&\left|
\begin{array}{lcl}
V_i &\mapsto& \widetilde V_i = (\dot{f}\, V_i)\circ f^{-1}\,,
\end{array}
\right.\\
\cT^+[g]:\,&\left|\begin{array}{lcl}
V_{1}&\mapsto& \widetilde{V}_{1}=V_{1}\,,\\
V_2 &\mapsto& \widetilde {V}_2  = V_2 + g\dot{V}_1  -  \dot{g}V_1\,,
\end{array}\right.\ee
\label{fgreparam}\esub
where we recall that all these quantities depend on $\tau$. In principle they are limited to be a M\"obius function for $f$ and a second degree polynomial for $g$, but we are going to consider here general functions. This is the most natural way of generalising to an infinite dimensional extension. They are given here in a form acting only on the dynamical fields $V_i$. To compute the transformation of the action it will be convenient to act instead with $\cL_f$ equivalently rewritten as
\be\label{freparam}
\cL[f]:\,\,\left|
\begin{array}{lcl}
\tau &\mapsto& \tilde{\tau} = f(\tau)\vspace*{1mm}\\
V_i &\mapsto& \widetilde V_i (\tilde{\tau}) = \dot f(\tau)\, V_i(\tau)
\end{array}
\right.\,,
\qquad\textrm{so}\q
\left|
\begin{array}{lcl}
\de \tau &\mapsto& \de \tilde{\tau} = \dot f \de \tau
\vspace*{1mm}\\
 {\de_{\tau} V_i } &\mapsto&
\de_{\tilde{\tau}} \widetilde V_i = {\de_{\tau} V_i } + V_i\,\de_{\tau}\ln \dot f
\end{array}
\right.\,.
\ee
We will also drop the index $(+)$ from the abelian transformation in order to lighten the notation. Following the terminology introduced for asymptotic symmetries, we will (abusively) refer to the transformations $\cL[f]$ as \textit{superrotations} and to $\cT[g]$ as \textit{supertranslations}.
In order to justify this name, we now show that the composition
\be\label{Ad definition}
\text{Ad}_{f,g}\coloneqq \cT[g]\circ \cL[f]
\ee
does indeed define a representation of BMS$_3$. For this, we use the properties
\bsub
\be
\cL[f^{-1}]\circ \cT[g]\circ \L[f]&=\cT[(g\circ f)/\dot{f}],
\\
\cL[f_1]\circ \cL[f_2]&=\cL[f_1\circ f_2]\,, \\
\cT[g_1]\circ \cT[g_2]&=\cT[g_1+g_2]\,,
\ee
\esub
which allow us to show that the composition law for \eqref{Ad definition} is given by
\be\label{BMScomposition}
\text{Ad}_{f_1,g_1}\circ \text{Ad}_{f_2,g_2}
&= \cT[g_1]\circ \cL[f_1]\circ \cT[g_2]\circ \cL[f_2]\cr
&=\cT[g_1]\circ (\cL[f_1]\circ \cT[g_2]\circ \cL[f_1]^{-1}) \circ (\cL[f_1]\circ \cL[f_2]) \cr
&=\cT[g_1+\dot{f}_1\,g_2 \circ f_1^{-1}]\circ \cL[f_1\circ f_2]. 
\ee
The inverse composition is
\be
(\cT[g]\circ \cL[f])^{{-1}}
=
\cL[f^{-1}]\circ \cT[-g]
=
\cT[-(g\circ f)/\dot{f}]\circ \cL_[f^{-1}]
\,.
\ee
As desired, this is the group multiplication and its inverse for the Lie group defined as the semi-direct product $\text{Diff}(S^1) \ltimes \text{Vect}(S^1)_\text{ab}$. The interpretation is then that the gravitational variables $V_1$ and $V_2$ belong to the Lie algebra
\be
\mathfrak{bms}_3=\text{Vect}(S^1) \loplus_\text{ad} \text{Vect}(S^1)_\text{ab}\ni(V_1,V_2)\,,
\ee
and are acted on by the group via the adjoint action \eqref{BMS_adjoint}.
 
\paragraph{Finite transformation of the action}

We now study how the action \eqref{new_lapse_action} transforms under finite BMS transformations. It turns out to be easier to compute this using the inverse transformations to \eqref{fgreparam}. Acting jointly with superrotations and supertranslations we find
\be
\label{transfo_action}
\text{Ad}_{(f^{-1},g)} \cS_0[V_i,\tau] &=\cS_0[V_i,\tilde \tau] - \f{\cV_0}{G}\int \de \tilde \tau\left [\vphantom{\f{\cV_0}{G}}(V_2-\dot{g}V_1+g \dot{V}_1)\,\text{Sch}[f^{-1}]-V_1g^{(3)} \right.\\
&\phantom{=\cS_0[V_i, \tilde \tau] +\f{1}{G} \int \de \tilde \tau\left[\right.~}+\left.\f{\de}{\de\tilde{\tau}} \left (\f{\ddot f}{\dot{f}}(V_2-\dot{g}V_1+g \dot{V}_1)+ \ddot{g}V_1 - \f{g \dot{V}_1^2}{2V_1} \right ) \right ]\,,\notag
\ee
where $\text{Sch}[\,\cdot\,]$ denotes the Schwarzian, whose properties have been previously recalled in \eqref{schw_dev}. As expected, this expression has cross terms involving both $f$ and $g$ since we have acted jointly with $\text{Ad}_{f,g}$, and the group is a semidirect product of the two sectors. Focusing on the sector $f$, we see immediately that the subgroup of transformations with vanishing Schwarzian derivative are symmetries of the reduced gravitational action. These are indeed given by M\"obius transformations, isomorphic to the group SL$(2,\R)$. Similarly, in the $g$ sector we obtain a symmetry if and only if $g$ has a vanishing third derivative, meaning that it is a second degree polynomial. This shows that the subgroup of \eqref{fgreparam} which describes the symmetries of \eqref{new_lapse_action} is given by $\text{SL}(2,\R)\ltimes\R^3$, which is the $2+1$ Poincar\'e group \cite{Geiller:2020xze, Geiller:2021jmg}.

When the transformation $\text{Ad}_{f,g}$ does not belong to the Poincar\'e subgroup, we see that the general BMS transformation produces terms proportional to the dynamical variables $V_1$ and $V_2$ in the action. We can see the same two terms in $V_1$ and $V_2$  appearing when introducing a cosmological constant or a scalar field in the model. For this we shall use the line element \eqref{Bh_minisup_line} with fixed lapse\eqref{Bh_minisup_line:null_conf} in terms of the $V_i$'s:
\be
\de s^2_\rm{BH} = -\f{V_1}{2V_2} \de \tau^2 + \f{V_2}{2V_1} (\de x + N^1 \de \tau)^2 + L_s^2 V_1 \, \de \Omega^2\,.
\label{Bh_minisup_line:Vi}
\ee
Working with proper time $\tau$ and introducing the IR cutoff $L_0$ to regulate the integration over $x$, the inclusion of a cosmological constant is done by adding to the Einstein--Hilbert action the volume term
\be\label{cosmo_const_in_S}
\cS_\Lambda = -\f{1}{8 \pi G}\int \de^4 x\,  \sqrt{-g}\, \Lambda= -\f{\cV_0}{G}\int \de \tau \, V_1\Lambda\,.
\ee
Similarly, the contribution of a minimally-coupled scalar field $\Phi$ in the spherically-symmetric and homogeneous spacetime \eqref{Bh_minisup_line:Vi} is described by the action
\be\label{scalar_field_in_S}
\cS_\phi = -\f{1}{16 \pi G}\int \de^4x\, \sqrt{-g} \left ( \f{1}{2} (\nabla \Phi)^2- \cV(\Phi) \right )= \f{\cV_0}{2G}\int \de \tau \left (\dot \Phi^2V_2  +V_1 \cV(\Phi)\right )\,.
\ee
Comparing these two actions with the transformation \eqref{transfo_action}, one can see that for pure supertranslations (i.e. with $f=\mathbb{I}$) we can interpret the third derivative of $g$ as an effective cosmological constant, while for pure superrotations (i.e. with $g=0$) we can be tempted to interpret the term created by the superrotation as the kinetic term of a scalar field (up to a redefinition $\text{Sch}[f^{-1}]\mapsto\dot{\phi}^2$). However, this interpretation does not strictly hold since, as can be seen on \eqref{scalar_field_in_S} a scalar field comes with specific time derivatives which cannot be obtained from $f$ and $g$ in \eqref{transfo_action}. Nonetheless, in the case of pure supertranslations, $g$ one can really interpret the term $g^{(3)}V_1$ as the contribution of a cosmological constant. The action of a further finite supertranslation then shifts the value of this cosmological constant. This is an important property that turns the BMS transformation into solution-generating maps, that in particular provide a conformal bridge between (A)dS black holes and the flat Schwarzschild solution. By taking the supertranslation with $\dddot g=const$, we are indeed able to turn on and off the cosmological constant or move between arbitrary values of $\Lambda$. This has been used for instance in \cite{Achour:2021dtj} to generalise the M\"obius covariance of the flat solution to the Schwarzschild-(A)dS, This is analogous to what happens in cosmology \cite{Gibbons:2014zla,BenAchour:2020xif}, even if therein the cosmological term is produced by a transformation with constant Schwarzian.

As an alternative to this interpretation in terms of a cosmological constant and a scalar field, it turns out that it is possible to write down a fully BMS invariant action by including extra fields in the theory. Indeed, once one realizes that the theory is not strictly invariant under the BMS group but only under its Poincar\'e subgroup, the natural question is whether this can be understood as the result of a symmetry breaking of a BMS-invariant theory. The symmetry reduction from BMS to Poincar\'e would then turn gauge degrees of freedom of the invariant theory into physical degrees of freedom of the black hole interior theory. Here, to his purpose, we introduce new fields that  transform in a specific way so as to compensate the Schwarzian and the $g^{(3)}$ terms produced by finite BMS transformations. This is achieved by the augmented action
\be
\label{invariant_1d_action}
\cS_\text{inv}[V_i,\Phi,\Psi] \coloneqq\cS_0[V_i] + \f{ \cV_0}{G}\int  \de \tau\left [(V_2+ \dot \Psi V_1- \Psi \dot{V}_1) \text{Sch}[\Phi]  + V_1\Psi^{(3)}  \right ],
\ee
provided the new fields $\Phi$ and $\Psi$ (which can be seen as St\"uckelberg fields for the BMS symmetry) transform as
\be
\label{FGreparam}
\text{Ad}_{f,g}:\,&\left|
\begin{array}{lcl}
\Phi &\mapsto& \widetilde \Phi = \Phi \circ f^{-1}\\
\Psi &\mapsto& \widetilde \Psi = g + (\dot{f} \Psi)\circ f^{-1}\\
\end{array}
\right.\,.
\ee
This BMS invariant action can be seen as the extension from the Virasoro group to the BMS group of the conformally-invariant action for FLRW cosmology that was introduced in \cite{BenAchour:2020xif}. The cosmological model can be understood as the $V_1=0$ regime of the theory with $V_2$ playing the role of the cosmological volume. In that case, there are no supertranslations but only superrotations and the conformally-invariant cosmological action was defined by introducing solely an extra $V_2\,\text{Sch}[\Phi]$ term.

In terms of group representation, this means that the fields are embedded in the BMS${}_3$ group itself:
\be
(\Phi,\Psi)\in\text{BMS}_3=\text{Diff}(S^1) \ltimes_\text{Ad} \text{Vect}(S^1)_\text{ab}.
\ee
This will inevitably modify the equations of motion for the initial fields $V_i$, providing again a solution-generating tool. We will come back to this point in the next sections while discussing the coadjoint orbits of the BMS group (see \cite{Geiller:2021jmg} for a more complete discussion).

Note that here we have simply introduced the BMS invariant action as an example of a 1-dimensional system with full BMS symmetry. This is to be contrasted with the 2-dimensional BMS--Liouville theory \cite{Merbis:2019wgk, Barnich:2017jgw, Barnich:2012rz} which will also have a close relationship with the action $\cS_0$ we started from.

\subsubsection{Infinitesimal transformations}

We now move on to the study of the infinitesimal transformations. Since the supertranslations are Abelian, their infinitesimal generator, which we will denote $\alpha$, acts in the same way as the finite generator $g$. For the superrotations, we parametrize the transformation infinitesimally as $f(\tau)= \tau + \epsilon X(\tau)$. The infinitesimal version of \eqref{fgreparam} reproduces by construction the $\mathfrak{bms}_3$ algebra \eqref{BMS_algebra}. We  have indeed 
\be\label{X_infinit}
\delta_X V_i:=\frac{\text{Ad}_{f,0} V_i -V_i}{\epsilon}= \dot XV_i-  X\dot{V}_i = [V_i,X],
\ee
and
\bsub\label{a_infinit}
\be
\delta_\alpha V_1 &= 0,\\
\delta_\alpha V_2 &= \alpha \dot{V}_1- \dot \alpha V_1 = -[V_1,\alpha].
\ee
\esub
With this, the infinitesimal variation of the action is
\be
\delta_{X,\alpha} \cS_0 =& \f{ \cV_0}{G}\int \de \tau \left[V_2 X^{(3)}-V_1 \alpha^{(3)} \right ] \\&+ \f{ \cV_0}{G} \de\tau\f{\de}{\de\tau}\left[\f{\dot{V}_1}{2V_1^2}\Big(2XV_1\dot{V}_2-\dot{V}_1(\alpha V_1+XV_2)\Big)+\ddot{\alpha}V_1-\ddot{X}V_2\right].\notag
\ee
Once again, we see that the Poincar\'e subgroup, which at the infinitesimal level has generators such that $X^{(3)}=0=\alpha^{(3)}$, generates symmetries of the system since this variation then reduces to a total derivative.

\subsection{Hamiltonian formulation and generators}
\label{subsec_hamilt_gen}

Having studied the action of finite and infinitesimal BMS transformations on our system $\cS_0$, we can now study the generators in the (covariant) Hamiltonian formulation. From the general variational expression of the action we can read the pre-symplectic potential $\theta$ and the equation of motions. These are explicitly given by
\bsub 
\be 
\delta \cS_0 =&  \int \de \tau \left [\cJ \delta V_1 + \cP \delta V_2 + \de_\tau \theta\right ]\,, \notag\\
\cJ&\coloneqq \f{\cV_0}{G}\left (\f{V_2 \dot{V}_1^2}{V_1^3} -\f{\dot{V}_1 \dot{V}_2}{V_1^2}- \f{V_2 \ddot V_1}{V_1^2} +\f{\ddot V_2}{V_1}\right) \,,\label{J def}\\
\cP&\coloneqq\f{\cV_0}{G}\left (\f{\ddot V_1 }{V_1}- \f{\dot{V}_1^2}{2 V_1^2}\right )\,,\label{P def}\\
\theta&\coloneqq\f{\cV_0}{G}\left (\f{\dot{V}_1 V_2 - \dot{V}_2 V_1}{V_1^2} \delta  V_1-\f{\dot{V}_1}{V_1} \delta V_2\right )\,.\label{V potential}
\ee
\label{J P theta}
\esub
The equations of motion are $\cJ=0=\cP$, and one can note that in terms of $W$ defined by $\de_\tau W=1/V_1$ we have $\cP=-\text{Sch}[W]$. From $\theta$ we find the symplectic structure $\omega := \delta \theta$, and read the canonical momenta by $\omega = \delta P_i \delta V_i$
\be\label{phase_space}
P_1 =\f{\cV_0}{G} \f{\dot{V}_1 V_2 -\dot{V}_2 V_1 }{V_1^2} \,,\q \q P_2=-\f{\cV_0}{G}\f{\dot{V}_1}{V_1}\,,\q\q\{V_i, P_j\}= \delta_{ij} \,,
\ee
that, without surprise, agree with the supermetric derivation $p_\mu = g_{\mu\nu} \dot x^\mu$. 
With the symplectic structure at our disposal, we can ask whether there are integrable generators associated with the infinitesimal transformations \eqref{X_infinit}, \eqref{a_infinit}. This is found by contracting $\omega$ with the variations $\delta_{X,\alpha}$. We adopt here the notations of the covariant phase space formalism, recalled in section \ref{sec1.3:boundary}, where we see the field variations as forms and the transformations as vectors. The contraction of the variation $\delta_X$ with the field-space two-form $\omega$ is then defined as $\delta_X\ipp \omega:=\omega[\delta_X,\delta]-\omega[\delta,\delta_X] \,$. We then use the equations of motion to get rid of the second order derivatives, and the definitions \eqref{phase_space} to rewrite the first derivatives in terms of canonical momenta. Quite surprisingly, even if the general BMS transformations do not describe symmetries of the action, and we cannot make use of the Noether theorem, they have an integrable generator. Separating the superrotations and supertranslations, we find
\bsub
\be
\delta_{X} \ipp \omega &= - \delta_{X} P_i\, \delta V_i + \delta P_i\, \delta_X V_i \cr
&\approx -\delta \left (\f{G}{\cV_0} P_2 \left (P_1 V_1 + \f{1}{2} P_2 V_2\right ) X + (V_2 P_2 +V_1 P_1) \dot X + \f{\cV_0}{G} V_2 \ddot X \right )\cr
&:= -\delta \cL_X \label{superrot_gen} \,,\\
\delta_{\alpha} \ipp \omega &= - \delta_{\alpha} P_i\, \delta V_i\, + \delta P_i\, \delta_\alpha V_i \cr
&\approx \delta \left (\f{G}{2\cV_0} {P_2}^2 V_1 \alpha + P_2 V_1   \dot \alpha + \f{\cV_0}{G}V_1 \ddot \alpha \right )\cr
&:= -\delta \cT_\alpha \label{supertran_gen} \,.
\ee
\esub
We can use the inverse of the change of variables \eqref{new_fields_v} to recast the generators in terms of the null configuration fields. Expanding the functions $\alpha$ and $X$ in power series ($\tau^{n+1}$) we obtain exactly the same expressions as in \eqref{sl_nform} and \eqref{trans_gen_null}, but with a generic integer $n$.

It is now interesting to compute the Poisson brackets between these generators. This can be one either by contracting twice the symplectic structure with a variation $\delta_{X,\alpha}$, or using the phase space expression for the generators and the usual definition of Poisson brackets. For example we have
\be
\{\cL_X, \cL_Y\}=- \delta_{X} \ipp \delta_{Y} \ipp \omega = \f{\de \cL_X}{\de V_i}\f{\de \cL_Y}{\de P_i}-\f{\de \cL_X}{\de P_i}\f{\de \cL_Y}{\de V_i}\,. 
\ee
Explicitly, we find
\bsub\label{poisson_bms}
\be
\{\cL_X, \cL_Y\}&=-\cL_{[X,Y]} + \f{\cV_0}{G} \big(X Y^{(3)}-Y X^{(3)} \big) V_2\,,\\
\{\cT_\alpha, \cT_\beta \}&=0\,,\\
\{\cD_X, \cT_\alpha \}&=-\cT_{[X,\alpha]} + \big(\alpha X^{(3)} - X \alpha^{(3)}\big) \f{\cV_0}{G} V_1\,,
\ee
\esub
where $[X,Y]=X\dot{Y}-Y\dot{X}$ and similarly for $[X,\alpha]$. These brackets are consistent with the equations of motion. For example, one can study the superrotation associated with a constant shift in time. This actually corresponds to the Hamiltonian of the system, i.e. $\cL_{X(\tau)=1}=-H$, as one can also check using a Legendre transform of the Lagrangian. Note that, as already pointed out, this is slightly peculiar since in the usual spacetime picture the Hamiltonian belongs to the Abelian sector of the symmetry algebra (i.e. Poincaré time translations), whereas here it belongs to the superrotations. For a general (time-dependent) phase space function $\cO(\tau)$, the time evolution is given by
\be 
\dot \cO\eqqcolon\de_\tau\cO = \partial_\tau \cO + \{\cO,H\}\,.
\ee
Using the commutation relations \eqref{poisson_bms} applied to $\cD_1$, and the fact that e.g. $\partial_\tau \cD_X = \cD_{\dot X}$, we then find
\bsub
\be
\de_\tau \cL_X &= \partial_\tau \cL_X-\{H,\cL_X\}= \f{\cV_0}{G}V_2 X^{(3)},\\
\de_\tau \cT_\alpha &= \partial_\tau \cT_\alpha-\{H,\cT_\alpha\}= V_1 \f{\cV_0}{G} \alpha^{(3)},
\ee
\esub
which, as expected, say that the generators are conserved only if the third derivative of the functions $\alpha$ and $X$ vanish. In spite of this consistency, we see that the brackets \eqref{poisson_bms} between the generators $\cL$ and $\cT$ do not close because of a remaining field-dependency on the right-hand side, although up to the $V_i$'s this would have been the centrally-extended BMS$_3$ algebra \eqref{bms_modes}. Rewriting the brackets in the mode expansion we have
\bsub\be
\{\cL_n , \cL_m\} &= (n-m) \cL_{n+m} +\tau^{m+n-1} \f{\cV_0 V_2}{G} (m^3 - m -n^3-n)\,,\\
\{\cT_n , \cT_m\} &= 0\,,\\
\{\cT_n , \cL_m\} &= (n-m) \cT_{n+m} +\tau^{m+n-1} \f{\cV_0 V_1}{G} (m^3 - m -n^3+n)\,.
\ee\esub
The non-closure of the generator algebra is to be expected since $\cL$ and $\cT$, although they are integrable, do not generate symmetries of the theory. Consistently, we recover a closed subalgebra for the Poincar\'e generators ($n \in \{\pm 1,0\}$) since these have a vanishing third derivative.

One way to obtain charges closing on the BMS$_3$ algebra is of course to consider a BMS invariant field theory. Although we have proposed a 1-dimensional such theory in \eqref{invariant_1d_action}, it turns out that a simpler field theory that has an interesting relationship with the black hole action \eqref{new_lapse_action} is the 2-dimensional BMS--Liouville theory of \cite{Merbis:2019wgk, Barnich:2017jgw, Barnich:2012rz}, that describe the boundary dynamics of 3D gravity. By a suitable redefinition of the fundamental fields therein, the two dimensional action is \footnote{In \cite{Merbis:2019wgk,Barnich:2017jgw,Barnich:2012rz} the action has the form 
\[
\cS_\rm{BMS} = \iint \de u\, \de \tau \big[ \partial_\tau\xi\, \partial_u \varphi - (\partial_\tau \varphi)^2 \big]\,.
\]
This is mapped to the one in the text by 
\[
\varphi=-\log V_1\,, \q\q \xi = \f{V_2}{V_1}\,.
\]
The fields now depend on two variables, but the symmetries are still in the form \eqref{fgreparam}, with 
\be
X(u,\tau) = Y(\tau)\,,\q\q \alpha(u,\tau)=T(\tau) + u \, \f{\de Y(\tau)}{\de \tau}\,.   
\ee
Usually the $\tau$ coordinate is compactified and represented by an angle $\phi$, but here we keep the name $\tau$ for a better comparison with our case. The evolution parameter $u$ is the null coordinate at infinity in 3D gravity, and shall not be confused with the null fields used in section \ref{sec2.2:mink&minisup}}
\be\label{BMS Liouville}
\cS_\rm{BMS} = \iint \de u\, \de \tau \left [ - \partial_\tau\left ( \f{V_2}{V_1}\right ) \f{\partial_u V_1}{V_1} - \left (\f{\partial_\tau V_1}{V_1}\right )^2 \right ]\,.
\ee

Performing the Hamiltonian analysis of this action with respect to the new evolution parameter variable, we find that the conserved charges generating the BMS$_3$ symmetries are given by (see also the formula \eqref{BMS_gravity_charges} in the first chapter)
\be\label{BMS Liouville charge}
\cQ= \int \de \tau \big[\cJ X+\cP \alpha \big],
\ee
where $\cJ$ and $\cP$ are precisely the quantities \eqref{J def} and \eqref{P def}, whose vanishing is equivalent to the equations of motion. 
As usual, one interprets this expression for the charges as a pairing between elements $(X,\alpha)$ in the adjoint representation and elements $(\cJ,\cP)$ in the coadjoint representation. This particular relationship between the charges of the 2d BMS$_3$ invariant theory and the equations of motion of the 1d model \eqref{new_lapse_action} becomes clearer if we look at the possibility of embedding our system in the coadjoint representation of BMS$_3$. This is the subject of the next section. We would like to note that for the action \eqref{invariant_1d_action} the BMS symmetries are \textit{gauge symmetries}, while in the case of \eqref{BMS Liouville} they are \textit{physical symmetries}.

\section{Coadjoint representation of BMS$_3$ and central charges}
\label{sec3.3:coadj_section}

The understanding of coadjoint orbits of the symmetry groups is crucial in order to go towards the quantum theory. In particular, the classification of the orbits of semi-direct product groups by their little groups will naturally lead to irreducible representations of the full group. In the case of the BMS$_3$ group, a detailed discussion is presented in \cite{Oblak:2016eij}, and we have reported its key ingredients in the first part of this chapter.

The main goal of this section is to explain how the identification of the gravitational fields with the $\mathfrak{bms}_3$ Lie algebra elements naturally leads to the introduction of 2-forms that transform as covectors under the group action. Moreover, we show here that the action itself can be seen as a bilinear form between vectors and covectors. For this, let us first consider the equations of motion for the configuration fields $V_i$, written in the form \eqref{J P theta}
\bsub\label{coadj_vector}
\be
\cP&= \f{\cV_0}{G}\left (\f{\ddot V_1}{V_1}-\f{\dot{V}_1^2}{2V_1^2}\right )\,, \label{supermomenta} \\
\cJ&= \f{\cV_0}{G}\left (-\f{\dot{V}_2 \dot{V}_1 }{V_1^2} + \f{\ddot{V_2}}{V_1}+\f{V_2}{V_1}\left ( \f{\dot{V}_1^2}{V_1^2} - \f{\ddot V_1}{V_1}\right )\right )\,.
\ee
\esub
A straightforward calculation reveals that $\cJ$ and $\cP$ transform exactly as in the coadjoint representation of the centrally-extended group $\text{BMS}_3$, i.e. as \eqref{BMS_coadj}, with central charges given by
\be
c_1= 0\,,\q\q c_2 = \f{\cV_0}{G}.
\ee 
The volume of the slice appears here in the central charge $c_2$. Nevertheless, by a constant and dimensionfull rescaling of the $V_i$'s in \eqref{J P theta} we can arbitrarily change the value of $c_2$ (and its dimension). The normalization in 3-dimensional gravity comes from the fact that the zero mode of the supermomentum represents the Bondi mass aspect, while here the on-shell value of $\cP$ is zero, so its dimension is not important. More precisely, defining the coadjoint action as
\be
\text{Ad}^*_{f^{-1},g} (\cJ,\cP):= \left (\text{D}_{f^{-1}}\circ \text{T}_{g}\right ) (\cJ,\cP) \eqqcolon(\tilde \cJ,\tilde \cP)\,,
\ee
we find the transformation laws
\be
\label{BMS_coadj_bh}
\tilde \cP=\dot{f}^2(\cP \circ f) - \,\text{Sch}[f]\,,
 \q\q
 \tilde{\cJ} =   \dot{f}^2\left(\cJ+g\dot{\cP}+2\dot{g} \cP - g^{(3)}\right )\circ f\,.
\ee

As a remark, one can note that the other central charge can be switched on by the shift $\cJ \to \cJ+c\cP$, which leads to $c_1 =c$. This shift can be obtained by adding a term $c \cP V_1$ to the action \eqref{new_lapse_action}. Remarkably, this shift does not affect the equations of motion and only modifies the symplectic potential. This comes from the fact that
\be\label{delta PV1}
\delta(\cP V_1)=\cP\delta V_1+\f{\de}{\de\tau}\left(\delta\dot{V}_1-\f{\delta V_1\dot{V}_1}{V_1}\right),
\ee
so that the variation of the new Lagrangian is $\delta(\cL_0+c\cP V_1)=(\cJ+c\cP)\delta V_1+\cP\delta V_2+\de_\tau\bar{\theta}$, where $\bar{\theta}$ is the sum of the boundary term \eqref{delta PV1} and the potential \eqref{V potential} coming from $\cL_0$. The equations of motion therefore combine to give once again $\cJ=0$ and $\cP=0$. It is very intriguing that this mechanism is the same as in 3-dimensional gravity, where the BMS$_3$ central charge $c_1$ (or a chiral mismatch $c_+\neq c_-$ between the two Brown--Henneaux central charges in the AdS$_3$ case) can be switched on by adding a Chern--Simons and a torsion term (together these constitute Witten's exotic Lagrangian \cite{MR974271}) to the first order Lagrangian, without however modifying the equations of motion \cite{Blagojevic:2006hh,Geiller:2020edh,Geiller:2020okp}.

By adding an appropriate boundary term to the initial action \eqref{new_lapse_action}, it is also possible to write a new action as the bilinear form between vectors and covectors. This rewriting is done by defining
\be
\label{BMS_geometric_action}
\cS[V_i]&:=\cS_0[V_i] - \f{1}{G}\int\de \tau \f{\de}{\de \tau}\left (\f{\dot{V}_1V_2}{V_1} \right ) \cr
&=  \f{1}{G}\int \de \tau \left[\f{ V_2 \dot{V}_1^2}{2 V_1^2} -  \f{\dot{V}_1 \dot{V}_2}{V_1} + \ddot V_2 \right]\cr
&=  \f{1}{G}\int \de \tau \big[\cJ V_1 + \cP V_2 \big].
\ee 
By construction this action leads to equation of motion equivalent to the ones of  $\cS_0$. It is however remarkable that the explicit variation is again of the form
\be
\label{BMS_geometric_action_variation}
\delta\cS=\f{1}{G}\int \de \tau\,\big[ \cJ\delta V_1+\cP\delta V_2+\de_\tau\tilde{\theta}\,\big]\,,
\ee
even though $\cJ$ and $\cP$ (expressed in terms of the $V_i$'s as in \eqref{coadj_vector}) have of course been varied in \eqref{BMS_geometric_action} to obtain this expression! Consistently, the only difference with \eqref{J P theta} is the form of the symplectic potential (which we have not displayed here because it won't play a role in what follows).

The result is the 1-dimensional action \eqref{BMS_geometric_action} for the black hole minisuperspace in gauge-fixed time $\tau$. This action, which is \textit{not} BMS$_3$ invariant, is however equivalent to the charge expression \eqref{BMS Liouville charge} for the 2-dimensional BMS$_3$ invariant theory \eqref{BMS Liouville}, provided we identify the infinitesimal parameters $(X,\alpha)$ in the charge with the algebra element $(V_1,V_2)$. We note that this new action transforms under BMS$_3$ as
\be
\label{transfo_geometric_action}
\text{Ad}^*_{(f^{-1},g)} \cS[V_i,\tau] = \cS[V_i,\tilde {\tau}] -\f{1}{G}\int \de \tilde \tau\left [(V_2-\dot{g}V_1 +g \dot{V}_1)\,\text{Sch}[f^{-1}] - V_1 g^{(3)} -\f{\de}{\de \tilde \tau} (g V_1 P) \right ],
\ee
which as expected differs from \eqref{transfo_action} only by a boundary term.

We can now give a nice interpretation of the Poincar\'e symmetry subgroup as an orbit stabilizer. So far we have embedded the fields $V_1$ and $V_2$ within the $\bms_3$ Lie algebra. We can build corresponding group elements by considering the functions
\be
f(\tau):=  \int^\tau \f{\de s}{V_1(s)}\,,\q\q g(\tau) :=  \int^\tau \de s \,\left (\f{V_2}{V_1} \right )\circ f^{(-1)}(s)\,.
\ee
Transforming these quantities under BMS$_3$ reveals that they are indeed group elements. If we now insert these functions into the transformation law \eqref{BMS_coadj_bh} for the covectors, starting from the point $(P=0,J=0)$ we find that the transformed pair $(\tilde J, \tilde P)$ is precisely equal to the right-hand side of \eqref{coadj_vector}. In other words, for the two functions $(f,g)$ given above we have
\be
-\text{Sch}[f]=\f{\ddot V_1}{V_1}-\f{\dot{V}_1^2}{2V_1^2}=\cP\,,\q\q -\dot{f}^2\, g^{(3)}\circ f = -\f{\dot{V}_2 \dot{V}_1 }{V_1^2} +\f{\ddot{V_2}}{V_1} +\f{V_2}{V_1}\left ( \f{\dot{V}_1^2}{V_1^2} - \f{\ddot V_1}{V_1}\right )=\cJ.
\ee
This shows that demanding that the fields satisfy the equations of motion is equivalent to the condition that ($\cP=0, \;\cJ=0$) be preserved under \eqref{BMS_coadj_bh}, which means in turn that $(f,g)$ belongs to the stabilizer of the orbit. In general, the coadjoint orbit of BMS$_3$ is defined by the Virasoro orbit spanned by the supermomentum $\cP$. These orbits are usually classified for the diffeomorphism group of the circle, while here we are dealing with functions on the real line since $\tau\in \R$. Through the de-compactification \eqref{decompact} introduced in the beginning we map point $(\cP=0,\;\cJ=0)$ to $(\cP=-c_2/2,\;\cJ=0)$, i.e. the constant representative of the so-called \textit{vacuum} orbit \cite{Oblak:2016eij}. The stabilizer of this orbit is known to be ISO(2,1), which of course is precisely the symmetry group of our theory.

\section{Weyl extension}
\label{sec3.4:Weyl}

In the previous section we have shown that even if the BMS transformations are not strictly speaking symmetries of the action, they nevertheless admit integrable generators on the phase space of the theory. Their charge algebra, given by \eqref{poisson_bms}, does however fail to reproduce the centrally-extended $\mathfrak{bms}_3$. This is to be expected since these are indeed not symmetries of the theory. Perhaps not surprisingly, the generators for the infinite-dimensional transformation that generalise \eqref{iso_bh} for any function $f,g$ corresponds to the extension of the global subalgebras $\cL_n$ and $\cT_n$, ($n \in\{0,\pm 1\}$) to any integer $n$.  
We recall however that the ISO(2,1) is part of a larger structure, corresponding to the quadratic or linear charge in the generators $\cS^\pm_n$ ($n =\pm 1/2$). It is then natural to beg the question of what happens if we relax the restriction on $n$ for the Heisemberg algebra. For this we restore the notation as in section \ref{sec2.2:mink&minisup}. Writing the phase space in null configuration variables and respective momenta, we want to see what happens to the generators $\cL_n$, $\cD_0$ \eqref{sl_2+R} and $\cS^\pm_n$ \eqref{h_2} for a general power\footnote{For the generator $\cD$, we define its infinite dimensional extension as $\cD_n:=\tau^{n} \cD_0$. And for the \textit{spin} generators $\cS_n$, the index $n$ is half integer, i.e $n =\f{1}{2}+m$, $m \in \Z$} $n$.
A straightforward calculation, by exponentiating the Poisson bracket on the configuration variables, gives the finite transformations
\bsub\be
\cL[f]:\,&\left|
\begin{array}{lcl}
u &\mapsto&  (\dot{f}^{1/2}\, u)\circ f^{-1}\\
v &\mapsto&  (\dot{f}^{1/2}\, v)\circ f^{-1}
\end{array}
\right.\,,\\
\cD[k]:\,&\left|
\begin{array}{lcl}
u &\mapsto&  k(\tau)\\
v &\mapsto&  v/k(\tau)
\end{array}
\right.\,,\\
\cS^+[g]:\,&\left|\begin{array}{lcl}
u&\mapsto& u+g(\tau)\,,\\
\end{array}\right.\\
\cS^-[h]:\,&\left|\begin{array}{lcl}
v&\mapsto& v+k(\tau)\,.\\
\end{array}\right.
\ee\label{fkgh_reparam}\esub
That corresponds to \eqref{mobius_sym} and \eqref{abelian_symm} with generic functions $f,k,g,h$, whose group law is isomorphic to $\left(\rm{Diff}(S^1) \ltimes \rm{Vect}(S^1)\right )\ltimes \left (\rm{Vect}(S^1)\times \rm{Vect}(S^1)\right ) $, upon (de-)compactification \eqref{decompact}. 

The enlargement of the $\bms$ algebra is not something new in gravity. For example in three-dimensional gravity, where the $\BMS_3$ group has been observed in the first place, the relaxation of the Bondi-Sachs gauge makes a four-parameter family appear (\cite{Geiller:2021vpg} and references therein). The new functions are associated with the symmetries that capture the Weyl rescaling of the boundary metric, which the extension is named after. 

Similarly as the $\bms$ extension, the brackets of the corresponding generators fails to reproduce a central extension of the Lie algebra of transformations. In their mode expansion, we have that
\be
\{\cL_n,\cL_m\} &= (n-m) \cL_{n+m} + \f{uv}{2} \left (m(m^2-1) - n(n^2-1)\right ) \tau^{m+n-1} \,,\notag\\
\{\cL_n,\cS_s^\pm\} &= \left (\f{n}{2}-s\right ) \cS_{n+s}^\pm +  \f{(u+v) \pm (v-u)}{8}\,(4 s^2-1) \tau^{n+s-1/2}\,,\notag\\
\{\cD_n,\cS_s^\pm\} &= \pm \cS_{s}^\pm + \f{(u-v)\pm (v+u)}{2}\, n\, \tau^{n+s-1/2}\,,\notag\\
\{\cL_n,\cD_m\} &= 0\,,\notag\\
\{\cS_s^\eta,\cS_{s'}^\epsilon\} &= (s-s')\delta_{\eta+\epsilon} \tau^{s+s'}\,,\notag
\ee
where $(m,n)$ are integers, and $(s, s')$ are half integers (i.e. $s = 1/2+n$, $n\in \Z$). It also turns out that the identities \eqref{identities} still hold, in the sense that
\bsub
\be
\cL_n &= \sum_{k=-1/2}^{1/2} k \left (\f{n}{2}+k\right ) (\cS^+_k \cS^-_{n-k}+\cS^-_k \cS^+_{n-k})\,,\\
\cD_n &=  \sum_{k=-1/2}^{1/2} k\, \left (\cS^-_k \cS^+_{n-k}- \cS^+_k \cS^-_{n-k}\right )\,,\\
\cT^\pm_n &= \sum_{k=-1/2}^{1/2} k \left (\f{n}{2}+k\right ) \cS^\pm_k \cS^\pm_{n-k}\,,\q \forall n \in \Z\,.
\ee \label{identities 2}
\esub
\ \\
As already discussed the enlargement of the $\bms$ algebra has been already witnessed in other gravitational setups. However, this is usually associated with the relaxation of some boundary conditions, to capture new degrees of freedom, for example, the Weyl rescaling of the boundary metric in 3D gravity \cite{Geiller:2021vpg}, or even radiative perturbation associated with the news tensor in four dimensions \cite{Barnich:2010eb, Barnich:2016lyg, Barnich:2019vzx, Freidel:2021fxf}. This is not the case here, where the extension from $\bms$ to a larger group is associated with a sort of Sugawara construction, that builds quadratic charges starting from an initial set, as shown in the formulas \eqref{identities 2} just above. We could wonder if it is possible to extend the previous discussion about coadjoint orbits in terms of the lager group $\left(\rm{Diff}(S^1) \ltimes \rm{Vect}(S^1)\right )\ltimes \left (\rm{Vect}(S^1)\times \rm{Vect}(S^1)\right ) $, but we  expect that due to the identities \eqref{identities 2}, seen as a restriction on the representation of the group, the conclusions might not change significantly. We plan to address the question shortly.

\ \newline 

In this chapter, we have discussed the possibility to extend the group of minisuperspaces symmetries to an infinite-dimensional group. We have found that this is not possible in the sense that the general transformations happen to not preserve the dynamics. Although this might seem disappointing at first, symmetry breaking is as important in physics as symmetry invariance, revealing degrees of freedom of the theory. Moreover it turns out that we are in an intermediate position, in the sense that the transformations act as a rescaling of the coupling constants, similarly to what we might expect from a renormalization group. We have also shown that the equation of motion themselves got an interpretation in terms of coadjoint orbits of the infinite-dimensional group, remarkably the fiducial size appears as the central charge of the coadjoint representation. Finally, we have built a relationship between the reduced action and the charges of the 2 dimensional Liouville theory, suggesting that maybe the extension to the midisuperspace accounting for spatial inhomogeneities (in the $x$ direction), might enlighten the symmetry breaking pattern. The non-trivial question is however that of the realisation of the symmetries in the precursor model. \newline \

This concludes the dive into the classical structure of the minisuperspaces symmetries, it is now time to discuss the consequences of these structures in the quantum theory.


\newpage
~
\thispagestyle{empty}

\newpage
~
\thispagestyle{empty}

\renewcommand{\afterpartskip}{}
\part*{Part II\\[.3cm]
From classical symmetries to canonical quantization} 
\addcontentsline{toc}{part}{II\ From classical symmetries to canonical quantization} \label{part:partII}
\newpage
~
\thispagestyle{empty}

\chapter{Canonical quantization of gravity}
\label{chap4}

The canonical approach to Quantum Gravity can be roughly summarized as the attempt of constructing a non-perturbative and background independent Quantum Field Theory of spacetime itself \cite{Thiemann:2007pyv}.  As already discussed in the introduction, the cornerstone of the canonical quantization of general relativity has been posed by Dirac, with the development of the technique to quantize constrained systems. Further works by Arnowit, Deser and Misner \cite{Arnowitt:1962hi} and finally, the ones by Wheeler and DeWitt \cite{Wheeler:1962, PhysRev.160.1113, PhysRev.162.1195, PhysRev.162.1239} contribute to the development of the theory. The main idea is to rewrite the Einstein-Hilbert action into its Hamiltonian form by splitting spacetime into three-dimensional hypersurfaces evolving in time (see \ref{sec1.1:Hamiltonian_GR}), and then apply Dirac's recipe to quantize the resulting constrained system.

In non-gauge theories the canonical quantization consists in the standard procedure of replacing the Poisson brackets on a phase space with commutators of operators acting on a Hilbert space. Conversely, gauge theories like gravity possess constraint and in Dirac's approach, the unconstrained (full) phase space is mapped to the so-called \textit{kinematical} Hilbert space. Defining then the constraint operators $\widehat \cC$ as quantum version of the constrain functions $\cC \approx 0$, we search for solutions $\ket \Psi$ of the quantum equation $\widehat \cC \ket \Psi =0$. The subset of states satisfying the quantum constrain, defines then the \textit{physical} Hilbert space.

\section{Wheeler-DeWitt quantization}

In metric general relativity (or \textit{geometrodynamics}) the set of constraints is provided by one scalar constraint \eqref{ADM_scalar_cstr} and three vector constraints \eqref{ADM_vect_cstr}, in particular the vanishing of the Hamiltonian constraints on the states $\hat H \ket \Psi =0 $ defines the so-called Wheeler-DeWitt equation, or \textit{quantum Einstein equations}. It resembles closely a Shr\"odinger equation, missing the time derivative term. This absence is a consequence of the fact that the coordinate time is unphysical and can be freely reparametrized, and is also a signature of the \textit{problem of time}.

Unfortunately, in their ADM formulation the constraints  \eqref{ADM_scalar_cstr}, \eqref{ADM_vect_cstr} are highly non-linear in the metric variables, leading to non-analytic second-order differential operators in the basic configuration variables so that a mathematically rigorous definition of a suitable Hilbert space is still unknown. 

Moreover, as a consequence of the 3+1 splitting, we lose a direct correspondence with the original four-dimensional covariance, which is on the other hand recovered as a gauge symmetry generated by the constraint. This means that we must avoid anomalies in the quantum realization of the algebra. Address the anomaly free condition in the Wheeler-DeWitt setup is cumbersome. There have been attempts to assure the absence of anomalies but this is at the price of having a Hilbert space that is too small and there is not enough room for a well defined classical limit \cite{Thiemann:2007pyv}.

Beside this technical problems that make the mathematical structure of the Wheeler-DeWitt equation still poorly understood, there are also some conceptual issues that complicate this approach \cite{Thiemann:2007pyv}. First of all, guessing that we are able to find the solution to quantum Einstein equations, one must find a set of Dirac observables (operators that leave the space of solutions invariant) to define the measurable gauge invariant quantities, which is extremely hard even in the classical setup. Equivalently we would need to deparametrise the theory to be able to give a relational description of the dynamics. In other words it is extremely difficult to define what is an observer in this setup. 

Moreover, having a background independent theory is synonymous of an emergent spacetime, in the sense that the quantum theory not defined \textit{on} a spacetime, but conversely \textit{it defines the spacetime}. It is therefore no longer clear if General Relativity is the corresponding classical limit. The issue must be approached by a semiclassical analysis, that is impossible to perform without an appropriate knowledge of the structure of the physical Hilbert space.

These major difficulties made the canonical quantization program stall for almost 60 years, until in 1986 Abhay Ashtekar realised that through a canonical transformation we could recast the constraint in an extremely easy form \cite{PhysRevLett.57.2244}, similar to an $\SU(2)$ Yang-Mills gauge theory. The key point was to understand that a $\SU(2)$ connection introduced some years before by Amitaba Sen (1982) \cite{Sen:1982qb}, satisfies the canonical Poisson algebra with the desensitized triad. These opened the door for the construction of Loop Quantum Gravity (LQG hereafter).

\section{Loop quantum gravity}
\label{sec4.1:LQG}

Loop quantum gravity (LQG) is a non-perturbative and background-independent theory of quantum gravity, where the geometry of spacetime is treated quantum mechanically. It relies on the key insight of writing general relativity as a gauge theory using the Ashtekar gauge connection as the configuration variable (conjugated to the triad). Originally introduced as a canonical (or Hamiltonian) approach, it also exists in a covariant (or path integral) formulation known as spin foams. One of the key results of LQG is the derivation of a notion of quantum Riemannian geometry. Indeed, one can show that in this theory operators measuring areas and volume of spatial regions have discrete spectra with a non-zero minimal eigenvalue. This is how quantum geometry sets a Planckian cut-off below which there exist no degrees of freedom.

First of all, instead of having the metric coefficient as dynamical variable (as in ADM), we define the triads $e^i_a$ on each slice\footnote{This is completely different from the minisuperspace reduction, where the triad appears as a splitting between spatial and temporal dependences. Here the triad depends on both time and space and the internal metric is just the diagonal $\delta_{ij}$} such that:
\be
q_{ab} = e^i_a e^j_b \delta_{ij}\,,
\ee
where $q_{ab}$ is the induced metric on the slice. The Ashetekar connection is given by
\be
A^i_a\coloneqq\Gamma^i_a[e]+\gamma K^i_a\,,
\ee
where $\Gamma^i_a$ is the torsionless spin connection of the triad (solution to the Cartan equation $\de e^i +\epsilon^i_{\;jk} \Gamma^j \wedge e^k = 0$), $K^i_a$ is the extrinsic curvature one-form, such that the first fundamental form \eqref{Extrinsic curvature} is given by $K_{ab} = K^a_i e^i_b$. Ashtekar had initially taken the imaginary value $\gamma=i$, it was later Barbero and Immirzi who introduced a real parameter. The parameter is named after them. The conjugate variable to the connection is the densitized triad
\be
\label{densitized triad}
E^a_i\coloneqq\det(e^i_a)\hat{e}^a_i=\f{1}{2}\eps^{abc}\eps_{ijk}e^j_be^k_c\,.
\ee
These variables obey the Poisson bracket relation \footnote{This is obtained from the original triad phase space $\lb k^i_a(x), e^b_j(y)\rb \sim \delta^i_j\delta^b_a\delta^3(x,y)$. The second bracket is highly non trivial and it is the main result of the Ashtekar seminal work. This because $A$ depends on the non-commuting quantities $K$ and $E$.}
\bsub\be\label{AE Poisson}
\lb A^i_a(x),E^b_j(y)\rb&=\delta^i_j\delta^b_a\delta^3(x,y)\,,\\
\lb A^i_a(x),A^j_b(y)\rb&= 0 \label{AA Poisson}\,.
\ee\esub
We remark that the phase space is polarized in a way that $A$ appears as the configuration variable and the triad as its momenta, contrary to what is the usual gravitodynamics convention. It's also possible to define the Ashtekar connection starting from the Palatini tetrad-connection action, in this case, we can complete the triad with the normal vector to the foliation to obtain the tetrad: $e^I_\alpha$ such that $g_{\alpha\beta}=\eta_{IJ} e^{I}_\alpha e^{J}_\beta$, and its spin connection $\omega^{IJ}_\alpha$ satisfying:
\be
\label{torsion free}
\partial_\alpha e^I_\beta - \Gamma^\sigma_{\,\alpha\beta}e^I_\sigma+\omega^I_{\alpha J}e^J_\beta=0\,.
\ee
We have that the connection can be written as:
\be
\label{barbero connection}
A^{i}_\alpha= \dfrac{1}{2} \epsilon^i_{\,jk}\omega^{jk}_\alpha + \gamma\omega^{0i}_\alpha\,.
\ee

In terms of $A$ and $E$ the action \eqref{ADM_Hamilton} takes the Hamiltonian form:
\be
\label{LQG_action}
\cS_\rm{LQG}=\int_\mathbb{R}\de t\int_\Sigma\de^3x\,\left[E^a_i\partial_0{A}^i_a-N \cC_\text{g}-N^a \cC_a-A^i_0G_i \right]\,,
\ee
where $N$ is the lapse function and $N_a$ is the shift vector enforcing respectively the scalar and vector constraints:  
\bsub
\be\label{scalar constraint_LQG}
\cC_\text{g}&=\f{1}{2\kappa} q^{-1/2} E^a_iE^b_j\left({\eps^{ij}}_kF^k_{ab}-2\left(1+\gamma^2\right)K^i_{[a}K^j_{b]}\right)\,,\\
\cC_a& =\f{1}{2\kappa} q^{-1/2} F^i_{ab} E^{bi}\,, \label{vector_constraint_LQG}
\ee\label{LQG_constraints}
\esub
where $q^{1/2}\coloneqq|\det(q_{ab})|^{1/2}=|\det(E^a_i)|^{1/2}=|\det(e^i_a)|$, and
\be
F^i_{ab}=\partial_aA^i_b-\partial_bA^i_a+{\eps^i}_{jk}A^j_aA^k_b\,.
\ee
The last constraint is the Gauss constraint $G_i= D_a E^a_i$. The last one is necessary to get rid of the antisymmetric part of the first fundamental form $k^a_i e^b_i$, and reduces the initial 9+9 dimensional phase space, back to the 6+6 metric gravity one ($q_{ab}, \pi^{ab}$). The scalar constraint takes a simple form only for the self-dual case $\gamma=i$. On the other hand, when we take the Ashtekar variables as the starting point for Loop Quantum Gravity (LQG hereafter) we will need to consider real Barbero-Immirzi parameters in order to have a well-defined Hilbert space.

Together with the use of Ashtekar's variables, the key point of Loop Quantization (which the theory is named after) was the realization by Ted Jacobson and Lee Smolin that it is possible to label the solutions of the constraints \eqref{LQG_constraints} by gauge-invariant loops. They take the name of \textit{holonomy} and measure the extent to which parallel transport around closed loops fails to preserve the geometrical data being transported. They are defined as the path oriented exponential of the integral of the connection along a closed loop.

The theory should be formulated in terms of intersecting loops, or graphs, which lead to the notion of spin-networks, which are graphs labelled by spins and provide a basis for states of quantum geometry. These are also the key ingredient to defining the spin foam covariant path integral formulation.
 
The discreteness of the area spectrum was firstly shown by Lee Smolin and Carlo Rovelli in 1994 \cite{Rovelli:1994ge}. It comes directly as a consequence of the quantization of the area operator in terms of triads and their corresponding $\SU(2)$ representations. For a surface $\Sigma$, the area operator acts on a Wilson loop (trace of the holonomy) eigenvector like:
\be
A_\Sigma = 8 \pi  \gamma \lp^2 \sum_I \sqrt{j_I(j_I+1)}\,,
\ee
where the sum is over all edges $I$  of the Wilson loop that pierce the surface $ \Sigma $, and $j$ is the label of the $SU(2)$ representation on the edge. This gives the minimal value for the area\footnote{We see here the importance of a real value of $\gamma$ to have a well-defined area operator}, when there is only one edge with $j=1/2$:
\be
\label{area_gap}
A_\rm{min} = 4 \sqrt{3} \pi \gamma \lp^2 := \Delta\,.
\ee
A complete review of Loop Quantum Gravity is beyond the scope of this thesis, and we refer the reader to the excellent literature already existing, for example, the book of Thiemann \cite{Thiemann:2007pyv} for the canonical approach, or the review by Rovelli and Vidotto \cite{Rovelli_introduction} for an introduction to the covariant approach, as well as the undergraduate level introduction by Gambini and Pullin \cite{Gambini_introduction}.

Despite its well-definiteness, LQG is a cumbersome theory and extracting physical information is a quite involved task. For instance, for several reasons, it is very complicated to study the semiclassical limit. Even worse, there is no known solution to determine the physical quantum states representing simple classical solutions, even for the homogeneous and isotropic cosmologies or spherically symmetric black holes.

It is nonetheless possible to incorporate some features of the full theories into reduced symmetry models. The application of LQG to homogeneous and isotropic minisuperspaces is known as Loop Quantum Cosmology (LQC). In the case of flat $\Lambda=0$ FLRW cosmologies, quantum effects become relevant in the Planck regime, giving rise to an effective repulsive force that enables dynamical trajectories to avoid the big bang singularity \citep{Ashtekar:2006uz}, \citep{Ashtekar:2006wn}. This was also extended to anisotropic Bianchi models and also to accommodate the presence of inhomogeneities.

\section{Loop quantum cosmology and polymerization}
\label{sec4.2:LQC}

In the following, we will briefly recall the results of homogeneous and isotropic FLRW loop quantum cosmology. This will serve as a basis for what follows in the next chapter, introducing the regularization procedure on the phase space known as \textit{polymerization}.

We investigate the symmetry reduced model obtained by restricting ourselves to spatially flat, homogeneous and isotropic cosmologies. This has already been presented in \ref{sec2.1:FLRW}, but we will rapidly review it here, this time in terms of Ashtekar variables. For this, we denote here the fiducial frame (without the time dependence) by an index $^o$. Explicitly the three-metric on the slice at constant time is given by  
\be
q_{ab}:=a^2(t)\left({}^o{q}_{ab}\right)=a^2(t)\left({}^o e^i_a\,{}^o e^j_b\delta_{ij}\right)\,.
\ee
A remark should be done at this point. In the first chapter, we have briefly used the triad formalism to define the minisuperspace models. Therein the triad has the role of a set of one-form fields that account for the spatial dependence of the metric, that is fixed \textit{a-priori}, when we select a particular minisuperspace model. On the other hand, in LQG the latin index $i$ stands for the coefficient of the $\su(2)$ basis carried along an edge, in this sense $A$ is an $\SU(2)$ connection:  
\be
A_a := -\f{i}{2} A_a^k \sigma_k\,,\q\q
e^a := -\f{i}{2} e^a_k \sigma_k\,,
\ee
where the $\sigma_k$ (or equivalently $\sigma^k$) are matrices reproducing the $\su(2)$ algebra. We then impose the condition of homogeneity and isotropy on the Wilson loops. The first one consists in taking a lattice of cubic cells, with edges that carry only the lower spin representation, so that the $\sigma_k$ are the 2x2 Pauli matrices. Simultaneously, isotropic connections and triads are characterized by:
\begin{subequations}
\ba
\label{gauge}
A^i_a&=&\cV_0^{-1/3}c(t)\,{}^o e^i_a\,,\\
E^a_i&=&\cV_0^{-2/3}\sqrt{{}^oq}p(t)\,{}^oe^a_i\,.
\ea
\end{subequations}
We recall that to regularize the integration over $\Sigma$, that is topologically $\mathbb{R}^3$ and so not compact, we restrict ourselves to a fixed fiducial cell of fiducial volume $\cV_0$. Where the scale factor $a(t)$ is related to the triad coefficient by $|p|=a^2 V_0^{2/3}$ (so that the comoving volume is $V= |p|^{3/2}$).

Notice that $p$ can be either positive or negative, which corresponds to two orientations of the triad. Since the metric is invariant under $e^a_i\rightarrow-e^a_i$ and we do not consider fermions in our model, the physics will also remain unchanged. Thus this orientation reversal of the triad is a gauge transformation that needs to be taken into account.

Because the model that we describe is spatially flat, the spin connection $\Gamma^i_a$ vanishes and $A^i_a=\gamma K^i_a$, and we find that the gravitational part of the Hamiltonian can be expressed as:
\be\label{LQCham}
\cC_\text{g}=-\frac{E^a_iE^b_j}{2\kappa\gamma^2\sqrt{q}}\epsilon^{ij}_{~~k}F^k_{ab}\,.
\ee
Having fixed the symmetry we are left with the sole scalar constraint. Following the literature on LQC we can add a minimally coupled massless scalar field as matter content, the action \eqref{LQG_action} thus rewrites:
\be
\cS=\int dt\left( p\dot{c}
+ p_\Phi \dot \Phi
+N\f{\kappa^2}{6}\sqrt{|p|}c^2-\frac{1}{2}\frac{p_\Phi^2}{|p|^{3/2}}\right)\,.
\ee
The canonical term gives us the Poisson brackets:
\be
\lbrace c,p\rbrace=1\,,
\qquad
\lbrace \phi, p_\phi \rbrace = 1\,,
\ee
and the Hamiltonian constraint is
\be\label{hamiltonian}
H= N\left(-\f{\kappa^2}{6}\sqrt{|p|}c^2+\frac{1}{2}\frac{p_\phi^2}{|p|^{3/2}}\right)
\equiv N(\mathcal{H}_\rm{g} + \cH_\Phi)\,.
\ee

The main feature of the full LQG theory is that there is no operator corresponding to the connection itself. The only well-defined operators are its holonomies along some edges. This means that there is no operator $\hat c$ on the gravitational Hilbert space $\mathfrak{H}^\text{grav}$. Since the holonomy $h^{(\mu)}_k$ of the connection $A^i_a$ along a line segment $e^a_k$ is given by 
\be
h_k^{(\mu)} = \cos \f{\mu c}{2} \I + 2 \sin \f{\mu c}{2} \tau_k\,,
\ee
the elementary dynamical variable to be quantized is taken to be $\exp(i \mu c)$. 
In the above formula $\I$ is the unit $2 \times 2$ matrix and $\tau_k=-i/2\, \sigma_k$, with $\sigma_k$ the Pauli matrices. Now in the full theory, a minimal value for the area is predicted \eqref{area_gap}. In order to take into account this fact we should impose that the area enclosed by the holonomy loops, measured on the physical metric $q_{ab}$ is $\Delta$. Since the area of faces of the elementary cell is $|p|$, and each side of the loop is $\mu$, we are led to choose $\mu=\bar \mu(p)$, such that:
\be
\bar \mu^2 |p| = \Delta\,.
\ee
The phase space dependence of the regularization parameter $\mu$ makes it difficult to well-define the Hilbert space. For the sake of convenience we reintroduce the canonical variables corresponding to the volume and its rate of change:
\be
B\equiv\frac{2c}{3\sqrt{|p|}}\,,
\qquad
V\equiv \sgn(p)|p|^{3/2}\,,
\ee
because they satisfy $\mb c= \lambda B$, with $\lambda=\sqrt{3\Delta/2}$. Now it's easier to work in the $V$ representation, where states on $\mathfrak{H}^\text{grav}$ are given by wave function $\Psi(V)$, and the holonomies are operators acting like:
\be
\widehat{e^{\pm i \mb c}}\Psi(V)\equiv\widehat{e^{\pm i\lambda b}}\Psi(V)=\Psi(V\pm \lambda)\,.
\ee

This means we need to work within a Hilbert space that supports well-normalized functions upon which the holonomy operator could act. This is provided by the functions on the compactified real line $\mathfrak{H}^\grav=L^2(\R_\rm{Bohr}, \de V_\rm{Bohr})$, meaning that the wavefunction $\Psi(V)$ is supported on a countable set of points on the real line. The holonomy operator acts then as a constant shift operator by step $\lambda$ on the wavefunction support. For a given parameter $\lambda$, the Hilbert space is divided into the so-called superselected sector, according to the position $(V)$ eigenstates, the latter taking discrete real values $\epsilon + n\lambda$, with a fixed offset $\epsilon$. The operator $e^{i\lambda B}$ creates a finite shift of step $\lambda$ and lets us move only within a given superselected sector labelled by the offset $\epsilon$.

An important remark should be done at this point. Loop Quantum Cosmology is not the homogeneous sector of Loop Quantum Gravity, because the symmetry reduction is done before the quantization. This makes the Hilbert space simpler but very different from the one of the full theory, and the huge simplification has led to numerous criticisms of LQC \cite{Bojowald:2020wuc}, in spite of the large effort to make contact with the full theory.

In order to introduce a well-defined operator on $\mathfrak{H}^\grav$ corresponding to the scalar constraint, we need to express the classical Hamiltonian in term of elementary variables $V$ and $h_k^{(\mu)}$. Following the general theory methods, the term involving the triads can be written as:

\be
\epsilon_{ijk} e^{-1} E^{aj} E^{bk} = \sum_k  \f{1}{2\pi\gamma G\mu } \epsilon^{abc} \;^o e_c^k \;\Tr \big(h_k^{(\mu)} \{ h_k^{(\mu)}{}^{-1}, V \} \tau_i \big)\,.
\ee

For the field strength, we use the strategy employed in lattice gauge theory. If we consider a face $\square_{ij}$ of the elementary cell, whose sides have length $\lambda V_o^{1/3}$, the curvature is given by:
\be
F_{ab}^k = -2  \lim_{{Ar}_\square \to 0} \Tr \left(
 \f{h^{(\mu)}_{\square_{ij}}-1}{\mu^2 }\right)\tau^k\; ^o e_a^i \;^o e_b^j\,. 
 \ee
The holonomy $h_{\square_{ij}}$ is just the product of holonomies along ${\square_{ij}}$ edges: $h^{(\mu)}_{\square_{ij}}= h_i^{(\mu)} h_j^{(\mu)} h_i^{(\mu)}{}^{-1} h_j^{(\mu)}{}^{-1}$. Now, the main concept of the LQG is not to shrink the area of the loop to zero, but only to its minimal value $\Delta$ , this means that the gravitational terms of the action becomes:
\be
\cC_{\text{g}} :=  -  \f{4}{2 \kappa^2 \gamma^3 \mb^3 } \sum_{ijk} \epsilon^{ijk} \,\Tr \big( h_i^{(\mb)} h_j^{(\mb)} h_i^{(\mb)}{}^{-1} h_j^{(\mb)}{}^{-1} h_k^{(\mb)} \{ h_k^{(\mb)}{}^{-1}, V \} \big)\,,
\ee
by explicit calculation it gives:
\be
\cC_{\text{g}} = \mathcal{H}_\text{g}=  -\f{3\kappa^2}{8} \f{\sin^2 \lambda B}{\lambda^2} V \,.
\label{polym_LQC}
\ee
Let us remark that we could have obtained the same result expressing the action in terms of $B$ and $V$, as in \eqref{Ham_prop_time_flrw}, and replace $B \to \f{\sin \lambda B}{\lambda}$. This procedure is known as \textit{polymerization} and it is the one we will discuss in the following for black holes. 

\paragraph{Polymerisation}
In classical and quantum mechanics with canonical phase space $(q,p)$ we define the Weyl algebra as the one generated by the exponentiated position and momentum operator $W = e^{i (b q -ap)}$. One of the most important theorems in quantum mechanics is the Stone-Von Neumann theorem \cite{doi:10.1073/pnas.16.2.172, Neumann1931DieED}, which states that any irreducible representation of the Weyl algebra that is weakly continuous is unitarily equivalent to the Schr\"odinger representation \cite{Corichi:2007tf}. Sometimes we are forced to work with the exponentiated operators, as in the famous  Aharonov-Bohm effect \cite{Aharonov_1969}, or whenever we want to measure a phase in an interference experiment. This makes us work with the so-called modular representation of the Weyl algebra, where we realize the Weyl algebra as commuting operator on a lattice. It is again equivalent to the Shr\"odinger representation but has found interesting application in gauge theories \cite{Strocchi_2016, Freidel:2015uug, Freidel:2016pls}. There are however other representations of the Weyl algebra where only half of the phase space is weakly continuous, evading the Stone-Von Neumman theorem, and thus being inequivalent to the Schr\"odinger picture. Among them, there is the polymer representation given by wavefunctions on the compactified real line, where the momentum operator $p$ is not well defined, only the exponentiated version acting as a finite translation it is. 

The interested reader might look at \cite{Strocchi_2016} for a complete review of the modular and polymer representations and \cite{Yargic:2020nkw} for its consequences on simple models like an harmonic oscillator.

Let us go back to LQC. To obtain the quantum Hamiltonian operator we now simply have to replace the classical variables in \eqref{polym_LQC} with the corresponding operator, and we can study how the expectations values evolve. On the polymer Hilbert space, only the exponentiated version of $B$ is well defined, and this is the one that appears in \eqref{polym_LQC}. For example, on a gaussian state, the effective dynamics of the volume prevent it from being zero. Conversely, it undergoes a bounce that replaces the classical big bang singularity. 

An important property of LQC, often used as an argument in favour of its robustness is the fact that the evolution of quantum coherent
states on the polymer Hilbert space follows the effective evolution described by the corresponding polymer Hamiltonian. Despite the apparent triviality of this statement due to Ehrenfest's theorem, the non-linearity introduced in the polymerization procedure makes it non-trivial anymore \cite{Corichi:2007tf, Strocchi_2016,Yargic:2020nkw}

For example, the majority of the works about black hole dynamics start directly with the heuristic effective dynamics, introducing by hand polymer corrections. In the absence of fully controlled quantum dynamics, the question of the equivalence between the effective evolution and the expectation values of quantum states is rarely addressed, making the models to lack a robust consistency check. 

The introduction of the quantum correction by a polymer regularization of the Hamiltonian also introduces a new set of ambiguities. On top of the usual ordering issues at the quantum level, already at the semiclassical level the regularization is ill-defined. When we replace the canonical momenta $p$ with the exponentiated version $\sin( \lambda p) /\lambda$, the parameter $\lambda$, typically related to the Planck length could be phase space dependent. Even worse, the choice of the $\sin( \lambda p)$ function is normally justified by analogy with the cosmological case, where on the other hand it is a simple consequence of the choice of lowest spin representation for the $\su(2)$ matrices. Including the higher spin representations introduces higher power of $\sin( \lambda p)$, with an opportune combination of the representations we will be able to freely choose any periodic function \cite{BenAchour:2016ajk} as \textit{polymerisation}. 

Furthermore, the choice of momenta that must be polymerized depends on the field space parametrization. For the black holes is it $(p_u, p_v)$ or $(p_a, p_b)$ or again $(P_1, P_2)$ that must be polymerised? While for cosmology there is a solid consensus in favour of the $\bar \mu-$scheme, presented above, no agreement has been found yet for the black holes, despite the large effort \cite{Ashtekar:2018cay,Ashtekar:2018lag,Bojowald:2018xxu,BenAchour:2018khr,Bodendorfer:2019xbp,Bodendorfer:2019cyv,Bodendorfer:2019nvy,Bodendorfer:2019jay,Ashtekar:2020ckv,BenAchour:2020mgu,Kelly:2020uwj,BenAchour:2020gon,Kelly:2020lec,Husain:2021ojz,Munch:2020czs,Munch:2021oqn}. This is where the symmetries introduced in the first part of the thesis become helpful.

\chapter{Minisuperspace quantization}
\label{chap5}

Throughout the whole thesis, the recurrent leitmotif has been the fact that symmetries provide an extremely powerful tool in physics. They give rise to conservation laws via Noether’s theorem, control the structure of classical solutions of a system, and also organize its quantum states via representation theory. In the first part, we have shown the existence of a group of physical symmetries that governs the minisuperspaces dynamic, and we dwelled on the classical structure of black holes and cosmologies. It is now time to discuss the consequences of these symmetries in the quantum regime. 

In the previous section, we have presented the canonical approach to quantum gravity, its LQG formulation and the corresponding cosmological application. The main point of the canonical approach is to realize an anomaly-free representation of the constraint algebra, which is a manifestation of the covariance of the classical theory. As we have seen in the first part this criterion becomes useless, as we are left with a unique constraint, trivially commuting with itself, leaving us without any criterion to test the covariance of the eventual regularization. 

On the other hand, the minisuperspace reduction makes us work with a mechanical model instead of a field theory, making the canonical quantization extremely simple. In this case, there are no technical obstacles to successfully performing the Dirac program, defining a Hilbert space for the Wheeler-DeWitt (WdW) operator and even explicitly finding the set of solutions. 

Moreover, they possess a finite-dimensional symmetry group, whose irreducible representation can be easily handled. The main idea behind this section is then to consider the minisuperspace symmetries as physically relevant, and demand that they are protected at the quantum level. This condition has two effects, on one side we can replace the missing anomaly free argument for the diffeomorphism algebra with the \textit{hidden} symmetries, this is reasonable if we think about the fact that this symmetry encodes the scaling properties of the minisuperspace. In some sense they provide the remnant of the reparametrization invariance, trading the status of gauge symmetries of the latter with the rigid (extended) conformal symmetry. Secondly, we can directly use the irreducible representation of the symmetry group to define a Hilbert space. Another motivation for preserving the conformal structure is the fact that it could also open the doors towards a bootstrap of the theory, solving the quantum theory only with conformal invariance arguments and leaving aside the knowledge of the quantum perturbations at the Hamiltonian level.

As in the first part of the manuscript, we will focus our attention on the black hole minisuperspace, spending some words on the more general case when needed.

This section collect results from \cite{Sartini:2021ktb}.

\section{Group quantization}
\label{sec5.1:WDW}
In gravitodynamics, the canonical quantization is performed by the usual recipe of replacing the Poisson brackets on a phase space with commutators of operators acting on a suitable Hilbert space. For minisuperspaces, this amounts to assigning to the configuration variables the corresponding multiplicative operator, and to momenta the respective linear derivative.

Let us illustrate this in detail for the black holes. Therein the classical phase space is given by $a,b$, which are the coordinates on the field space ($q^\mu)$, and the respective momenta $p_a, p_b$. The natural quantization scheme consists in promoting the field variables to multiplicative operators and the momenta to derivatives $\hat p_\mu=-i \partial_{\mu}$. The Hilbert space is given by square normalizable wavefunctions $L^2(\R^2,\de a \de b)$.

This corresponds to the \textit{position} polarization of the wavefunctions, which is known to be equivalent to the \textit{momenta} polarization up to a Fourier transform. 

This is of course well-defined, we could obtain the Hamiltonian operator and study the corresponding dynamical solution. However, this construction suffers from a major ambiguity. The presentation of the first part of this thesis strongly relies on the ability to freely redefine coordinates on the field space. The latter is indeed seen as a manifold endowed with a Riemannian supermetric, and classical trajectories for the black hole were interpreted as geodesics on the field space. The freedom of the choice of variables on the field space is translated into the invariance of the Hamiltonian formalism under a canonical transformation. This has been widely used in the first part to swap between scales factor $a,b$, null coordinate fields $u,v$ and $\bms_3$ vectors $V_1, V_2$. 

Unfortunately, not all the canonical transformations are implemented by a unitary transformation on the respective quantum mechanical description, so the choice of field coordinates might change the conclusions that we infer for the quantum evolution.

A way out of this puzzle is provided by the symmetry group of the theory. Instead of sticking to the usual WdW prescription, we start directly from the irreducible representation of the symmetry group. The cumbersome task to find Dirac observables for the Wheeler-DeWitt quantization is trivially solved here, the rigid symmetry naturally providing a complete set of conserved charges. The simplicity of the model also offers the opportunity to study the semiclassical limit.

There are several ways to use the symmetries for quantization. We can for instance quantize the conserved charges \eqref{charges}, that because of their explicit time dependence will give a Heisenberg-like picture, or conversely we can quantize the kinematical version of the algebra, that depends only on the phase space, and obtain the corresponding Schr\"odinger representation. 

We also shall come back here to the discussion about the size of the algebra. In the first sections we have pointed out the need for having a number of independent generators that is not higher than the original gravitational phase space. For this we have identified the fundamental nucleus of the charge algebra, which surprisingly is given by a two dimensional Heisenberg algebra. We recall that the classical expression for its generators in terms of the scale factors and their momenta is given by
\be
\cS^+_{1/2} &=  \sqrt{\f {\cV_0} G} \f{8 a b^2 \cV_0+p_a \tau G}{4 a b \cV_0}\,,
&\cS^-_{1/2} &= \sqrt{\f {\cV_0} G} \f{8 a^2 b^2 \cV_0+ (2 b p_b-a p_a)\tau G}{4 b \cV_0} \,,\notag\\  
\cS^+_{-1/2} &=  \sqrt{\f G {\cV_0}} \f{p_a}{4 a b} \,,
&\cS^-_{-1/2} &= \sqrt{\f G {\cV_0}} \f{2 b p_b-a p_a}{4 b} \,,
\ee
which measure the initial condition for the phase space quantities \footnote{These are exactly the null coordinates on the field space \eqref{conf_to_null_bh} and their momenta} 	
\be
u &=  \sqrt{\f {\cV_0} G} 2 b\,,
&v &= \sqrt{\f {\cV_0} G} 2 a^2 b  \,,\notag\\  
p_u &= \sqrt{\f G {\cV_0}} \f{2 b p_b-a p_a}{4 b}\,,
&p_v &=  \sqrt{\f G {\cV_0}} \f{p_a}{4 a b} \,. 
\label{h_2_quantum}
\ee
We can thus promote them to operators\footnote{The quantities $u,v,p_u,p_v$ are like the position and momenta operators in quantum mechanics in the Shr\"odinger picture, where the states are wavefunction evolving in time. The conserved quantities that measures the respective classical initial condition are obtained by integrating backward along the Hamiltonian flow $e^{\{\bullet, \tau H\}}$.  A straightforward calculation leads to verify that this will give the conserved charges $\cS^\pm_n$. The construction is implemented in quantum mechanics by the unitary transformation that maps the Schr\"odinger picture to the Heisenberg one, were the latter is obtained by the quantization of the explicitly time dependent operators acting on a fixed wavefunction.} by using the fundamental representation of the Heisenberg algebra
\be
\hat u\, \Psi(u,v) = u \Psi (u,v) \,,\q\q \widehat p_u \,\Psi(u,v) = -i \partial_u \Psi(u,v)\,, \notag\\
\hat v\, \Psi(u,v) = v \Psi (u,v) \,,\q\q \widehat p_v\, \Psi(u,v) = -i \partial_v \Psi(u,v)\,, 
\ee
on the Hilbert space of square integrable wavefuctions $L^2(\R^2, \de u \de v)$.

We have now two main questions to address: \textit{(i)} does this Hilbert space provide an irreducible representation of the larger Poicar\'e group? \textit{(ii)}  Is this equivalent to the WdW quantization of the scale factors? 

The answer to the first question is hidden in the fact that the Poincar\'e generators are not independent, as seen in equation \eqref{poincar_casimir}. The identities that allow identifying the Heisenberg algebra as the fundamental blocks of the symmetry are translated into the vanishing condition on the Casimir of the Poicar\'e algebra. In particular, this means that the black hole interior must carry a massless and zero spin representation of the Poincar\'e group, which turns out to be unitary equivalent to the fundamental realization of the $\h_2$ algebra, as it will be discussed in a moment.

Answering the second question is less trivial, as we need to find a unitary transformation on the square-integrable wavefunctions that implement the coordinate transformation \eqref{h_2_quantum}. The same happens if we perform the WdW quantization of the $\bms_3$ vector fields and their momenta. This is a quite involved task because of the inverse power of the configuration variable that appears in the map between the momenta. The Poincar\'e structure will also help to enlighten this point.

\subsection{ISO(2,1) representations}
\label{iso_irreps}

The unveiling of the conformal symmetry for cosmology and its extension to black holes suggests that for minisuperspaces we can replace the constraint algebra with the new symmetry algebra. Here this means that the quantum Hilbert space must contain an irreducible representation of Poincar\'e algebra. It is then smart to directly start by exploiting the well known irreducible representation of (the universal cover of) $\ISO(2,1)$. For this, the reader shall refer to an exhaustive discussion in \cite{Grigore:1993}, of which we will recall here some key features. 

First of all want to define the kinematical algebra that once quantized will provide the set of fundamental operator in the Shr\"odinger picture. These are defined by the following quantities, expressed in terms of the $\bms$ vector fields and the respective momenta 
\bsub
\be
&J= \f{V_2}{2\lambda} - \lambda \f{\cV_0}{G} H\,,&  
&K_x= \lambda \f{\cV_0}{G} 	 H  + \f{V_2}{2\lambda}\,,& 
&K_y = -C_\rm{d}\,, \label{Poincare_alg_lorenz} \\
&\Pi_0 = \f{V_1}{2\lambda}+ \lambda\cA\,,&
&\Pi_x= V_1 P_2\,,& 
&\Pi_y= \f{V_1}{2\lambda}-  \lambda\cA\,,\label{Poincare_alg_trans}
\ee	
\label{Poincare_alg_gen}
\esub
for a real dimensionless parameter $\lambda$.
The boosts $K_x$, $K_y$ and the rotation generator $J$ are a linear combinations of the evolving quantities of whom the $\cL_n$'s \eqref{superrot_gen} measures the initial condition, while the translations $\Pi$ have the same role with respect to $\cT$ \eqref{supertran_gen}.  Computing the Poisson brackets between the 6 generators gives the $\iso(2,1)$ algebra
\be
\label{Poincare_alg}
\{J,K_i\}& = \epsilon_{ij}K_j\,,&
\{K_x,K_y\}&= -J\,,&
 &\cr
\{J,\Pi_i\}&=\epsilon_{ij}\Pi_j\,, &
\{K_i,\Pi_0\} & = \Pi_i \,,&  
\{K_i,\Pi_j\} & = \delta_{ij} \Pi_0\,.
\ee 
The quantities $\cA$ and $\cB$ are the first integrals introduced in \ref{sec2.3:BH_minisup}, and $C_\rm{d}$ is defined in \eqref{HKV_mink}, all of them rewritten in terms of $V_i, P_i$:
\be
\cA= \f{P_2^2 V_1}{2} \,,\q\q \cB= P_1 V_1 \,,\q\q C_\rm{d}= P_1 V_1 + P_2 V_2\,,
\ee
while the Hamiltonian reads
\be
H= -\f{P_2 (2 P_1 V_1 + P_2 V_2) G}{2 \cV_0}\,.
\ee
We consider the realization of the algebra \eqref{Poincare_alg} as self-adjoint operators, acting on wavefunctions on $\R^2$ in polar coordinates, with the scalar product
\be 
\label{Hilbert_scalar}
\braket{\psi|\chi} =\int_0^\infty \de \rho \int_0^ {2 \pi} \de\phi\, \psi^*(\rho,\phi) \chi(\rho,\phi)\,.
\ee
The $\iso(2,1)$ generators are realized as a one parameter family, with $s\in\R (\rm{mod}	\, 2)$
\bsub
\be
(\widehat{\Pi_0}^{(s)} \psi)(\rho,\phi) =& \rho \,\psi(\rho,\phi)\,,\\
(\widehat{\Pi_x}^{(s)} \psi)(\rho,\phi) =& \rho \sin \phi \,\psi(\rho,\phi)\,,\\
(\widehat{\Pi_y}^{(s)} \psi)(\rho,\phi) =& \rho \cos \phi\, \psi(\rho,\phi)\,,\\
(\widehat{J}^{(s)} \psi)(\rho,\phi)  =& \left[ i \f{\partial}{\partial \phi}  -\f{s}{2} \right ]\psi(\rho,\phi)\,,\\
(\widehat{K_x}^{(s)}  \psi)(\rho,\phi)  =& \left[i \rho \left(\sin \phi\f{\partial}{\partial \rho} +\f{\cos \phi}{\rho}\f{\partial}{\partial \phi} \right ) -\f{s}{2}\cos \phi \right] \psi(\rho,\phi)\,,\\
(\widehat{K_y}^{(s)}  \psi)(\rho,\phi) =& \left[i \rho \left(\cos \phi\f{\partial}{\partial \rho} -\f{\sin \phi}{\rho}\f{\partial}{\partial \phi} \right ) +\f{s}{2}\sin \phi \right] \psi(\rho,\phi)\,.
\ee
\label{Poincare_quant}
\esub
By a straightforward calculation, we can verify that they satisfy the quantum version of the algebra \ref{Poincare_alg}, replacing the Poisson bracket with the commutator between operators. We would like to stress that in this polarization for the wavefunctions the translations $\Pi$ act like multiplicative operators, and the $\sl(2,\R)$ sector as derivatives. Through a Fourier transformation of the wave-packet, we can of course inverse this property, or have a mixed polarization if we perform the transformation only for one of the variables, or a combination of them.
The realization \eqref{Poincare_quant} satisfies the condition that both the Casimirs are zero,
\bsub
\be
\widehat{\mathfrak{C}_1} &:= -\widehat{\Pi_0}^2 +\widehat{\Pi_x}^2+\widehat{\Pi_y}^2 \,,\\
\widehat{\mathfrak{C}_2} &:= \f{1}{2} \left (\widehat{J} \widehat{\Pi_0} +\widehat{\Pi_0} \widehat{J} +\widehat{K_x} \widehat{\Pi_y} +\widehat{\Pi_y} \widehat{K_x}- \widehat{K_y} \widehat{\Pi_x}- \widehat{\Pi_x} \widehat{K_y} \right  )\,,\\
&\widehat{\mathfrak{C}_2}\, \psi(\rho,\phi) = 0 =\widehat{\mathfrak{C}_1}\, \psi(\rho,\phi)\,.\notag
\ee
\esub
We can explicitly calculate the action of the $\mathfrak{sl}(2,\R)$ Casimir, and see that it does not depend on $s$
\be
({\widehat{\mathfrak{C}}_{\mathfrak{sl}(2,\R)}}^{(s)}  \psi)(\rho,\phi) =(-\widehat{J}^2+\widehat{K_x}^2+\widehat{K_y}^2)\psi (\rho,\phi) = -\rho \left(2\f{\partial}{\partial \rho} +\rho\f{\partial^2}{\partial \rho^2} \right) \psi(\rho,\phi)\,.
\ee
Let us finally remark that we can also define the \textit{square root} of this operator, which corresponds to the classical integral $\cB$:
\be
(\widehat \cB  \psi)(\rho,\phi) :=i\left (\rho \f{\partial}{\partial \rho} +\f{1}{2} \right) \psi(\rho,\phi)\,,\q\q \widehat{\mathfrak{C}}_{\mathfrak{sl}(2,\R)}=\widehat{\cB}^2+\f{1}{4}\,.
\ee
The factor $1/4$ represents a quantum correction to the Casimir with respect to its value in terms of the classical integration constant $\cB$. 

It is now time to discuss the equivalence with the canonical gravitodynamics quantization. Let us begin with the phase space defined by the $\bms_3$ fields $V_i$ and the respective momenta. The question is whether the Hilbert space that has been just introduced carries or not an irreducible representation of the Heisenberg algebra of $(V_i,P_i)$. For this we see that the classical version of the multiplicative operators $\widehat \Pi$ depend only the phase space functions $V_1$, $P_2$. For example if we consider the change of coordinates for the wavefunctions
\be
\begin{array}{rl}
V_1&=\lambda \rho (1+ \cos \phi)\,,\\
P_2&=\df{1}{\lambda}\tan\df{\phi}{2}\,,\\
\end{array}
\label{WdW_1}
\ee
we see that the quantum realization of the $\Pi$'s takes exactly the same form as the classical expressions \eqref{Poincare_alg_gen}, e.g. $\hat \Pi_x \ket \psi =  \rho \sin \phi \ket \psi = V_1 P_2 \ket \psi$. From the change of variables above, we can also try to find out the respective conjugate operators $V_2,P_1$. This corresponds to finding a definition for the action of $-i \partial_{V_1}$ and $i \partial_{P_2}$, seeking for a mixed polarization of the standard Schoedinger picture, where we perform a Fourier transform on the usual wavefunction depending on the configuration variables $V_i$. A straightforward calculation leads to
\be
\begin{array}{rl}
\widehat{P_1} \psi(\rho,\phi) &= -i \df{\partial_\rho \psi}{\lambda + \lambda\cos \phi}\,,\\
\widehat{V_2} \psi(\rho,\phi) &= i \lambda \left ( (1+\cos \phi) \partial_\phi+\rho \sin\phi\partial_\rho\right )\psi\,.\\
\end{array}
\label{WdW_2}
\ee
Of course, they reproduce the Heisenberg algebra $\h_2$. The subtlety here is given by the range of the coordinates $V_1, P_2$ induced by the definition of the scalar product \eqref{Hilbert_scalar}. Because of the fact  that $\rho\in \R^+$ and $\phi \in S^1$, we have that $V_1$ is always positive (or zero) while $P_2$ is real. Conversely, $P_2$ spans the whole real line. This in turn implies that the conjugate momenta are well-defined hermitian operators with continuous spectra, but $P_1$ has only positive eigenvalues, while $V_2$ can change its sign. This means that the Hilbert space \eqref{Hilbert_scalar} is unitarily equivalent to the WdW quantization only if we restrict the $g_{\Omega\Omega}$ coefficient to positive (or zero) values, and we allow the $g_{xx}$ to change its sign. Again we stress that this corresponds exactly to the extension from the interior to the exterior of the horizon. The presence of the symmetry forces us to extend the minisuperspace to cover the whole black hole spacetime.

Moreover, this remark implies that the square-integrable functions within the Hilbert space must not diverge faster than $1/\sqrt{\rho}$ to zero when $\rho \to 0$, or equivalently $V_1 \to 0$. We recall that classically this point in the phase space represents the singularity of the black hole, although this does not completely solve the singularity, the presence of the symmetry, already imposes a strong bond on the values of the wavefunctions near it, on a simple kinematical level. To discuss the fate of the singularity more in detail we must take into account the dynamics and find the physical Hilbert space.

\section{Hamiltonian eigenstates and dynamics}
First of all we search for eigenstate of the rotation generator $J$ and the Casimir $\cB$. These provide a complete basis of the Hilbert space, and are labelled by a real number $B$ and an integer $m$
\bsub
\be
\widehat{J}\ket{B,m}&=m\ket{B,m}\,,\q\q m\in \Z\,,\\
\widehat{\cB}\ket{B,m}&=B\ket{B,m}\,,\q\q B\in \R\,,\\
\widehat{K_\pm} \ket{B,m}&= \left (m\pm \left (\f 1 2+i B\right )\right )\ket{B,m\pm 1}\,,\\
\widehat{\Pi_0} \ket{B,m}&= \ket{B+i,m}\,,\\
\widehat{\Pi_\pm} \ket{B,m}&= \ket{B+i,m\pm 1}\,,
\ee
\label{j_eigen}\esub
where $K_\pm=K_x\pm i K_y$, and $\Pi_\pm=\Pi_y\mp i \Pi_x$ \footnote{In the last two lines the complex shift of the Casimir must be read as a formal replacement of $B$ into the wavefunction, for example
\be
\braket{\rho,\phi|\widehat \Pi_0|B,m}&= \f{1}{2\pi} \f{1}{\sqrt{\rho}} e^{-i (B+i) \log (\rho)} e^{-\f{1}{2}i (2m+s) \phi}= \f{1}{2\pi} \sqrt{\rho}\, e^{-i B \log (\rho)} e^{-\f{1}{2}i (2m+s) \phi} \notag
\ee
Actually the momenta operator $\Pi_i$ acting on a single eigenstate of the Casimir labelled by $B$ maps it to a combination of eingenstates, exactly like the momentum operator acting on a position eigenstate in standard quantum mechanics. 
}.

The formulas above can be verified by direct computation of the realization \eqref{Poincare_quant} on the normalized wavefunctions
\be
\braket{\rho,\phi|B,m} = \f{1}{2\pi} \f{1}{\sqrt{\rho}} e^{-i B \log (\rho)} e^{-\f{1}{2}i (2m+s) \phi}\,,\q\q
\braket{B', m'|B,m} = \delta_{m,m'} \,\delta(B-B')\,.
\ee
We shall remark that the first lines in \eqref{j_eigen} provide a representation of the $\sl(2,\R)$ algebra at fixed Casimir, but the presence of the abelian sector, represented by translations of the Poincar\'e group, allows to move between different values of $B$.

An interesting role is played by the parameter $s$. It is very similar to the one played by the superselection parameter in LQC \cite{Ashtekar:2011ni}. With respect to the scalar product \eqref{Hilbert_scalar}, two states with different $s$ are always orthogonal, but belong to two unitarily equivalent representations of the Poincar\'e group, exactly as two superselected lattices in LQC with respect to the polymer representation of Weyl algebra \cite{Ashtekar:2002sn, Corichi:2007tf}. Hereafter, without loss of generality, we will set $s=0$.

If we want to impose some dynamics on the Hilbert space we must recall that the classical evolution is generated by the Hamiltonian $H$. At the quantum level, it means that we need to search for eigenstates of $K_x-J$. The Hamiltonian operator is defined as
\be
(\widehat H\,\psi)(\rho,\phi) = \f{G}{2\lambda\cV_0}(\widehat{K_x}-\widehat{J})\psi(\rho,\phi) = E\,\psi(\rho,\phi)\,.
\ee
The associated differential equation has an analytical solution. Diagonalising with respect to the Casimir $\mathfrak{C}_{\mathfrak{sl}(2,\R)}$, we obtian again a complete basis for the Hilbert space,
\bsub
\be
\widehat H\ket{B,E}&=E\ket{B,E}\,,\\
\widehat \cB\ket{B,E}&=B\ket{B,E}\,,\\		
\braket{\rho,\phi|B,E} &= \f{\sqrt{\lambda \cV_0/G}}{2\pi} \f{1}{\sqrt{\rho} \sin(\phi/2)}  e^{-i B \log \left (\rho \sin^2(\phi/2)\right )} e^{-i \left (2\lambda \cV_0/G E \cot\left (\f{\phi}{2}\right ) 
\right )}\,,\\
\braket{B', E'|B,E} &= \delta(E-E') \,\delta(B-B')\,.
\ee
\esub
We shall remark that the spectrum of the Hamiltonian is continuous and unbounded from below, exactly as it happens for cosmology and conformal mechanics \cite{deAlfaro:1976vlx}. This should lead to a catastrophic instability when we consider (multiple) interacting systems, but our formalism is valid only for a single static black hole, without any matter content, and thus the question of stability cannot be addressed here. For instance, it does not make any sense to couple with a thermal bath and to look at the partition function $e^{-\beta H}$.

In order to verify the orthogonality, as well as the completeness, of the basis it is convenient to perform a change from the $(\rho,\phi)$ polarization to a new set of variables, with respect to whom the eigenfunctions look like plane waves (see  \cite{Sartini:2021ktb} for the realization of the Poincar\'e algebra on the new variables):
\bg
z:= 2\lambda \f{\cV_0}{G} \cot \left (\f{\phi}{2}\right )\,,\q\q a:= \log \left (\rho \sin^2\f{\phi}{2}\right )\,,\q\q(z,a)\in \R^2\,,\\
\braket{z,a|B,E} := \f{1}{2\pi} e^{-i B a} e^{-i E z 
}\,.
\label{za_wave}
\eg
We choose to call the second variable $a$ because it is actually related to the $A$ operator:
\be
(\hat \cA\, \psi)(\rho,\phi) = \f{1}{2\lambda} (\Pi_0 - \Pi_y)\psi (\rho,\phi) = \f{1}{\lambda} \left (\rho\, \sin^2\f{\phi}{2} \right )\psi(\rho,\phi):=\f{1}{\lambda}(\widehat{e^a}\, \psi)(\rho,\phi)\,.
\ee
On the other hand, $z$, being conjugated to the energy, is expected to be related to time. This is indeed what happens once we consider physical states satisfying the quantum dynamics. The imposition of the dynamics on the Hilbert space will be the subject of the next subsection.


\subsection{Physical solution and semiclassical states}
\label{physical space}
In the previous section, we found the eigenvectors that diagonalise both the Hamiltonian ($K_x-J$) and one of the classical first integral represented by the $\mathfrak{sl}(2,\R)$ Casimir ($\cB$). We can exploit this basis to impose the dynamics of the system. We face two different possibilities to do so: we recall that at the classical level we can both impose the constraint $\cC_\cH=H-\cV_0/(G L_s^2)=0$, or equivalently see $H$ as a true Hamiltonian generating the time evolution with respect to $\tau$. At the quantum level, the two interpretations (hereafter denoted respectively by \textit{strong} and \textit{weak} constraint) will lead to drastically different semiclassical evolutions. In both cases, we will reconstruct the metric as an emergent quantity, based on expectation values of the fundamental operators. Before starting to discuss the quantum evolution we would like to recall that the classical solutions for $V_i$ can be inferred from \eqref{Bh_sol}, and the definition \eqref{new_fields_v}. They are
\bsub \be
V_1(\tau) &= \frac{\cA G^2 (\tau-\tau_0)^2}{2 \cV_0^2}\,,\\
V_2(\tau) &= \frac{(\tau-\tau_0) \left(2\cB L_s^2 G/\cV_0-\tau+\tau_0 \right)}{2 L_s^2}\,.
\ee \label{classic_traj} \esub
\paragraph{Strong constraint}
Following the Dirac procedure, we implement the strong constraint requiring that the physical states are the ones that satisfy $\cC_\cH {\ket\psi}_\text{phys} =0$, or equivalently:
\be
{\ket \psi}_\text{phys} :=\int\de E\int  \de B\, \delta\left (E-\f{\cV_0}{ GL_s^2}\right ) \,\psi (B) \,\ket{B,E} = \int_\R \de B \,\psi (B) \,\ket{B,\f{\cV_0 }{G L_s^2}}	\,,
\ee 
where the energy scales and their fiducial volume  $\cV_0$ are given \textit{a priori}, and they act as multiples of the identity operator. These physical states are of course not normalized within the original Hilbert space, and we need to introduce a new inner product to obtain the physical space. To this purpose, we make use of \textit{group averaging} (or refined algebraic quantization) \cite{Thiemann:2007pyv,Marolf:1995cn,Marolf:2000iq} and define the projector
\be
\label{group_averag}
\delta(\cC_\cH)=\f{1}{2\pi}\int_{-\infty}^\infty \de x \exp\left ( i\, x\,\cC_\cH\right )\,,	\q\q \delta(\cC_\cH): \cH_\rm{kin}\to \cH_\rm{phys}\,,
\ee
that induces the inner product
\be
{\braket{\chi |\psi}}_\text{phys} = {\braket{\chi|\delta(\cC_\cH)|\psi}}_\rm{kin} = \int \de B\, \chi(B)^* \psi(B)\,.
\ee
Notice that on $H_\rm{phys}$, only the quantum operators $\widehat \cO$ that commute with the constraint are well-defined $[\widehat \cO,\widehat{\cC_\cH}]=0$, otherwise their action will map out of the physical subspace. 
We need then to deparametrize the dynamics with respect to a preferred clock (here $\tau$). More explicitly this means that the observables we can measure are one-parameter families, e.g:
\be
\label{Dirac_obs}
\widehat{V_1}(\tau) = \f{G^2\tau^2}{2\cV_0^2} \widehat {\cA}\,,\q\q \widehat{V_2}(\tau) = \widehat {\cB}\, \f{G \tau}{\cV_0} -\f{G \tau^2}{2\cV_0} \hat {H}\,.
\ee
From the expression \eqref{za_wave} we can infer that $B$ and $a$ are conjugated variables, while $z$ is conjugated to the energy $E$, and it is thus traced out in the group averaging \eqref{group_averag}. The most convenient way of representing the physical space is in terms of functions of the variable $a$ (i.e. Fourier transform of $\psi(B)$), upon which the physical observables in \eqref{Dirac_obs} act as
\be 
\hat \cA \psi(a) = \f{1}{\lambda} e^a \psi(a)\,,\q\q \hat{\cB}\psi (a) = i \partial_a \psi(a)\,,\q\q \hat{H} \psi(a) = \f{\cV_0}{G L_s^2} \psi(a)\,.
\label{ba_operators}
\ee

The semiclassical states can be obtained by picking $\psi(a)$ to be a Gaussian distribution peaked around some classical values $(a_*,B_*)$,
\be
\psi_*(a) = \,\f{1}{(2\pi \sigma^2)^{1/4}}\,e^{-\f{(a-a_*)^2}{4 \sigma^2}} e^{-i B_* a}\,,\q\q
\left|
\begin{array}{rl}
\braket{\widehat{V_1}(\tau)}&=\dfrac{e^{a_*+ \f{\sigma^2}{2}}}{2\lambda} \,\dfrac{\tau^2 G^2}{2\cV_0^2} :=\dfrac{A_*}{2}\dfrac{\tau^2G^2}{2\cV_0^2}\,,\\
\braket{\widehat{V_2}(\tau)}&=\dfrac{ B_*\tau G}{\cV_0} -\dfrac{\tau^2}{2L_s^2} \,.\\
\end{array}
\right .
\label{v_strong_evolution}
\ee
This means that the expectation values follow the classical trajectories, up to a constant rescaling of the first integral $A$, due to quantum indetermination. This comes without much surprise, as the classical evolution has already been imposed in the deparametrization of the dynamics, in the definition of the one-parameter family of Dirac observables \eqref{Dirac_obs}.

\paragraph{Weak constraint}
The other possibility to impose the dynamics consists in asking that the constraint is satisfied in a weaker sense: $\braket{ \psi|H|\psi}= \f{\cV_0}{G L_s^2}$. One could imagine that somehow this would account for some deep fuzziness of the geometry, contributing as an effective stress-energy tensor, that allows some fluctuations around the classical energy value. The idea is similar to the so-called \textit{deconstantization} of unimodular gravity, where the cosmological constant is allowed to fluctuate around its classical value \cite{Unruh:1988in,Amadei:2019ssp,Perez:2019gyd,Magueijo:2021rpi}. In the following we will see how this statement should be correctly interpreted, the uncertainty on the energy level coming explicitly into the game.

The imposition of the weak constraint is easily achieved by Gaussian wavepackets, peaked on some semiclassical values for the pairs of conjugated variables $(B_*, a_*)$ and $(\cV_0/(G L_s^2), z_*)$:
\bsub
\be
\ket {\psi_{*}}&:= \int\de E\int  \de B \,\f{1}{(2\pi \sigma_B \sigma_E)^{1/2}}\,e^{-\f{(B-B_*)^2}{4 \sigma_B^2}} \,e^{-\f{(E-\cV_0/(G L_s^2))^2}{4 \sigma_E^2}} \,e^{i B a_*}e^{i E z_*} \,\ket{B,E}\,,\\
\braket {a,z|\psi_{*}}&=  \,\sqrt{\f{2\sigma_B \sigma_E}{\pi}} \,e^{-\left (a-a_* \right )^2 \sigma_B^2} \, e^{-\left (z- z_*\right )^2 \sigma_E^2}\, e^{i B_*\left (a_*- a\right )} e^{i (z_*- z )\cV_0 /(G L_s^2) }\,.
\ee 
\label{semiclass:weak}
\esub
Now all the operators corresponding to the $\mathfrak{iso}(2,1)$ generators are well defined.  If we want to reconstruct the evolution of the black hole, we simply need to map these generators back to the gravitational phase space. This is exactly the same as the mapping to the WdW quantization \eqref{WdW_1}, \eqref{WdW_2} :
\be
\widehat{V_1}\, \psi(\rho, \phi) =\lambda (\widehat{\Pi_0}+\widehat{\Pi_y})\psi(\rho,\phi)\,,&\q\q
\widehat{V_2}\, \psi(\rho, \phi) =\lambda (\widehat{J}+\widehat{K_x})\,\psi(\rho,\phi)\,.
\ee
A straightforward calculation, using the properties of Gaussian integrals, gives the expectation values
\bsub
\be
\braket{\widehat{V_1}}&=\f{G^2}{2\lambda \cV_0^2} e^{a_*+ \f{1}{8\sigma_B^2}}\,\left (z_*^2+\f{1}{4 \sigma_E^2}\right ):=\f{G^2 A_*}{2 \cV_0^2}\,\left (z_*^2+\f{1}{4 \sigma_E^2}\right ) \,,\\
\braket{\widehat{V_2}}&= \f{B_* G}{\cV_0} \, z_* -\f{1}{2 L_s^2} \,\left (z_*^2+\f{1}{4 \sigma_E^2}\right )\,.
\ee
\label{expected_gaussian_weak}
\esub
Comparing with the classical solutions, we identify $\tau= z_*$, and we see that the quantum evolution closely follows a classical trajectory, for large time $\tau\gg1/4\sigma_E^2$, up to a shift in the relationship between the classical constant of motion $A_*$ and its quantum realization $e^{a_*}$, due to the quantum indetermination. On the other hand for small $\tau$, the quantum correction comes into play and it actually prevents $V_1$ from being zero, avoiding the singularity. 

We must also notice that the imposition of the weak constraint coincides with considering $H$ as a true Hamiltonian so that the states must satisfy the Schr\"odinger equation
\be
i \partial_\tau \ket{\psi(\tau)} = H \ket{\psi(\tau)}\,,
\ee
whose solution is indeed given by \eqref{semiclass:weak} replacing $z-z_* \mapsto \tau-\tau_0$. In the following figure, there is a comparison between the classical solution and the expectation values of a Gaussian semiclassical state. Therein the classical trajectory is represented by the black dashed line. This one is separated into two branches, the upper one is for positive $\tau$ (we set $\tau_0=0$) and represents the black hole solution. The interior is for $V_2>0$ and the exterior for $V_2<0$. The lower branch, for negative $\tau$, represents the white hole side. Because of the negative mass, this has no horizon and $V_2$ is always negative. The red solid line in \ref{fig:effective_gaussian} represents the quantum effective evolution, resulting from a weakly constrained state, this has the property of smoothly connecting the two regions, without going through a singularity, i.e $\braket{V_1}$ is never zero. 

\vbox{
\begin{center}
\begin{minipage}{0.4\textwidth}
		\includegraphics[width=\textwidth]{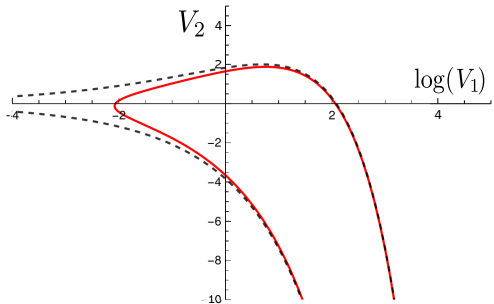}
\end{minipage}
~\hspace{0.5cm}
\begin{minipage}{0.5\textwidth}
\captionof{figure}{\small{Plot (in solid red) of the expectation values of the gaussian state \eqref{expected_gaussian_weak} in the $V_i$ plane, with $A_*= 1\,; B_*=2	\,;L_s=\cV_0/G\,$ and $\sigma_E=G/\cV_0\,$, compared with the respective classical trajectory labeled by the first integrals $\cB=B_*\,,\cA=A_*$, in dashed black.} \label{fig:effective_gaussian}}
\end{minipage}
\end{center}
}
If we want to reconstruct the effective metric we need to insert the expectation values into the line element \eqref{Bh_minisup_line:Vi}. For this, we shall also choose a shift function, and for the moment we will restrict ourselves to the case $N^1=0$. We will come back to the freedom of the shift later. We shall remark here that we actually calculate the expectation values of the fundamental fields $V_i$, and not of the metric coefficients. The two can differ from some $\sigma$ correction. We consider indeed the effective metric to be:
\be
\de s^2_\rm{eff} = - \f{\braket{V_1}}{2\braket{V_2}}\de \tau^2+\f{\braket{V_2}}{2	\braket{V_1}}\de x^2+\braket{V_1}L_s^2\de\Omega^2\,.
\ee
Doing so we have a well defined operator $V_i$, otherwise we should introduce some regularization to deal with the inverse operator. With a suitable change of coordinates similar to \eqref{change_classic}\footnote{The change of coordinate is here
\be
\tau= z_*=\sqrt{\f{2}{A_*}}\f{\cV_0 }{G L_s} R	\,,\q\q x= \f{L_s G}{\cV_0}\sqrt{2A_*} T \notag\,.
\ee}, we can recast the metric into the form
\be
\de s^2_\rm{eff} = -f(R)^{-1} \de R^2 + f(R) \de T^2 + (R^2 + \Delta) \de \Omega^2\,,\q f(R) =\f{2 M R - R^2 - \Delta}{R^2 + \Delta}\,,
\label{effective_metric}
\ee
where the quantum corrections are encoded in the parameter
\be
\Delta = \f{A_* L_s^2 G^2}{8 \cV_0^2 \sigma_E^2} = \f{2 A_* \lp^4}{L_0^2 L_s^2 \sigma_E^2}\,.
\ee
It depends on the scales of the system, but also the quantum states through the uncertainty on the energy $\sigma_E$ and the classical first integral $A_*$. For small quantum correction ($\Delta/M^2\ll 1$), the region where $\tau$ is timelike is bounded by two horizons for the Killing vector $\partial_r$, in correspondence of the zero of $V_2$. The outer one is close to $R\approx 2M$ and represents the event horizon for the outside of the black hole. The inner one is close to $R\approx 0$. The interior structure resembles closely the Reissner-Nordstr\"om solution of general relativity, bounded by two null horizons. Extending the solution outside the horizons \cite{Ashtekar:2018cay, Achour:2021dtj}, we merge two asymptotically flat regions at $R\to\pm\infty$, without any singularity. This also solves the problem of the hole in the conformal causal diagram that appears at the classical level.
 
Figure \ref{fig:lighcone} represents a schematic diagram for the lightcones structure in the three regions and the related conformal causal structure.

\vbox{
\begin{center}
\begin{minipage}{0.35\textwidth}
		\includegraphics[width=\textwidth]{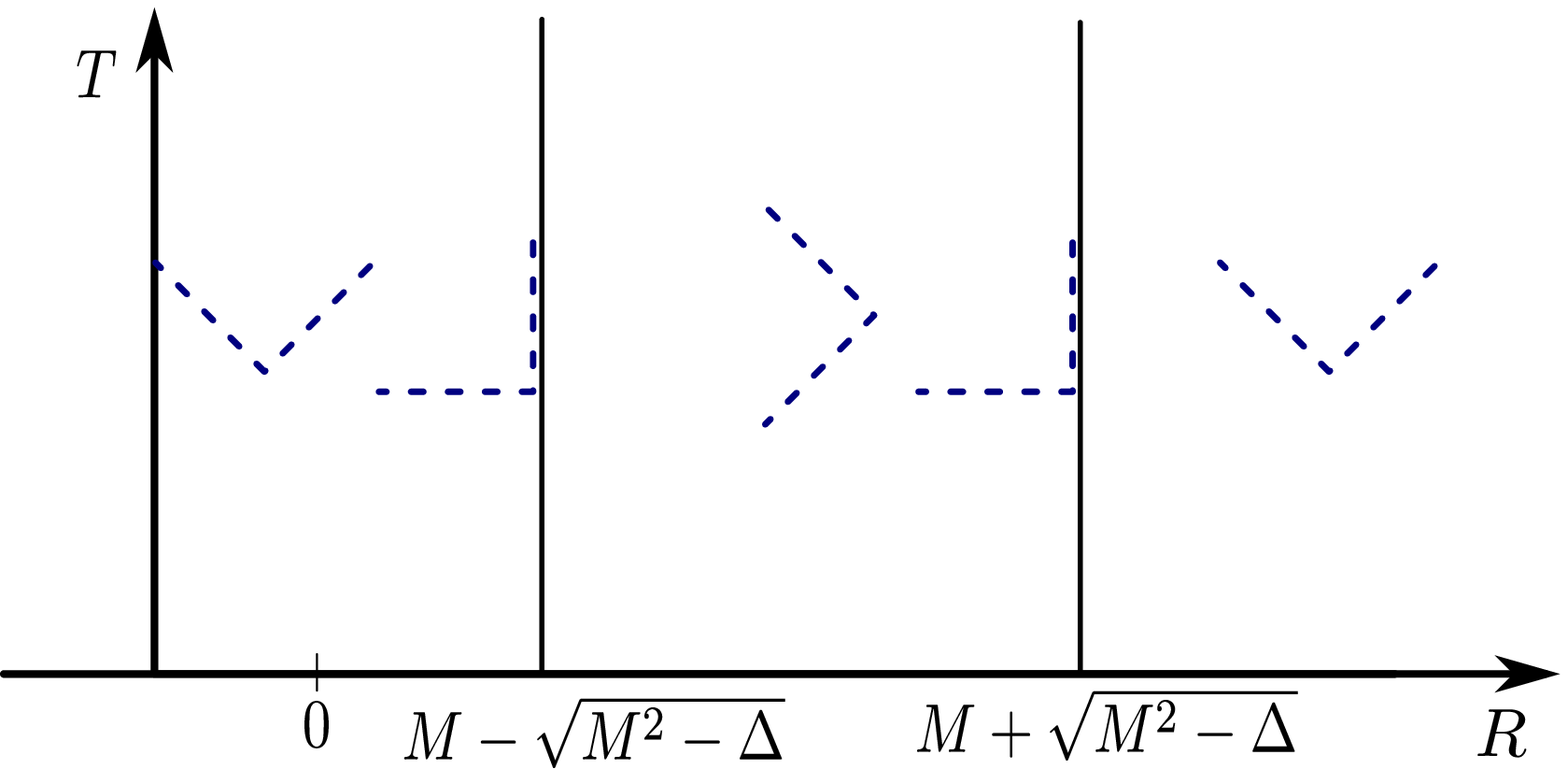}
\end{minipage}
~\hspace{0.5cm}\begin{minipage}{0.35\textwidth}
		\includegraphics[width=\textwidth]{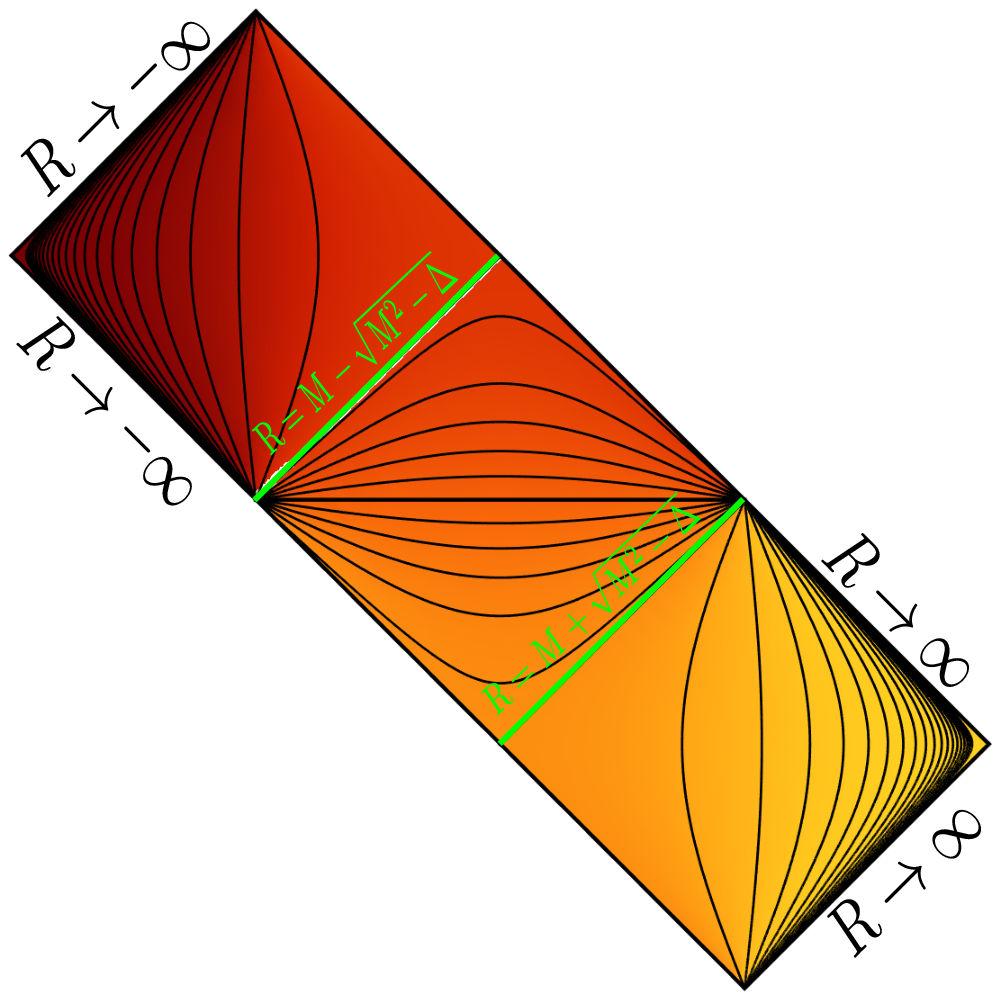}
\end{minipage}\\
\begin{minipage}{\textwidth}
\captionof{figure}{\small{On the left we have the future directed lightcones in the effective solution with two horizons, the $R$ coordinates spans the whole real line, the vertical lines correspond to the locations of the horizons, on the right the conformal Penrose diagram, where the lines at constant $R$ are drawn in black. The backgound color is lighter for $R\to \infty$, and gets darker when $R$ decrease. With respect to the classical setup we see that there is no singularity and the hole in the diagram disappears.} \label{fig:lighcone}}
\end{minipage}
\end{center}
}

Looking at the exterior region for positive $R$ this effective solution will give a new class of stationary modified black holes, and it would be interesting to further study the correction to standard black hole physics (e.g. Hawking radiation or quasi-normal modes) on such an effective spacetime \cite{Arbey:2021jif, Arbey:2021yke}. Concerning the region behind the inner horizon, it represents a white hole \textit{outside} region, where $\partial_T$ is timelike, and the \textit{effective} mass is negative $-M$. This could also have interesting perspectives on the dark matter issue.

We shall nevertheless remark that the locations (and even their existence) of the two horizons depend on the quantum states. Moreover in the extremal limit $\Delta\to M^2$ the two horizons coincide and the quantum correction becomes relevant at a macroscopic scale, meaning that there are large deviations from the classical solution in the low curvature regime, near the horizon. In order to see this, we could also look at the value of the Kretschmann scalar at the transition surface ($R=0$), where the radius of the two-sphere ($V_1$) is minimal. It is given by
\be
\cK_{R=0} = \f{12}{\Delta^2}\,.
\ee
All the corrections to the standard Schwarzschild solution are encoded in the parameter $\Delta$, which in turn depends on the uncertainty on the energy $\sigma_E$ and the classical first integral $A$. The importance of the quantum correction is directly proportional to $\Delta$, so inversely proportional to $\sigma_E$. By calculating the expectation values of the squared operators it is possible to show that the uncertainties on the metric coefficients ($\delta V_i:= \braket{V_i^2}-{\braket{V_i}}^2$) near the minimal radius $(z_*=0)$ also grow inversely proportional to the dispersion $\sigma_E$: 
\be
\delta V_i\big |_{z_*=0} \xrightarrow[\sigma_E \to 0]{} \cO\left (\f{1}{\sigma_E^4}\right )\,,\q\q
\begin{array}{rl}
\delta V_1\big |_{z_*=0} \xrightarrow[\sigma_E \to \infty]{} \cO\left (\dfrac{1}{\sigma_E^4}\right )\,\\[.4cm]
\delta V_2\big |_{z_*=0} \xrightarrow[\sigma_E \to \infty]{} \cO\left (\dfrac{1}{\sigma_E^2}\right )\,\\
\end{array}\,.
\ee
It is then logical to expect that for a heavily fluctuating metric (small $\sigma_E$), the quantum correction becomes important, and this is indeed what happens. The problem of this model is that for any given dispersion $\sigma_E$, playing with $\cA$, coupled with the other integral $\cB$, it is possible to make $\Delta$ as big as desired, without changing $M$. Even for small metric fluctuation (big $\sigma_E$), the deviation from the Schwarzschild solution could be relevant as close as desired to the horizon, or even cancel the horizon itself (if $\Delta>M^2$). This means that we need to add by hand a first-class constraint on the integral $\cA$, that fixes its value. We also would like to eliminate the dependence on the fiducial scale of the quantum correction. This uniquely fixes $A_* \propto L_0^2 L_s^2/L_P^4$ and implies that all the corrections come from the uncertainty on the energy. We shall remark that this represents a huge difference with respect to cosmology, where the appearance of quantum correction for the Wheeler-DeWitt quantization has not been observed \cite{Ashtekar:2011ni}. Nonetheless, this dependence on the energy uncertainty disappears for non-squeezed states, where the dispersion on conjugate variables is minimized, by e.g. fixing $\sigma_{E, B} = 1/2$.

We will see in the next chapter (\ref{chap6}) how the results are modified if we introduce a regularization. Before moving to the study of possible regularizations we could exploit again the Poincar\'e structure to discuss the mass spectrum, this is the subject of the next subsection.

\section{Mass operator}
\label{mass_spectra}
We recall that at the classical level we have a degeneracy on the definition of the mass, we have indeed two first integral $\cA$ and $\cB$, that combine into \eqref{classic mass} to give the only physical quantity that is relevant in the GR framework, the mass, that labels diffeo-inequivalent solutions. But if we look at the quantum theory, we expect that both $\cB$ and $\cA$ acquire some fluctuation contributing to the mass. Moreover, the two observables do not commute, and we had to build coherent states to represent semiclassical solutions with a finite spread on both of them. 

This in turn inevitably forces us to work with semiclassical states that are not eigenvectors of the mass operator. Nevertheless, the group quantization provides interesting information about the mass spectrum. For this purpose, we need to further investigate the properties of the mass operator. 
For this, we start by mapping the classical observable measuring the squared mass to a combination of Poincar\'e generators

Despite its non-analytical form, with the definition of the $\cA$ operator and the \textit{square root} of the $\sl(2,\R) $ Casimir ($\cB$) provided at the beginning of this section \eqref{ba_operators}, we can easily build the self-adjoint mass operator
\be
\widehat{M} \psi (a,z)&:= L_M \widehat \cA^{1/4} \widehat \cB  \widehat \cA^{1/4} \, \psi(a,z) := i \f{L_M}{\sqrt{\lambda}} e^{a/4} \partial_a \left (e^{a/4} \psi(a,z) \right )\,,\\ 
L_M&:= \f{L_s^3 G^2}{\sqrt{2} \cV_0^2} = \f{8 \sqrt{2} \lp^4 }{ L_0^2 L_s}\notag \,,
\ee
where we have introduced the constant length $L_M$, indicating how the UV
fundamental length and the fiducial scales couple into the definition of the mass observable. We shall remark that, despite the apparent dependence on the fiducial scale, the mass is unchanged by a rescaling of the IR length, because also $\cA$ and $\cB$ change under the rescaling \cite{Geiller:2020edh}.

We shall remark that this operator sees only the $a$ dependence of the wavefunction. Unsurprisingly this means that the mass is a Dirac observable commuting with the Hamiltonian, which in turn implies that it can be measured without any problem on both the weakly and strongly constrained states. For the sake of simplicity, in the following we will consider states on the strongly constrained physical space, tracing out the time-energy dependence of the wavefunction. The conclusions about the spectrum will not be affected by this simplification.

We can explicitly calculate the wavefunctions that diagonalise the mass operator and provide a complete basis for the physical wavefunctions. These are given by the set
\be
\braket{a|M} &:= \psi_M(a)= \sqrt{\frac{2 \lambda^{1/2}}{\pi L_M}} e^{-a/4} \cos \left(\frac{2 M \sqrt \lambda}{L_M e^{a/2}}\right) \,,\q \widehat{M}\ket{M} = M\ket{M}\,. \label{a_mass}
\ee
%
where $M$ is a real positive continuous parameter. By virtue of the integral properties of the cosine functions we can prove the orthogonality and completeness of the basis. With  $y= 2 \sqrt \lambda e^{-a/2}/L_M$, we have indeed
\be
\braket{M'|M}= \f{2}{\pi} \int_0^\infty \de y \, \rm \cos (M y) \rm \cos (M' y) &= \delta(M-M')\, ,\q\q
\int_0^\infty \ket{M}\bra{M} =\I\,.
\ee
The existence of the Poincar\'e structure forces the mass to have a continuous spectrum, as it has been pointed out in \cite{Achour:2021dtj}. This property is in contrast with several other investigations of black hole spectra \cite{Bekenstein:1995ju, Berezin:1998xf} where a discrete spectrum is postulated or obtained \cite{Vaz:1998gm, Louko:1996md, Bojowald:1999ex}. In particular, it means that the black hole could emit particles with any given mass, and not only the ones corresponding to the gap between eigenstates. 

\ \newline

In this chapter, we have seen how to exploit the existence of the rigid symmetry to build a quantum theory of the black hole minisuperspace. The simplicity of the model allows us to overcome many difficulties of the covariant approach to quantum gravity. Moreover, the remark that the symmetry is physical, in the sense that it reflects the scaling property of the system, together with containing the mass observable, makes it natural to request its protection on a quantum level. In addition, the conserved charges naturally provide a full set of Dirac observables, and the well-known representations of $\ISO(2,1)$ make the semiclassical limit easy to handle. We have seen however a huge difference between the imposition of the strong constraint and the addition of some energy fluctuation. This results in a completely different dynamics for the latter, where the singularity is replaced by a transition surface and a second horizon. In the next chapter, we will see to what extent the introduction of a regularisation on the phase space closes the gap between \textit{strong} and \textit{weak constraints}.

\chapter{Polymerisation and singularity resolution}
\label{chap6}

In this section, we will discuss how it is possible to define a \textit{polymer} quantization that preserves the $\ISO(2,1)$ symmetry. For this regularization, the evolution of the coherent state will reproduce the effective metric \eqref{effective_metric} for both the strong and weak constraints.

In the chapter \ref{chap4} we stressed that the main ingredient of Loop Quantum Cosmology is a realization of the Weyl algebra on a non-separable Hilbert space, inequivalent to the standard Schr\"odinger representation. For a given configuration variable (say $q\in \R$) the space is spanned by orthogonal vectors $\ket{q}$ and it contains functions that are non-vanishing only on a countable subset of $\R$. The lack of weak continuity implies that the momentum operator (say $p=-i\partial_q$) is not defined, but only its finite exponential $e^{i \lambda p}$. This leads to the necessity to introduce a regularized Hamiltonian, where the momenta are replaced by (combination of) their exponentiated version. This is usually done by the substitution $p \mapsto \sin (\lambda p)/\lambda$, but other regularizations are possible as well, and the exact form of the effective Hamiltonian has been heavily debated, especially in the context of black hole interior \cite{Ashtekar:2018cay, Ashtekar:2018lag, Bodendorfer:2019xbp, Bodendorfer:2019cyv, Bodendorfer:2019nvy, Bodendorfer:2019jay, Ashtekar:2020ckv}. In any case, the regulator $\lambda$ is claimed to encode the fundamental discreteness of spacetime, relating its value to the Planck length. In the limit where it becomes negligible $\lambda \to 0$, we shall recover the classical evolution. For a given parameter $\lambda$ the Hilbert space is divided into so-called superselected sectors, according to the position eigenstates, the latter taking discrete real values $\epsilon + n \lambda$, with a fixed offset $\epsilon$. The operator $e^{i \lambda p}$ creates a finite shift of step $\lambda$ and lets us move within a given superselected sector. 

The problem with introducing a regularization scheme for the Hamiltonian is that, in general, it spoils the classical Poincar\'e symmetry, unless we extend the regularization to the other observables. A systematic way to ensure that any Poisson structure on a phase space is preserved is to look at the regularization as a canonical transformation \cite{Geiller:2020xze, BenAchour:2019ywl}, represented by a symplectomorphism $(V_i, P_i) \mapsto (v_i,p_i)$. For this to describe a different physics we also need to change the mapping between metric coefficient and field space, this means that we replace the new variables $v_i$ at the place of the corresponding \textit{classical} $V_i$ into the line element \eqref{Bh_minisup_line:Vi}, instead of writing the old variables in terms of the new ones. To make this clearer let us imagine a point transformation $V_i = V_i(v_i)$, the usual classical picture is unchanged if we calculate the evolution of the new variables and than we insert them into
\be
\de s^2 = - \f{V_1(v_i)}{2V_2(v_i)}\de \tau^2+\f{V_2(v_i)}{2	V_1(v_i)}\de x^2+V_1(v_i) L_s^2\de\Omega^2\,,
\ee
on the other hand the metric
\be
\de s^2 = - \f{v_1}{2v_2}\de \tau^2+\f{v_2}{2	v_1}\de x^2+v_1 L_s^2\de\Omega^2\,
\ee
might be completely different from the original one, in this sense the polymerization as canonical transformation changes the physics of the system.

In \cite{Geiller:2020xze} we have proposed a possible canonical transformation that realises the simple polymerisation for the $P_i$ variables. This was only performed on an effective level, as the majority of the works in the framework of black hole minisuperspace. We recall rapidly some of its features here. The aim is to obtain the replacement $P_i \mapsto \sin(\lambda P_i)/\lambda_i$ in the Hamiltonian, for some dimensionless constants $\lambda_i$. In \cite{Geiller:2020xze} we gave explicitly the transformation, that is:
\bsub\be
V_1 &= v_1 \cos^2\left (\f{ \lambda_1 p_1}{\lambda_1}\right )\,,  & V_2 &= v_2 \cos^2\left (\f{ \lambda_2 p_2}{\lambda_2}\right ) + \f{\lambda_2^2 \cV_0^2}{2 G^2 L_s^2} \,,\\  
P_1 &= \f{2}{\lambda_1} \tan \left ( \f{\lambda_1 p_1}{\lambda_1}\right )\,,  &P_2 &= \f{2}{\lambda_2} \tan \left ( \f{\lambda_2 p_2}{\lambda_2}\right )\,.
\ee \label{canonic_transf_1}\esub
This was obtained by starting from the hypothesis of a point symplectomorphism for the momenta $P_i\mapsto p_i$, with $P_i = F_i(p_i)$, and the polymer Hamiltonian
\be
H_\rm{poly} =  -\f{G}{2 \cV_0} \f{\sin(\lambda_2 P_2)}{\lambda_2} \left (2 \f{\sin(\lambda_1 P_1) V_1}{\lambda_1} + \f{\sin(\lambda_2 P_2) V_2}{\lambda_2}\right ) \,. \label{h_poly_1}
\ee
We then calculate the evolution of the \textit{polymer} variables, via the Poisson bracket with $H_\rm{poly}$. Using the on shell value of the Hamiltonian ($H_\rm{poly}\approx H\approx \cV_0/GL_s^2$) and the requirement for the right semiclassical limit when $\lambda_i \to 0$, we finally obtain the form of $F_i$ and the transformation \eqref{canonic_transf_1}. We have also studied the causal structure of the effective solution. The non-vanishing of the polymer field $v_1$ is synonymous with a singularity resolution, that is replaced by a bounce in the two-sphere radius. This is seen in the effective trajectories in the figure below \eqref{fig:effective_polym}

\vbox{
\begin{center}
\begin{minipage}{0.4\textwidth}
		\includegraphics[width=\textwidth]{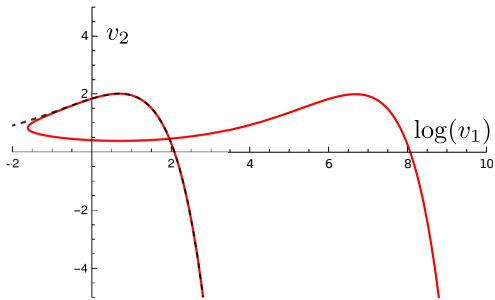}
\end{minipage}
~\hspace{0.5cm}
\begin{minipage}{0.5\textwidth}
\captionof{figure}{\small{Plot (in solid red) of the expectation values of the polymer variables in the $v_i$ plane, the initial conditions are given by $A= 1\,; B=2	\,; L_s=\cV_0/G\,$ and the polymerisation parameters are $\lambda_i =0.1\,$. The effective evolution is compared with the respective classical trajectory for the $V_i$'s in dashed black} \label{fig:effective_polym}}
\end{minipage}
\end{center}
}
The solutions are easily obtained by inverting the canonical transformation \eqref{canonic_transf_1}, and replacing the solutions into the modified line element gives the effective spacetime \footnote{Actually in \cite{Geiller:2020xze} we point out that the modified line element is not simply obtained by a replacement $V_i\to v_i$ in \eqref{Bh_minisup_line:Vi}, but there is a modification in the clock choice, the $\de \tau$ term pick up a factor $\cos^2(\lambda_2 p_2/2)$, to have a coherent relationship between the mapping between Hamiltonians ($H_\rm{polym}$ and $H$) and the energy levels. Thus the effective line element is
\be
\notag \de s^2_\rm{poly} = -\f{v_1}{2v_2}\f{\de \tau^2}{\cos^4(\lambda_2 p_2/2)}   + \f{v_2}{2v_1} \de x^2 + L_s^2 v_1 \de \Omega^2\,.
\ee
}. Just like in most of the effective studies of the LQC black hole interior, we find that the singularity is replaced by a black-to-white hole transition. At the minimal radius (minimum of $v_1$) we have a sphere which is a transition surface between the trapped (BH) and anti-trapped (WH) space-like regions. This can be checked by computing the expansion of the future pointing null normal to the 2-sphere at constant time, which changes sign while passing through the transition surface. We refer the reader to \cite{Geiller:2020xze}6 for details about this model.

We just would like to point out some substantial differences with respect to the effective evolution obtained in the previous chapter, represented in figure \eqref{fig:effective_polym}. Both the models are regular in the sense that there is no singularity, which is replaced by an inner horizon, connecting the interior to an asymptotically flat \textit{white hole} region. The main difference concerns the effective mass of the white hole: while it is negative for the first case (figure \ref{fig:effective_gaussian}) but its absolute value equals the mass on the black hole side, it is positive in the second case (figure \ref{fig:effective_polym}). A stronger contrast between the two effective metrics is however on a more conceptual level. The quantum effects of the polymerisation are indeed introduced here purely on a semiclassical level, performing a canonical transformation and changing the variables in the line element. Due to the high non-linearity of the Hamiltonian, we are however no longer sure that the evolution of coherent states for the corresponding quantum evolution reproduces the same dynamics. For instance, we have seen that when we consider the weakly constrained states for the non-polymer dynamics, the quantum effects are already relevant and prevent the system to go through a singular point. If now we promote the variables $v_i, p_i$ to operators and we make the states evolve under the quantum version of \eqref{h_poly_1}, it is reasonable to expect that some other quantum effect comes into play, dramatically modifying the effective structure. In other words, the quantum effects of the semiclassical polymerisation will differ from the effects of the quantization! 
Even worse, in general, we are not even able to describe such coherent states, due to the complexity of the quantum Hamiltonian. 

As a final remark, we point out that this way of looking at the polymerisation as a symplectomorphism does not solve all its ambiguities. The choice of momenta to be polymerised is still arbitrary, as well as the choice of \textit{lowest spin representation}, manifest in the choice of $\sin$ functions. If we stick to the effective level, there is an infinite amount of canonical transformations that possibly define a polymerisation, in the sense that we replace some momenta with a periodic function. It would be interesting to study whether there are some global features, sheared by any choice of periodic function, or not. This is however beyond the scope of this thesis and we defer it to future works. On the other hand, we can make use of the explicit realization of the symmetry on a quantum level to constrain the polymerisation. This is the subject of the next section.

\section{Regularisation from symmetry representations}

In the previous chapter, we have detailed the properties of the Poincar\'e representations, and we now want to make use of them to discuss the polymerisation. In LQC the inequivalence between the Wheeler-DeWitt and the loop quantization is obtained by avoiding the Stone-Von Neumann theorem through a modular Hilbert space. However, if we want the symmetry to be preserved upon regularisation, we need to work within the Hilbert space presented in the previous chapter, which is equivalent to the Wheeler-DeWitt quantization (with appropriate boundary conditions).

The idea is then to mix the two ingredients at our disposal: the representations of $\ISO(2,1)$ and the polymerisation as canonical transformation. We recall that in general a canonical transformation could be implemented by a non-unitary transformation at the quantum level. The difference between the WdW quantization and the polymer one is not seen as the result of Hilbert spaces that are unitarily inequivalent, but as the consequence of considering inequivalent operators (not related by a unitary transformation, but still canonical on the semiclassical level) on the \textit{same} Hilbert space, that in addition carries an irreducible representation of $\ISO(2,1)$. Moreover, we expect to be able to rewrite the new variables $v_i$ in terms of $\iso (2,1)$ generators and calculate their expectation values.

Nevertheless, we cannot freely choose any transformation, but we want it to satisfy a set of properties:
\begin{itemize}
\item the transformation must be such that the effective metrics are asymptotically-flat,
\item we want the phase space functions representing the \textit{polymer} coefficients to have a quantum realization with discrete spectra.
\end{itemize}

We already have at our disposal an operator whose eigenvalues are discrete, and with a superselected sector,  the rotation generator $J$. The idea is then to take the regularized metric coefficient to be \cite{Bodendorfer:2020ovt}
\be
v_2 = 2\lambda J = V_2 - 2\f{\cV_0}{G}\lambda^2 H\,.
\label{poly_v2}
\ee
The $\lambda$ parameter must be the same as in the mapping from the original phase space to the $\iso(2,1)$ generator to ensure the right limit $\lambda\to 0$, that maps back to the original phase space. For the superselected sector chosen in the previous sections $(s=0)$, the eigenvalues of $v_2$ are discrete real values $2 n \lambda\,, n \in \Z$.   

Concerning $V_1$ the polymerization is less straightforward. Assuming the transformation \eqref{poly_v2} for $V_2$, we find that a compatible canonical transformation is of the form:
\be
\left |
\begin{array}{rl}
v_1&= v_1(\Pi_0, \cB) = v_1(V_1 + \lambda V_1 P_2^2, V_1 P_1)\,,\\
p_1&= p_1(\Pi_0, \cB) = p_1(V_1 + \lambda V_1 P_2^2, V_1 P_1)\,,\\
v_2&= 2\lambda J = V_2 - 2\lambda^2\df{\cV_0}{G}H\,,\\
p_2&= \df{1}{\lambda} \arctan (\lambda P_2)\,,\\
\end{array} \right.
\ee
with two function satisfying $\{v_1, p_1\}=1$. We thus need to find an operator, built with $\Pi_0$ and $B$ that has a discrete spectrum. Unfortunately, this is not achievable through a linear combination, but we need at least a quadratic operator. The simplest one is given by
\be 
\widehat{v_1^{\,2}} := 4 \lambda^2 \widehat \Pi_0^2 - \mu^2 \widehat \cB^2\,,
\label{polymer_v1}
\ee
with a real parameter $\mu$. Its discrete eigenvalues are $ 4\mu^2 n^2\,,n \in \Z$.  In order to see this explicitly we shall in a first place look at the linear combination
\be 
\widehat{v_1} := 2 \lambda \widehat{\Pi_0} + \mu \widehat{\cB}\,.
\ee
By acting on a wavefunction in the polar representation, we can analytically solve the eigenvalue problem and find the eigenvectors
\be
\braket{\rho,\phi|v_1,m} = \f{1}{2\pi \sqrt{\mu \rho}}e^{-i\left (2\f{\rho\lambda}{\mu}-\f{v_1}{\delta} \log \rho\right )} e^{-i m \phi}\,,\q\q \widehat{v_1}\ket{v_1,m} =v_1 \ket{v_1,m}\,,
\ee
that provide an orthonormal basis for a continuous spectrum $v_1 \in \R$: 
\be
\braket{v_1',n|v_1,m} =\delta_{n,m}\delta(v_1-v_1')\,.
\ee
On the other hand the quadratic combination  \eqref{polymer_v1} has eigenvectors
\be
\braket{\rho,\phi|v_1^{\,2},m} = \sqrt{\f{v_1}{2\lambda\rho}}\,J_{\f{v_1}{\mu}}\left (\f{2 \lambda \rho}{\mu}\right )\, \f{1}{\sqrt{2\pi}}e^{-i m \phi}\,,\q\q \widehat{v_1^2}\ket{v_1^2,m} = v_1^2 \ket{v_1^2,m}\,,
\ee
where $J_n$ are the first kind Bessel fucntions. The eigenstates are normalized for a discrete spectrum $v_1/2\mu \in Z$. This is easily shown by using the integral property of the Bessel functions
\be
\int_0^\infty \f{\de y}{y}J_{2n}(y)J_{2m}(y) = \f{1}{2\pi}\f{\sin\left (\pi(n-m)\right)}{n^2-m^2} = \f{1}{2 n}\delta_{n,m}\q \Leftrightarrow\q n,m \in \Z\,.
\ee

In this case, we will not have access to the quantum operator measuring $v_1$, but only its square. From the point of $v_1$, this is similar to what happens in cosmology, where the fundamental discreteness is imposed on the volume, the third power of the scale factor. From the point of view of the scale factor in FLRW cosmology and $v_1$ here, it looks like a so-called $\bar \mu$-scheme. Nevertheless, at the semiclassical level, it is possible to take the square root and implement the canonical transformation: 
\be
\left |
\begin{array}{rl}
v_1&= \sqrt{4 \lambda^2 \Pi_0^2 - \mu^2 \cB^2} = V_1 \sqrt{(1+P_2 \lambda^2)^2 - \mu^2 P_1^2}\,,\\
p_1&= \df{1}{\mu} \arctan \left (\df{\mu P_1}{\sqrt{(1+P_2 \lambda^2)^2 - \mu^2 P_1^2}}\right )\,,\\
v_2&= V_2 +\lambda^2 P_2 (2 P_1 V_1 + P_2 V_2)\,,\\
p_2&= \df{1}{\lambda} \arctan (\lambda P_2)\,.\\
\end{array}\right.
\ee
For the previous construction to make sense, we need to add a constraint on $\mu$. The minimal value of $\Pi_0$ during the classical evolution is provided by $2\lambda A$. If we want a positive definite square $v_1^2$, we need
\be
\mu  \leq \f{2 \lambda^2 \cA }{\cB}\,.
\ee
When the strict inequality holds, we have that $v_1$ is never zero and there is no singularity in the effective metric. On the other hand, if  $\mu B =2 \lambda^2 A$ there is a singularity at $T=0$. If the inequality is not satisfied the canonical transformation is not defined on the whole trajectory and the singularity occurs before the classical one. The only way to have a canonical transformation that is well defined on the whole phase space is to take the limit $\mu\to 0$, which corresponds to the situation where $V_1$ is not polymerised. This is somehow analogous to what has been found for the deformation of the constraint algebra for spherically symmetric spacetimes \cite{Arruga:2019kyd}, where only one of the momenta gets polymerized. When $\mu$ goes to zero, we gain back an operator measuring directly the metric coefficient $v_1$, instead of its square value, but now it has continuous eigenvalues:
\be
\mu \to 0\q \Rightarrow\q
\left |
\begin{array}{rl}
v_1&= 2\lambda \Pi_0 = V_1 (1+P_2 \lambda^2)\,,\\
p_1&= \df{P_1}{1+P_2 \lambda^2}\,,\\
v_2&= 2\lambda J\,,\\
p_2&= \df{1}{\lambda} \arctan (\lambda P_2)\,.\\
\end{array}\right.
\ee
Inverting the canonical transformation we can express $H$ in terms of the polymerised variables and it gives
\be
\f{\cV_0}{G} H= -v_1 \f{\tan \mu p_1}{\mu} \f{\sin (2\lambda p_2)}{2\lambda} - v_2 \f{\sin^2(\lambda p_2)}{2\lambda^2} \q\xrightarrow[\mu\to 0]{} \q-v_1 p_1 \f{\sin (2\lambda p_2)}{2\lambda} - v_2 \f{\sin^2(\lambda p_2)}{2\lambda^2}
\label{poly_hamilton}\,.
\ee
The evolution of $v_1$ and $v_2$, generated by the effective Hamiltonian, can be easily solved by inverting the canonical transformation above. 

In the appendix \ref{app:4_LQG_phase_space} we review the relationship between the Ashtekar variables and the $\bms$ vectors $V_i$. It happens that from the point of view of the triad connection variables, the polymerization proposed here is translated to a half $\bar\mu$-scheme, where the connection along the angular direction is polymerised with a $\bar \mu$-like scheme, while the connection along the radial direction is not polymerised, this is equivalent to introducing a discretization only along the edges perpendicular to the angular, direction, i.e. along the $x$ coordinate.

The main advantage of looking at the regularization as a canonical transformation is that the description of both classical and quantum mechanics in terms of Poincar\'e generators is not modified. In other words, the evolution is always generated by $K_x -J$ and the dynamical quantum states are the same as in the previous section, according to which kind of constraint we want to impose. We simply need to change the operators whose expectation values we want to replace in the semiclassical metric. Contrary to the heuristic approach presented at the beginning of the chapter, we path followed here keeps us closer to the $\ISO(2,1)$ representation theory, making it straightforward to analyse the question of the semiclassical limit. We study here the case where $\mu \to 0$, and the effective metric corresponds to  
\be
\de s^2_\rm{eff} = - \f{\braket{\Pi_0}}{2\braket{J}}\de \tau^2+\f{\braket{J}}{2\braket{\Pi_0}}\de x^2+2\lambda \braket{\Pi_0}L_s^2\de\Omega^2\,.
\label{regular_metric}
\ee
Without much surprise, deparametrizing the dynamics with respect to the time $\tau$, the evolution of $J$ and $\Pi_0$ on the physical Hilbert space satisfying the strong constraint follows the respective classical trajectories
\be
\begin{array}{rl}
\braket{\lambda  \widehat{\Pi_0}(\tau)}&=\dfrac{A_*G^2}{2\cV_0^2}\,\left (\tau^2 + \df{4\cV_0^2}{G^2} \lambda^2\right ) \,,\\[9pt]
\braket{\lambda  \widehat{J}(\tau)}&= \dfrac{B_* G}{\cV_0}\, \tau -\dfrac{1}{2L_s^2}\,\left (\tau^2 + \df{4\cV_0^2}{G^2} \lambda^2\right )\,.\\
\end{array}
\ee
The expectation values are on the Gaussian physical state, as in \eqref{v_strong_evolution}, and the deparametrization has been done by using the classical solution for $\Pi_0$ and $J$ given directly by \eqref{classic_traj}, and then we replace the first integrals $\cA,\cB$ by the corresponding quantum operator. This leads to the same effective metric as in \eqref{effective_metric}, where now $\Delta= 2\lambda^2 A_*L_s^2$. 

Perhaps even more surprisingly, the $\sigma$ correction appearing in the evolution on the weakly constrained states goes in the same direction: more precisely, on the weakly constrained Gaussian wavepackets \eqref{semiclass:weak} we have
\be
\begin{array}{rl}
\braket{\widehat{\lambda  \Pi_0}}&=\df{G^2}{2\lambda \cV_0^2} e^{a_*+ \f{1}{8\sigma_B^2}}\,\left (z_*^2+\df{4 \lambda^2 \cV_0^2}{G^2}+\df{1}{4 \sigma_E^2}\right ):=\df{A_* G^2}{2\cV_0^2 }\,\left (z_*^2+ 4 \lambda^2 \df{\cV_0^2}{G^2}+ \df{1}{4 \sigma_E^2}\right ) \,,\\[9pt]
\braket{\widehat{\lambda J}}&= \df{B_* G}{\cV_0} \, z_* -\df{1}{2L_s^2}\,\left (z_*^2+ 4 \df{\cV_0^2}{G} \lambda^2 +\df{1}{4 \sigma_E^2}\right )\,.
\end{array}\label{polymer_gaussian_weak}
\ee
The effective structure is again given by \eqref{effective_metric}, but now the quantum correction is encoded in 
\be
\Delta = 2\left (\lambda^2 L_s^2+ \f{ \lp^4}{L_0^2 L_s^2 \sigma_E^2}\right ) A_*\,.
\ee
We shall remark that this does not solve the apparent paradox of quantum correction at a macroscopic scale. If we leave $A$ free, even for metrics with small fluctuation (big $\sigma_E$), the inner horizon can come as close as desired to the external horizon. We still need to add a condition on $\cA$. 

Taking a closer look at the parameter $\Delta$, we see that it is exactly the sum of the one obtained for the weakly constrained WdW states and the strongly constrained polymer states. It is natural to interpret the two contributions as taking into account respectively the quantum uncertainty on the metric coefficients and the deep discreteness of the spacetime. For wavefunctions that are well localised, the first one is negligible compared to the second one, i.e. $\f{\cV_0^2 }{\sigma_E^2 G^3} \ll \lambda^2 $. In this case, we expect that the quantum corrections are of Planck size, meaning that the first-class constraint would impose 
\be
2 \lambda^2 L_s^2 A_* \approx \lp^2\q \Rightarrow \q
\Delta \approx \lp\left (1 + \f{G^2}{16 \cV_0^2 \sigma_E^2 \lambda^2}\right ) \,.
\label{area_constr}
\ee

This in turn implies that it is impossible to fully get rid of the fiducial scale, entering the game through the central charge $\cV_0/G$. It would be worth studying the role of the cut-off scales, looking at them as running renormalization parameters. The imposition of a first-class constraint relating the polymerization parameter to one of the first integrals is analogous to the construction in \cite{Ashtekar:2018cay}.  Despite \eqref{area_constr} being more reasonable than the one imposed in the WdW setup, the question of if and how we can infer this kind of constraint from the full LQG theory is still unanswered. However, we can have a hint about its origin by noticing the presence of $A_*$ in the coordinate redefinition \eqref{change_classic}. The relationship \eqref{area_constr} is expected to be related somehow to the introduction of a Planck length ruler on spacetime. Using the relationship \eqref{area_constr} makes the change of coordinate \eqref{change_classic} to have the interesting form 
\be
\tau -\tau_0 =\f{2 \cV_0 \lambda }{\lp^3} R\,,\q\q x= \f{\lp^3}{\cV_0 \lambda} T\,,
\label{change_quantum_2}
\ee
that could be interpreted as swapping between IR and UV scales at the coordinate level.

Furthermore, the impossibility of eliminating the dependence of the effective metric on the fiducial scales points again towards a physical role of the boundary. We would like to stress that similar behaviour has been observed in cosmology \cite{Mele:2021gzx}.

In the previous section, we saw that the evolution of the operators measuring the metric coefficients produces drastically different metrics, depending on whether we allow some energy fluctuation or not. Here the main features of the effective metric are the same in both cases, and they agree with the classical line element corresponding to \eqref{effective_metric}. The evolution of quantum coherent states on the polymer Hilbert space follows the effective evolution described by the corresponding polymer Hamiltonian, and this is stable for non zero energy fluctuation. This property has been used as a consistency check in favour of the robustness of FLRW polymerization and can be here extended to the black hole minisuperspace.

In other words, we have shown here that the effective evolution described by the polymer Hamiltonian \eqref{poly_hamilton} is the same as the expectation values of the operators $J, \Pi_0$ on some coherent state, upon identification $v_2 = 2 \lambda J$ and $v_1  = 2 \lambda \Pi_0$. Focusing on the effective approach, we would like to spend some more words on its dynamics and causal structure. For this we want to reintroduce a non-vanishing shift and consider the line element
\be
\de s^2_\rm{eff} = - \f{v_1}{2 v_2}\de \tau^2+\f{v_2}{2 v_1}(\de x +N^1 \de \tau)^2+ v_1 L_s^2\de\Omega^2\,.
\label{polym_metric_shift}
\ee
In the section \ref{sec2.3:BH_minisup} we have studied the classical solution and we have shown that by properly choosing the shift we can select the role of the homogeneous coordinate $x$, respectively as the time measured by an asymptotic observer, as the null infalling coordinate or the proper time of a raindrop. We could wonder if we have a similar property for the effective metric. The definition of the three \textit{times} investigated in the classical case do not depends on the dynamics of the fields. For instance the null coordinate is always selected by the choice
\be
N^1_\rm{EF} = -\f{v_1}{v_2}\,,
\ee
such that the $\de \tau^2$ term disappears (i.e. $g_{\tau\tau} =0$). Consistently, if we insert the effective solution into the line-element with this choice of shift and we do the same change of variables as \eqref{change_classic} we recover the null parametrization of the effective metric  \eqref{effective_metric}, that is
\be
\de s^2_\rm{EF} = -2 \de R \de T + f(R) \de T^2 + (R^2 + \Delta) \de \Omega^2\,,\q f(R) =\f{2 M R - R^2 - \Delta}{R^2 + \Delta}\,.
\label{effective_metric_null}
\ee
Concerning the raindrop proper time, the situation is a bit more subtle, because of the presence of the momenta in the shift. The definition of the raindrop time is equivalent to the requirement that $g^{TT}=-1$. On the other hand, the very definition of $T$ depends on the integrals of motion, that in turn depends on the momenta. At the classical level the following procedures commute: 
\begin{itemize}
\item[\textit{(i)}] we impose $N^1 = \left (\f{2 V_1^3 P_1}{V_2^2(2P_1 V_1 + P_2 V_2)}\right  )^{1/2}$ as described in section \ref{sec2.3:BH_minisup}, we use the classical solution and the change of variables \eqref{change_classic} to obtain the Gullstrand-Painlev\'e coordinates for the static black hole, or conversely 
\item[\textit{(ii)}] we leave a general $N^1$, we ask that $g^{TT}=-1$, where the latter is seen as an equation for the shift, where all the constant of motion $\cA$ and the energy are expressed in terms of phase space functions. 
\end{itemize}
If we want the two to commute also in the polymer case we need to choose
\be
N^1_\rm{GP} = \left (\f{2 V_1^3 P_1}{V_2^2(2P_1 V_1 + \tan (\lambda P_2) V_2/\lambda)}\right  )^{1/2} \,,
\ee
where $\lambda$ is the regularization parameter. We see that we need to introduce a regularization pattern whenever the momenta appear in the metric. This is somehow similar to the construction in \cite{Kelly:2020uwj}. Therein the metric depends on the momenta and a regularization on it is imposed by the request of closure of the constraint algebra. Here we trade the constraint algebra for the new symmetry to constrain the polymerisation, but we still need to introduce a coherent regularization pattern to have a consistent spacetime description. The corresponding solution for the line element is
\be
\de s^2_\rm{EF} =  \de R^2  +2\sqrt{\f{2M R}{R^2 +\Delta}} \de T - f(R) \de T^2 + (R^2 + \Delta) \de \Omega^2\,,
\label{effective_metric_rain}
\ee
where $f$ is given in \eqref{effective_metric_null}. The lines at constant $T$ are the same as in figure \ref{fig:coord_level}, but now the hole in the conformal diagram is filled as in figure \ref{fig:lighcone}.

The invariance under shift is an important property pointing towards a sort of covariance of the polymerisation. However, the invariance of the dynamics, both classically and at an effective level, under different choices of shift is not realised directly as a dynamical symmetry, making the connection with the covariance of the full theory difficult to establish. This is because of the lack of vector constraint from the onset and it suggests that a clear understanding of the model could be achieved by studying the symmetries of the inhomogeneous midisuperspace black hole, with both $\tau$ and $x$ seen as evolution parameters. \newline
\ 

To summarise, we have used the irreducible representation of the symmetry group to build a quantum theory of the black hole minisuperspace. We have then introduced a regularization scheme by changing the map between metric observables and quantum operators. In this setup, the inequivalence between the classical and regularised descriptions is given by a non-unitary mapping between different operators on the same Hilbert space. Nonetheless, at the effective level, this is implemented by a canonical transformation, that in turn assures the protection of the symmetry. 

However, the choice of polymerisation is not unique, as there is an infinite set of canonical transformations that could give inequivalent dynamics. We have proposed here a possible choice by adding some further requirements on the spectra of the fundamental regularised fields $v_i$, but this is still only one possibility, as nothing forbids to do the polymerisation starting from another parametrization of the field space. The argument in favour of our choice is the \textit{a posteriori} result that the effective dynamics is stable under energy perturbation and reproduces the WdW weakly constrained dynamics. We expect however that the requirement of protection of the symmetry, together with the necessity to have the right asymptotic limit, might leave unaltered some of the features of the effective dynamics studied here like the resolution of the singularity, the filling of the conformal diagram like in figure \eqref{fig:lighcone} and the presence of a black-to-white hole transition.

The generalisation to the other Bianchi models, presented in the first part of the thesis, also opens the door toward a systematic study of the quantization of the minisuperspaces exhibiting a finite group of symmetry. Along the lines of this section, we could define the corresponding quantum theory starting from the irreducible representations of the symmetry group, and introduce a regularisation pattern that protects the symmetry. However, we defer these studies to future works.

\pagestyle{Conclusion}
\chapter*{Conclusions}
\addcontentsline{toc}{part}{Conclusion}

The study of the symmetries of minisuperspaces in gravity has been the central point of my thesis. A huge amount of work has been produced more or less recently about the effect of boundaries in gauge theories, especially gravity. The symmetry reduced models have also drawn attention, because of their simplicity and ease to be handled. In particular, they reduce the infinite-dimensional phase space into a simple mechanical model. In quantum gravity we expect them to be toy models, enlightening some properties of the full theory, but also a tool to find physically motivated solutions of the full theory, necessary to build a perturbation scheme and study the semiclassical limit.

In this thesis, we try to merge the two approaches by the study of the symmetry structure of minisuperspaces. Although this is not something completely new (\cite{Christodoulakis:1991ky, Christodoulakis:2018swq} and references therein for some generalities on the mechanical field space approach or \cite{BenAchour:2017qpb, BenAchour:2019ufa, BenAchour:2019ywl, BenAchour:2020njq} for a  more developed approach to cosmology), most of the previous works \cite{Christodoulakis:1991ky, Christodoulakis:2018swq} seem to neglect the role of the regulators and the associated boundary.

In the first part of the thesis, we have concentrated on the classical setup. We have defined the conditions for the symmetry reduction, for its dynamics to be equivalent to Einstein's equations. We have then introduced the field space approach, endowing the configuration space of evolving metric coefficient with a supermetric, we have interpreted its dynamics in terms of the geometrical properties of the supermetric. This has allowed us to reveal the existence of a rigid symmetry living on top of the residual diffeomorphism gauge invariance.

The associated conserved charges are \textit{evolving constants of motion} depending explicitly on a gauge fixed time. This makes them non-local with respect to a more general coordinate time, as they depend on the history of the system.

In an appropriate set of coordinates, for two-dimensional field spaces, these symmetries decompose into a conformal $\sl(2,\R)$ sector acting as M\"obius transformations of the gauge fixed time, a constant rescaling of the configuration fields and a 4-dimensional Abelian part generating linear time translations of the configuration variables. This build the symmetry group $(\SL(2,\R) \times \R) \ltimes \R^4$. Thanks to the Noether theorem we have found the corresponding conserved charges that allow us to (completely or partially) integrate the motion in terms of symmetries.

The existence of the rigid symmetry is related to the intrinsic properties of the evolving hypersurfaces, by the contribution of its Ricci scalar as a potential for the reduced action. Even though we were not able yet to give a general systematic criterion to directly determine the existence of (part of) the symmetry, we were able to explicitly discuss some simple, and yet physically relevant, systems.

For example, we have seen that the symmetry allows us to completely integrate the motion of the black hole minisuperspace, as well as the curved FLRW metric, coupled with a scalar field. It also encodes the full dynamics of the III, VI and VI Bianchi models, while it contains information on only half of the phase space of the Bianchi VIII and IX. This does not come as a big surprise, because of their well-known chaotic properties. We also find out that the symmetry does not exist for the Bianchi IV and VII models.  

Focusing the attention on the black hole case has allowed us to better understand the role of boundaries in the minisuperspace setup. For it, we have singled out a subgroup of symmetries that is isomorphic to the Poincar\'e group in 2+1 dimensions

In section \ref{sec2.3:BH_minisup}, we have studied in detail the Poincar\'e symmetries of the classical action. We stressed once again that these are not mere time reparametrizations leftover from diffeomorphism invariance by the fixing to homogeneous metrics. These are \textit{new symmetries} existing on top. Importantly, we have computed the action of the symmetries of the Lagrangian on the physical trajectories of the system and found that they act indeed as \textit{physical symmetries} changing the mass of the black hole.

Furthermore, the fiducial scales representing the spatial boundaries of the slice happen to be changed by the symmetry. An intriguing role of the IR cutoff $\cV_0$ has been witnessed at various stages of our study
\begin{itemize}
\item First of all it is necessary to regulate the divergent integration of the Einstein-Hilbert action. For the latter to define a good variational principle we were forced to constrain the integration on the spatial slice $\Sigma$ to a finite region of fiducial (coordinate) volume $\cV_0$.
\item Once we have gauge fixed the time, the lapse function disappears from the action and we lose the scalar constraint generating the time diffeomorphisms, the correctness of the symmetry reduction is insured if we fix the energy level of the system, it must be a function of the volume $\cV_0$ and the others fiducial scales $L_0$, $L_s$. 
\item The rigid symmetry happens to change the energy value of the trajectory, and so we need to modify accordingly the fiducial scales of the system $L_0$ and $L_s$, in order to protect the scalar constraint
\item $\cV_0$ has the role of central charge in the infinite-dimensional extension
\item It labels the quantum solution and allowing some fluctuation on it will prevent the system to withstand a singularity. Furthermore, it happens to control the regulation parameter in the polymerization scheme.
\end{itemize}

The inclusion of the fiducial scales (representing the boundary) into the degrees of freedom of the theory allows also to solve an apparent contradiction between the number of independent charges, and the results expected from the no-hair theorem, assigning only to the mass the status of physically relevant quantity. Nonetheless, the precise meaning of the physical relevance of the boundary is yet to be determined. We have seen that it happens to control the effective metric near the singularity for a particular polymerization procedure, is this property more general? Some recent work \cite{Achour:2022syr} has pointed out a possible relationship between the conformal properties of the Schwarzschild background and its static perturbation. There is maybe a more general relationship between the propagation of test fields and the rigid symmetry presented here? The intimate connection between perturbations and Hawking radiation, together with the role of boundaries in the symmetry has some interesting perspectives on the holographic properties of gravity.

We have revealed in chapter \ref{chap3} that the newly discovered rigid symmetries actually descend from an infinite-dimensional set of transformations. Always focusing on the black hole case, we have rewritten the homogeneous action \eqref{new_lapse_action} describing the black hole interior spacetime, as
\be
\cS_0
=\f{\cV_0}{G}\int \de \tau \, \f{ \dot{V}_1 ( V_2 \dot{V}_1 - 2  V_1  \dot{V}_2)}{2V_1^2 }\,,
\tag{\ref{new_lapse_action}}
\ee
where $V_1$ and $V_2$ are components of the 4-metric. We have shown in section \ref{sec3.2:bms_adjoint} how finite and infinitesimal BMS$_3$ transformations act on \eqref{new_lapse_action}. These transformations do not leave the action invariant, and only the subgroup of transformations corresponding to Poincar\'e does. Nevertheless, in the case of the supertranslations, one can interpret the non-invariance of the action as creating a term corresponding to a cosmological constant. Acting with a further supertranslation then preserves the form of the Lagrangian while however changing the value of the cosmological \textit{constant}. We have shown that even if the BMS transformations are not strictly speaking symmetries of the action, they nevertheless admit integrable generators on the phase space of the theory. Their charge algebra, given by \eqref{poisson_bms}, does however fail to reproduce the centrally-extended $\mathfrak{bms}_3$. This is to be expected since these are indeed not symmetries of the theory. We have explained how a fully BMS-invariant action \eqref{invariant_1d_action} can be obtained by promoting the BMS$_3$ group elements to dynamical variables. Aside from this 1-dimensional invariant action, there are 2-dimensional geometrical actions typically arising in studies of the boundary dynamics of 3-dimensional gravity \cite{Barnich:2012rz, Barnich:2013yka, Merbis:2019wgk}. This is for example the BMS$_3$ invariant action \eqref{BMS Liouville}. Intriguingly, the charges of this 2-dimensional action are written in terms of coadjoint vectors which turn out to be precisely the equations of motion of our starting point action \eqref{new_lapse_action}. This has motivated the study of the coadjoint representation of BMS$_3$ in section \ref{sec3.3:coadj_section}. 

We have then shown that the action \eqref{new_lapse_action} for the black hole interior can be rewritten (up to a boundary term innocent for the equations of motion) as the compact geometric action
\be
\cS=  \f{\cV_0}{G}\int \de \tau \left[\f{ V_2 \dot{V}_1^2}{2 V_1^2} -  \f{\dot{V}_1 \dot{V}_2}{V_1} + \ddot V_2 \right]
=  \f{\cV_0}{G}\int \de \tau \big[\cJ V_1 + \cP V_2 \big]\,.
\tag{\ref{BMS_geometric_action}}
\ee
This is made possible by the fact that $\cJ$ and $\cP$ (which are defined in terms of $V_1, V_2$ and their time derivatives)
are coadjoint vectors under BMS$_3$ with central charge $c_2=\cV_0/G$, while the configuration variables $V_1$ and $V_2$ are elements of the Lie algebra $\mathfrak{bms}_3$.
Remarkably, the variation of this action \eqref{BMS_geometric_action} leads to the field equations $\cJ=\cP=0$, which are exactly equivalent to the original equations of motion.
%
This property of the geometric action also enables to turn on the other central charge $c_1$ by adding a term to the Lagrangian without however modifying the equations of motion. In this construction, the Poincar\'e subgroup corresponds to the stabilizer of the vacuum orbit of BMS$_3$.

These results show that even if the action \eqref{new_lapse_action} is only invariant under Poincar\'e symmetries, there is a meaningful way in which one can understand this invariance as a broken BMS$_3$ symmetry. This confirms the intuition that this latter symmetry group does indeed plays a foundational role in the physics of the black hole interior, although it does not appear here in the more usual way as an asymptotic or horizon boundary symmetry. 

In order to investigate further the origin and the physical role of these symmetries (both the Poincar\'e and the extended BMS one), there are several directions to be developed:
\begin{itemize}
\item One should investigate which physical systems correspond to the other BMS coadjoint orbits. A related question is how the phase space structure of these coadjoint orbits can be used to define a quantization of the system. This can already be investigated in terms of Poincar\'e representations (\cite{Sartini:2021ktb} and chapter \ref{chap5}), but the embedding into BMS could allow to shift the Casimirs and describe new physical processes related to the dynamics of the black hole interior.
\item An intriguing question remains that of the relationship between the BMS group appearing in the present context and that appearing for asymptotic or near-horizon symmetries. It is particularly interesting how the Poincar\'e symmetry of \eqref{new_lapse_action} can be seen as a broken BMS symmetry as in the case of SL$(2,\mathbb{R})$ and the Schwarzian action in JT gravity.
\item A related question is that of understanding BMS-invariant actions in mechanics and field theory. It could be that \eqref{new_lapse_action} descends from a higher-dimensional BMS-invariant action with a gauge-fixing and a dimensional reduction.
\item We have seen that the BMS transformations do not leave the action invariant, but that in a sense (which is clear for supertranslations) they generate new terms which can be interpreted e.g. as a cosmological constant. Then the BMS transformations generate a flow in a space of theories and preserve the form of the action while changing its couplings. It would be very interesting to investigate the extension of Noether's theorem(s) to this type of transformation (which are very similar to renormalization flow transformations).
\end{itemize}

The existence of the \textit{hidden} rigid symmetry becomes a key ingredient to consider in any quantization approach. Taking here the conservative approach of preserving the classical structure, has provided a criterion to constrain the quantization. Concretely, this means that any Hilbert space we would like to choose, being it the standard Schr\"odinger WdW picture, or a regularized polymer space, must contain an irreducible representation of the symmetry group. In this work, I focus the attention on the black hole interior, but the construction can be generalised to any minisuperspace model that exhibits the same symmetries. 

The well known irreducible representations of $\ISO(2,1)$ are used to build a consistent quantum theory, providing an explicit example of observables and their spectra. The most striking consequence of this construction is that we obtain a mass operator with a continuous spectrum. This has important consequences on the emission spectra of black holes and is in contrast to what has been postulated in various works on black hole quantum physics. The existence of this \textit{hidden} symmetry plays also a role in perturbation theory, providing an interesting interpretation in terms of conserved quantities associated to test fields propagating on the black hole background.

On a more concrete playground, we have started with a quantization equivalent to the standard Shr\"odinger representation of Wheeler-DeWitt gravity, calculating the expected values of the metric coefficients on some semiclassical states. Classically, only a particular combination of first integrals (namely the mass) is physically relevant. On the other hand, if we allow some fluctuation on the energy, we have seen that the effective metric, emerging as a result of the quantum evolution, strongly depends on both the first integrals and the amplitude of the fluctuations. 

In the last chapter \ref{chap6}, we propose a \textit{half-polymerized} regularization, reminiscent of the modification allowed in the context of deformed constraint algebra for spherically symmetric spacetime \cite{Arruga:2019kyd}. The apparent puzzle of introducing a discretization on the configuration space, keeping the invariance under Poincar\'e group is solved here by looking at the regularised variables as a set of operators that satisfy the polymer-Weyl algebra on the same Hilbert space as the usual Shr\"odinger operator, but the two sets are not related by a unitary transformation.

We find out that the quantum corrections come from two terms going in the same direction, summing up into the parameter $\Delta$ that modifies the classical spacetime structure as in \eqref{effective_metric}. The singularity is replaced by a Killing horizon, leading to a white hole region. The two contributions have been interpreted as the effect of a quantum uncertainty on the metric coefficients, and a constant piece proportional to the Planck length, encoding the fundamental discreteness of spacetime. This implies that the light cone structure is the same for both the weakly and strongly constrained states, where the effective structure is achieved by evolving the metric coefficients with respect to the polymerized Hamiltonian \eqref{poly_hamilton}.

Despite the common feature of replacing the singularity with a black-to-white hole transition, the metric presented in this article is different to the one usually considered for the study of properties of regular black holes (\citep{Arbey:2021jif, Arbey:2021yke} and reference therein), and it would be interesting to see how this affects the phenomenology.

The extension of the \textit{hidden} Poincar\'e structure to the case with a nonvanishing cosmological constant \cite{Achour:2021dtj} and here to some of the Bianchi models open the doors towards a systematic approach to singularity issues and related regularization in many minisuperspaces.

Finally, let us point out that many times along the thesis we have spotted the necessity of generalising the study to the inhomogeneous black hole (midisuperspace). It is possible to consider for example a spherically-symmetric ansatz for the metric, and thereby reduce the problem to a two-dimensional $(r,t)$-plane gravitational theory with an inhomogeneous radial direction. Upon imposition of homogeneity, this reduces to the Kantowski--Sachs model studied here, and for which the symmetries have been identified. In this framework we lose the simplicity of dealing with a mechanical model, going back to the field theory approach, but we might expect to have a clearer setup to deal with the boundaries, the effect of the choices of the shift and hopefully enlighten the role of the $\bms$ charges, the symmetry breaking and the 2d Liouville action. The question is there that of the origin and the realization of these symmetries and the corresponding non-local charges in the inhomogeneous precursor model.\newline \

We have reached the end of my thesis, however, I realise that the three-year work I have presented here has just opened the door towards a lot of new leads: we still need a deeper understanding of the physical role of the symmetries, their relationship with other boundaries degrees of freedom and hopefully their relevance in some bulk-boundary correspondence; we could also wonder if there are experimental setup that mimics the minisuperspaces dynamics helping to shed a light onto the physical information of the charges; we can as well hope to establish a connection between the requirement of the preservation of the symmetries in the regularisation with the studies of solutions of full quantum theories, ....

\newpage
~
\thispagestyle{empty}
\newpage
~
\thispagestyle{empty}

\pagestyle{Regular}
\appendix

\chapter{Triad homogeneous decomposition of ADM Hamiltonian}
\label{app:1_ADM_triad}

We have defined the minisuperspaces as a manifold sliced in a way that the line-element separates into the temporal (orthogonal direction to the slice) and spatial (tangential directions on the slice) dependence as \eqref{minisup_metric}. This implies that the trace of the extrinsic curvature (GHY term) and the ADM kinetic term depend only on the internal metric $\gamma$, up to the determinant of the triad. This is once we have imposed the shift condition
\be
0\doteq \pounds_{\vec{\cN}} q_{ab} = N_i D_{(a} e^i_{b)}\,,\tag{\ref{shift_condition}}
\ee
In this case the extrinsic curvature \eqref{extrins_curv_tetrad} becomes
\be
K_{ab} = \f{1}{2N} \dot q_{ab}\q \Rightarrow\q \sqrt{q} \left (K^2 - K^{ab}K_{ab}\right ) = \f{ |e| \sqrt{\gamma}}{4 N^2} \left ((\gamma^{ij} \dot \gamma_{ij})^2+ \dot \gamma_{ij} \dot \gamma^{ij}\right )\,,
\ee
where $e:= \det(e_a^i)$ and $\gamma =\det(\gamma_{ij})$. The former, once integrated on a compact fiducial slice, gives the volume $\cV_0$. To analyse the three dimensional curvature it turns out to be useful to introduce the spin connection
\be
\omega^{ij}_{a} :=& e^{b i} \partial_{[a}^{\phantom{j}} e_{b]}^j -e^{b j} \partial_{[a}^{\phantom{i}} e_{b]}^i - e^{c i} e^{d j} e^k_a \partial_{[c} e_{d] k}\\ 
=&\gamma^{\ell i} (e^{b}_\ell \partial_{[a}^{\phantom{j}} e_{b]}^j )- \gamma^{\ell j} (e^{b}_\ell \partial_{[a}^{\phantom{i}} e_{b]}^i) -  \gamma_{\ell k}  \gamma^{n i}  \gamma^{m j} (e^{c}_n e^{d}_m  e^\ell_a \partial_{[c}^{\phantom k} e_{d]}^k)\notag\\
=& \gamma^{\ell [i} (e^{b}_\ell \partial_{a}^{\phantom{j}} e_{b}^{j]} )- \gamma^{\ell [i} (e^{b}_\ell \partial_{b}^{\phantom{j}} e_{a}^{j]}) -  \gamma_{\ell k}  \gamma^{n [i}  \gamma^{j] m} (e^{c}_n e^{d}_m  e^\ell_a \partial_{c}^{\phantom k} e_{d}^k)\,.
\ee
As we can see from the last line it does not simply split into the product of functions depending separately on triad and on the internal metric, the same happens for the curvature:
\be
F^{ij}_{ab} &:= 2 \left (\partial_{[a}^{\phantom j} \omega^{ij}_{b]} + \gamma_{k \ell} \omega^{i\ell}_{[a} \omega^{kj}_{b]} \right ) \,,	\\
R^{(3)} &=\f{1}{|e|} \epsilon^{abc}\epsilon_{ijk} e^k_c F^{ij}_{ab}\,.
\ee
On the other hand the vector constraint \eqref{minisup_vect_constr} depends explicitly on the spin connection:
\be
0=\cH^a &= 2 D_b \left (K q^{ab}-K^{ab} \right )\notag\\
&= \f{1}{N} D_b \left (\gamma^{k \ell} \dot \gamma_{k \ell} e^a_i\, e^b_j\, \gamma^{ij}+ e^a_i\, e^b_j\, \dot \gamma^{ij} \right )\notag	\\
&= \f{1}{N} D_b (e^a_i e^b_j) \left (\gamma^{k \ell} \dot \gamma_{k \ell} \, \gamma^{ij}+ \dot \gamma^{ij} \right )\notag\\
&=- \left (\partial_b(e^a_i\, e^b_j) + (e^{a}_k e_{\sigma \ell} \omega^{k\ell}_b + e^a_k \partial_b e^k_\sigma) e^\sigma_i\, e^b_j + (e^{b}_k e_{\sigma \ell} \omega^{k\ell}_b + e^b_k \partial_b e^k_\sigma) e^a_i\, e^\sigma_j \right ) \f{\pi^{ij}}{\cV_0 \sqrt{\gamma}}\notag\\
&=-\f{\pi^{ij}}{\cV_0 \sqrt{\gamma}} \left (\cancel{\partial_b(e^a_i\, e^b_j)} + e^{a}_k \gamma_{i\ell}\, e^b_j \omega^{k\ell}_b -  \cancel{\partial_b (e^a_i)\, e^b_j} +  e^a_i\, \gamma_{\ell j} e^{b}_k  \omega^{k\ell}_b -  \cancel{\partial_b (e^b_j)\, e^a_i} \right ) \notag\\
\Leftrightarrow &\q \pi_i^j e^{(a}_k e^{b)}_j \omega^{ik}_b =0 \,.
\ee
where the momenta are exactly the conjugate momenta to the internal metric coefficients in the ADM action \eqref{Einstein_mini}
\be
\pi^{ij}:=\f{\delta \cL_{\rm {ADM}}}{\delta \dot \gamma_{ij}} = - \cV_0 \f{ \sqrt{\gamma}}{N} \left ((\gamma^{k\ell} \dot \gamma_{k\ell})\,  \gamma^{ij} + \dot \gamma^{ij}\right ) \,.
\ee

\chapter{Metric and algebras in Bianchi models}
\label{app:2_bianchi}

The appendix gathers all the properties of the Bianchi models which are needed for the study of the phase space symmetry algebra. First, we give a list of the Bianchi line elements which satisfy the vector constraint \eqref{minisup_vect_constr}. These are
\bsub
\be
(\rm{III})&&\de s^2 &= -N^2 \de t^2 + a^2 \de x^2 + b^2 L_s^2 \left (\de y^2+\sinh^2 y\, \de \phi^2\right )\,,\\
(\rm{IV})&&\de s^2 &= -N^2 \de t^2 + a^2 L_s^2 \de x^2 + \f{a^4}{b^2}e^{-2x}\de y^2 +b^2 e^{-2x}(\de z-x\,\de y)^2\,,\\
(\rm{V})&&\de s^2 &= -N^2 \de t^2 + a^2 L_s^2 \de x^2 + \f{a^4}{b^2}e^{-2x} \de y^2 +b^2 e^{-2x} \de z^2\,,\\
(\rm{VI}_0)&&\de s^2 &= -N^2 \de t^2 + a^2 L_s^2 \de x^2 + b^2 \left (e^{-2x} \de y^2 +e^{2x}\de z^2\right )\,,\\
(\rm{VI}_h)&&\de s^2 &= -N^2 \de t^2 + a^2 L_s^2 \de x^2 + a^{\f{4h}{h-1}} b^{\f{2(1+h)}{1-h}} e^{\f{2(1-h)x}{1+h}} \de y^2 + b^2 e^{-2x}\de z^2 \,,\\
(\rm{VII})&&\de s^2 &= -N^2 \de t^2 + a^2 L_s^2 \de x^2  + b^2 e^{-2hx}(\cos x\, \de z-\sin x\, \de y)^2\cr
&&&\phantom{=\ -N^2 \de t^2 + a^2 L_s^2 \de x^2 }+\f{a^4}{b^2} e^{-2hx}(\cos x \, \de y + \sin x \,\de z)^2 \,,\\
(\rm{VIII})&&\de s^2 &= -N^2 \de t^2 + a^2 (\de x + L_s \cosh y\, \de \phi)^2  + L_s^2  b^2\, (\de y^2 +\sinh y\,  \de \phi)^2 \,,\\
(\rm{IX})&&\de s^2 &=  -N^2 \de t^2 + a^2 (\de x + L_s \cos \theta\, \de \phi)^2  + L_s^2  b^2\, (\de \theta^2 +\sin \theta\,  \de \phi)^2 \,.
\ee
\esub
The length scale $L_s$ has been introduced in order to have dimensionless fields. In terms of the decomposition \eqref{minisup_metric}, the fundamental triads are
\bsub
\be
&&&e^1=&&e^2=&&e^3=\notag\\
(\rm{III})\quad&&
&\de x &
& L_s\, \de y &
& L_s\, \sinh y\,\de z \,,\\
(\rm{IV})\quad&&
& L_s\,\de x &
& e^{-x}\, \de y &
&e^{-x}(\de z- x\,\de y)\,,\\
(\rm{V})\quad&&
&L_s\,\de x&
&e^{-x}\, \de y &
&e^{-x}\, \de z \,,\\
(\rm{VI}_0)\quad&&
&L_s\,\de x&
&e^{-x}\, \de y &
&e^{x}\, \de z \,,\\
(\rm{VI}_h)\quad&&
&L_s\,\de x&
&e^{\f{(1-h)x}{1+h}}\, \de y &
&e^{-x}\, \de z \,,\\
(\rm{VII})\quad&&
&L_s\,\de x&
&e^{-hx}(\cos x\, \de y+\sin x\, \de z) &
&e^{-hx}(\cos x\, \de z-\sin x\, \de y)\,,\\
(\rm{VIII})\quad&&
&\de x+ L_s \cosh y\, \de z&
&L_s \de y &
&L_s \sinh y\, \de \phi\,,\\
(\rm{IX})\quad&&
&\de x+ L_s \cos y\, \de z&
&L_s  \de \theta &
& L_s  \sin\theta\, \de \phi\,.
\ee
\esub
One should note that for each triad the line elements given above are not the only solutions to the vector constraint. The finite volumes of the fiducial cells are $\cV_0 = \f{1}{16 \pi}\int_\Sigma |e|$ and given by
\bsub
\be
(\text{III, VIII})&& \cV_0=&\f{1}{4}L_x L_s^2\sinh^2 \left (\f{y_0}{2}\right) &&x\in[0,L_x]\,,y\in[0,y_0]\,, \phi \in[0,2\pi]\,, \\
(\text{IV, V})&& \cV_0=&\f{1}{16 \pi}L_s L_y L_z e^{-x_0}\sinh x_0 &&x\in[0,x_0]\,,y\in[0,L_y]\,, z \in[0,L_z] \,,\\
(\rm{VI}_0)&& \cV_0=&\f{1}{16 \pi}L_s L_y L_z x_0  &&x\in[0,x_0]\,,y\in[0,L_y]\,, z \in[0,L_z] \,,\\
(\rm{VI}_h)&& \cV_0=&\f{1+h}{32 \pi h}L_s L_y L_z (1-e^{-2\f{h x_0}{1+h}})  &&x\in[0,x_0]\,,y\in[0,L_y]\,, z \in[0,L_z] \,,\\
(\rm{VII})&& \cV_0=&\f{1}{32 \pi h}L_s L_y L_z (1-e^{-2h x_0}) &&x\in[0,x_0]\,,y\in[0,L_y]\,, z \in[0,L_z]\,, \\
(\rm{IX})&& \cV_0=&\f{1}{4}L_x L_s^2  &&x\in[0,L_0]\,,\theta\in[0,\pi]\,, \phi \in[0,2\pi]\,. 
\ee
\esub
Note that there are several fiducial scales entering the equations. $L_s$ is used to ensure that the dynamical fields are dimensionless, while $L_x,L_y,\dots$ are dimensionful cut-offs (fiducial lengths) on the variables $x,y,\dots$ and are equivalent to $L_0$ for the black hole example, while $x_0,y_0,\dots$ are dimensionless cut-offs. The metrics are divided in three categories depending on the internal metric
\bsub
\be
\gamma_{ij} &= \rm{diag}\left(a^2,b^2,b^2\right) &&\; \rm{III}\,,\; \rm{VI}_{0}\,,\; \rm{VIII}\,,\; \rm{IX}\,.\\
\gamma_{ij} &= \rm{diag}\left (a^2,\f{a^4}{b^2},b^2\right) &&\; \rm{IV}\,,\; \rm{V}\,,\; \rm{VII}\,.\\
\gamma_{ij} &= \rm{diag}\left (a^2, b^{\f{2(1+h)}{1 - h}} a^{\f{4 h}{h -1}},b^2\right) &&\; \rm{VI}_h\,.
\ee
\esub
To the first class we can also add the Kantowski sachs cosmology
\bsub
\bg
\de s^2_{\rm{KS}} = -N^2 \de t^2 + a^2 \de x^2 + b^2 L_s^2 \left (\de \theta^2+\sin^2 \theta\, \de \phi^2\right)\,,\\
e^1 = \de x\,,\q e^2 = L_s\, \de \theta \,,\q e^3 = L_s\, \sin \theta \, \de \phi \,,\\
\cV_0 = \f{L_0 L_s^2}{4} \,,\q \q x\in[0,L_0]\,,\theta\in[0,\pi]\,, \phi \in[0,2\pi].
\eg
\esub
The latter does not belong to the Bianchi classification because it doesn't have three independent spacelike killing vector forming a closed Lie algebra. We write here the potential terms coming from the minisuperspace reduction of the three dimensional Ricci scalar.
\be
\int_\Sigma \f{|e| \sqrt{\gamma} R^{(3)}}{16 \pi G} = \f{1}{G}N \cV_0\Omega_\text{model}
\ee
with
\bsub
\be
\Omega_\text{KS}&=-2\df{a}{L_s^2}=-\df{3^{1/3 }}{L_s^2 \sigma^{4/3}}\df{\tilde u}{\tilde{v}^{1/3}}\,,\\
\Omega_\text{III}&= 2\df{a}{L_s^2}=\df{3^{1/3 }}{L_s^2 \sigma^{4/3}}\df{\tilde u}{\tilde{v}^{1/3}}\,,\\
\Omega_\text{IV}&=\df{6a}{L_s^2}+ \df{b^4}{2 L_s^2 a^3 }=\f{1}{4 L_s^2}(3\tilde u \tilde v)^{1/3}\left (12+16^{\sqrt{3}} 81^{1/\sqrt{3}} (\tilde{v} \sigma^2 /\tilde u)^{4/\sqrt{3}}\right )\,,\\
\Omega_\text{V}&=\df{6a}{L_s^2}=\df{3(3\tilde u \tilde v)^{1/3}}{L_0^2}\,,\\
\Omega_\text{VI$_0$}&= 2\df{b^2}{L_s^2 a}=\df{(3 \sigma^8 \tilde v^5)^{1/3}}{\tilde u L_s^2}\,,\\
\Omega_\text{VI$_h$}&= 2\df{(1+3 h^2)}{1+h^2}\df{b^{\f{2}{1-h}}a^{\f{1+h}{h-1}}}{L_s^2}
=\left(\df{3G}{\cV_0}\right )^{\f{1}{3}}\df{1+3 h^2}{(1+h)^2 L_s^2} \tilde u^{\f{4}{3\sqrt{3 h^2+1}}+\f{1}{3}} \tilde v^{-\f{4}{3\sqrt{3 h^2+1}}+\f{1}{3}} \,,\\
\Omega_\text{VII$_h$}
&= \df{a^8+2(6h^2-1) a^4 b^4 +b^8}{2 L_0^2 a^3b^4}\\
&=\left (\df{3G}{\cV_0}\right )^{1/3} \df{\left (\left ( \left (\df{\tilde u}{\tilde v}\right )^{\f{4}{\sqrt{3}}}-1  \right )^2 +12 h^2 \left (\df{\tilde u}{\tilde v}\right )^{\f{4}{\sqrt{3}}} \right )\tilde  u^{-\f{4}{\sqrt{3}}+\f{1}{3}}\tilde  v^{\f{4}{\sqrt{3}}+\f{1}{3}}} {4 L_s^2}\,,\q\notag\\
\Omega_\text{VIII}&=\df{a^3+4a b^2 }{2L_s^2 b^2}=\df{3^{1/3}\left (\tilde u^3 + 4 \tilde u \tilde v^2 \sigma^4 \right )}{4 L_s^2 \tilde v^{7/3} \sigma^{16/3}}\,,\\
\Omega_\text{IX}&=\df{a^3-4a b^2}{2L_s^2 b^2}=\df{3^{1/3}\left (\tilde u^3 - 4 \tilde u \tilde v^2 \sigma^4 \right )}{4 L_s^2 \tilde v^{7/3} \sigma^{16/3}}	\,.
\ee
\esub
Note that the transformation that maps the \textit{scale factors} to the conformal null coordinates on the field space changes between the three families with different internal metric, and is given by
\be
\tag{\ref{family_1_conf}}
&\left|
\begin{array}{rl}
\tilde u &= 2  \sqrt{\df {2\cV_0}{3G}}\, a \sqrt{b} \\
\tilde v &= 2 \sqrt{\df{2\cV_0}{3G}}\, b^{3/2} \\
\end{array} \right.\,, &&\;\rm{III}\,,\; \rm{VI}\,,\; \rm{VIII}\,,\; \rm{IX}\,,\\
\tag{\ref{family_2_conf}}
&\left|
\begin{array}{rl}
\tilde u &= 2 \sqrt{\df{2\cV_0}{3G}} a^{\f{3+\sqrt 3}{2}} b^{-\f{\sqrt{3}}{2}} \\
\tilde v &= 2 \sqrt{\df{2\cV_0}{3G}} a^{\f{3-\sqrt 3}{2}} b^{\f{\sqrt{3}}{2}}\\
\end{array}
\right. \,,&&\; \rm{IV}\,,\; \rm{V}\,,\; \rm{VII}\,,\\
\tag{\ref{family_3_conf}}
&\left|
\begin{array}{rl}
\tilde u &= 2 \sqrt{\df{2\cV_0}{3G}} a^\f{1 - 3 h- \sqrt{1 + 3 h^2}}{2 -2 h}  b^\f{2 + \sqrt{1 + 3 h^2}}{2 -2 h}\\
\tilde v &= 2 \sqrt{\df{2\cV_0}{3G}} a^\f{1 - 3 h+ \sqrt{1 + 3 h^2}}{2 -2 h}  b^\f{2 - \sqrt{1 + 3 h^2}}{2 -2 h} \\
\end{array} \right.\,,
&&\; \rm{VI}_h\,.
\ee
We see that for Bianchi III and the black hole we can choose the same lapse, that gives an opposite sign to the constant potential, changing the sign of the "mass of the particle moving on the field space" (see main text for the algebra). We can also put the potential to a constant with a flat conformal factor in the V and VI cases. This will respectively give the null coordinates:
\bsub \be
N&= \f{1}{2  a}\,, && \begin{array}{rlcrl}
u&= 2 \sigma a^2 b	\,,		&\q\q & p_u &=-\df{p_a}{4 a b \sigma}\,,\\[9pt]
v&= 2 \sigma b\,,		&\q\q & p_v &=  \df{a p_a + 2 b p_b}{4b \sigma}\,,\\
\end{array} \q\q\rm{III}\,,\; \rm {KS}&&\\
N&= \f{1}{6  a}\,,&& \begin{array}{rlcrl}
u&= 3 \sigma a^{2+\f{2}{\sqrt{3}}}	b^{-\f{2}{\sqrt{3}}}	\,,		
							&\q & p_u &=\left (\df a b\right )^{-2/\sqrt{3}}\df{a p_a +(1-\sqrt{3}) b p_b)}{12 a^2 \sigma}\,,\\[9pt]
v&= 3 \sigma a^{2-\f{2}{\sqrt{3}}}	b^{\f{2}{\sqrt{3}}}\,,
							&\q & p_v &=  \left (\df a b \right )^{2/\sqrt{3}}\df{a p_a +(1+\sqrt{3}) b p_b}{12 a^2\sigma }\,,\\
\end{array}&&\rm{V}\\
N&= \f{a}{2  b^2}\,, && \begin{array}{rlcrl}
u&= 2\sigma \log a^2 b	\,,	&\q\q & p_u &=\df{a p_a}{2\sigma}\,,\\[9pt]
v&= {\sigma b^4}\,,	&\q\q & p_v &=  \df{-a p_a + 2 b p_b}{8b^4\sigma}\,,\\
\end{array} &&\rm{VI}_0\,,
\ee
where $\sigma= \sqrt{\df{\cV_0}{G}}$, finally for the VI$_h$ case
\be
N&= \f{(1+3h^2)}{2h^2 }b^{\f{2}{h-1}}a^{\f{1+h}{1-h}}\,,\notag \\
& \begin{array}{rl}
u&=\sigma a^{\f{2h (1+h +\sqrt{1+3 h^2})}{(h-1)\sqrt{1+3 h^2}}} b^{\f{2 (1+h^2+\sqrt{1+3 h^2})}{(1-h)\sqrt{1+3 h^2}}}\,,		\\[9pt]
p_u &= a^{\f{2h (1+h +\sqrt{1+3 h^2})}{(1-h)\sqrt{1+3 h^2}}} b^{\f{2 (1+h^2+\sqrt{1+3 h^2})}{(h-1)\sqrt{1+3 h^2}}}\df{b p_b h (1+h -\sqrt{1+3 h^2}) +a p_a (m^2+1 - \sqrt{1+3 h^2})}{4\sigma h^2}\,, \\[9pt]
v&=\sigma a^{\f{2h (h-1 +\sqrt{1+3 h^2})}{(h-1)\sqrt{1+3 h^2}}} b^{\f{2 (1+h^2-\sqrt{1+3 h^2})}{(h-1)\sqrt{1+3 h^2}}}\,,	\\[9pt] 
p_v &=  a^{\f{2h (1+h +\sqrt{1+3 h^2})}{(h-1)\sqrt{1+3 h^2}}} b^{\f{2 (1+h^2+\sqrt{1+3 h^2})}{(1-h)\sqrt{1+3 h^2}}}\df{b p_b h (1+h +\sqrt{1+3 h^2}) +a p_a (m^2+1 + \sqrt{1+3 h^2})}{4\sigma h^2}\,, \\
\end{array} &&\rm{VI}_h\,.
\ee \esub
We just need to preform the change of coordinate above into the expression of the generators, to get the algebra in terms of scale factors and their momenta. This is not possible for the other models, on the other hand, looking at the potentials for VIII and IX, we see that the two monomials are the same with a sign difference, and they satisfy the condition \eqref{power_condition}. We have than two choices for the lapse that gives the $\SL(2,\R)$ algebra
\be
N&= \f{1}{2 a} \,,&& \left |\begin{array}{rl}
u&=2 \sigma a^2 b\\
v&=2\sigma b \\
\end{array} \right .\,,
&C_\rm{d}&= \f{b b_p}{2}\,,
&V_3&= 8 \sigma^2 a^2 b^2 \,,
&\cQ_0&=\f{p_a (a p_a -2 b p_b)}{16 \sigma a b^2}+ \f{\cV_0^2 a^2}{4  L_s^2 G^2 b^2}\,.\\
N&= - \f{2 b^3}{ a^3}\,,&& \left|\begin{array}{rl}
u&= \df{\sigma}{2} a^4 b^2\\[9pt]
v&= \df{\sigma}{2  b^2}\\
\end{array} \right .\,,
&C_\rm{d}&=\f{b b_p}{2}\,,
&V_3&=\sigma^2 \f{a^4}{2}\,,
&\cQ_0&= -\f{p_a (a p_a -2 b p_b)}{4 a^3}\mp \f{4 \cV_0^2 b^2}{  L_s^2 G^2 a^2}\,.
\ee
where the upper sign is for VIII, and the lower one is for IX. Unfortunately \eqref{power_condition} is not satisfied for IV and VII, meaning that they do not exhibit the $\SL(2,\R)$ symmetry, because of the non-conformal potential.

\chapter{Connection triad variables}
\label{app:4_LQG_phase_space}

Early work on the effective dynamics of the Kantowski--Sachs black hole interior used a slightly different language based on the connection-triad variables of loop quantum gravity \cite{Bojowald:2001xe,Bohmer:2007wi}. In this appendix write down the relationship between these and our variables, and also recall the structure of the CVH algebra of full LQG.

The line element used in e.g. \cite{Bohmer:2007wi} is
\be
\de s^2 = - N^2 \de t^2 + \f{p_b^2}{L_0^2 p_c} \de x^2 + p_c \de \Omega^2\,.
\ee
This corresponds to the choice of a densitized triad
\be\label{KS E}
E=E^a_i\tau^i\partial_a=p_c\sin\theta\,\tau_1\partial_x+\f{p_b}{L_0}\sin\theta\,\tau_2\partial_\theta+\f{p_b}{L_0}\tau_3\partial_\phi\,.
\ee
In LQG this densitized triad is canonically conjugated to the $\text{SU}(2)$ Ashtekar--Barbero connection, and the symplectic structure is
\be
\{A^i_a(x),E_j^b(y)\}= 8 \pi\gamma \delta_a^b \delta_j^i \delta(x-y)\,.
\ee
In the case of the homogeneous spherically-symmetric BH interior the connection is
\be\label{KS A}
A=A^i_a\tau_i\de x^a=\f{c}{L_0}\tau_1\de x+b\tau_2\de\theta+b\sin\theta\,\tau_3\de\phi+\cos\theta\,\tau_1\de\phi\,,
\ee
and the Poisson brackets reduce to
\be
\{c,p_c\}= 2 \gamma\,,\q\q\{b,p_b\}= \gamma\,.
\ee
These phase space variables are related to the $\bms_3$ vectors by
\be
\label{transf_bc_to_12}
 V_1 = L_s^2 p_c \,,\q V_2 = 2L_s^2 \f{p_b^2}{L_0^2}\,,\q P_1=-\f{1}{2L_s^2\gamma} c\,, \q P_2 = -\f{L_0^2}{4 L_s^2\gamma} \f{b}{p_b}\,.
\ee

We now close this appendix by commenting on the CVH algebra structure of full general relativity in Ashtekar--Barbero canonical variables. This follows the appendix of \cite{BenAchour:2017qpb} and corrects a minor point there. We then connect the results to the spherically-symmetric homogeneous spacetime considered here. The connection-triad variables are built from the spatial triad $e^i_a\de x^a$ and the extrinsic curvature $K^i_a \de x^a$, where $i,j,k$ are internal $\mathfrak{su}(2)$ indices, as
\be
E^a_i=\det(e^i_a) e^a_i\,, \q\q A^i_a=\Gamma_a^i[E] + \gamma K^i_a\,,
\ee
with the torsion-free spin connection $\Gamma^i_a[E]$ is
\be
\Gamma_a^i= \f{1}{2} \epsilon^{ijk} E^b_k \left(2 {\partial}_{[b} {E^j}_{a]} + E^c_j E^l_a \partial_b E^l_c\right) +\f{1}{4}\f{ \epsilon^{ijk} E^b_k}{\det E} \big(2E^j_a \partial_b \det E-E^j_b \partial_a \det E\big)\,.
\ee
The canonical pairs $(A_a^i,E^i_a)$ are subject to seven first class constraints given by Gauss, diffeomorphism and scalar constraints. Here we are interested in the scalar constraint
\be
H[N]=\int_\Sigma \de^3 x\,\cH=\f{1}{16\pi}\int_\Sigma \de^3 x \f{N}{\sqrt{q}}E^a_iE^b_j\Big({\eps^{ij}}_kF^k_{ab}-2(1+\gamma^2)K^i_{[a}K^j_{b]}\Big)\,.
\ee
It is useful to split the constraint into the so-called Euclidean $H_\text E$ and Lorentzian $H_\text K$ parts given by
\bsub
\be
H_{\text E}[N]&=\int_\Sigma \de^3 x\,N \cH_\text{E} = \f{1}{16\pi}\int_\Sigma \de^3 x N\f{E^a_iE^b_j}{\sqrt{\det E}}\big(\eps^{ij}_{\;\;k}F^k_{ab}\big)\,,
\\
H_{\text K}[N]&=\int_\Sigma \de^3 x\,N \cH_\text{K} =-\f{1+\gamma^2}{8\pi}\int_\Sigma \de^3 x N \f{E^a_iE^b_j}{\sqrt{\det E}}K^i_{[a}K^j_{b]}\,.
\ee
\esub
The generator of dilation in the phase space is here the trace of the extrinsic curvature, also known as the complexifier (hence its name $C$ in the main text) as it plays a primary role in defining the Wick transform between real and self-dual version of LQG. The last quantity to consider is the volume of the space-like hypersurface. These quantities are given by
\be
C=\f{1}{8\pi }\int_\Sigma \de^3 x \,E^a_i K^i_a\,, \q\q
V=\int_\Sigma \de^3 x \sqrt{\det(E^a_i)}\,.
\ee
A straightforward calculation shows that, together with the Lorentzian part for a constant lapse function, they form an $\mathfrak{sl}(2,\R)$ CVH algebra, that is the LQG analogous of the kinematical ADM algebra presented in the main text. For $N=1$ we have indeed 
\be
\label{HKCV}
\{ C,V \} =	\f{3}{2 } V\,, \q \q 
\{ V, H_{\text K} \}= (1+\gamma^2){8\pi} C\,,  \q \q
\{ C, H_{\text K} \}= -\f{3}{2} H_{\text K} \,.
\ee
On the other hand, if we also consider the Euclidean part of the Hamiltonian constraint the algebra fails to be closed and we find for $N=1$
\bg
\label{HECV}
\{ V, H_{\text E} \}= -{8 \pi \gamma^2} C \,,\q \q
\{ C, H_{\text E} \}= \f{1}{2} H_{\text E} + 2\f{\gamma^2}{1+\gamma^2} H_{\text K} \,,\\
\{ H_{\text E} , H_{\text K} \} = \f{1+\gamma^2}{16\pi}\int_\Sigma \de^3 x	\,\left [ \epsilon^{abc} F_{bc}^i K_a^i - 6 \gamma^2 \det (K^i_a) \right ].
\notag
\eg
Note the particular role played by the self-dual value $\gamma=\pm i$.

In flat FLRW cosmology, it turns out that the Euclidean and Lorentzian parts of the Hamiltonian constraint are proportional, as a consequence of the vanishing of $\Gamma$, that is indeed synonymous of a flat three dimensional slice. This is however not the case for the Kantowski--Sachs geometry, where with the variables introduced above and for $N=1$ we get
\be
\label{cb_hamilt}
H_{\text E}=\f{2bcp_c+(b^2-1)p_b}{2\sqrt{p_c}}\,,
\q\q
H_{\text K}=-(1+\gamma^{2})\f{ 2bcp_c+b^2p_b}{2\sqrt{p_c}}\,,
\ee
while the complexifier and the volume are given by
\be
\label{bc_cvh}
C=\f{2 b p_b +c p_c}{2\gamma}\,,\q\q
V= 4 \pi p_b \sqrt{p_c}\,.
\ee
The algebra is then obviously given again by \eqref{HKCV} and \eqref{HECV}, with the last bracket changed to
\be
\{ H_{\text E} , H_{\text K} \} = -\f{1 + \gamma^2}{2 \gamma} c \,.
\label{CVH_reduced}
\ee
The key idea which leads to the CVH algebra for the black hole interior is to change the lapse so as to recover the same property as in FLRW, namely a simple phase space independent relationship between $H_\text{E}$ and $H_\text{K}$. This choice corresponds to
\be
N=\f{2\sqrt{p_c}}{p_b} \q\Rightarrow\q H_{\text K}[N] = -(1+\gamma^2) (H_{\text E }[N]+1)\,.
\ee
Using \eqref{transf_bc_to_12} shows that this lapse is indeed (up to a prefactor depending on the fiducial lengths) the one used in section \ref{sec2.3:BH_minisup}. Reabsorbing the lapse and redefining the volume and the complexifier, we see that the modified CVH algebra gives indeed the $\mathfrak{sl}(2,\R)$ sector of the $\mathfrak{iso}(2,1)$ structure presented in the main text, with
\bsub
\ba
\f{1}{\gamma^2(1+\gamma^2)} H_{\text K}[N] &=&  -\f{L_s^6}{2\cV_0^2}(2 P_1 V_1 + P_2 V_2) P_2 = \f{L_s^6}{\cV_0 G}H \,, \\
V &\to& \f{V}{ 16 \pi N} = \f{p_b^2}{8} \propto V_2\,, \\
\phantom{\f{1}{2}}C &\to& C_\rm{d} = \left \{V_2,\f{\cV_0}{G}H\right \} = - V_1 P_1 - V_2 P_2\,.
\ea
\esub

Finally, we use the relationship \eqref{transf_bc_to_12} to rewrite the polymer Hamiltonian \eqref{poly_hamilton} in terms of connection and triads. This reads
\be
H_\rm{poly}= -\f{c p_c }{4\gamma \lambda} \sin \left (\f{L_0^2 \lambda b}{2\gamma L_s^2 p_b}\right ) -  \f{L_s^2 p_b^2}{\lambda^2 L_0^2} \sin^2 \left (\f{L_0^2\lambda b}{2\gamma L_s^2 p_b}\right )\,.
\ee
\chapter{Properties of homothetic killing vectors}
\label{app:5_HKV_prop}

In the main text of the article we have used some properties of the conformal Killing vectors, for example the fact that they are solution to the geodesic deviation equation. We will prove here this statement and other properties of the conformal vector. We recall that given an invertible metric $g_{\mu\nu}$, we define an homothetic Killing vector by the property
\be
\label{HKV_append} \nabla_\mu \xi_\nu + \nabla_\nu \xi_\mu : =  2  \nabla_{(\mu} \xi_{\nu)} =\lambda g_{\mu\nu}\,,\q\q \lambda=\rm{const}\,.
\ee

\begin{theorem}
Any homotetic killing vector $\xi$ is a solution of the geodesic deviation equation \cite{Caviglia:1982}:
\be
\label{geodesic deviation}
p^\mu p^\nu \nabla_\mu \nabla_\nu \xi_\rho = - R_{\rho\mu\sigma\nu}p^\mu p^\nu \xi^\sigma\,.
\ee
Where $p^\mu$ is the tangential vector to a geodesic (i.e. a curve describing a solution of the equations of motion). It satisfy the property: $p^\mu \nabla_\mu p^\nu =0$.
\end{theorem}
\begin{proof}
We use the definition of HKV \eqref{HKV_append} on the l.h.s.
\be
p^\mu p^\nu \nabla_\mu \nabla_\nu \xi_\rho 
&= - p^\mu p^\nu \nabla_\mu \nabla_\rho \xi_\nu + \lambda p^\mu p^\nu \cancel{\nabla_\mu g_{\nu \rho}}\notag\\
&= - p^\mu p^\nu  R_{\nu\sigma\mu\rho} \xi^\sigma +  p^\mu p^\nu \nabla_\rho \nabla_\mu \xi_\nu\notag\,,
\ee
From the first line to second one we use the definition of Riemann tensor. The last term is zero:
\be
p^\mu p^\nu \nabla_\rho \nabla_\mu \xi_\nu = p^\mu p^\nu \nabla_\rho \nabla_{(\mu} \xi_{\nu)} = \lambda p^\mu p^\nu \nabla_\rho g_{\mu\nu }=0\notag\,,
\ee
so we prove the statement, using the properties of $R$:
\be
p^\mu p^\nu \nabla_\mu \nabla_\nu \xi_\rho  = - p^\mu p^\nu  R_{\nu\sigma\mu\rho} \xi^\sigma=
- R_{\rho\mu\sigma\nu}p^\mu p^\nu \xi^\sigma\,.\notag
\ee
\end{proof}
The homothetic Killing vectors form an algebra:
\be
\label{structure}
[\xi_{(i)},\xi_{(j)}] &= c_{ij}^{\;\;k} \xi_{(k)}\,,\\
[\xi_{(i)},\xi_{(j)}]^\nu &:= \xi_{(i)}^\mu \nabla_\mu \xi_{(j)}^\nu -\xi_{(j)}^\mu \nabla_\mu \xi_{(i)}^\nu\,,\notag
\ee
and the derived subalgebra is given by Killing vectors
\be
\cL_{[i,j]} g_{\mu\nu} &= \nabla_ \mu (\xi_{(i)}^\sigma \nabla_\sigma \xi_{(j)\,\nu} -\xi_{(j)}^\sigma \nabla_\sigma \xi_{(i)\,\nu}) + \left (\mu \leftrightarrow \nu \right )\notag\\
&= \nabla_ \mu \xi_{(i)}^\sigma \nabla_\sigma \xi_{(j)\,\nu} + \xi_{(i)}^\sigma \nabla_ \mu \nabla_\sigma \xi_{(j)\,\nu} -\nabla_ \mu \xi_{(j)}^\sigma \nabla_\sigma \xi_{(i)\,\nu} + \xi_{(j)}^\sigma \nabla_ \mu \nabla_\sigma \xi_{(i)\,\nu} + \left (\mu \leftrightarrow \nu \right )\notag\\
&= \nabla_ \mu \xi_{(i)\,\nu} \lambda_{(j)} - \cancel{\nabla_ \mu \xi_{(i)}^\sigma \nabla_\nu \xi_{(j)\,\sigma}}  + \xi_{(i)}^\sigma \nabla_ \mu \nabla_\sigma \xi_{(j)\,\nu} -\nabla_ \mu \xi_{(j)\,\nu} \lambda_{(i)}\cr 
&\;\;+ \cancel{\nabla_ \mu \xi_{(j)}^\sigma \nabla_\nu \xi_{(i)\,\sigma}} + \xi_{(j)}^\sigma \nabla_ \mu \nabla_\sigma \xi_{(i)\,\nu} + \left (\mu \leftrightarrow \nu \right )\notag\\
&= \cancel{g_{\mu\nu}\lambda_{(i)} \lambda_{(j)}} + \xi_{(i)}^\sigma R_{\nu\rho\mu\sigma}\xi_{(j)}^\rho - \xi_{(i)}^\sigma \nabla_ \sigma \nabla_\mu \xi_{(j)\,\nu} -\cancel{g_{\mu\nu}\lambda_{(j)}\lambda_{(i)}} \cr
&\;\;- \xi_{(j)}^\sigma R_{\nu\rho\mu\sigma}\xi_{(i)}^\rho + \xi_{(j)}^\sigma \nabla_ \sigma \nabla_\mu \xi_{(i)\,\nu}+ \left (\mu \leftrightarrow \nu \right )\notag\\
&= \xi_{(i)}^\rho \xi_{(j)}^\sigma (R_{\nu\mu\sigma\rho} + R_{\mu\nu\sigma\rho})  -\xi_{(i)}^\sigma \nabla_\sigma (g_{\mu\nu}) \lambda_{(j)} + \xi_{(j)}^\sigma \nabla_\sigma (g_{\mu\nu}) \lambda_{(i)}\cr
& =0\,.
\ee
Here we have used \eqref{HKV_append} when going from the second to the third line, and eliminated the antisymmetric terms in $\mu\,,\nu$. We have then used \eqref{HKV_append} once again as well as the definition of the Riemannn tensor, and finally concluded by using the antisymmetry of the Riemann tensor. This result implies that the structure constants of the algebra of homothetic Killing vectors satisfy
\be
c_{ij}^{\;\;k} \lambda_k =0\,.
\ee
Concerning the quadratic charges, we have that their third time derivative vanish if the field space is equivalent to Minkowski (i.e. if the Riemann tensor vanishes)
\be
\f{\de^3}{\de t^3}V_{ij}&=
p^\mu p^\nu p^\rho \nabla_\mu \nabla_\nu \nabla_\rho \left (\xi_{(i)}^\sigma \xi_{(j)\sigma}\right )\notag\\&=2p^\mu p^\nu p^\rho \nabla_\mu\left ( (\nabla_\nu \xi_{(i)}^\sigma) (\nabla_\rho \xi_{(j)\sigma}) -  \xi_{(i)}^\sigma R_{\sigma\nu\kappa\rho} \xi_{(j)}^\kappa) \right ) \notag\\
&= - 2 p^\mu p^\nu p^\rho \left (2(\nabla_\rho \xi_{(j)}^\sigma) R_{\sigma\mu\kappa\nu} \xi_{(i)}^\kappa +2(\nabla_\rho \xi_{(i)}^\sigma) R_{\sigma\mu\kappa\nu} \xi_{(j)}^\kappa  + \xi_{(i)}^\sigma\xi_{(j)}^\kappa (\nabla_\rho R_{\sigma\mu\kappa\nu})\right)\,.\ee
This shows in particular that $\dddot{V}_{(ij)}=0$ whenever the Riemann tensor vanishes (this is of course sufficient but not necessary), as in the case of the flat field space geometry discussed in section \ref{sec2.2:mink&minisup}.

We are interested in the condition for which the functions $V_{(ij)}$ and $C_{(i)}$ form a closed algebra with the Hamiltonian $H$. In the free case where $H=p_\mu p^\mu/2$ we find
\be
\lb V_{(ij)},H\rb
&=p^\mu \partial_\mu \big( \xi_{(i)}^\nu \xi_{(j)\,\nu} \big)\,\notag\\
&=p^\mu \xi_{(i)}^\nu \nabla_\mu   \xi_{(j)\,\nu} +p^\mu \xi_{(j)}^\nu \nabla_\mu   \xi_{(i)\,\nu}\notag\\
&= \lambda_{(i)} C_{(j)}+\lambda_{(j)} C_{(i)}- p^\mu\left ( \xi_{(i)}^\nu \nabla_\nu   \xi_{(j)\,\mu} +\xi_{(j)}^\nu \nabla_\nu   \xi_{(i)\,\mu}\right ) \,,
\ee
and
\be
\lb V_{(ij)}, C_{(k)}\rb
&=\xi_{(k)}^\mu \partial_\mu  \big( \xi_{(i)}^\nu \xi_{(j)\,\nu} \big)\,\notag\\
&=\xi_{(k)}^\mu \xi_{(i)}^\nu \nabla_\mu   \xi_{(j)\,\nu} +\xi_{(k)}^\mu \xi_{(j)}^\nu \nabla_\mu   \xi_{(i)\,\nu}\notag\\
&= \lambda_{(i)} V_{(jk)}+\lambda_{(j)} V_{(ik)}- \xi_{(k)}^\mu\left ( \xi_{(i)}^\nu \nabla_\nu   \xi_{(j)\,\mu} +\xi_{(j)}^\nu \nabla_\nu   \xi_{(i)\,\mu}\right ) \,,
\ee
which closes iff
\be\label{condition on xi's}
\xi_{(i)}^\nu \nabla_\nu   \xi_{(j)\,\mu} +\xi_{(j)}^\nu \nabla_\nu   \xi_{(i)\,\mu}=\sum_i \alpha^i\xi_{(i)\,\mu}\,,
\ee
for some combination of the vectors on the RHS.
\chapter{List of variables and operators for the BH minisuperspace}
\label{app:6_Variables}

In this appendix we gather some of the notations which are used throughout the paper to denote certain variables. In particular, the Poincar\'e generators have appeared in many different forms depending on the context. We give here these different forms as well as the equations defining them. 

We start by the line element that defines the minisuperspace. This has been written in terms of three different set of variables according to the parametrization of the field space that we choose. First of all there is the \textit{scale factor} choice that gives
\be
\de s^2_\rm{BH} = -N^2 \de t^2 + a^2 (\de x + N^1 \de t)^2 +L_s^2 b^2\, \de \Omega^2\,,
\tag{\ref{Bh_minisup_line}}
\ee
then the conformal null parametrization
\be
\de s^2_\rm{BH} = -N^2 \de t^2 + \f{u}{v} (\de x + N^1 \de t)^2 + \f{G L_s^2}{4\cV_0} v^2 \, \de \Omega^2\,,
\tag{\ref{Bh_minisup_line:null_conf}}
\ee
and finally the $\BMS_3$ vectors $V_i$
\be
\de s^2_\rm{BH} = -N^2 \de t^2 + \f{V_2}{2V_1} (\de x + N^1 \de t)^2 + L_s^2 V_1 \, \de \Omega^2\,,
\tag{\ref{Bh_minisup_line:Vi}}
\ee
The three are related via
\be
\begin{array}{rcl}
V_1 =& b^2 &= \df{G v^2}{4 \cV_0} \,,\\[9pt]
V_2 =& 2a^2 b^2 &= \df{G uv}{2 \cV_0}\,.
\end{array}
\tag{\ref{new_fields_v}}
\ee
The mechanical model is in priciple invariant under the change of coordinates \eqref{new_fields_v}, but it might happen that different choices of field space parametrization do not cover the whole evolution, exactly like in GR, where a choice of coordinate can be restricted to a patch of the whole manifold. However there is in principle no need to restrict the field to some range of values, it is the evolution that will tell us how the trajectories span the field space.

We have also chosen to gauge fix the lapse in order to have a constant potential term, this is given by
\be
N^2=\f{1}{4a^2} = \f{v}{4u} = \f{V_1}{2 V_2}\,, 
\ee
that in turn make us picking the supermetric
\be
g_{\mu\nu} \dot  x^\mu \dot  x^\nu = \left \{
\begin{array}{l}
-\df{\cV_0}{G}\left (4 a^2\, \dot b^2 +8 b a\, \dot a\, \dot b\right )\\[9pt]
- 2 \dot u \dot v\\
\df{\cV_0}{G} \left (\df{\dot V_1 (V_2 \dot V_1 - 2 V_1 \dot V_2)}{V_1^2} \right )\,	\\
\end{array}
\right. \,.
\ee
From \ref{new_fields_v} we infer the relationship between the canonical momenta
\be
\begin{array}{rcl}
P_1 =& \df{b p_b- a p_a}{2 b^2} &= \df{2\cV_0}{G}\df{v p_v- u p_u}{v^2}\,,\\[9pt]
P_2 =& \df{p_a}{4 a b^2} &= \df{2 \cV_0}{G}\df{p_u}{v}\,.
\end{array}
\tag{\ref{new_fields_v}}
\ee
Because of the dimension of the phase space, the dynamics need four integrals of motion in order to be solved, these are provided by
\be
\cA &= \f{p_a^2}{32 a^2 b^2} = \f{\cV_0}{2G} p_u^2 =\f{P_2^2 V_1}{2}\\
\cB &= \f{b p_b -a p_a}{2} = \f{v p_v -u p_u}{2} = P_1 V_1\\
\tau_0 &= \tau + \f{8 a^2 b \cV_0}{p_a G}=   \tau +\f{v}{p_u}= \tau +\f{2\cV_0}{G P_2} \\
H &= g^{\mu\nu} p_\mu p_\nu \approx \f{\cV_0}{G L_s^2}
\ee
For the initial time $\tau$ we necessarily have a time dependent phase space function, because of the fact that it cannot commute with the Hamiltonian it must be expressed as an evolving constant of motion. Inverting these relations provides the solutions for the equations of motion. From them we infer that the $\bms_3$ and null fields can be used all along the trajectory, from $\tau=-\infty$ to $\tau= + \infty$, while there is a problem with the scale factors. It turns out that $V_1$ is always positive, while $V_2$ changes its sign at the point representing the horizon. This means that for negative $V_2$, it is not possible to invert \eqref{new_fields_v} for real $a$, meaning that the scale factor parametrization is valid only in the interior of the black hole. However it is also possible that for another clock the trajectory happens to be prolonged outside the chart spanned by $(V_1,V_2) \in \R^+ \times \R$.

\vbox{
\begin{flushleft}
\small{(a)}\begin{minipage}{0.30\textwidth}
		\includegraphics[width=\textwidth]{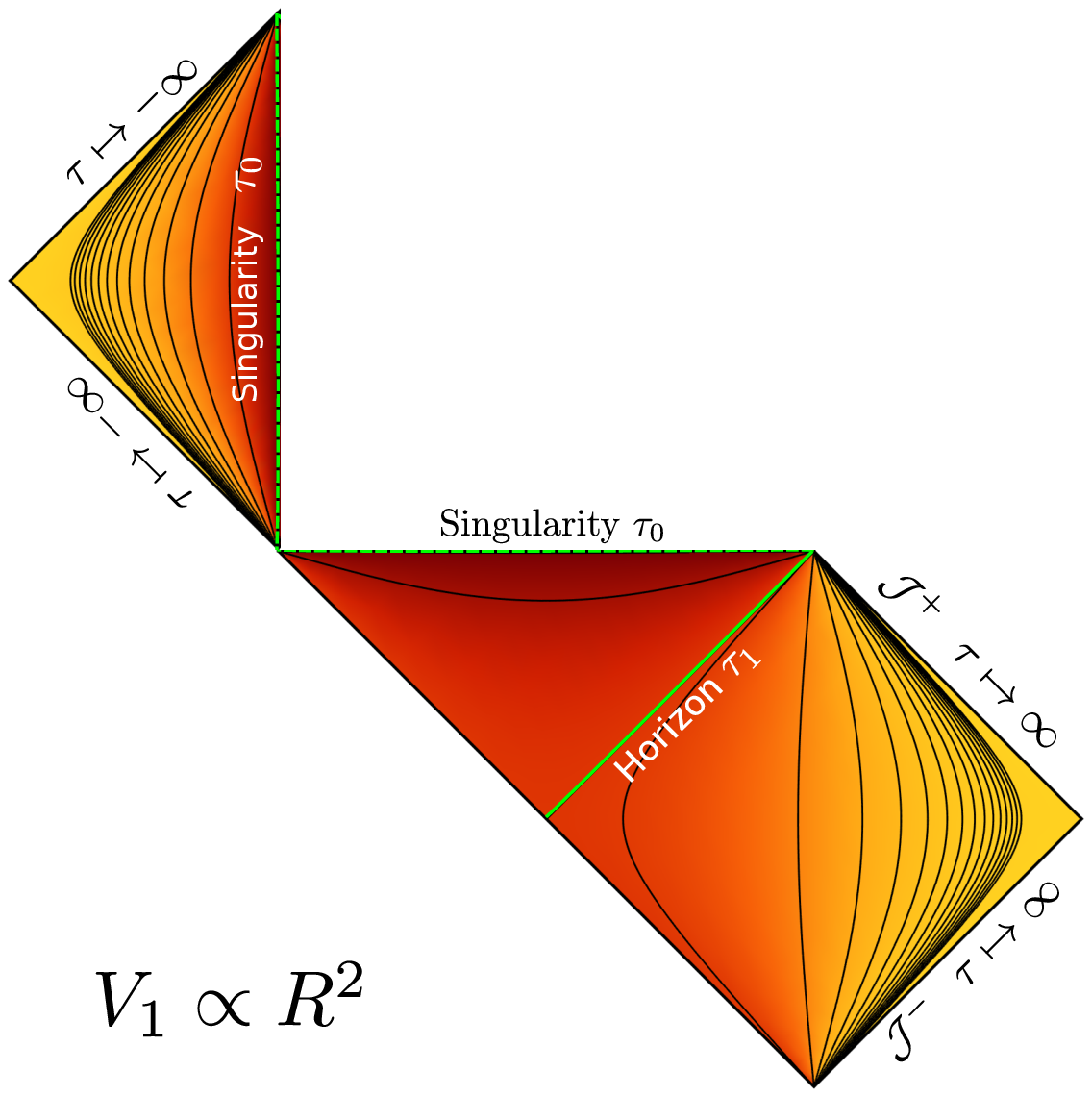}
\end{minipage}
~\hspace{2cm}
\small{(b)}
\begin{minipage}{0.30\textwidth}
		\includegraphics[width=\textwidth]{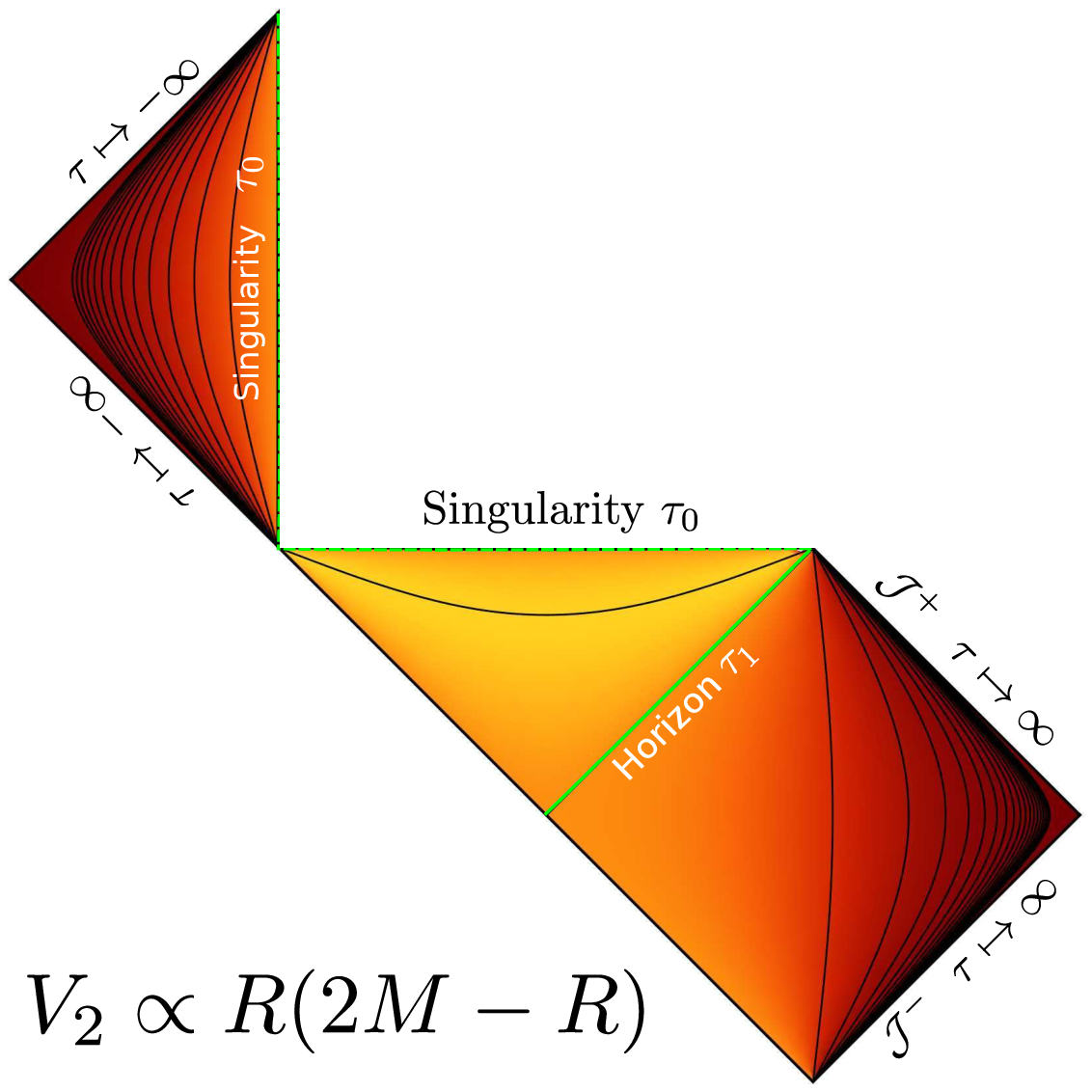}
\end{minipage}\\
\begin{minipage}{\textwidth}
\captionof{figure}{\small{In these figures we have drawn the levels of the fuctions $V_1,V_2$ in terms of the Schwarzschild radius $R$ on a Penrose diagram, completing the ones represented in figure \eqref{fig:kruskal_level}. The black lines are hypersurfaces that equally space the levels of the respective functions, the lighter background color represent the higher values of the functions and it gets darker as the function decreases, we recall that $R \propto (\tau-\tau_0)$ and the Penrose diagram is extented behind the singularity if we add the patch with negative coordinate radius. The first one starting from the left (a) represents $V_1(\tau)$ as in \eqref{classic_traj}. In (b) we have $V_2$, both are well-defined for any real $\tau$.} \label{fig:kruskal_level_2}}
\end{minipage}
\end{flushleft}}

Finally the symmetry generators have the expressions
\be
\begin{array}{ll}
\hline
 & \\[\dimexpr-\normalbaselineskip+4pt]
       \multirow{3}{*}{$\cL_n$} 
&  \multicolumn{1}{l}{p_u p_v \tau^{1 + n} +  \df{\tau^n}{2} (p_u u+ p_v v) (1 + n) +  \df{\tau^{n-1}}{2} u v n (1 + n)} \\[9pt]
& \multicolumn{1}{l}{\df{\tau^{n-1} (2 b p_a p_b \tau^2 G^2 + 
   32 a^3 b^4 \cV_0^2 n (1 + n) + a \tau G (-p_a^2 \tau G + 8 b^3 p_b \cV_0 (1 + n))}{16 a b^2 \cV_0 G}} \\[9pt]
& \multicolumn{1}{l}{\df{\tau^{n-1} (2 P_1 \tau V_1 G (\cV_0 + 
      P_2 \tau G + \cV_0 n) + 
   V_2 (P_2^2 \tau^2 G^2 + 
      2 P_2 \tau \cV_0 G (1 + n) + 
      2 \cV_0^2 n (1 + n))}{2 \cV_0 G}} \\[9pt]\hline
       & \\[\dimexpr-\normalbaselineskip+4pt]
      \multirow{3}{*}{$\cT^+_n$}  
& \multicolumn{1}{l}{\df{1}{4} \tau^{n-1} (2 p_u^2 \tau^2 + 2 p_u \tau v (1 + n) +    v^2 n (1 + n))} \\[9pt]
& \multicolumn{1}{l}{\df{1}{4} \tau^{n-1} \left (\df{p_a^2 \tau^2 G}{
   8 a^2 b^2 \cV_0} + \df{pa \tau (1 + n)}{a} + \df{4 b^2 \cV_0 n (1 + n)}{G}\right )} \\[9pt]
& \multicolumn{1}{l}{\df{1}{2} \tau^{n-1} V_1 \left (\df{P_2^2 \tau^2 G}{\cV_0 }+ 
   2 P_2 \tau (1 + n) + \df{2 \cV_0 n (1 + n)}{G}\right )} \\[9pt]\hline
    \end{array}\notag
\ee
\be
\begin{array}{ll}
& \phantom{\df{\tau^{n-1} (2 P_1 \tau V_1 G (\cV_0 + 
      P_2 \tau G + \cV_0 n) + 
   V_2 (P_2^2 \tau^2 G^2 + 
      2 P_2 \tau \cV_0 G (1 + n) + 
      2 \cV_0^2 n (1 + n))}{2 \cV_0 G}} \\[9pt]\hline
   & \\[\dimexpr-\normalbaselineskip+4pt]
       \multirow{3}{*}{$\cS^+_n$} 
& \multicolumn{1}{l}{\df{1}{2} \tau^{n-\f 1 2} (2p_u \tau + v + 2 v n)} \\[9pt]
& \multicolumn{1}{l}{\df{\tau^{n- \f 1 2} (p_a \tau G+ 
   4 a b^2 \cV_0 (1 + 2 n))}{4 a b \sqrt{\cV_0 G}}} \\[9pt]
& \multicolumn{1}{l}{\tau^{n-\f 1 2} \sqrt{\df{V_1}{\cV_0 G}} \left (\cV_0 + P_2 \tau G +  2 \cV_0 n\right )} \\[9pt]\hline
 & \\[\dimexpr-\normalbaselineskip+4pt]
	\multirow{3}{*}{$\cS^-_n$} 
& \multicolumn{1}{l}{\df{1}{2} \tau^{n-\f 1 2} (2p_v \tau + u + 2 u n)} \\[9pt]
& \multicolumn{1}{l}{\df{\tau^{n- \f 1 2} (-a p_a \tau G+ 
   2 b pb \tau G+  4 a^2 b^2 \cV_0 (1 + 2 n))}{4 b \sqrt{\cV_0 G}}} \\[9pt]
& \multicolumn{1}{l}{\tau^{n-\f 1 2}  \df{\tau (2 P_1 V_1 + P_2 V_2) G + 
   \cV_0 (V_2 + 2 V_2 n)}{2 \sqrt{\cV_0 V_1 G}}} \\[9pt]\hline
    & \\[\dimexpr-\normalbaselineskip+4pt]
      \multirow{3}{*}{$\cD_n$}  
& \multicolumn{1}{l}{\tau^n (p_u u - p_v v)} \\[9pt] 
& \multicolumn{1}{l}{\tau^n (p_a a - p_b b)} \\[9pt]
& \multicolumn{1}{l}{-2 \tau^n P_1 V_1} \\[9pt]\hline
        \end{array}
\ee
These can be extended to an infinite tower of integrable generators by relaxing the conditions on $n$, but on the other hand they do not provide conserved charges. The $\iso(2,1)$ subgroup is recast in the form
\bsub
\be
&J^{(\tau)}= \f{G}{4 \lambda \cV_0} \cL_1 + \f{\cV_0 \lambda}{G} \cL_{-1} \,,&  
&K_x^{(\tau)}=  \f{G}{4 \lambda \cV_0} \cL_1 - \f{\cV_0 \lambda}{G} \cL_{-1} \,,& 
&K_y^{(\tau)} = -\cL_0\,, \\
&\Pi_0^{(\tau)} =  \f{G}{4 \lambda \cV_0} \cT_1^+ + \f{\cV_0 \lambda}{G} \cT_{-1}^+\,,& 
&\Pi_y^{(\tau)}=  \f{G}{4 \lambda \cV_0} \cT_1^+ - \f{\cV_0 \lambda}{G} \cT_{-1}^+\,,&
&\Pi_x^{(\tau)}= \cT^+_0\,.
\ee	
\esub
This the evolving version of the kinematical algebra used in the quantization \eqref{Poincare_alg_gen}, they of course do form the same algebra and e.g. $J$ in\eqref{Poincare_alg_lorenz} is $J = J^{(\tau)} (\tau =0)$. The evolving BMS$_3$ generators are linked to the first integrals via (this is the same as evaluating them \textit{on-shell})
\be
\cL_n &= \cB \tau^{n-1}(n+1)\left (n(\tau-\tau_0)-\tau\right ) + \f{H}{2} \tau^{n-1} \left ( n (\tau^2-\tau_0^2)-n^2 (\tau-\tau_0)^2 -2 \tau \tau_0 \right ) \,,\\
\cT_n^+ &= \f{\cA \tau^{n-1}G}{2 \cV_0} \left (n^2 (\tau-\tau_0)^2- n (\tau^2-\tau_0^2) +2 \tau \tau_0 \right ) \,, \q\q \forall n \in \Z\,.
\ee
These simplify to constants if and only if $n\in \{\pm 1 , 0\}$
\be
\cL_n &= -(n+1)\tau_0^n \cB - \tau_0^{n+1} H \tag{\ref{gen&integral_1}}\,,\\
\cT_n^+ &= \f{\cA G}{\cV_0} \tau_0^{n+1} \tag{\ref{gen&integral_2}}\,, \q\q  n \in \{0,\pm 1\}\,.
\ee

\newpage
~
\thispagestyle{empty}
\pagestyle{Bibliography}
\bibliography{./biblio/Biblio,./biblio/mypaper,./biblio/other_field,./biblio/gravity_general,./biblio/minisuper}
\bibliographystyle{bib-style}

\end{document}